%% file: thesis.tex
\documentclass{ociamthesis}
\usepackage[utf8]{inputenc}

\usepackage{algorithm}
\usepackage{algorithmic}
\usepackage{amsfonts}
\usepackage{amsmath}
\usepackage{amsthm}
\usepackage{amssymb}
\usepackage{braket}
\usepackage{comment}
\usepackage{dsfont}
\usepackage{enumitem}
\usepackage{framed} 
\usepackage{hyperref}
\usepackage{ifthen}
\usepackage{stmaryrd}
\usepackage{tocvsec2}
\usepackage{tikz}
\usepackage{url}
\usepackage{verbatim}
\usepackage{xspace}

\usepackage{mathpartir}
\usepackage{mathsdefs}
\usepackage{tikz-cd}
\usepackage{tikzfig}
\usepackage{tikzfigures}
\usepackage{tikzcdiag}  

    {\begin{framed}\begin{sloppypar}\noindent}%
    {\end{sloppypar}\end{framed}}

\def\clap#1{\hbox to 0pt{\hss\strut #1\hss}}

\title{Reasoning with $!$-Graphs}
\author{Alexander Merry}
\college{Magdalen College}
\degree{Doctor of Philosophy}
\degreedate{Trinity 2013}

\theoremstyle{plain} 
\newtheorem{theorem}{Theorem}[section]
\newtheorem{proposition}[theorem]{Proposition}
\newtheorem{lemma}[theorem]{Lemma}
\newtheorem{corollary}[theorem]{Corollary}

\theoremstyle{definition}
\newtheorem{definition}[theorem]{Definition}
\newtheorem{definitions}[theorem]{Definitions}

\newtheorem{example}[theorem]{Example}

\newtheorem{remark}[theorem]{Remark}

\begin{document}
\baselineskip=18pt plus1pt

\maketitle

\begin{acknowledgements}
    Firstly, I would like to thank my supervisors, Professor Samson Abramsky
    and Professor Bob Coecke, for giving me the opportunity to do this
    research and supporting me in my efforts, as well as for starting the
    field that was the motivation for this work.  I also owe an immense debt
    of gratitude to Dr. Aleks Kissinger, who both paved the way for my work
    and gave me invaluable advice, encouragement and feedback on the way.

    I wish to thank Dr. Lucas Dixon and Dr. Ross Duncan for the work they put
    into Quantomatic, both the theory and the implementation, without which
    this work could have taken much longer.  Additionally, I owe thanks to
    Matvey Soloviev for his bright ideas in the DCM paper, which really opened
    up the possibility of reasoning about $!$-graphs using graph rewriting.

    Thanks also to the administrative staff of the Department of Computer
    Science for helping me navigate university bureaucracy, and especially to
    Julie Sheppard and Janet Sadler.

    Finally, a huge thank you to my family and friends, who have supported and
    encouraged me over the last four years, through the hard times and the
    good.  I couldn't have done it without you.
\end{acknowledgements}

\begin{abstract}
    The aim of this thesis is to present an extension to the string graphs of
    Dixon, Duncan and Kissinger that allows the finite representation of
    certain infinite families of graphs and graph rewrite rules, and to
    demonstrate that a logic can be built on this to allow the formalisation
    of inductive proofs in the string diagrams of compact closed and traced
    symmetric monoidal categories.

    String diagrams provide an intuitive method for reasoning about monoidal
    categories.  However, this does not negate the ability for those using
    them to make mistakes in proofs.  To this end, there is a project
    (Quantomatic) to build a proof assistant for string diagrams, at least for
    those based on categories with a notion of trace.  The development of
    string \textit{graphs} has provided a combinatorial formalisation of
    string diagrams, laying the foundations for this project.

    The prevalence of commutative Frobenius algebras (CFAs) in quantum
    information theory, a major application area of these diagrams, has led to
    the use of variable-arity nodes as a shorthand for normalised networks of
    Frobenius algebra morphisms, so-called ``spider notation''.  This notation
    greatly eases reasoning with CFAs, but string graphs are inadequate to
    properly encode this reasoning.

    This dissertation firstly extends string graphs to allow for
    variable-arity nodes to be represented at all, and then introduces
    $!$-box notation -- and structures to encode it -- to represent string
    graph equations containing repeated subgraphs, where the number of
    repetitions is abitrary.  This can be used to represent, for example, the
    ``spider law'' of CFAs, allowing two spiders to be merged, as well as the
    much more complex generalised bialgebra law that can arise from two
    interacting CFAs.

    This work then demonstrates how we can reason directly about $!$-graphs,
    viewed as (typically infinite) families of string graphs.  Of particular note
    is the presentation of a form of graph-based induction, allowing the
    formal encoding of proofs that previously could only be represented as a
    mix of string diagrams and explanatory text.
\end{abstract}

\begin{romanpages}
\tableofcontents
\end{romanpages}

\chapter{Introduction}
\label{ch:introduction}

In the early 1950s, as a research student attending lectures on differential
geometry, Roger Penrose developed a pictoral notation for dealing with
tensors.  He started using this in publications some twenty years
later, most notably in \cite{Penrose1971}.  His motivation for using this
notation was as an exposition and working aid when dealing with the tensor
index notation common in various areas of physics.  \cite{Penrose1971}, in
which Penrose generalised and formalised tensor index notation as
\textit{abtract tensors}, is liberally illustrated with these diagrams.  In a
tensor expression like
\[ A^{ip}_n B^k_{ik} C^l_p D^{nq}_m \]
where we can consider the large letters to be morphisms with subscripts as
inputs and superscripts as outputs, we have to keep careful track of which
indices are repeated (and therefore should be \textit{contracted} -- the
notion of composition for abstract tensors).  This is made clear in the
following diagrammatic representation
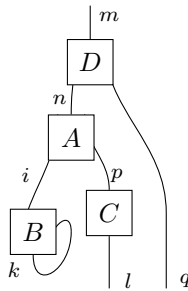
\[ \input{tensor-ex.tikz} \]

These \textit{string diagrams} (or \textit{wire diagrams}) are particularly
good at capturing information with a planar (or mostly planar), rather
than linear, structure.  Thus it was quite natural that they came to be used
for monoidal categories, which have an associative, unital bifunctor
$\bigotimes : \mathcal{C} \times \mathcal{C} \rightarrow \mathcal{C}$ that
acts as a form of composition orthogonal to the categorial composition.

Note that, unlike in the traditional diagrams of category theory, string
diagrams represent morphisms as nodes and objects as wires.  This allows
categorical composition to be represented by joining wires, as in
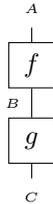
\[ \input{f-then-g.tikz} \]
while the monoidal composition can be represented via juxtaposition
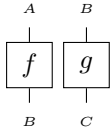
\[ \input{f-tensor-g.tikz} \]

Twenty years after Penrose's publication, Joyal and Street formalised the link
between string diagrams and monoidal categories, at least in
part\cite{Joyal1991}.  In particular, they showed that \textit{progressive
plane diagrams}, a generalised topological graph with directed wires that do
not cross or have loops, exactly capture the axioms of monoidal categories.
Additionally, they showed that allowing crossings captures the axioms of
symmetric monoidal categories, which have a natural isomorphism witnessing
that $A \otimes B \cong B \otimes A$ for all objects $A$ and $B$, and that
other variations on these diagrams capture braided and balanced monoidal
categories.

Joyal and Street planned a follow-up paper to treat non-progressive diagrams,
such as those with loops, but this was never published.  The link between
these diagrams and either \textit{traced symmetric monoidal categories} or
\textit{compact closed categories} (depending on the restrictions placed on
the diagrams) was shown in \cite{KissingerDPhil}.

We concentrate on traced symmetric monoidal categories and compact closed
categories in this thesis.  Selinger's survey paper\cite{Selinger2009} gives a
more comprehensive overview of graphical languages for various types of
monoidal categories.

These diagrammatic languages found a home in the study of quantum computer
science, when in 2004 Samson Abramsky and Bob Coecke recast von Neumann's
axiomatisation of quantum mechanics into category theory, allowing the use of
string diagrams.  The diagrams proved to be particularly helpful in
illuminating the role of entanglement in various quantum mechanical and
quantum informational protocols, and provided a grounding for intuition in
something familiar (the physical manipulation of a diagram) for something
that, at many times, appears to behave counter-intuitively (quantum
mechanics).

In this setting, the wires of a string diagram represent quantum mechanical
systems, and the nodes are operations on those systems.  Of particular
relevance to this thesis is the importance of commutative Frobenius algebras
in these languages.

A Frobenius algebra in a monoidal category $\mathcal{C}$ consists of a monoid
$(A,\mu,\eta)$ and a comonoid $(A,\delta,\epsilon)$ on an object of
$\mathcal{C}$ -- the monoid being an associative multiplication operation $\mu
: A \otimes A \rightarrow A$ with unit $\eta : I \rightarrow A$, and the
comonoid being its dual $\delta : A \rightarrow A \otimes A$ with counit
$\epsilon : A \rightarrow I$ -- that interact ``nicely''.  Specifically, they
obey the Frobenius law:
\[
    (1_A \otimes \mu) \circ (\delta \otimes 1_A)
    = \delta \circ \mu
    = (\mu \otimes 1_A) \circ (1_A \otimes \delta)
\]
If we represent the monoid and comonoid graphically as $(A,\mult,\unit)$ and
$(A,\comult,\counit)$, then this law can be drawn
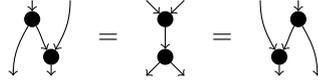
\[ \input{frob-law.tikz} \]

Commutative Frobenius algebras (CFAs) are those whose monoid is commutative
and comonoid is cocommutative.  They have the property that the value of any
connected (in the graphical language) network of components of a given
CFA is determined purely by the number of inputs, outputs and loops it has.
This gives rise to a consise and useful representation of a CFA, the
``spider'':
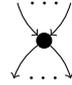
\[ \input{spider.tikz} \]
Loops are encoded in this representation as self-loops (by connecting outputs
to inputs).

This leads to the \textit{spider law}, which allows us to merge two connected
spiders of the same type
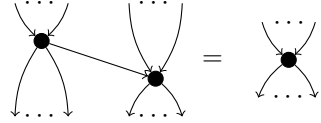
\[ \input{spider-law-informal.tikz} \]
and this law significantly simplifies working with these graphical languages.

While the use of string diagrams can help greatly with understanding and
intuition, constructing large proofs by hand can still be tedious and
error-prone.  This is where projects like Quantomatic\cite{quanto}, a
collection of automated tools for graphical languages, come in.

The goal of Quantomatic is to produce a useful proof assistant for graphical
languages; as its name suggests, it was originally aimed at the languages
being developed for quantum computation.  For example, Dixon and Duncan used
it to verify the correctness of the Steane code in \cite{Duncan2013}.  It is
still primarily used for quantum languages, but it aims to be useful in other
fields as well.

Dixon, Duncan and Kissinger provided a theoretical underpinning for
Quantomatic's automated reasoning in the form of \textit{string graphs},
introduced as \textit{open graphs} in \cite{Dixon2010} and further refined by
Dixon and Kissinger in \cite{Dixon2010a} and then by Kissinger in
\cite{KissingerDPhil}.  These represent string diagrams as typed graphs, where
the vertex types are sorted into \textit{node-vertices} and
\textit{wire-vertices}, the former being the ``real'' vertices from the string
diagram and the latter a kind of ``dummy'' vertex holding the wires in place
(as well as carrying the wire type).  For example, the string diagram
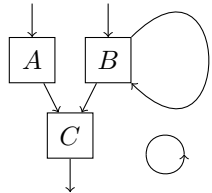
\[ \input{string-diag-ex.tikz} \]
could be represented by the graph
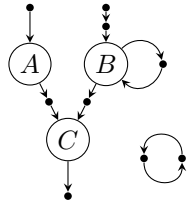
\[ \input{string-graph-ex.tikz} \]
This allows us to do equational reasoning for diagrams using graph rewriting,
in a similar manner to how traditional proof assistants use term rewriting.

Note that string graphs are not the only combinatorial representation for the
categorical structures that underlie string diagrams.  In particular,
Hasegawa, Hofmann and Plotkin use a bijective map to describe the connections
formed by wires in their paper showing that finite dimensional vector spaces
are complete for traced symmetric monoidal categories\cite{Hasegawa2008}.
Using graphs, however, allows us to build on a substantial body of established
work on graph rewriting.

The first aim of this dissertation is to extend this formalism to allow tools
like Quantomatic to work with the spider law and other, more complex, rules
using the \textit{$!$-graphs} posited in \cite{Dixon2009} as graph patterns.
In particular, we aim to give a mathematical underpinning to these, which were
previously only treated informally.

$!$-graphs are string graphs with demarked subgraphs called \textit{$!$-boxes}
that can be repeated any number of times, allowing an infinite family of
graphs to be represented as a single graph:
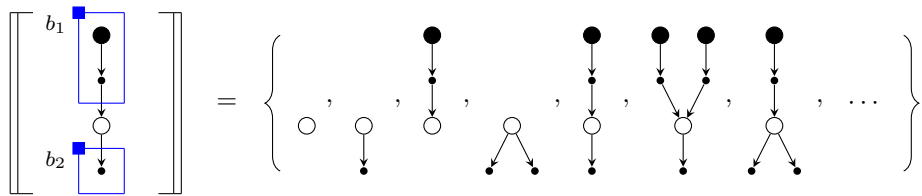
\[ \input{bang-graph-ex.tikz} \]
This can be extended to string graph equations, allowing the spider law to
be represented as a $!$-graph equation:
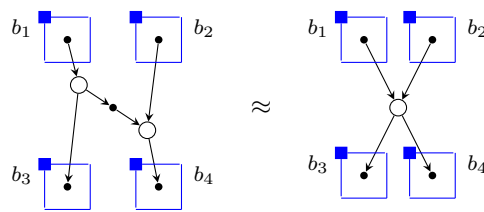
\[ \input{spider-law-graph.tikz} \]
We call a string graph or string graph equation represented by a $!$-graph
or $!$-graph equation an \textit{instance} of that graph or rule.

This notation may seem overpowered for simply representing the spider law, but
it also allows many equations describing interactions between certain kinds of
CFA, such as
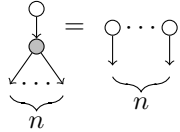
\[ \input{x-copies-z-spider.tikz} \]
which can be represented
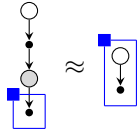
\[ \input{x-copies-z-spider-graph.tikz} \]
and even more complex examples that we will describe in later chapters.

The second aim of this dissertation is to provide the beginnings of a formal
system for reasoning \emph{about} (and not just \emph{with}) these $!$-graphs.
We show that rewrite rules composed of $!$-graphs can be used to rewrite
$!$-graphs and, in doing so, perform equational reasoning on these infinite
families of string graphs.  What is more, we introduce some additional rules
to this logic of $!$-graphs, including a form of graphical induction that
encodes inductive proofs that could previously only be done with a mix of
diagrams and explanatory text.

Chapter \ref{ch:diag-reas} describes various types of monoidal category and
the graphical languages they induce.  It also provides some motivating
examples in the form of a subset of the Z/X calculus for quantum information
processing.

Chapter \ref{ch:rewriting} provides a brief introduction to rewriting as a way
of doing equational reasoning, using the well-established field of term
rewriting.  It then presents an equational logic for traced symmetric monoidal
categories and compact closed categories using string graphs, and shows how
this can be implemented using string graph rewriting.  We present a slightly
different construction of string graphs to \cite{KissingerDPhil}, as
Kissinger's construction does not allow for variable-arity nodes, required by
the spiders of the quantum graphical languages.

Chapter \ref{ch:bang-graphs} describes the $!$-box notation and its encoding
in $!$-graphs.  It provides a set of operations on $!$-boxes ($\COPY$, $\DROP$
and $\KILL$) that are used to define the concrete instances of a $!$-graph
(ie: those string graphs it represents).

Chapter \ref{ch:rw-bang-graphs} introduces $!$-graph equations to represent
families of string graph equations, and demonstrates how their directed form,
$!$-graph rewrite rules, can be used to rewrite not only string graphs but
also $!$-graphs.  The latter produces further $!$-graph equations that are
sound with respect to their interpretation as families of string graph
equations.  This effectively allows us to rewrite a potentially infinite
family of string graphs simultaneously.

Chapter \ref{ch:rules} builds on this to construct a nascent logic for
reasoning about $!$-graphs.  In particular, we construct an analogue of
induction for string graphs which functions as a $!$-box introduction rule,
and demonstrate how this can be used to derive
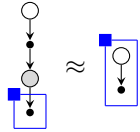
\[ \input{x-copies-z-spider-graph.tikz} \]
in the Z/X calculus.  We also demonstrate a derivation of the generalised
bialgebra law, which we introduce at the end of chapter \ref{ch:diag-reas}.

Chapter \ref{ch:implementation} demonstrates an algorithm for finding
instances of $!$-graphs that match string graphs or other $!$-graphs, allowing
for practical implementations of rewriting with $!$-graph rewrite rules.  The
complete matching process implemented in Quantomatic for rewriting string
graphs with $!$-graph rewrite rules is described, as well as the planned
extension to it that will allow $!$-graphs to be rewritten.

Finally, chapter \ref{ch:conclusions} presents our conclusions and areas of
further work.

\chapter{Diagrammatic Reasoning}
\label{ch:diag-reas}

This chapter provides a brief summary from the existing literature of
various types of monoidal categories that are relevant to this thesis,
and the diagrammatic representations of their morphisms.  It also
demonstrates their utility in reasoning about such categories using
examples drawn from existing work in the area of graphical languages for
quantum computation.

\section{Monoidal Categories} 
\label{sec:monoidal-cats}

Monoidal categories provide a useful setting for reasoning about
processes of all kinds.  In addition to the usual functional
composition, they allow a form of parallel composition via the tensor
product.

\begin{definition}[Monoidal Category]
    A \textit{monoidal category} is a category $\mathcal{M}$ equipped
    with
    \begin{itemize}
        \item a bifunctor $\bigotimes : \mathcal{M} \times \mathcal{M}
            \rightarrow \mathcal{M}$, called the \textit{tensor
            product};
        \item a distinguished object $I$ of $\mathcal{M}$, called the
            \textit{tensor unit}; and
        \item three natural isomorphisms
            \begin{align*}
                \alpha_{A,B,C} &: (A \otimes B) \otimes C
                \cong A \otimes (B \otimes C) \\
                \lambda_A &: I \otimes A \cong A \\
                \rho_A &: A \otimes I \cong A
            \end{align*}
            where $\lambda_I = \rho_I$ and the following diagrams
            commute:
            \[\begin{tikzcd}
                ((A \otimes B) \otimes C) \otimes D
                    \arrow{r}{\alpha}
                    \arrow[swap]{d}{\alpha \otimes 1}
                & (A \otimes B) \otimes (C \otimes D)
                    \arrow{r}{\alpha}
                & A \otimes (B \otimes (C \otimes D))
                \\
                (A \otimes (B \otimes C)) \otimes D
                    \arrow{rr}{\alpha}
                && A \otimes ((B \otimes C) \otimes D)
                    \arrow[swap]{u}{1 \otimes \alpha}
            \end{tikzcd}\]
            \[\begin{tikzcd}
                (A \otimes I) \otimes B
                    \arrow{rr}{\alpha}
                    \arrow[swap]{dr}{\rho \otimes 1}
                && A \otimes (I \otimes B)
                    \arrow{dl}{1 \otimes \lambda}
                \\
                & A \otimes B &
            \end{tikzcd}\]
    \end{itemize}
    \label{def:mon-cat}
\end{definition}

\begin{example}
    The category \catSet of sets is monoidal, with the cartesian product
    as $\bigotimes$ and the single-element set as $I$.  Alternatively,
    the disjoint union of sets and the empty set can be used as the
    monoidal structure.

    The category \catVectK of vector spaces over a field $K$ and linear
    functions between them is a monoidal category, where $\bigotimes$ is
    the tensor product of vector spaces, and $I$ is $K$ as a
    one-dimensional vector space.

    In later examples, we will use the fact that if $U$, $V$ and $W$ are
    vector spaces over $K$, then each bilinear map $f : U \times V
    \rightarrow W$ induces a unique linear map $\tilde{f} : U \otimes V
    \rightarrow W$.  We can use this to construct the natural
    isomorphisms $\alpha_{U,V,W}$, $\lambda_V$ and $\rho_V$ from the
    maps
    \begin{mathpar}
        ((u,v),w) \mapsto u \otimes (v \otimes w)
        \and
        (k,v) \mapsto kv
        \and
        (v,k) \mapsto kv
    \end{mathpar}
    \label{ex:vect-mon}
\end{example}

When considering (categorical) diagrams of a monoidal category
$\mathcal{M}$ involving the above isomorphisms, Mac Lane's coherence
theorem\cite{CWM} allows us to work instead in a category where those
isomorphisms are identities; we call this category the
\textit{strictification} of $\mathcal{M}$.  Formal categorical diagrams
(of the sort seen in definition \ref{def:mon-cat}) commute in
$\mathcal{M}$ if and only if they commute in its strictification.
Therefore, we abuse notation slightly by omitting paretheses and,
insofar as possible, $I$, writing
\[ A \otimes B \otimes C \]
instead of
\[ (A \otimes I) \otimes (B \otimes C) \]
for example.

In \cite{Penrose1971}, Penrose presented a graphical language he had
developed for representing tensor networks.  Restricted variants of this
language can be used for talking about more general monoidal categories.

In this notation, we represent objects of the category (and identity
morphisms on those objects) as lines, and morphisms as boxes, also
called nodes.  For example, the morphism $f : A \rightarrow B$ could be
depicted
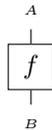
\[ \input{f.tikz} \]
where we read the diagram in a downwards direction.  The tensor product
is depicted by placing things side-by-side, so if $g : B \rightarrow C$
then $f \otimes g : A \otimes B \rightarrow B \otimes C$ is
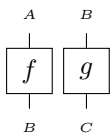
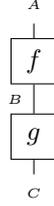
\[ \input{f-tensor-g.tikz} \]
and $g \circ f : A \rightarrow C$ is
\[ \input{f-then-g.tikz} \]

One convenience that this immediately provides, other than requiring
less mental effort to parse than conventional symbolic notation, is that
the equation
\[ (g_1 \otimes g_2) \circ (f_1 \otimes f_2) = (g_1 \circ f_1) \otimes
(g_2 \circ f_2) \]
-- a consequence of the fact that $\bigotimes$ is a bifunctor -- is
implicit in the graphical notation; both sides of the equation are, in
fact, represented by the same diagram
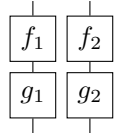
\[ \input{bifunc.tikz} \]

The tensor product has a rigid order; it is common to allow more
flexibility, in a manner akin to rerouting information.  This is
provided by a symmetric monoidal category.

\begin{definition}[Symmetric Monoidal Category]
    A \textit{symmetric monoidal category}, or SMC, is a monoidal
    category $\mathcal{M}$ equipped with a natural isomorphism
    \[ \gamma_{A,B} : A \otimes B \cong B \otimes A \]
    such that $\gamma_{A,B} \circ \gamma_{B,A} = 1$, $\rho_A = \lambda_A
    \circ \gamma_{A,I}$ and the following diagram commutes:
    \[\begin{tikzcd}
        (A \otimes B) \otimes C
            \arrow{r}{\alpha}
            \arrow{d}{\gamma \otimes 1}
        & A \otimes (B \otimes C)
            \arrow{r}{\gamma}
        & (B \otimes C) \otimes A
            \arrow{d}{\alpha}
        \\
        (B \otimes A) \otimes C
            \arrow{r}{\alpha}
        & B \otimes (A \otimes C)
            \arrow{r}{1 \otimes \gamma}
        & B \otimes (C \otimes A)
    \end{tikzcd}\]
\end{definition}

\begin{example}
    \catSet with either cartesian products or disjoint unions is
    symmetric, as is \catVectK.  In the latter case, $\gamma$ can be
    defined (see example \ref{ex:vect-mon}) using the map
    \[ (v,w) \mapsto (w \otimes v) \]
\end{example}

Diagrammatically, we represent $\gamma_{A,B}$ as
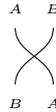
\[ \input{swap-no-arrows.tikz} \]
The fact that $\gamma$ is a natural isomorphism is represented in this
notation as an ability to ``slide'' morphisms up and down wires:
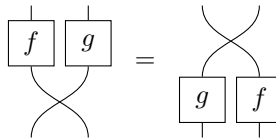
\[ \input{swap-nat-tens.tikz} \]

Joyal and Street presented\cite{Joyal1991} a formalisation of Penrose's
diagrams based on topological graphs.  They showed that
\textit{progressive polarised diagrams}, essentially those without
loops, can be used to form free symmetric monoidal categories;
effectively, they exactly capture the axioms of a symmetric monoidal
category.

There are plenty of scenarios where we do want to introduce loops,
whether to represent the trace (or partial trace) operation on a matrix,
looping in a computation or even the effects of entanglement in a
quantum mechanical system, where information can sometimes appear to
flow backwards in time.

A \textit{trace operation}, if one exists, can be used to capture a
looping construction.

\begin{definition}[Traced Symmetric Monoidal Category]
    A \textit{traced symmetric monoidal category} is a symmetric
    monoidal category $\mathcal{M}$ together with a \textit{trace
    operation}: a function
    \[ \Tr^X_{A,B} : \hom_{\mathcal{M}}(A \otimes X,B \otimes X)
        \rightarrow \hom_{\mathcal{M}}(A,B) \]
    for all objects $X,A,B$ that satisfies the following conditions:
    \begin{enumerate}
        \item \label{tr-cond-contract} $\Tr^X((g \otimes 1_X) \circ f
            \circ (h \otimes 1_X)) = g \circ \Tr^X(f) \circ h$
        \item \label{tr-cond-slide} $\Tr^Y(f \circ (1_A \otimes g)) =
            \Tr^X((1_B \otimes g) \circ f)$
        \item \label{tr-cond-I} $\Tr^I(f) = f$
        \item \label{tr-cond-obj-tensor} $\Tr^{X \otimes Y}(f) = \Tr^X(\Tr^Y(f))$
        \item \label{tr-cond-fn-tensor} $\Tr^X(g \otimes f) = g \otimes \Tr^X(f)$
        \item \label{tr-cond-yank} $\Tr^X(\gamma_{X,X}) = 1_X$
    \end{enumerate}
\end{definition}

\begin{example}
    \catVectK is a traced symmetric monoidal category, with the partial
    trace of linear maps as the trace operation.
\end{example}

$\Tr^X_{A,B}(f)$ is represented as
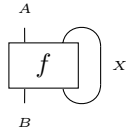
\[ \input{Tr-X-A-B.tikz} \]
and, as before, the properties of the trace operation are subsumed into
its representation in the graphical language, either directly (as with
\ref{tr-cond-obj-tensor} or \ref{tr-cond-fn-tensor}, for example, where
both sides of the equation have the same representation) or via
``natural'' graphical manipulations.  For example, condition
\ref{tr-cond-slide} corresponds to sliding a morphism round a trace loop
from one end to the other:
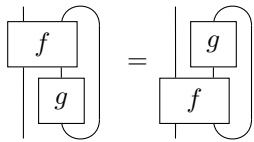
\[ \input{tr-slide.tikz} \]
although care must be taken, as the following diagram has no meaning in
a traced SMC:
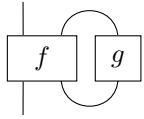
\[ \input{bad-trace.tikz} \]
(but see remark \ref{rem:formal-dual} below).  Condition
\ref{tr-cond-yank} is visually represented as an ability to ``yank''
loops with no (non-identity) morphisms straight:
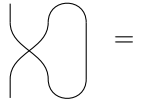
\[ \input{tr-yank.tikz} \]

Finally, we come to \textit{compact closed categories}.  These are based
around \textit{dual objects}, which are the categorisation of dual
spaces from linear algebra.

\begin{definition}[Dual Objects]
    Let $A$ and $A^*$ be objects in a monoidal category.  $A^*$ is
    called a \textit{left dual} of $A$ (and $A$ a \textit{right dual} of
    $A^*$) if there exist morphisms $d_A : A \otimes A^* \rightarrow I$
    and $e_A : I \rightarrow A^* \otimes A$ such that
    \[ (d_A \otimes 1_A) \circ (1_A \otimes e_A) = 1_A \]
    and
    \[ (1_{A^*} \otimes d_A) \circ (e_A \otimes 1_{A^*}) = 1_{A^*} \]
\end{definition}

Note that in a symmetric monoidal category, if $A^*$ is a left dual of
$A$, it is also a right dual of $A$ and vice versa, so we simply say
that $A^*$ is the dual of $A$.

\begin{definition}[Compact Closed Category]
    A \textit{compact closed category} is a symmetric monoidal category
    where every object has a dual.
\end{definition}

We add the notion of dual objects to the graphical language using
arrows.  If an object $A$ is represented by a downward-directed wire,
its dual is represented by an upward-directed one.  The maps $d_A$ and
$e_A$ can then be drawn
\[ \input{dA.tikz} \qquad \textrm{and} \qquad \input{eA.tikz} \]
respectively.  The equations in the above definition can then be drawn
\[ \input{dual-eq-1.tikz} \qquad \textrm{and} \qquad \input{dual-eq-2.tikz} \]

\begin{example}
    Consider a finite-dimensional vector space vector space $V$ and its
    (algebraic) dual space, $V^*$.  Given a basis $\{u_1,\ldots,u_n\}$
    of $V$, there is a corresponding basis $\{u_1^*,\ldots,u_n^*\}$ of
    $V^*$ such that
    \[ u_i^*(a_1 u_1 + \cdots + a_n u_n) = a_i \]
    We can then define the map $d_V : V \otimes V^* \rightarrow K$ to be
    induced (see example \ref{ex:vect-mon}) by
    \[ (v,f) \mapsto f(v) \]
    and $e_V : K \rightarrow V^* \otimes V$ to be
    \[ k \mapsto k(u_1^* \otimes u_1 + \cdots + u_n^* \otimes u_n) \]

    So if we consider $V$ as an object of \catFDVectK, the category of
    finite-dimensional vector spaces over $K$, $V^*$ is a categorical
    dual of $V$.  Thus \catFDVectK is compact closed (although \catVectK
    is not).  Note, however, that our definitions of $d_V$ and $e_V$ are
    not the only possibilities.
    \label{ex:fdvect-ccc}
\end{example}

Given the graphical notation, we would expect to be able to construct a
trace operation using these maps; this is indeed the
case\cite{Joyal1996}.
$\Tr^X_{A,B}(f)$ is
\[ (1_B \otimes d_X) \circ (f \otimes 1_{X^*}) \circ (1_A \otimes e_X)
\]
or, graphically,
\[ \input{compact-closed-trace.tikz} \]

\begin{remark}\label{rem:formal-dual}
    Joyal, Street and Verity's $\textrm{Int}$
    construction\cite{Joyal1996} allows a traced SMC to be embedded in a
    compact closed category with formal dual objects.  This allows the
    diagrams for compact closed categories to be given meaning when
    working with traced SMCs.
\end{remark}

\section{Quantum Computation with Diagrams} 
\label{sec:quantum}

An area which is making active use of these diagrammatic languages is
quantum information processing.  Starting with a paper by Abramsky and
Coecke in 2004\cite{Abramsky2004}, there has been considerable work
in axiomatising quantum mechanics in a categorical setting that admits
diagrammatic reasoning.  Indeed, a major application of the ideas set
out in this thesis is to make this work easier.  We will therefore
present a brief introduction to one of these languages, both as
motivation and as a source of examples.

Most of the languages that have been developed in this area are based on
compact closed categories; in particular, the usual model for them is
the category \catFDHilb of finite dimensional \textit{Hilbert spaces}.

\begin{definition}
    A \textit{Hilbert space} is a real or complex inner product space
    that is a complete metric space with respect to the distance metric
    \[ d(x,y) = \sqrt{\langle x-y,x-y \rangle} \]
\end{definition}

Note that we can ignore the compactness requirement, as it is satisfied
by all finite-dimensional inner product spaces.

\catFDHilb is a compact closed category, with the categorical dual of a
Hilbert space $H$ being the algebraic dual $H^*$, as in \catFDVect
(example \ref{ex:fdvect-ccc}).  As such, we can make use of the
diagrammatic language already introduced for such categories.

A structure that has particular importance in this formulation of quantum
mechanics is the Frobenius algebra, a monoid and comonoid pair that
interact particularly well.

A \textit{monoid} in a monoidal category $\mathcal{C}$ as a triple
$(A,\mu,\eta)$ where $A$ is an object of $\mathcal{C}$, and $\mu : A
\otimes A \rightarrow A$ and $\eta : I \rightarrow A$ are morphisms
satisfying unit and associativity laws:
\[\label{pg:cfa} \mu \circ (\eta \otimes 1_A) \circ \lambda_A^{-1}
= 1_A = \mu \circ (1_A \otimes \eta) \circ \rho_A^{-1} \]
\[ \mu \circ (\mu \otimes 1_A) = \mu \circ (1_A \otimes \mu) \]
Graphically, we can represent the monoid as $(A,\mult,\unit)$; these
laws can then be drawn
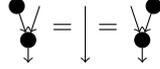
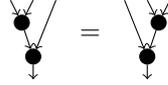
\[ \input{unit-law-diag.tikz} \]
and
\[ \input{assoc-law-diag.tikz} \]

A \textit{comonoid} is the dual of this: a triple $(A,\delta,\epsilon)$
with $\delta : A \rightarrow A \otimes A$ and $\epsilon : A \rightarrow
I$ satisfying counit and coassociativity laws:
\[ \lambda_A \circ (\epsilon \otimes 1_A) \circ \delta
= 1_A = \rho_A \circ (1_A \otimes \epsilon) \circ \delta \]
\[ (\delta \otimes 1_A) \circ \delta = (1_A \otimes \delta) \circ \delta
\]
Graphically, this is $(A,\comult,\counit)$ satisfying
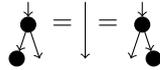
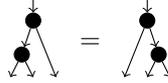
\[ \input{counit-law-diag.tikz} \]
and
\[ \input{coassoc-law-diag.tikz} \]

A monoid and comonoid on the same object form a \textit{Frobenius
algebra} if they satisfy the Frobenius law:
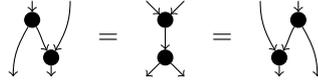
\[ \input{frob-law.tikz} \]
or, symbolically,
\[ (1_A \otimes \mu) \circ (\delta \otimes 1_A) = \delta \circ \mu
= (\mu \otimes 1_A) \circ (1_A \otimes \delta) \]
A Frobenius algebra is \textit{commutative} if the monoid is commutative
and the comonoid is cocommutative:
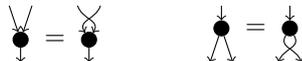
\[ \input{commutative-law-diag.tikz} \qquad \qquad
   \input{cocommutative-law-diag.tikz} \]
\[ \mu = \mu \circ \gamma_{A,A} \qquad \qquad \delta =
    \gamma_{A,A} \circ \delta \]

Commutative Frobenius algebras (CFAs) appear repeatedly in quantum
graphical languages, whether to describe
entanglement\cite{Coecke2009,Coecke2010} or classical
data\cite{Coecke2006}.  Thus the practical utility of these languages is
greatly improved by a consequence of the Frobenius law, namely that any
connected (in the diagrammatic language) network of multiplications,
comultiplications, units and counits of a commutative Frobenius algebra
has a unique normal form consisting of a series of multiplications
followed by a series of loops followed by a series of comultplications:
\[ \input{frob-normal-form.tikz} \]
This allows a more compact representation, generally referred to as a
``spider'':
\[ \input{spider.tikz} \]
where the loops are encoded by connecting outputs of the spider to
inputs.  Note that our original multiplication, comultiplication, unit
and counit are subsumed into this notation.  This leads to the
\textit{spider law}, which allows us to merge two connected spiders of
the same type
\[ \input{spider-law-informal.tikz} \]
The spider law and the spider identity law
\[ \input{spider-id-law.tikz} \]
completely capture the laws of a CFA.

There are various quantum graphical languages built around CFAs,
including the Z/X calculus\cite{Coecke2009}, the trichromatic
calculus\cite{Lang2012} and the GHZ/W calculus\cite{Coecke2010}.  There
are also several languages built on top of the Z/X calculus, such as
Duncan and Perdrix's calculus for measurement-based quantum
computation\cite{Duncan2010}.  The GHZ/W calculus even admits an
encoding of rational arithmetic\cite{Coecke2011}.

We will use a restricted version of the Z/X calculus (without angles or
scalars) to provide motivating examples.  This consists of two
\textit{special} CFAs, called the X and Z CFAs, together with
certain rules governing their interaction.

\begin{definition}
    A CFA $(\mult, \unit, \comult, \counit)$ is \textit{special} if it
    satisfies
    \[ \input{special-law-diag.tikz} \]
    The ``spiderised'' version of this is
    \[ \input{special-spider-law-diag.tikz} \]
\end{definition}

Traditionally, the X CFA has been represented with red nodes, and the
Z CFA with green nodes.  However, we will follow \cite{KissingerDPhil}
in using grey for X and white for Z, if only to spare colourblind
readers some frustration.  Thus our CFAs are
\[(\mu_X, \eta_X, \delta_X, \epsilon_X) = (\greymult, \greyunit,
\greycomult, \greycounit)\]
and
\[(\mu_Z, \eta_Z, \delta_Z, \epsilon_Z) = (\whitemult, \whiteunit,
\whitecomult, \whitecounit) \].

In addition to the laws of a special CFA for each colour of node, we
have the following laws:
\begin{mathpar}
    \input{z-copies-x.tikz}
    \and
    \input{x-copies-z.tikz}
    \and
    \input{bialg.tikz}
    \\
    \input{z-x-mult-comult-cancel.tikz}
    \and
    \input{z-x-duals-coincide.tikz}
\end{mathpar}
We name these laws, for ease of reference, \textit{Z copies X},
\textit{X copies Z}, the \textit{bialgebra law}, \textit{scalar
elimination} and the \textit{dual law}, in order.

For the dual law we might expect certain symmetries:
\[ \input{z-x-duals-symmetries.tikz} \]

Indeed, we can use the laws of a CFA together with the dual law to
derive these.  For example, the last one can be derived as follows:
\[ \input{z-x-duals-derivation.tikz} \]
This fairly simple proof can be written in ``traditional'' mathematical
notation as well, but -- because we have to explicitly use various laws
and axioms that are inherent in the diagrammatic language -- it ends up
being long and tedious.  We present such a proof below to demonstrate
this.  For reference, the dual law written symbolically is
\[ \rho \circ (1_A \otimes (\epsilon_X \circ \mu_X)) \circ
((\delta_Z \circ \eta_Z) \otimes 1_A) \circ \lambda^{-1} =
1_A \]
Then we have
\begin{align*}
    1_A = &\textrm{\{(co)unitality laws\}} \\
          &\mu_X \circ (1_A \otimes \eta_X) \circ \rho^{-1} \circ \lambda
           \circ (\epsilon_X \otimes 1_A) \circ \delta_X \circ \\
          &\mu_Z \circ (1_A \otimes \eta_Z) \circ \rho^{-1} \circ \lambda
           \circ (\epsilon_Z \otimes 1_A) \circ \delta_Z
           \displaybreak[0]\\
        = &\textrm{\{naturality of $\lambda$, $\rho$\}} \\
          &\lambda \circ (1_I \otimes \mu_X) \circ (1_I \otimes 1_A
           \otimes \eta_X) \circ (\epsilon_X \otimes 1_A \otimes 1_I)
           \circ (\delta_X \otimes 1_I) \circ \rho^{-1} \circ \\
          &\lambda \circ (1_I \otimes \mu_Z) \circ (1_I \otimes 1_A
           \otimes \eta_Z) \circ (\epsilon_Z \otimes 1_A \otimes 1_I)
           \circ (\delta_Z \otimes 1_I) \circ \rho^{-1}
           \displaybreak[0]\\
        = &\textrm{\{bifunctoriality of $\otimes$\}} \\
          &\lambda \circ (\epsilon_X \otimes 1_A) \circ (1_A \otimes
           \mu_X) \circ (\delta_X \otimes 1_A) \circ (1_A \otimes
           \eta_X) \circ \rho^{-1} \circ \\
          &\lambda \circ (\epsilon_Z \otimes 1_A) \circ (1_A \otimes
           \mu_Z) \circ (\delta_Z \otimes 1_A) \circ (1_A \otimes
           \eta_Z) \circ \rho^{-1}
           \displaybreak[0]\\
        = &\textrm{\{Frobenius law\}} \\
          &\lambda \circ (\epsilon_X \otimes 1_A) \circ (\mu_X \otimes
           1_A) \circ (1_A \otimes \delta_X) \circ (1_A \otimes
           \eta_X) \circ \rho^{-1} \circ \\
          &\lambda \circ (\epsilon_Z \otimes 1_A) \circ (\mu_Z \otimes
           1_A) \circ (1_A \otimes \delta_Z) \circ (1_A \otimes
           \eta_Z) \circ \rho^{-1}
           \displaybreak[0]\\
        = &\textrm{\{bifunctoriality of $\otimes$\}} \\
          &\lambda \circ ((\epsilon_X \circ \mu_X) \otimes 1_A) \circ
           (1_A \otimes (\delta_X \circ \eta_X)) \circ \rho^{-1} \circ
           \\
          &\lambda \circ ((\epsilon_Z \circ \mu_Z) \otimes 1_A) \circ
           (1_A \otimes (\delta_Z \circ \eta_Z)) \circ \rho^{-1}
           \displaybreak[0]\\
        = &\textrm{\{naturality of $\lambda$, $\rho$\}} \\
          &\lambda \circ (1_I \otimes \lambda) \circ (1_I \otimes
           (\epsilon_X \circ \mu_X) \otimes 1_A) \circ (1_I \otimes 1_A
           \otimes (\delta_X \circ \eta_X)) \circ \\
          &((\epsilon_Z \circ \mu_Z) \otimes 1_A \otimes 1_I) \circ (1_A
           \otimes (\delta_Z \circ \eta_Z) \otimes 1_I) \circ (\rho^{-1}
           \otimes 1_I)\circ \rho^{-1}
           \displaybreak[0]\\
        = &\textrm{\{bifunctoriality of $\otimes$\}} \\
          &\lambda \circ ((\epsilon_Z \circ \mu_Z) \otimes 1_A) \circ
           (1_A \otimes 1_A \otimes \lambda) \circ
           (1_A \otimes 1_A \otimes (\epsilon_X \circ \mu_X) \otimes 1_A)
           \circ \\
          &(1_A \otimes (\delta_Z \circ \eta_Z) \otimes 1_A \otimes 1_A) \circ
           (\rho^{-1} \otimes 1_A \otimes 1_A) \circ (1_A \otimes
           (\delta_X \circ \eta_X))\circ \rho^{-1}
           \displaybreak[0]\\
        = &\textrm{\{laws of $\lambda$, $\rho$\}} \\
          &\lambda \circ ((\epsilon_Z \circ \mu_Z) \otimes 1_A) \circ
           (1_A \otimes \rho \otimes 1_A) \circ
           (1_A \otimes 1_A \otimes (\epsilon_X \circ \mu_X) \otimes 1_A)
           \circ \\
          &(1_A \otimes (\delta_Z \circ \eta_Z) \otimes 1_A \otimes 1_A) \circ
           (1_A \otimes \lambda^{-1} \otimes 1_A) \circ (1_A \otimes
           (\delta_X \circ \eta_X))\circ \rho^{-1}
           \displaybreak[0]\\
        = &\textrm{\{bifunctoriality of $\otimes$\}} \\
          &\lambda \circ ((\epsilon_Z \circ \mu_Z) \otimes 1_A)
           \circ\\
          &(1_A \otimes (\rho \circ (1_A \otimes (\epsilon_X \circ \mu_X)) \circ
          ((\delta_Z \circ \eta_Z) \otimes 1_A) \circ \lambda^{-1}) \otimes 1_A) \circ \\
          &(1_A \otimes (\delta_X \circ \eta_X))\circ \rho^{-1}
           \displaybreak[0]\\
        = &\textrm{\{dual law\}} \\
          &\lambda \circ ((\epsilon_Z \circ \mu_Z) \otimes \lambda)
           \circ\\
          &(1_A \otimes 1_A \otimes 1_A) \circ \\
          &(\rho^{-1} \otimes (\delta_X \circ \eta_X))\circ \rho^{-1}
           \displaybreak[0]\\
        = &\textrm{\{identity\}} \\
          &\lambda \circ ((\epsilon_Z \circ \mu_Z) \otimes \lambda)
           \circ
           (\rho^{-1} \otimes (\delta_X \circ \eta_X))\circ \rho^{-1} \\
\end{align*}
Even with explanatory notes about what law is being used when, this
proof is hard to follow -- and was hard to construct -- compared to the
graphical notation.

Note that we can make the graphical version even cleaner by using spider
notation:
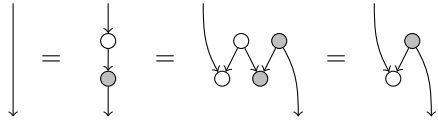
\[ \input{z-x-duals-derivation-spider.tikz} \]

The other two symmetries rely on commutativity.  In particular, we have
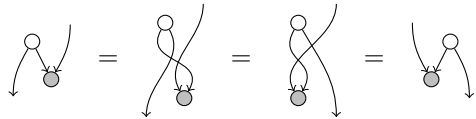
\[ \input{z-x-duals-symmetry-derivation.tikz} \]
and similarly with the colours swapped.

We can create ``spiderised'' versions of the copying laws
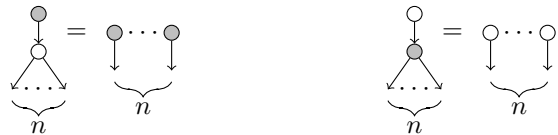
\begin{mathpar}
    \input{z-copies-x-spider.tikz}
    \and
    \input{x-copies-z-spider.tikz}
\end{mathpar}
and we can prove that these are correct by induction on $n$.  The proof
is fairly simple, and we will just do the X copies Z version, since
the other is symmetric.  The base case ($n = 0$) follows by the scalar
elimination laws.  So consider $n = k + 1$:
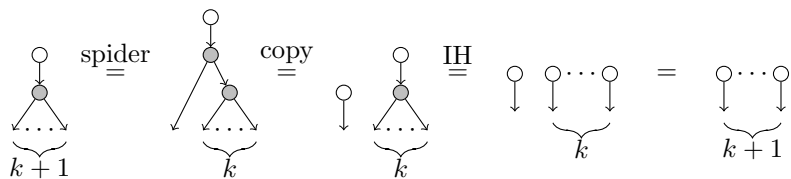
\[ \input{x-copies-z-spider-proof.tikz} \]

We can construct a similarly spiderised version of the bialgebra law
with some trivial applications of the spider law
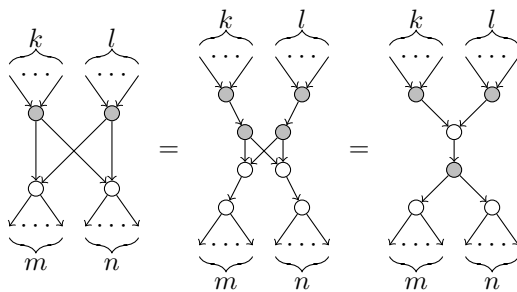
\[ \input{bialg-spider.tikz} \]
but there is also a generalised version of this involving the
$(s,t)$-bipartite graph of $X$ and $Z$ spiders (the normal bialgebra law
is the case when $s = t = 2$).  Graphically, this looks something like
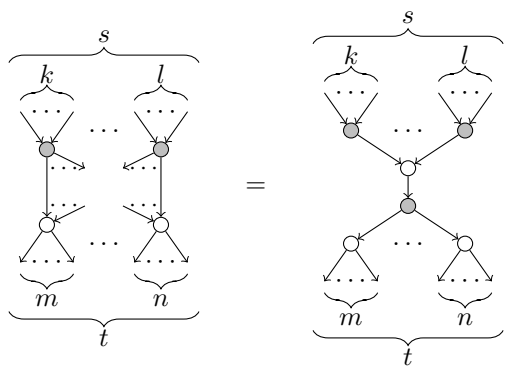
\[ \input{gen-bialg-attempt.tikz} \]
but it is hard to see what this is supposed to represent without an
accompanying textual description.  For the same reason, proving this
graphically is a messy affair, and convincing yourself that such a proof
is correct is even harder.

Chapters \ref{ch:bang-graphs} and \ref{ch:rw-bang-graphs} will introduce
a formal notation capable of expressing this generalised bialgebra law,
and chapter \ref{ch:rules} will provide the techniques necessary to
prove it formally, together with an informal demonstration of the proof.


\chapter{Rewriting}
\label{ch:rewriting}

Equational reasoning is core to many problems in mathematics and
computer science, from proving that a left unit is a right unit in a
commutative ring to optimising functional programs.

The traditional approach to mechanising equational reasoning is to use
term rewriting to implement equational logic.  We will briefly describe
this approach before presenting a related mechanisation scheme, adapted
from the existing literature (and \cite{KissingerDPhil} in particular),
that is much more suited to the examples in the preceeding chapter.

\section{Term Rewriting and the Word Problem} 
\label{sec:term-rewriting}

In order to reason formally about mathematical equations, we need a
symbolic language to express those equations.  This section summarises
the relevant parts of a large body of existing literature on term
rewriting and universal algebra; a more comprehensive introduction can
be found in \cite{TRaAT}.

A \textit{signature} provides symbols for the ``fixed'' elements of a
theory: for example, a multiplication operator or its unit.  It is
simply a set of symbols $\Sigma$, each with an associated non-negative
integer known as the arity\footnote{For simplicity, we are presenting
\textit{untyped} terms here}.  In the case of group theory, we might have
$\Sigma = \{e,i,\grpmult\}$, with respective arities $0$, $1$ and $2$.

If we also have a set of \textit{variables} $V$ disjoint from $\Sigma$,
we can start to write down \textit{terms} such as $\grpmult(i(x),e)$ (where
$x$ is a variable).  The set of possible $\Sigma$-terms over $V$,
$T(\Sigma,V)$, is defined inductively:
\begin{itemize}
    \item if $x \in V$ then $x \in T(\Sigma,V)$
    \item if $f \in \Sigma$ with arity $n$ and $t_1,\ldots,t_n \in
    T(\Sigma,V)$ then $f(t_1,\ldots,t_n) \in T(\Sigma,V)$
\end{itemize}

We can now encode equations using terms.  A \textit{$\Sigma$-equation}
is just a pair of terms that we wish to consider equal.  We write such
equations $s \rweq t$, where $s$ and $t$ are $\Sigma$-terms, to make it
clear that $s$ and $t$ are not the same term, but should be considered
equal in the theory.

The axioms of our theory, for example the axioms of a group, are encoded
as a set of $\Sigma$-equations, and equational logic\cite{Gries1993}
describes how we can use those axioms to derive further equalities via
the following inference rules, adapted for the notion of signature we
are using in \cite{TRaAT}
\begin{mathpar}
    (\textsc{Axiom})\enskip
    \inferrule{s \rweq t \in E}{E \vdash s \rweq t}
    \and
    (\textsc{Refl})\enskip
    \inferrule{ }{E \vdash t \rweq t}
    \and
    (\textsc{Sym})\enskip
    \inferrule{E \vdash s \rweq t}{E \vdash t \rweq s}
    \and
    (\textsc{Trans})\enskip
    \inferrule{E \vdash s \rweq t \\ E \vdash t \rweq u}{E \vdash s \rweq u}
    \and
    (\textsc{Subst})\enskip
    \inferrule{E \vdash s \rweq t}{E \vdash \sigma(s) \rweq \sigma(t)}
    \and
    (\textsc{Leibniz})\enskip
    \inferrule{E \vdash s_1 \rweq t_1 \\ \ldots \\ E \vdash s_n \rweq t_n}
              {E \vdash f(s_1,\ldots,s_n) \rweq f(t_1,\ldots,t_n)}
\end{mathpar}
where $\sigma$ is a substitution that replaces variables with other
terms and $f$ is an $n$-ary function symbol of $\Sigma$.

We commonly wish to determine whether two $\Sigma$-terms $s$ and $t$ are
equal by this logic under a certain set of axioms $E$; in other words,
is it the case that $E \vdash s \rweq t$? This is known as the
\textit{word problem} and is a major application of term rewriting.

The obvious way to approach the word problem is to attempt to use the
rules to construct a proof that two terms are equal directly.  Term
rewriting, in constrast, turns one term into another term that is equal
according to the \textsc{Axiom}, \textsc{Subst} and \textsc{Leibniz}
rules.  This is done by directing equations $s \rweq t$ to make them
\textit{rewrite rules}, written $s \rewritesto t$.

A set of terms $E$ can be turned into a \textit{term rewrite system} $R$
in this way, and we define the relation $\rewritestoin{R}$ such
that $s \rewritestoin{R} t$ if and only if there is a substitution
$\sigma$ and a rewrite rule $l \rewritesto r$ in $R$ such that
$\sigma(l)$ occurs in $s$, and replacing that occurrence by $\sigma(r)$
yields $t$.  For example, if $\grpmult(e,x) \rewritesto x \in R$, then
we can deduce that
\[ i(\grpmult(e,i(y))) \rewritestoin{R} i(i(y)) \]

\textit{Matching} is the process of finding an appropriate substitution
$\sigma$ and locating occurences of $\sigma(l)$ in $s$, and is
decidable.  Hence, given a finite set of rewrite rules $R$ and a term
$s$, it is possible to find all terms $t$ such that $s
\rewritestoin{R} t$.

The utility of this approach comes from the observation that the
reflexive, symmetric, transitive closure of $\rewritestoin{R}$ --
which we write $\rewriteequivin{R}$ -- corresponds exactly to the
$\rweq$ of equational logic; in other words, if we allow ourselves to
view $R$ as a set of equations,
\[ s \rewriteequivin{R} t \quad \Leftrightarrow \quad R \vdash s \rweq t \]

This observation means, in particular, that if we can find a common term
$u$ such that $s \rewritetransin{R} u$ and $t \rewritetransin{R} u$
(where $\rewritetransin{R}$ is the transitive closure of
$\rewritestoin{R}$), then we can deduce that $R \vdash s \rweq t$.  For
term rewrite systems that are \textit{terminating} (there is no infinite
chain $s_1 \rewritestoin{R} s_2 \rewritestoin{R} \ldots$) and
\textit{confluent} (if $s \rewritetransin{R} t_1$ and $s
\rewritetransin{R} t_2$ then there is a term $u$ such that $t_1
\rewritetransin{R} u$ and $t_2 \rewritetransin{R} u$), the word problem
is decidable.

While terms, which can have multiple inputs (the variables) but only one
output, are well-suited for encoding algebras, they are less useful for
structures like coalgebras, where we can have multiple outputs.  For
example, while a multiplication map $m : A \otimes A \rightarrow A$ in a
monoidal category is easy to represent as a binary term, a
comultiplication map $d : A \rightarrow A \otimes A$ does not have such
a clear symbolic representation.

Taking our cue from the diagrams of chapter \ref{ch:diag-reas}, we will
use graphs to construct a combinatorial symbolic representation of the
sort of maps found in a monoidal category.

\section{Monoidal Signatures} 
\label{sec:monoidal-sigs}

As with term rewriting, we need a way to represent items of interest
from the theory we are reasoning about; in this case, these items
correspond to objects and morphisms in a monoidal category.

\begin{definition}[Monoidal Signature: \cite{KissingerDPhil}, pp 30]
    A (small, strict) \textit{monoidal signature} $(O,M,\dom,\cod)$
    consists of two sets $M$ and $O$ together with a pair of functions
    $\dom,\cod : M \rightarrow O^*$ into lists of elements of $O$.
\end{definition}

The intended interpretation of a monoidal signature is that the elements
of $O$ are representations of objects of a monoidal category, $M$ is
(some of) its morphisms and $\dom$ and $\cod$ give the types of the
domain and codomain of the morphisms.  An empty list corresponds to $I$,
the tensor unit.

\begin{definition}[Monoidal Homomorphism: \cite{KissingerDPhil}, pp 31]
    If $S$ and $T$ are monoidal signatures, a \textit{monoidal signature
    homomorphism} $f : S \rightarrow T$ consists of functions $f_O : O_S
    \rightarrow O_T$ and $f_M : M_S \rightarrow M_T$ making the
    following diagrams commute:
    \[
        \begin{tikzcd}
            M_S \arrow{r}{\dom_S} \arrow[swap]{d}{f_M} &
            O_S^* \arrow{d}{f_O^*} \\
            M_T \arrow[swap]{r}{\dom_T} &
            O_T^*
        \end{tikzcd}
        \qquad
        \begin{tikzcd}
            M_S \arrow{r}{\cod_S} \arrow[swap]{d}{f_M} &
            O_S^* \arrow{d}{f_O^*} \\
            M_T \arrow[swap]{r}{\cod_T} &
            O_T^*
        \end{tikzcd}
    \]
\end{definition}

Recall from chapter \ref{ch:diag-reas} the ``spiders'' that arise from
Frobenius algebras.  These are families of morphisms that differ only in
how many inputs or outputs of a certain type they have.  Rather than
having an infinite monoidal signature with an object in $M$ for each
possible combination of inputs and outputs, it would be useful to have a
single symbol with variable arity.  To this end, we adapt the definition
of a monoidal signature:

\begin{definition}[Compressed Monoidal Signature]
    \label{def:compressed-monoidal-sig}
    If $O$ and $M$ are sets and $\dom,\cod : M \rightarrow (O \times
    \{\varar,\fixar\})^*$ are functions into lists of pairs of elements
    of $O$ and either $\varar$ or $\fixar$, then $(O,M,\dom,\cod)$ is a
    \textit{compressed monoidal signature}.
\end{definition}

$\varar$ and $\fixar$ are intended to indicate whether the input or
output is of variable or fixed (single, in fact) arity.  If it has
variable arity, it can be removed or duplicated.

Of course, if this is to be a true compressed representation of a larger
monoidal signature, we need a way to reconstruct that monoidal signature
from the compressed version.

\begin{definition}[Expansion]
    \label{def:mon-sig-exp}
    Let $S = (O,M,\dom,\cod)$ be a compressed monoidal signature.

    For each $f \in M$, let $D_f$ be the indices of the \varar-tagged
    elements of $\dom(f)$, and $C_f$ the indices of the \varar-tagged
    elements of $\cod(f)$.  The \textit{expansion} of $S$, $\expn(S)$,
    is $(O,M',\dom',\cod')$ where
    \begin{itemize}
        \item $M' = \bigcup\{\{f\}\times \mathbb{N}^{D_f} \times
            \mathbb{N}^{C_f} : f \in M\}$
        \item $\dom'(f,d_f,c_f)$ be $\dom(f)$ with, for all indices $i$,
            the $ith$ element replaced with $x$ if it is $(x,\fixar)$
            and $d_f(i)$ copies of $x$ if it is $(x,\varar)$
        \item $\cod'(f,d_f,c_f)$ be $\cod(f)$ with, for all indices $i$,
            the $ith$ element replaced with $x$ if it is $(x,\fixar)$
            and $c_f(i)$ copies of $x$ if it is $(x,\varar)$
    \end{itemize}
\end{definition}

$S$ is said to be a compressed monoidal signature for a monoidal
signature $T$ if the expansion of $S$ is isomorphic to $T$.
Compressed monoidal signatures that make no use of $\varar$ correspond
exactly to ordinary monoidal signatures.

\begin{example}
    \label{ex:compressed-mon-sig}
    Consider the compressed monoidal signature where
    \begin{align*}
        O &= \{A,B\} \\
        M &= \{f,g\} \\
        \dom(f) &= [(A,\varar)] \\
        \dom(g) &= [(A,\fixar),(B,\varar)] \\
        \cod(f) &= [(A,\fixar),(A,\fixar)] \\
        \cod(g) &= [(A,\fixar)]
    \end{align*}
    To make the description simpler, and more evocative of the intended
    interpretation, we use the following notation:
    \begin{align*}
        f &: [A^\varar] \rightarrow [A^\fixar,A^\fixar] \\
        g &: [A^\fixar,B^\varar] \rightarrow [A^\fixar]
    \end{align*}

    The expansion of this signature would include
    \begin{align*}
        (f,1 \mapsto 0,\varnothing) &: [\,] \rightarrow [A,A] \\
        (f,1 \mapsto 1,\varnothing) &: [A] \rightarrow [A,A] \\
        (f,1 \mapsto 2,\varnothing) &: [A,A] \rightarrow [A,A] \\
        \vdots
    \end{align*}
    and
    \begin{align*}
        (g,\varnothing,2 \mapsto 0) &: [A] \rightarrow [A] \\
        (g,\varnothing,2 \mapsto 1) &: [A,B] \rightarrow [A] \\
        (g,\varnothing,2 \mapsto 2) &: [A,B,B] \rightarrow [A] \\
        \vdots
    \end{align*}
\end{example}

Since the elements of a compressed monoidal signature are meant to
represent families of morphisms in a category, we need a way of
determining \emph{which} morphisms they are supposed to represent.  This
is provided by a \textit{valuation}.

Note: if $X$ is a list, we use $\mathcal{I}(X)$ to refer to the set of indices
of $X$.

\begin{definition}[Arity Count]
    Let $(M,O,\dom,\cod)$ be a compressed monoidal signature and $f \in
    M$.  Then an \textit{arity count for $f$} is a pair of maps $d :
    \mathcal{I}(\dom(f)) \rightarrow \mathbb{N}$ and $c :
    \mathcal{I}(\cod(f)) \rightarrow \mathbb{N}$ such that the index of
    each $\fixar$-tagged element maps to $1$.
\end{definition}

\begin{definition}[Valuation]
    \label{def:valuation}
    A \textit{valuation} $v$ of a compressed monoidal
    signature $(M,O,\dom,\cod)$ over a monoidal category $\mathcal{C}$ is
    a map $v_o$ from $O$ to the objects of $\mathcal{C}$
    together with a map $v_m^f$ for each $f \in M$ that takes
    arity counts for $f$ to morphisms of $\mathcal{C}$ with
    \[ v_m^f(d,c) :
            v_o(\alpha_1)^{d(1)} \otimes \ldots
            \otimes v_o(\alpha_l)^{d(l)}
            \rightarrow
            v_o(\beta_1)^{c(1)} \otimes \ldots
            \otimes v_o(\beta_n)^{c(n)}
    \]
    where $\dom(f) = [\alpha_1,\ldots,\alpha_l]$, $\cod(f) =
    [\beta_1,\ldots,\beta_n]$, $A^0$ is the monoidal identity $I$, and
    $A^{n+1} = A \otimes A^n$.

    The valuation is said to be \textit{expansion-order invariant} if
    $\mathcal{V}$ is a symmetric monoidal category and for all $f \in M$
    and all arity counts $(d,c)$ for $f$, if $i$ is the index of a
    variable-arity entry of $\dom(f)$ and $1 \leq j < d(i)$, then
    \[ v_m^f(d,c) \circ \textrm{swap}_{i,j} = v_m^f(d,c) \]
    where $\textrm{swap}_{i,j}$ is the map
    \[
        1_{X_1^{d(1)}} \otimes \ldots \otimes
        1_{X_{i-1}^{d(i-1)}} \otimes
        1_{X_{i}^{j-1}} \otimes
        \gamma_{X_{i},X_{i}} \otimes
        1_{X_{i}^{d(i)-j-1}} \otimes
        1_{X_{i+1}^{d(i+1)}} \otimes \ldots \otimes
        1_{X_l^{d(l)}}
    \]
    where $X_i = v_o(\alpha_i)$, and
    \[  \textrm{swap}'_{i',j'} \circ v_m^f(d,c) = v_m^f(d,c) \]
    where $i'$, $j'$ and $\textrm{swap}'_{i',j'}$ are defined
    analagously for $\cod(f)$ and $c$.
\end{definition}

Expansion-order invariance is intended to capture the idea that, in the
expansion of a compressed monoidal signature, a morphism should be
commutative on the inputs derived from the same $\varar$-tagged input,
and cocommutative on outputs derived from the same $\varar$-tagged
output.

\begin{example}
    Suppose we have a CFA (page \pageref{pg:cfa}) on an object $X$ of a
    symmetric monoidal category
    \begin{align*}
        \mu &: X \otimes X \rightarrow X \\
        \eta &: I \rightarrow X \\
        \delta &: X \rightarrow X \otimes X \\
        \epsilon &: X \rightarrow I
    \end{align*}
    and the compressed monoidal signature
    \[ s : [A^\varar] \rightarrow [A^\varar] \]
    If we want $s$ to represent the spiderised version of the CFA, we
    can construct a valuation $v$ such that $v_o(A) = X$ and $v_m^s(1
    \mapsto k,1 \mapsto l)$ is $k-1$ copies of $\mu$ followed by $l-1$
    copies of $\delta$, arranged in the following form:
    \[ \input{frob-normal-form-no-loops.tikz} \]
    If $k = 0$, we use $\eta$ to get a domain of $I$, and if $l = 0$
    we use $\epsilon$ to get a codomain of $I$.

    This valuation is expansion-order invariant because CFAs are
    (co-)commutative and (co-)associative, so precomposing or
    postcomposting the map above with swap maps does not affect the
    value of the morphism.
\end{example}

\section{Graphs} 
\label{sec:graphs}

A well-established categorical construction for directed graphs is to
use functors into \catSet, the category of (small) sets:

\begin{definition}
    \catGraph is the category of functors from the category
    \[
        \begin{tikzcd}
            E \arrow[yshift=0.4em]{r}{$s$}
                \arrow[swap,yshift=-0.2em]{r}{$t$}
            & V
        \end{tikzcd}
    \]
    to \catSet.
\end{definition}

Each element of \catGraph therefore consists of two sets, one of edges
(the image of $E$) and one of vertices (the image of $V$), and a pair of
functions selecting the source ($s$) and target ($t$) of each edge.  We
say that an edge is \textit{incident} to both its source and its target,
and call the set of edges incident to a vertex $v \in V$ the
\textit{edge neighbourhood} of $v$, $E(v)$.  Also, the source of an edge
is \textit{adjacent} to its target, and vice versa.

We can separate vertices and edges into types through the use of a
\textit{typegraph}.  This also allows the possible sources and targets
of the edges to be constrained.  For a given graph $\mathcal{G}$, we can
construct the category $\slicecat{\catGraph}{\mathcal{G}}$, the slice
category over $\mathcal{G}$.  The elements of this category can be
viewed as pairs $(H,t_H)$ of a graph $H$ and a graph morphism $t_H : H
\rightarrow \mathcal{G}$, known as the \textit{typing morphism}, and the
morphisms of the category are those graph morphisms $f : G \rightarrow
H$ that respect the typing morphisms, in the sense that
\[
    \begin{tikzcd}
        G \ar{rr}{f} \ar[swap]{dr}{\tau_G}
        && H \ar{dl}{\tau_H}
        \\ & \mathcal{G} &
    \end{tikzcd}
\]
commutes.  We call $(H,t_H)$ a \textit{$\mathcal{G}$-typed graph}.

If $\mathcal{G}$ is a subgraph of $\mathcal{H}$, then any
$\mathcal{G}$-typed graph can be viewed as a $\mathcal{H}$-typed graph,
simply by reinterpreting the codomain of the typing function (or,
equivalently, composing the typing function with the embedding morphism
from $\mathcal{G}$ to $\mathcal{H}$).  The general theorem, stated in
categorical terms, is as follows:

\begin{theorem}\label{thm:slice-cat-mono-func}
    Let $C$ be a category, and $m : S \rightarrow T$ a monomorphism in
    $C$.  Then $m$ induces a full embedding functor $E_m : C/S
    \rightarrow C/T$.  If $C$ has pullbacks along monomorphisms, $E_m$
    has a right adjoint $U_m$, and $U_m \circ E_m$ is the identity
    functor.
\end{theorem}
\begin{proof}
    If $(X,\tau_X)$ is an object of $C/S$, we define $E_m(X,\tau_X) =
    (X,m \circ \tau_X)$, and let $E_m$ be the identity on morphisms.
    Recall that a morphism $f : X \rightarrow Y$ of $C$ is a morphism $f
    : (X,\tau_X) \rightarrow (Y,\tau_Y)$ of $C/S$ if and only if
    \[ \tau_Y \circ f = \tau_X \]
    But if this is the case, then we must have
    \[ (m \circ \tau_Y) \circ f = m \circ \tau_X \]
    and hence $E_m$ is a well-defined functor.  Since $m$ is monic, the
    converse is also true, and so $E_m$ is full.

    Now suppose $C$ has pullbacks along monomorphisms; we need to show
    that $E_m$ has a right adjoint $U_m : C/T \rightarrow C/S$.  Given
    an object $(X',\tau_{X'})$ in $C/T$, we define $U_m(X',\tau_{X'})$
    to be $(X,\tau_X)$ in the following pullback:
    \[
        \begin{tikzcd}
            X' \arrow{r}{\tau_{X'}}
            & T
            \\ X \arrow{u}{\iota_X} \arrow[swap]{r}{\tau_X}
                 \NEbracket
            & S \arrow[swap]{u}{m}
        \end{tikzcd}
    \]
    which exists since $m$ is a monomorphism.  Now let $f' : X'
    \rightarrow Y'$ be a morphism of $C/T$.  We define $U_m(f')$ to be
    $f$ in the following diagram, which exists and is unique by
    pullback:
    \[
        \begin{tikzcd}
            X' \arrow[swap]{r}{f'}
               \arrow[bend left]{rr}{\tau_{X'}}
            & Y' \arrow[swap]{r}{\tau_{Y'}}
            & T
            \\
            & Y \arrow{u}{\iota_Y} \arrow{r}{\tau_Y}
                \NEbracket
            & S \arrow[swap]{u}{m}
            \\
            X
                \arrow{uu}{\iota_X}
                \arrow[bend right]{urr}{\tau_X}
                \arrow[dashed]{ur}{f}
            &&
        \end{tikzcd}
    \]

    The natural transformations $\eta_{(X,\tau_X)} = 1_X$ and
    $\epsilon_{(X',\tau_{X'})} = \iota_X$ witness that $U_m$ is right
    adjoint to $E_m$, and that $U_m \circ E_m = 1_{C/S}$.
\end{proof}

\begin{remark}
    Note that it is the case that, for any graph $T$, monomorphisms are
    exactly injective functions in both \catGraph and $\catGraph/T$.  It
    follows that if $m : S \rightarrow T$ is a graph monomorphism, $E_m$
    and $U_m$ in the above theorem preserve and reflect monomorphisms.
\end{remark}

We will sometimes, as in the next proposition, use the term
\textit{pushout of monomorphisms}; this refers to the pushout of a span
of monomorphisms, as distinct from a \textit{pushout along a
monomorphism} where only one arrow in the span needs to be monic.

\begin{proposition}\label{prop:slice-cat-func-pb}
    Let $S$ and $T$ be graphs and $m : S \rightarrow T$ a monomorphism.
    Then the functors $E_m : \catGraph/S \rightarrow \catGraph/T$ and
    $U_m : \catGraph/T \rightarrow \catGraph/S$ from theorem
    \ref{thm:slice-cat-mono-func} preserve pushouts of monomorphisms.
\end{proposition}
\begin{proof}
    All left adjoint functors preserve pushouts, and so $E_m$ does in
    particular.

    Consider the following pushout of monomorphisms in $\catGraph/T$:
    \[
        \begin{tikzcd}
            A \rar{a_1} \dar[swap]{a_2} & B \dar{b} \\
            C \rar[swap]{c} & D \NWbracket
        \end{tikzcd}
    \]
    By the \catSet-based construction of $\catGraph$, we know that the
    above is a pushout if and only if $D$ is covered by $\im(b)$ and
    $\im(c)$ and the intersection of those images is exactly the image
    of $b \circ a_1 = c \circ a_2$ (note that this is only true because
    all the arrows in the diagram are monomorphisms).

    We know that the following diagram commutes and consists of
    monomorphisms in $C/S$:
    \[
        \begin{tikzcd}[column sep=large]
            U_m(A) \rar{U_m(a_1)} \dar[swap]{U_m(a_2)}&
            U_m(B) \dar{U_m(b)} \\
            U_m(C) \rar[swap]{U_m(c)} &
            U_m(D)
        \end{tikzcd}
    \]

    Recall the construction of, for example, the map $U_m(b)$:
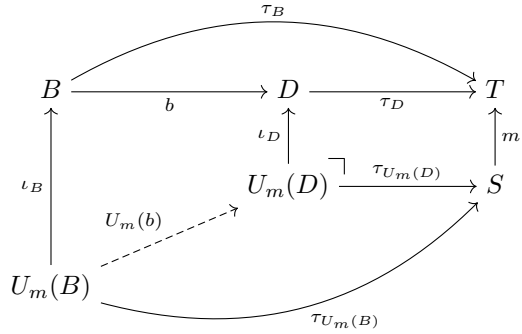
    \[
        \begin{tikzcd}[column sep=huge]
            B \arrow[swap]{r}{b}
               \arrow[bend left]{rr}{\tau_{B}}
            & D \arrow[swap]{r}{\tau_{D}}
            & T
            \\
            & U_m(D) \arrow{u}{\iota_D} \arrow{r}{\tau_{U_m(D)}}
                \NEbracket
            & S \arrow[swap]{u}{m}
            \\
            U_m(B)
                \arrow{uu}{\iota_B}
                \arrow[bend right,swap]{urr}{\tau_{U_m(B)}}
                \arrow[dashed]{ur}{U_m(b)}
            &&
        \end{tikzcd}
    \]
    It suffices to show that, for each $x \in U_m(D)$, $x \in
    \im(U_m(b))$ if and only if $\iota_D(x) \in \im(b)$ and similarly
    for the other morphisms of the pushout.  The argument is the same
    for each morphism, so we will only treat $b$.  We can see from the
    left square of the diagram that if $x \in \im(U_m(b))$, then
    $\iota_D(x) \in \im(b)$.

    For the converse, suppose $\iota_D(x) \in \im(b)$, and let $y$ be
    its (unique) preimage under $b$.  Now, $\tau_D(\iota_D(x)) \in
    \im(m)$, by the pullback square defining $U_m(D)$.  So then
    $\tau_B(y) \in \im(m)$, since $\tau_D \circ b = \tau_B$.  So $y \in
    \im(\iota_B)$ by the pullback defining $U_m(B)$ and its preimage
    must map to $x$ under $U_m(b)$.  Hence $x \in \im(U_m(b))$.
\end{proof}

\subsection{String Graphs} 
\label{sec:string-graphs}

Arbitrary graphs, however, are too general for our purposes.  They have
no concept of inputs or outputs, and so there is no clear way to
represent a morphism.  We therefore make use of \textit{string graphs},
introduced as \textit{open graphs} in \cite{Dixon2010} and further
refined in \cite{KissingerDPhil}.  These separate vertices into two
kinds: \textit{node-vertices} and \textit{wire-vertices}.  One way of
viewing these is that, in translating the diagrams we saw in chapter
\ref{ch:diag-reas} into string graphs, node-vertices correspond to the
nodes of those diagrams, while wire-vertices ``hold the wires in
place''.  For example, the diagram
\[ \input{mult-diag.tikz} \]
could be represented as
\[ \input{mult-graph.tikz} \]
where the larger white circle is a node-vertex and the smaller black
ones are wire-vertices.

Just as terms were defined relative to a signature, we will define
string graphs relative to a compressed monoidal signature.  We will use
a typegraph, so that the type of each node-vertex indicates which
morphism symbol it represents, and the type of each wire-vertex
indicates which object symbol it represents.

The following is adapted from the definition of a derived typegraph
(\cite{KissingerDPhil}, pp 88).

\begin{definition}\label{def:sg-monoidal-typegraph}
    Given a compressed monoidal signature $T =
    (O,M,\dom,\cod)$, the \textit{derived compressed typegraph}
    $\mathcal{G}_T$ has vertices $O + M$,
    a self-loop $\textrm{mid}_X$ for every $X \in O$ and, for every $f
    \in M$,
    \begin{itemize}
        \item an edge $\textrm{in}_{f,i}^a$ from $X$ to $f$, where
            $(X,a) = \dom(f)[i]$, for each index $i$ of the list
            $\dom(f)$, and
        \item an edge $\textrm{out}_{f,i}^a$ from $f$ to $X$, where
            $(X,a) = \cod(f)[i]$, for each index $i$ of the list
            $\cod(f)$.
    \end{itemize}
\end{definition}

\begin{example}
    Recall the compressed monoidal signature from example
    \ref{ex:compressed-mon-sig}:
    \begin{align*}
        f &: [A^\varar] \rightarrow [A^\fixar,A^\fixar] \\
        g &: [A^\fixar,B^\varar] \rightarrow [A^\fixar]
    \end{align*}
    The derived compressed typegraph would then be
    \[ \input{example-typegraph.tikz} \]

    The typegraph for the (spider) Z/X calculus (section
    \ref{sec:quantum}) is very simple:
    \[ \input{z-x-typegraph.tikz} \]
\end{example}

Let $T = (O,M,\dom,\cod)$ be a compressed monoidal signature.  Then for
a graph $(G,\tau)$ in $\catGraph/\mathcal{G}_T$ and a vertex $v$ in $G$,
if $\tau(v)$ is in $O$, we call $v$ a \textit{wire-vertex}.  Otherwise,
$\tau(v)$ must be in $M$, and we call $v$ a \textit{node-vertex}.  We
denote the set of wire-vertices $W(G)$ and the set of node-vertices
$N(G)$.  If $v$ is a node-vertex and $e$ is an edge incident to $v$, we
call $e$ \textit{variable-arity} if $\tau(e)$ is \varar-tagged (eg:
$\textrm{in}^\varar_{f,1}$) and \textit{fixed-arity} if it is
\fixar-tagged (eg: $\textrm{in}^\fixar_{g,1}$).  The \textit{fixed edge
neighbourhood} $N^\fixar(v)$ of a node-vertex $v$ is the set of edges in
the edge neighbourhood of $v$ that are fixed-arity.

In general, we will depict wire-vertices as small black dots.

\begin{definition}[Arity-matching]\label{def:sg-arity-matching}
    A map $f$ between $\mathcal{G}_T$-typed graphs $G$ and $H$ is
    \textit{arity-matching} if for every $v \in N(G)$, the restriction
    of $f$ to the fixed edge neighbourhood of $v$ is a bijection onto
    the fixed edge neighbourhood of $f(v)$.
\end{definition}

Note that, by considering $\mathcal{G}_T$ as the typed graph
$(\mathcal{G}_T,1_{\mathcal{G}_T})$, we can view typing morphisms as
morphisms of $\catGraph/\mathcal{G}_T$, and hence apply the above
terminology of arity-matching to typing morphisms.

The following is adapted from the definition of a string graph in
\cite{KissingerDPhil}, pp 89; it is equivalent to the original
definition if $\mathcal{G}_T$ has no variable-arity edges.

\begin{definition}[String Graph]\label{def:string-graph}
    A $\mathcal{G}_T$-typed graph $(G,\tau_G) \in
    \catGraph/\mathcal{G}_T$ is a \textit{string graph} if $\tau_G$ is
    arity-matching and each wire-vertex in $G$ has at most one incoming
    edge and at most one outgoing edge.  The category $\catSGraph_T$ is
    the full subcategory of $\catGraph/\mathcal{G}_T$ whose objects are
    string graphs.
\end{definition}

We refer to a wire-vertex of a string graph $G$ with no incoming edges
as an \textit{input}, and write the set of all inputs $\In(G)$.
Similarly, a wire-vertex with no outgoing edges is called an
\textit{output}, and the set of all such vertices is written $\Out(G)$.
The inputs and outputs together form the \textit{boundary} of the graph,
written $\Bound(G)$.

We will only consider finite string graphs and finite graphs in
$\catSGraph_T$, since we are only interested in graphs that are
tractable for computers.

\begin{lemma}
    \label{lemma:sg-morphism-bounds}
    If $f:G \rightarrow H$ is a morphism in $\catSGraph$ and $v$ is a vertex in
    $G$ that maps to an input (respectively output) of $H$ under $f$, then $v$
    must be an input (respectively output) of $G$.
\end{lemma}
\begin{proof}
    Suppose $e$ is an edge of $G$ such that $t_G(e) = v$.  Then
    $t_H(f(e))$ must be $f(v)$, and so $f(v)$ cannot be an input unless
    $v$ is.  The output case is symmetric.
\end{proof}

It is worth noting that not every subgraph of a string graph (in
$\catGraph/\mathcal{G}_T$) is a string graph.  In particular, if $G$ is a string
graph, $n$ is a node-vertex in $G$ and $e$ is a fixed-arity edge
incident to $n$, a subgraph of $G$ that contains $n$ but not $e$ will
fail to satisfy the arity-matching requirement of the typing morphism.
However, the intersection and union of two string graphs are still
string graphs.

\begin{proposition}\label{prop:sg-subgraph-ops}
    Let $G$ and $H$, both string graphs, be subgraphs of the string
    graph $K$.  Then $G \cap H$ and $G \cup H$ are both string graphs.
\end{proposition}
\begin{proof}
    Any subgraph of a string graph satisfies the requirements about
    wire-vertices, so we only have to consider whether the typing
    morphisms are arity-matching.  What is more, the typing morphism of
    any subgraph of a string graph must be injective at the fixed-arity
    neighbourhood of any of its node-vertices; we just need to consider
    sujectivity.

    Let $n$ be a node-vertex in $G \cap H$ and let $\tau_{G\cap H}^n$ be
    the restriction of $\tau_{G\cap H}$ to the fixed-arity
    neighbourhood of $n$.  Let $e$ be a fixed-arity edge adjacent to
    $\tau_{G\cap H}(n)$ in the typegraph.  Then both $\tau_G^n$ and
    $\tau_H^n$ must have $e$ in their images, and hence so much
    $\tau_{G\cap H}^n$.  Thus $\tau_{G\cap H}^n$ is surjective onto the
    fixed-arity neighbourhood of the typegraph, and so $\tau_{G\cap H}$
    is arity-preserving.

    The argument for $G \cup H$ is similar.
\end{proof}

\begin{remark}
    The construction of the typegraph $\mathcal{G}_T$ means that there
    can be no edges between node-vertices; any connection between
    node-vertices is mediated by wire-vertices.  Since wire-vertices in
    string graphs have at most one input and one output, every string
    graph $G$ is simple, in the sense that for any two vertices $v,w$ of
    $G$, there is at most one edge between $v$ and $w$ in each
    direction.
\end{remark}

\begin{proposition}\label{prop:sg-arity-matching}
    Every morphism in $\catSGraph_T$ is arity-matching.
\end{proposition}
\begin{proof}
    Let $G$ and $H$ be string graphs and $f : G \rightarrow H$ a graph
    morphism.  Consider $v \in N(G)$.  Since $\tau_G = \tau_H \circ f$,
    we know that the restriction of these maps to $N^\fixar(v)$ is
    identical.  But $\tau_G$ is arity-matching, and so its restriction
    to $N^\fixar(v)$ is a bijection, and hence the same is true of
    $\tau_H \circ f$.  But that means that $f$ restricted to
    $N^\fixar(v)$ must also be a bijection, as required.
\end{proof}

The following is adapted from a similar proof in \cite{KissingerDPhil},
pp 90.  We can conclude from it that $\catSGraph_T$ is what Kissinger
calls a \textit{partial adhesive category}; we have chosen not to use
this notion on the basis that it does not improve the clarity of the
work.

\begin{proposition}\label{prop:sg-mono-inj}
    A morphism of $\catSGraph_T$ is monic if and only if it is
    injective.
\end{proposition}
\begin{proof}
    Since this holds in $\catGraph/\mathcal{G}_T$, any injective map in
    $\catSGraph_T$ must be monic in $\catGraph/\mathcal{G}_T$, and hence also
    monic in $\catSGraph_T$.

    Suppose we have a non-injective morphism $f: G \rightarrow H$ in
    $\catSGraph_T$.  Then $f$ must either map two or more edges in $G$
    to the same edge in $H$ or map two or more vertices in $G$ to the
    same vertex in $H$.  In fact, since $G$ and $H$ are both simple, $f$
    must map two or more vertices $v_i$ in $G$ to a single vertex $v$ in
    $H$.  Consider the subgraph $K$ of $H$ containing $f(v_i)$.  This
    must be the vertex, together with any incident fixed-arity edges
    (and the wire-vertices at the other end of those edges).  Then for
    each $v_i$, we construct the string graph morphism $g_i$ that takes
    the single vertex in $K$ to $v_i$, and the fixed-arity edges to the
    correct incident edges of $v_i$ (to make the typing morphisms
    commute).  Now all the $f \circ g_i$ are the same morphism, but the
    $g_i$ morphisms are distinct, so $f$ is not monic.
\end{proof}

For the rest of this chapter, we fix an arbitrary compressed monoidal
signature $T$, and objects and morphisms will be implicitly drawn from
$\catSGraph_T$ unless otherwise stated.

\subsubsection{Wires} 
\label{sec:sg-wires}

It will sometimes be useful to have a different view of a string graph,
where we treat each connected string of edges and wire-vertices as a
single unit.  As the name \textit{wire-vertex} suggests, we will refer
to these as \textit{wires}.

We use a different definition of wires (and wire homeomorphisms) to
\cite{KissingerDPhil}; our definition has several advantages, notably
allowing for proposition \ref{prop:wires-cover-graph} and being able to
refer to explicit wire homeomorphisms.  However, despite the different
approach, the two definitions are morally equivalent.

\begin{definition}[Wire]
    A \textit{wire path} of a string graph $G$ is a path in $G$
    containing at least one edge, whose internal vertices are all
    wire-vertices, and which is not contained in a longer path of this
    kind.

    A \textit{closed wire} of $G$ is a pair $(W_V,W_E)$
    such that there is a wire path $W$ of $G$ where $W_E$ is the set
    of edges in $W$ and $W_V$ is the set of vertices of $W$.

    An \textit{open wire} of $G$ is a pair $(W_V,W_E)$
    such that there is a wire path $W$ of $G$ where $W_E$ is the set
    of edges in $W$ and $W_V$ is the set of internal vertices of $W$,
    plus the start/end vertex if $W$ is a cycle.

    $\Wires(G)$ is the set of all open wires of a string graph $G$.
\end{definition}

Note that an isolated wire-vertex (with no incident edges) is
\textit{not} a wire path.

\begin{proposition}
    Let $G$ be a string graph.  Then two distinct wire paths either
    describe the same cycle or overlap only on their start and end
    vertices.
    \label{prop:no-wire-overlap}
\end{proposition}
\begin{proof}
    The maximality requirement of wire paths means that, except in the
    case of cycles, the start or end vertex of one path cannot be an
    internal vertex of another.  So distinct paths with a shared
    component must diverge at some point.  However, the requirements of
    a string graph mean this cannot happen at a wire-vertex, and so
    can only happen at the start or end vertex.  Thus either the paths
    are identical (or, at least, describe the same cycle) or are
    distinct everywhere except the start and end vertices.
\end{proof}

This means there is a one-to-one correspondence between open and closed
wires.  In most cases, it will not matter whether we are referring to an
open or a closed wire, and so we will simply use the term \textit{wire}.
It also means that each wire path gives rise to only one closed wire and
one open wire, and only wire paths that describe the same cycle produce
the same wire.

\begin{proposition}
    The wire paths, and hence wires, of a string graph $G$ can be sorted
    into the following disjoint \textit{kinds}, based on their start and
    end vertices:
    \begin{itemize}
        \item \textit{circles}, where the start and end vertices are the
            same;
        \item \textit{interior wires}, where the start and end vertices
            are both node-vertices;
        \item \textit{bare wires}, where the start and end vertices
            are distinct wire-vertices, the start in $\In(G)$ and the
            end in $\Out(G)$;
        \item \textit{input wires}, where the start vertex is a
            wire-vertex in $\In(G)$ and the end vertex is a node-vertex;
            and
        \item \textit{output wires}, where the start vertex is a
            node-vertex and the end vertex is a wire-vertex in
            $\Out(G)$.
    \end{itemize}
    Additionally, all the wire-vertices of a wire path have the same
    type, and this is called the \textit{type} of the wire path (and of
    its derived open and closed wires).
\end{proposition}
\begin{proof}
    The main thing we need to note for the categorisation is that if the
    path is not a cycle (and hence not a circle), the maximality
    requirement on a wire path requires the start or end vertex to be in
    $\Bound(G)$ if it is a wire-vertex, since otherwise it could be an
    internal vertex of a longer wire path.

    The second part is a consequence of the construction of
    $\mathcal{G}_T$, which disallows edges between wire-vertices of
    distinct types.
\end{proof}

For wires other than circles, we call the start vertex of the path it is
derived from the \textit{source} of the wire, and the end vertex its
\textit{target}.  Closed wires include their source and target, while
open wires do not.

\begin{proposition}
    \label{prop:wires-cover-graph}
    Let $G$ be a string graph.  Then $N(G)$, $\Bound(G)$ and the
    elements of $\Wires(G)$ cover $G$ exactly (with no overlap).
\end{proposition}
\begin{proof}
    Suppose $v$ is a wire-vertex of $G$ not in $\Bound(G)$.  Then it
    must have both an incoming and an outgoing edge.  These, together
    with their respective source and target, form a path whose internal
    vertices are wire-vertices.  This must either be a wire path of $G$
    or else be part of a longer wire path, and hence $v$ is part of an
    open wire of $G$.  Similarly, any edge of $G$ trivially forms a
    path whose internal vertices are wire-vertices, and is therefore
    contained in an open wire of $G$.  So these components cover $G$.

    $N(G)$ and $\Bound(G)$ cannot overlap, due to the typegraph.
    Similarly, open wires contain no node-vertices, and so cannot
    overlap $N(G)$.  Any vertex in an open wire is an internal vertex of
    a path, and so cannot be in $\Bound(G)$.  Finally, proposition
    \ref{prop:no-wire-overlap} gives us that no two open wires
    intersect.
\end{proof}

Wire homeomorphisms provide a way to map between graphs that are
equivalent under this view.

\begin{definition}[Wire Homeomorphism]
    A \textit{wire homeomorphism} $f : G \sim H$ between two string
    graphs $G$ and $H$ consists of three bijective type-preserving
    functions
    \begin{itemize}
        \item $f_N : N(G) \leftrightarrow N(H)$
        \item $f_B : \Bound(G) \leftrightarrow \Bound(H)$
        \item $f_W : \Wires(G) \leftrightarrow \Wires(H)$
    \end{itemize}
    such that for each wire $w$ in $\Wires(G)$ with a source $v$ in
    $N(G)$ (resp. $\Bound(G)$), the source of $f_W(w)$ in $H$ is
    $f_N(v)$ (resp.  $f_B(v)$), and similarly for targets of wires.

    If there is a wire homeomorphism from $G$ to $H$, then $G$ and $H$
    are said to be \textit{wire homeomorphic}.
\end{definition}

\section{Valuation} 
\label{sec:sg-val}

In order for reasoning using string graphs to be useful, we need some
way to interpret them.  Node-vertices are intended to correspond to
morphisms of a category, wire-vertices to objects of the category and
edges to some concept of ``information flow''.  The inputs of a graph
naturally correspond in some way to the domain of the morphism
represented by the graph, and the outputs to the codomain.

The construction of the value of a graph is essentially the same as the
one Kissinger uses when he constructs a free traced SMC and a free
compact closed category using string graphs (\cite{KissingerDPhil}, pp
96-99, 104-112).  We bring the notion of the value of a graph to the
fore here, however, as we are more concerned with that than with free
categories.

Our approach will be to fix an order on the inputs and outputs of a
graph, using \textit{framed cospans}, and then to break down the graph
into \textit{elementary subgraphs} that we can use a valuation
(definition \ref{def:valuation}) to assign values (morphisms) to.  For
example, the graph
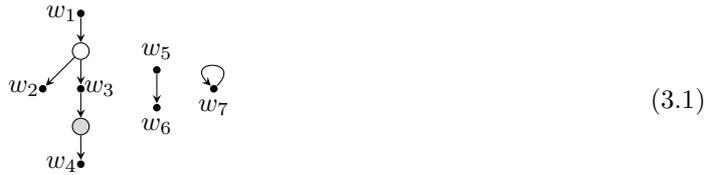
\begin{equation}\label{eq:value-graph}
    \input{val-ex-1a.tikz}
\end{equation}
will be broken down into
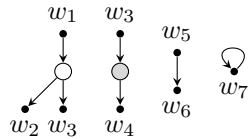
\[ \input{val-ex-2a.tikz} \]
which will be assigned morphisms, say $f$, $g$, $1_A$ and $1_B$
respectively.  We then take the tensor product of these morphisms
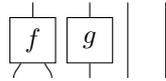
\[ \input{val-ex-3a.tikz} \]
and use the trace, in the form of a \textit{contraction}, to link the
morphisms together in the manner indicated by the connections of the
original string graph
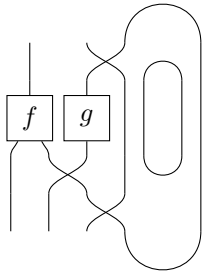
\[ \input{val-ex-5a.tikz} \]
Swap maps and dual maps are then applied as necessary to get the inputs
and outputs in the right order.

\subsection{Framed String Graphs} 
\label{sec:sg-framed}

Order is important.  In a monoidal category, the morphisms
\[
    f : A \otimes B \rightarrow C \otimes D
    \qquad \textrm{and} \qquad
    g : B \otimes A \rightarrow D \otimes C
\]
cannot be the same, although in a symmetric monoidal category, it is
possible that
\[ f = \gamma_{D,C} \circ g \circ \gamma_{A,B} \]

In particular, consider the swap and identity morphisms for $A \otimes
B$ in a symmetric monoidal category
\[
    \gamma_{A,B} : A \otimes B \rightarrow B \otimes A
    \qquad \textrm{and} \qquad
    1_{A \otimes B} : A \otimes B \rightarrow A \otimes B
\]
This would naturally be represented in string graph form as
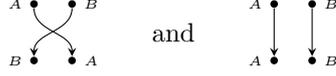
\[
    \input{swap.tikz}
    \qquad \textrm{and} \qquad
    \input{identity-A-B.tikz}
\]
However, these are the \emph{same graph} (or, at least, isomorphic),
and so there is no sensible way to assign one the value of
$\gamma_{A,B}$ and the other $1_{A \otimes B}$.  In a compact closed
category, we have a similar issue with the morphisms
\[
    d_A : A \otimes A^* \rightarrow I
    \qquad \textrm{,} \qquad
    1_A : A \rightarrow A
    \qquad \textrm{and} \qquad
    e_A : I \rightarrow A^* \otimes A
\]
whose graph representations should look like
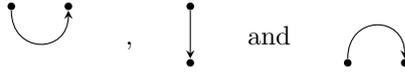
\[
    \input{d_A.tikz}
    \qquad \textrm{,} \qquad
    \input{1_A.tikz}
    \qquad \textrm{and} \qquad
    \input{e_A.tikz}
\]

To deal with this, we use \textit{framed cospans}\cite{KissingerDPhil}.
\begin{definition}[Framed Cospan]
    \label{def:framed-cospan}
    A \textit{string graph frame} is a triple $(X,<,\textrm{sgn})$ where
    $X$ is a string graph consisting only of isolated wire-vertices, $<$
    is a total order on $V_X$, the vertices of $X$, and $\textrm{sgn} :
    V_X \rightarrow \{+,-\}$ is the \textit{signing map}.

    A cospan $X \xrightarrow{d} G \xleftarrow{c} Y$ is called a
    \textit{framed cospan} if
    \begin{enumerate}
        \item $X$ and $Y$ are string graph frames
        \item $G$ contains no isolated wire-vertices
        \item the induced map $[d,c] : X + Y \rightarrow G$ restricts to
        an isomorphism\\$[d,c]' : X + Y \cong \Bound(G)$
        \item for every $v \in V_{X}$, $d(v) \in \In(G) \Leftrightarrow
        \textrm{sgn}(v) = +$
        \item for every $v \in V_{Y}$, $c(v) \in \Out(G) \Leftrightarrow
        \textrm{sgn}(v) = +$
    \end{enumerate}
\end{definition}
The frames can be seen as ``holding the boundary in place''.  The
signing map indicates inputs that appear in the codomain of the span,
and outputs that appear in the domain, with a $-$.  These correspond to
dual objects in compact closed categories.  A framed cospan where every
vertex in both frames is $+$ is called \textit{positive}.

Given a string graph $G$, we will often use $\hat G$ to refer to some
fixed \textit{framing} of it (ie: a framed cospan where the shared
domain is $G$).  The particular framing, if it is relevant, will be
given by the context.

The composition $\hat H \circ \hat G$ of two framed cospans
\[
    X \xrightarrow{d} G \xleftarrow{c} Y
    \qquad
    Y \xrightarrow{d'} H \xleftarrow{c'} Z
\]
is formed by the pushout
\[
    \begin{tikzcd}[ampersand replacement=\&]
        \&\&Y \arrow[swap]{dl}{c} \arrow{dr}{d'}\&\& \\
        X \arrow{r}{d} \& G \arrow[swap]{dr}{p_1} \&\&
        H \arrow{dl}{p_2} \& Z \arrow[swap]{l}{c'} \\
        \&\& H \circ G \Nbracket \&\&
    \end{tikzcd}
\]
where the resulting cospan is
\[ X \xrightarrow{p_1 \circ d} H \circ G \xleftarrow{p_2 \circ c'} Z \]

\begin{remark}
    We remarked earlier that proposition \ref{prop:sg-arity-matching}
    means that $\catSGraph_T$ is what Kissinger calls a partial adhesive
    category in \cite{KissingerDPhil}.  This, and the value construction
    above, form the necessary components for applying his proofs
    (\cite{KissingerDPhil}, pp 104-112) that that classes of framed
    cospans up to wire homeomorphism form a compact closed category, and
    homeomorphism classes of positive framed cospans form a symmetric
    traced category, to our extended version of string graphs.
\end{remark}

\subsection{Indexings and Contractions} 
\label{sec:sg-ind-contr}

Before we define the value of a graph, we will need some notation.  We
fix a traced symmetric monoidal category $\mathcal{C}$.

\begin{definition}[Indexing]
    Let $O$ be a (small) set of objects of $\mathcal{C}$.  For an object
    $X$ of $\mathcal{C}$, an \textit{$X$-indexing} is an indexed tensor
    product of the form
    \[ X_{i_1} \otimes \ldots \otimes X_{i_n} \]
    that is isomorphic to $X$.  We call a morphism $f : X \rightarrow Y$
    \textit{indexed} when we have fixed an $X$-indexing and a
    $Y$-indexing.
\end{definition}

For an $X$-indexing as above, we define the map $\textrm{shunt}_{X:k}$
to be the (unique) isomorphism constructed from the identity and swap
maps of $\mathcal{C}$ that moves the $X_k$ element of the $X$-indexing
to the end and leaves all others the same, as in the following diagram:
\[ \input{shunt.tikz} \]

\begin{definition}[Contraction]
    Given an indexed map $f : X_{i_1} \otimes
    \ldots \otimes X_{i_m} \rightarrow Y_{j_1} \otimes \ldots
    \otimes Y_{j_n}$, where $X_k = Y_l$, the \textit{contraction of $f$
    from $k$ to $l$}, $C_k^l(f)$, is the morphism
    \[ \Tr^{X_k}(\textrm{shunt}_{Y:l} \circ f \circ
    \textrm{shunt}_{X:k}^{-1}) \]
\end{definition}

A contraction of an indexed morphism yields another indexed morphism.
We can therefore apply multiple contractions to the same morphism, and
these contractions commute:
\begin{proposition}[Kissinger]
    For an indexed morphism $f$ and distinct indices $i$ ,$i'$, $j$ and
    $j'$,
    \[ C_i^j(C_{i'}^{j'}(f)) = C_{i'}^{j'}(C_i^j(f)) \]
\end{proposition}

\begin{proposition}
    Given $X$- and $Y$-indexings $X_{i_1} \otimes \ldots \otimes
    X_{i_m}$ and $Y_{j_1} \otimes \ldots \otimes Y_{j_n}$, and morphisms
    \[ f : X_{i_1} \otimes \ldots \otimes X_{i_k} \rightarrow
        Y_{j_1} \otimes \ldots \otimes Y_{j_l}
    \]
    and
    \[ g : X_{i_{k+1}} \otimes \ldots \otimes X_{i_m} \rightarrow
        Y_{j_{l+1}} \otimes \ldots \otimes Y_{j_n}
    \]
    if $p = i_a$ for some $a > k$ and $q = j_b$ for some $b > l$, then
    \[ C_p^q(f \otimes g) = f \otimes C_p^q(g) \]
    \label{prop:index-split}
\end{proposition}
\begin{proof}
    This follows from the definition and the laws of traces.
\end{proof}

\subsection{Value} 
\label{sec:sg-val-value}

Let $\mathcal{V}$ be a symmetric traced category, and let $v$ be an
expansion-order invariant valuation of $T$ over $\mathcal{V}$.
For a framed cospan $\hat G = X \rightarrow G \leftarrow Y$ of
$\catSGraph_T$, we define the \textit{value} of $\hat G$ under $v$ in
the following manner.

The \textit{elementary subgraphs} of $G$ are those subgraphs that
consist of one of the following:
\begin{itemize}
    \item a single node-vertex together with its incident edges and
        adjacent wire-vertices
        \[ \input{elem-g-node.tikz} \]
    \item an edge between wire-vertices, together with those
        wire-vertices (which may be the same vertex)
        \[ \input{elem-g-wire.tikz} \qquad \textrm{or}
           \qquad \input{elem-g-circle.tikz} \]
    \item a wire-vertex that has no incident edges in $G$
        \[ \input{elem-g-wire-vertex.tikz} \]
\end{itemize}
Note that these are all valid string graphs.  The definition of a string
graph (and, in particular, the restriction on incident edges of
wire-vertices) means that:
\begin{itemize}
    \item $G$ is covered by its elementary subgraphs,
    \item no node-vertex or edge appears in more than one elementary
        subgraph,
    \item wire-vertices appear in at most two elementary subgraphs, and
    \item any wire-vertex that does appear in two elementary subgraphs
        is an input in one and an output in the other.
\end{itemize}

\begin{example}
    The string graph
    \begin{equation}\label{eq:value-graph}
        \input{val-ex-1.tikz}
    \end{equation}
    has the following four elementary subgraphs:
    \[ \input{val-ex-2.tikz} \]
\end{example}

We then assign a value to each of these subgraphs, as well as a
wire-vertex-based indexing for the domain and codomain of that value.

For an elementary subgraph $H$ of $G$ with a node-vertex $n$, let $f =
\tau_H(n)$, and for each incoming edge $\textrm{in}^a_{f,i}$ of $f$ in
$\mathcal{G}_T$, let $d(i)$ be the size of the preimage of the edge
under $\tau_H$ (ie: the number of edges adjacent to $n$ that map to it),
and define $c(j)$ similarly for each outgoing edge
$\textrm{out}^a_{f,j}$.  Then these are arity counts for $f$, since
$\tau_G$, and hence $\tau_H$, is arity-matching.  The value of $H$ is
then $v_m^f(d,c)$.

Consider $A$, the domain of $v_m^f(d,c)$, and let $D(H)$ be the set of
wire-vertices of $H$ that are the source of an incoming edge of $n$.  If $D(H)$
is empty, $A = I$.  Otherwise, for each input $w$ of $H$ with adjacent
edge $e$, let $\lambda(w) = i$, where $\tau_H(e) = \textrm{in}^a_{f,i}$.
We can then choose a total order for $D(H)$ such that $\lambda$ is
non-decreasing, and this ordered set provides an $A$-indexing with the
property that if $\dom(f) = [\alpha_1,\ldots,\alpha_m]$, then $A_w =
\alpha_{\lambda(w)}$.  We can construct a similar indexing for the
codomain of $v_m^f(d,c)$ using $C(H)$, the set of targets of edges whose
source is $n$ (unless this is empty).

If $H$ is instead an edge $e$ connecting two wire-vertices, we assign the
value $1_{v_o(Z)} : v_o(Z) \rightarrow v_o(Z)$, where $Z$ is the image
of $s(e)$ (and hence also of $t(e)$).  We set $D(H) = \{s(e)\}$ and
$C(H) = \{t(e)\}$ (note that these may be the same).  If $H$ is an
isolated wire-vertex $w$ of type $Z$, we similarly assign the value
$1_{v_o(Z)}$ and set $D(H) = C(H) = \{w\}$.

We then let $g_0$ be the tensor product of the values of each elementary
subgraph of $G$.  A diagrammatic representation of this morphism for the
graph \eqref{eq:value-graph} is
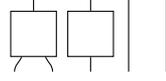
\[ \input{val-ex-3.tikz} \]
Let $W_1(G)$ be the union of $D(H)$ and $W_2(G)$ the union of $C(H)$ for the
elementary subgraphs $H$ of $G$.  The total orders we placed on the
inputs of each elementary subgraph of $G$, together with the order of
the tensor product $g_0$, gives us an order on $W_1(G)$, and (providing
$W_1(G)$ is not empty) this gives rise to a $W_1(G)$-based indexing of the
domain of $g_0$ that agrees with the indexings we constructed for the
elementary subgraphs; the same can be done for the codomain of $g_0$
using $W_2(G)$.

Now we order the wire-vertices $w_1,\ldots,w_m$ of $G$ that appear in
both $W_1(G)$ and $W_2(G)$, but are not isolated wire-vertices in $G$,
and for $1 \leq i \leq m$ let
\[ g_i = C_{w_i,w_i}(g_{i-1}) \]
For the previous example, this would result in
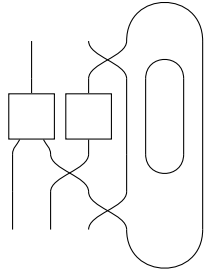
\[ \input{val-ex-5.tikz} \]

Note that $\In(G) \setminus \Out(G) = W_1(G) \setminus W_2(G)$ and,
conversely, $\Out(G) \setminus \In(G) = W_2(G) \setminus W_1(G)$.
This means that, given that the isolated wire-vertices of $G$ are the
only vertices to be both inputs and outputs of $G$, the remaining
indices in the indexing of the domain of $g_m$ are exactly the inputs of
$G$, and those in the indexing of the codomain are its outputs.

This means we can simply precompose and postcompose $g_m$ with the
necessary swap maps to re-order the indexings of its domain and codomain
to agree with the orderings of the frames of $\hat G$. We call the
resulting morphism $v(\hat G)$, the value of $\hat G$.

This final reordering means that the choice of order of elementary
subgraphs when constructing $g_0$ will not affect the value of $\hat G$.
The fact that contractions commute means that the choice of order of
wire-vertices in $G$ when constructing the maps $g_1,\ldots,g_m$ does
not affect the value of $\hat G$.  Finally, since the valuation $v$ is
expansion-order invariant, the choice of indexing for the domains and
codomains of the elementary subgraphs of $G$ does not affect the value
of $\hat G$.

In a compact closed category, a combination of swap maps and the maps
associated with dual objects can be used to produce $v(\hat G)$ from
$g_m$.

Note that the domain and codomain of $v(\hat G)$ depend only on the
frames, not on $G$.  Note also that the frames only affect the final
part of the construction, and so the values of two different framed
cospans with the same internal string graph are the same up to pre- and
postcomposition with swap maps if both are positive, and swap maps and
dual object maps otherwise.

We will assume all valuations are expansion-order invariant from now on.

\begin{proposition}
    Suppose $\hat G = X \xleftarrow{g_x} G \xrightarrow{g_y} Y$ and
    $\hat H = Z \xleftarrow{h_z} H \xrightarrow{h_w} W$ are framings
    and $f : G \rightarrow H$ is a wire homeomorphism such that the
    following diagram commutes
    \[
        \begin{tikzcd}
            X \ar{dr} \ar[<->]{rrr}{\cong}
            &&&
            Z \ar{dl}
            \\
            G \ar{u}{g_x} \ar[swap]{d}{g_y}
            &
            \Bound(G) \ar[right hook->]{l} \ar[<->]{r}{f_B} &
            \Bound(H) \ar[left hook->]{r} &
            H \ar[swap]{u}{h_z} \ar{d}{h_w}
            \\
            Y \ar{ur} \ar[<->]{rrr}{\cong}
            &&&
            W \ar{ul}
        \end{tikzcd}
    \]
    where the triangles on the left and right are taken from the
    definition of a cospan frame, and the top and bottom morphisms are
    order- and sign-preserving.  Then, for every valuation $v$,
    \[ v(\hat G) = v(\hat H) \]
    \label{prop:value-wire-homeo}
\end{proposition}
\begin{proof}
    We start by noting that the diagram above means we can identify $X$
    with $Z$, $Y$ with $W$ and $\Bound(G)$ with $\Bound(H)$ in a
    consistent manner.

    Now we make the observation that if $G$ and $H$ are
    wire-homeomorphic, then there is a third graph $K$ that is
    wire-homeomorphic to both where no two wire-vertices are adjacent to
    each other.  What is more, $K$ can be reached from both $G$ and $H$
    by a (finite) series of $0$ or more \textit{wire contractions},
    where an edge $e$ between two distinct wire-vertices is removed, and
    $s(e)$ and $t(e)$ are identified.  Thus we just have to show that
    the value of a framed cospan is stable under wire contraction, and
    the result follows by transitivity.

    So let $e$ be an edge of $G$ whose source and target are distinct
    wire-vertices.  Then $e$ (together with its source and target) is an
    elementary subgraph of $G$.  If $s(e)$ is an input of $G$ and $t(e)$
    is an output, contracting $e$ will make no difference to the value
    of $\hat G$, as we will simply be replacing one elementary subgraph
    with another that has the same value, and there are no contractions
    in either case.

    So suppose $s(e)$ is \textit{not} an input of $G$.  Then $s(e) \in
    W_2(G)$, and there is a contraction $C_{s(e)}^{s(e)}$.  We can
    choose to place the elementary subgraph containing $e$ last and
    apply the contraction on its source first.  $C_{s(e)}^{s(e)}(g_0)$
    will then have the form (where we only display the last swap of the
    contraction)
    \[ \input{val-e-loop.tikz} \]
    where $g'_0$ is the tensor product of the elementary subgraphs of the
    contracted form of $G$ together with all but one swap from the
    contraction of $s(e)$.  Then we have
    \[ C_{s(e)}^{s(e)}(g_0) = g'_0 \]
    by the trace axioms.  The extra swaps are irrelevant to the value of
    the contracted form of $\hat G$ as they are either absorbed into
    $C_{t(e)}^{t(e)}$ if $t(e)$ is not an output, or are corrected at
    the end to account for the frame ordering.

    The case where $t(e)$ is not an output of $G$ is symmetric.  So the
    value of $G$ remains the same after contracting $e$, as required,
    and hence we have the required result.
\end{proof}

\section{Graph Equations} 
\label{sec:graph-eqs}

Our aim is to be able to mechanise equational reasoning for diagrams
using string graphs.  However, we first need to establish what we even
mean by a ``graph equation'', and develop an equational logic for this
concept.

Recall that in the equational logic for terms, there are rules that
state that $\rweq$ is closed under contexts and substitutions.  We can
view substitutions as a sort of internal context.  For example, if we
say
\[ f(x) \rweq g(x) \]
we mean that we can replace $f$ with $g$ (or vice versa) regardless of
what is either inside or surrounding that symbol.  This approach is
well-suited to the tree-like structure of terms, where there is a single
external context and variables provide an easy way of labelling internal
contexts, but does not work so well for graphs, which have only an
external context, but may have multiple points of interaction with that
context.

Declaring that two morphisms are equal only makes sense if both
morphisms have the same type.  Given our intended interpretation of the
string graphs, this means that they need the same inputs and outputs
(respecting types).  But that is not enough; consider the associativity
law for multiplication.  For terms, we would have something like
\[ \grpmult(\grpmult(x,y),z) \rweq \grpmult(x,\grpmult(y,z)) \]
while with string graphs we would want something like
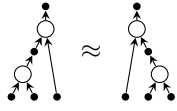
\[ \input{assoc-law.tikz} \]
Visually, the meaning of this is clear, but we need some way to encode
the correspondance between inputs; after all, one way of correlating the
inputs would result in a law that was simply a consequence of
commutativity:
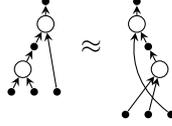
\[ \input{assoc-law-swapped.tikz} \]

We encode the correspondance of inputs and outputs using a span of graph
morphisms whose images are the boundaries of the two graphs.  We place
some extra coherence requirements on the span to ensure that inputs are
linked to inputs and outputs to outputs.

\begin{definition}[String Graph Equation; \cite{KissingerDPhil}, pp 93]
    \label{def:graph-eq}
    A \textit{string graph equation} $L \rweq_{i_1,i_2} R$ is a span ${L
        \xleftarrow{i_1} I \xrightarrow{i_2} R}$ where:
    \begin{itemize}
        \item $L$ and $R$ contain no isolated wire-vertices;
        \item $\In(L) \cong \In(R)$ and $\Out(L) \cong \Out(R)$;
        \item $\Bound(L) \cong I \cong \Bound(R)$; and
        \item the following diagram commutes, where $j_1$, $j_2$, $k_1$
            and $k_2$ are the coproduct inclusions into the boundary
            composed with the above isomorphisms:
            
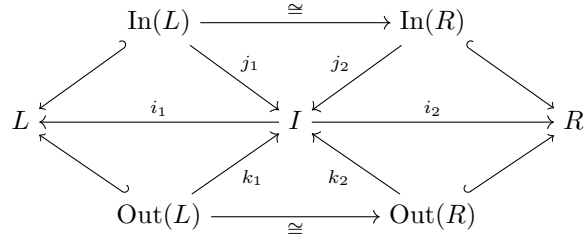
\[
            \begin{tikzcd}[ampersand replacement=\&]
                \& \In(L)
                    \arrow[left hook->]{dl}
                    \arrow{dr}{j_1}
                    \arrow{rr}{\cong}
                \&\& \In(R)
                    \arrow[swap]{dl}{j_2}
                    \arrow[right hook->]{dr}
                \&
                \\
                L \&\& I \arrow[swap]{ll}{i_1} \arrow{rr}{i_2} \&\& R
                \\
                \& \Out(L)
                    \arrow[left hook->]{ul}
                    \arrow[swap]{ur}{k_1}
                    \arrow[swap]{rr}{\cong}
                \&\& \Out(R)
                    \arrow{ul}{k_2}
                    \arrow[right hook->]{ur}
                \&
            \end{tikzcd}
            \]
    \end{itemize}
\end{definition}

We can assemble similar laws to those we had in the term-based
equational logic.  Axiom, reflexivity and symmetry are straightforward:
\begin{mathpar}
    (\textsc{Axiom}) \enskip
    \inferrule{G \rweq_{i,j} H \in E}{E \vdash G \rweq_{i,j} H}
    \and
    (\textsc{Refl}) \enskip
    \inferrule{ }{E \vdash G \rweq_{b_G,b_G} G}
    \and
    (\textsc{Sym}) \enskip
    \inferrule{E \vdash G \rweq_{i,j} H}{E \vdash H \rweq_{j,i} G}
\end{mathpar}
where $b_G : \Bound(G) \hookrightarrow G$ is the obvious inclusion of
the boundary of $G$ into $G$.  Transitivity requires a bit more work to
express.  If we have the graph equations
\[ G \xleftarrow{i} I \xrightarrow{j} H \]
and
\[ H \xleftarrow{k} J \xrightarrow{l} K \]
then the constraints on graph equations mean that $k^{-1} \circ j$ is
both well-defined and an isomorphism $I \cong J$.  If we take $p := i$
and $q := l \circ k^{-1} \circ j$, we can express transitivity as
\[
    (\textsc{Trans}) \enskip
    \inferrule{E \vdash G \rweq_{i,j} H \\ E \vdash H \rweq_{k,l} K
             }{E \vdash G \rweq_{p,q} K}
\]

For this to reflect intuitive equational reasoning, though, we need to
be able to apply equations in the context of larger graphs.  Given
graph equations
\[ G \xleftarrow{i} I \xrightarrow{j} H \]
and
\[ G' \xleftarrow{i'} I' \xrightarrow{j'} H' \]
such that there are monomorphisms making the following diagram commute,
and the squares in it pushouts,
\[ \begin{tikzcd}
    G \arrow[swap,right hook->]{d}
    & I \arrow[swap]{l}{i} \arrow{r}{j} \arrow[right hook->]{d}
    & H \arrow[right hook->]{d}
    \\ G' \NEbracket
    & D \arrow[left hook->]{l} \arrow[right hook->]{r}
    & H' \NWbracket
    \\
    & I' \arrow{ul}{i'} \arrow[swap]{ur}{j'}
         \arrow[right hook->]{u} &
\end{tikzcd} \]
then we have the following rule
\[
    (\textsc{Leibniz}) \enskip
    \inferrule{E \vdash G \rweq_{i,j} H}{E \vdash G' \rweq_{i',j'} H'}
\]

We actually need one more rule, to account for wire homeomorphisms.
Given graph equations
\[ G \xleftarrow{i} I \xrightarrow{j} H \]
and
\[ G' \xleftarrow{i'} I \xrightarrow{j'} H' \]
we say they are wire homeomorphic if there are wire homeomorphisms $g :
G \sim G'$ and $h : H \sim H'$ such that
\begin{equation}
    \begin{tikzcd}
        \Bound(G) \arrow[swap]{dd}{g_B}
        && \Bound(H) \arrow{dd}{h_B}
        \\
        & I \arrow[swap]{ul}{\cong} \arrow{ur}{\cong}
            \arrow{dl}{\cong} \arrow[swap]{dr}{\cong} &
        \\
        \Bound(G') && \Bound(H')
    \end{tikzcd}
    \label{eq:homeo-bounds}
\end{equation}
where the isomorphisms are those induced by $i$, $j$, $i'$ and $j'$ (see
definition \ref{def:graph-eq}). Given two such wire homeomorphic
equations, we have the following rule
\[
    (\textsc{Homeo}) \enskip
    \inferrule{E \vdash G \rweq_{i,j} H}{E \vdash G' \rweq_{i',j'} H'}
\]

We should note here that multiple stacked applications of the
\textsc{Leibniz} rule can be condensed into a single application.
\begin{proposition}
    Suppose we have the proof tree fragment
    \[
        \inferrule*[Left=Leibniz]{
            \inferrule*[Left=Leibniz]{
                E \vdash G \rweq_{i,j} H
            }{E \vdash G' \rweq_{i',j'} H'}
        }{E \vdash G'' \rweq_{i'',j''} H''}
    \]
    Then
    \[
        \inferrule*[Left=Leibniz]{E \vdash G \rweq_{i,j} H
        }{E \vdash G'' \rweq_{i'',j''} H''}
    \]
    is a valid proof tree fragment.
    \label{prop:leibniz-compact}
\end{proposition}
\begin{proof}
    We must have the following commuting diagram
    \[ \begin{tikzcd}
        G \arrow[swap,right hook->]{dd}
        & I \arrow[swap]{l}{i} \arrow{r}{j} \arrow[right hook->]{d}
        & H \arrow[right hook->]{dd}
        \\
        & D \arrow[left hook->]{dl} \arrow[right hook->]{dr}
        &
        \\[-2em]
        G' \arrow[swap,right hook->]{dd} \NEbracket
        &&
        H' \arrow[right hook->]{dd} \NWbracket
        \\[-2em]
        & I' \arrow{ul}{i'} \arrow[swap]{ur}{j'}
             \arrow[right hook->]{uu} \arrow[right hook->]{d} &
        \\ G'' \NEbracket
        & D' \arrow[left hook->]{l} \arrow[right hook->]{r}
        & H'' \NWbracket
        \\
        & I'' \arrow{ul}{i''} \arrow[swap]{ur}{j''}
             \arrow[right hook->]{u} &
    \end{tikzcd} \]
    We just need to find a string graph $D''$ and suitable monomorphisms
    to make
    \begin{equation}\label{eq:leibniz-compacted-diag}
        \begin{tikzcd}
            G \arrow[swap,right hook->]{d}
            & I \arrow[swap]{l}{i} \arrow{r}{j} \arrow[right hook->]{d}
            & H \arrow[right hook->]{d}
            \\ G'' \NEbracket
            & D'' \arrow[left hook->]{l} \arrow[right hook->]{r}
            & H'' \NWbracket
            \\
            & I'' \arrow{ul}{i''} \arrow[swap]{ur}{j''}
                 \arrow[right hook->]{u} &
        \end{tikzcd}
    \end{equation}
    commute.  We will only demonstrate how to construct the left part of
    the diagram, as the right side can be constructed in the same
    manner.  First of all, $D''$ and its inclusion into $G''$ are formed
    by pushout:
    \[ \begin{tikzcd}
        I' \rar[right hook->] \dar[right hook->] \ar[bend left]{rr}{i'}
        & D \rar[right hook->] \dar[right hook->]
        & G' \dar[right hook->]
        \\ D' \rar[right hook->] \ar[left hook->,bend right]{rr}
        & D'' \rar[right hook->,dashed] \NWbracket
        & G''
    \end{tikzcd} \]
    What is more, because the left square and outer squares are both
    pushouts, the right square must also be a pushout.

    Now the left square of \eqref{eq:leibniz-compacted-diag} can easily be
    seen to commute:
    \[ \begin{tikzcd}
            I \rar[right hook->] \dar[right hook->]
            & G \dar[right hook->]
            \\ D \rar[right hook->] \dar[right hook->]
            & G' \dar[right hook->] \NWbracket
            \\ D'' \rar[right hook->]
            & G''
    \end{tikzcd} \]
    What is more, we have just established that the bottom square is a
    pushout, so the outer square must also be a pushout.

    The triangle between $I''$, $D''$ and $G''$ follows easily:
    \[ \begin{tikzcd}
        I'' \rar[right hook->] \ar[bend right,swap]{drr}{i''}
        & D' \rar[right hook->] \ar[right hook->]{dr}
        & D'' \dar[right hook->]
        \\ && G''
    \end{tikzcd} \]
\end{proof}

\subsection{Soundness}

We want to make sure these rules are \textit{sound} with respect to the
valuation we described in section \ref{sec:sg-val}.  In other words, if
the assumptions of one of the preceding rules holds under a particular
valuation, then the conclusion must hold under the same valuation.

However, the rules above deal with string graph equations, and
valuations apply to framed cospans.  We need to introduce frames to both
sides of a string graph equation in a way that is compatible with the
span that forms the equation.

\begin{definition}
    A \textit{framing} of a string graph equation $L \rweq_{i_1,i_2} R$
    is a pair of framed cospans $X \xrightarrow{a} L \xleftarrow{b} Y$
    and $X \xrightarrow{c} R \xleftarrow{d} Y$ such that there are
    morphisms $j_1$ and $j_2$ forming the coproduct (ie: disjoint union)
    $X \xrightarrow{j_1} I \xleftarrow{j_2} Y$ and making the following
    diagram commute:
    \[
        \begin{tikzcd}[ampersand replacement=\&]
            \& X
                \arrow[swap]{dl}{a}
                \arrow{d}{j_1}
                \arrow{dr}{c}
                \&
            \\
            L \& I \arrow[swap]{l}{i_1} \arrow{r}{i_2} \& R
            \\
            \& Y
                \arrow{ul}{b}
                \arrow[swap]{u}{j_2}
                \arrow[swap]{ur}{d}
                \&
        \end{tikzcd}
    \]
\end{definition}

Note that it is always possible to frame a string graph equation in the
following way: let $X$ be a subgraph of $I$, and $Y$ its complement.
Place an arbitrary ordering on $X$ and $Y$, and construct signing maps
so that $\textit{sgn}_X(x) = +$ if and only if $i_1(x)$ (and hence also
$i_2(x)$) is an input, and $\textit{sgn}_Y(y) = +$ if and only if
$i_1(y)$ (and hence also $i_2(y)$) is an output.  $X$ and $Y$ are then
frames.  We set $j_1$ and $j_2$ to be the natural inclusions, and define
$a$, $b$, $c$ and $d$ by the above diagram.  $X \xrightarrow{a} L
\xleftarrow{b} Y$ and $X \xrightarrow{c} R \xleftarrow{d} Y$ are then
framed cospans, and so we have a framing of the string graph equation.

In addition, fixing a framed cospan for either $L$ or $R$ fixes a
framing for the equation.  This is because fixing $a$ and $b$, for
example, fixes $j_1$ and $j_2$ (since $i_1$ is monic).

We say that a string graph equation $L \rweq_{i_1,i_2} R$ holds
under a valuation $v$ if for all framings $(\hat L,\hat R)$ (all
positive framings if the valuation is not over a compact closed
category) we have that $v(\hat L) = v(\hat R)$.  In fact, this is the
same as saying we have this for \textit{any} framing.

\begin{proposition}
    For any valuation $v$ and any string graph equation $L
    \rweq_{i_1,i_2} R$, if $(L_1,R_1)$ and $(L_2,R_2)$ are framings of
    $L \rweq_{i_1,i_2} R$ (where both are positive unless $v$ is over a
    compact closed category), then $v(L_1) = v(R_1)$ if and only if
    $v(L_2) = v(R_2)$.
    \label{prop:soundness-is-framing-invariant}
\end{proposition}
\begin{proof}
    Consider the domains and codomains of $v(L_1)$ and $v(L_2)$, indexed
    by the vertices of their respective frames, $X_1$, $Y_1$, $X_2$ and
    $Y_2$.

    Using the same construction as for the value, we can find a tensor
    product of identities $i$ and maps $c$ and $d$ composed of
    identities, swaps and dual maps such that
    \[ v(L_1) = c \circ (v(L_2) \otimes i) \circ d \]
    and this construction depends only on the frames of the cospans,
    since these are what determine the domains and codomains of the
    values and their indexings.  But this means that we also have
    \[ v(R_1) = c \circ (v(R_2) \otimes i) \circ d \]
    since $L_1$ and $R_1$ have the same frames, and $L_2$ and $R_2$ also
    share their frames.  Thus if $v(L_2) = v(R_2)$, then $v(L_1) =
    v(R_1)$.

    The same argument works in reverse, giving us that if $v(L_1) =
    v(R_1)$, then $v(L_2) = v(R_2)$.
\end{proof}

Soundness of the reflexivity and symmetry axioms follows from the
reflexivity and symmetry of $=$.  For transitivity, we just need to note
that, given a framing $(\hat G,\hat K)$ of $G \rweq_{p,q} K$, there is a
framed cospan $\hat H$ such that $(\hat G,\hat H)$ and $(\hat H,\hat K)$
are framings of $H \rweq_{k,l} K$ and $G \rweq_{i,j} H$ respectively.

\begin{proposition}
    The \textsc{Leibniz} rule for string graphs is sound.  In other
    words, if we have the following commuting diagram
    \[ \begin{tikzcd}
        G \arrow[swap,right hook->]{d}{g}
        & I \arrow[swap]{l}{i_1} \arrow{r}{i_2} \arrow[right hook->]{d}{d}
        & H \arrow[right hook->]{d}{h}
        \\ K \NEbracket
        & D \arrow[swap,left hook->]{l}{k} \arrow[right hook->]{r}{l}
        & L \NWbracket
        \\
        & J \arrow{ul}{j_1} \arrow[swap]{ur}{j_2}
             \arrow[right hook->]{u}{d'} &
    \end{tikzcd} \]
    where the top and bottom spans are string graph equations, a framing
    $(\hat G,\hat H)$ of $G \rweq_{i_1,i_1} H$ and framing $(\hat
    K,\hat L)$ of $K \rweq_{j_1,j_2} L$, then for all valuations $v$,
    \[ v(\hat G) = v(\hat H)\;\Rightarrow\;v(\hat K) = v(\hat L) \]
\end{proposition}
\begin{proof}
    By proposition \ref{prop:soundness-is-framing-invariant}, it
    suffices to choose a specific framing for each equation, instead of
    showing it for all framings.

    Since all the arrows in the above diagram are monic (and hence
    injective), we will consider them to be subgraph relations (so $G$
    is a subgraph of $K$ by $g$, and so on).  For convenience, we will
    use $D'$ to refer to $D$ less any isolated wire-vertices.

    We start by noting that $I$ consists only of isolated wire-vertices,
    but $G$ and $K$ have \textit{no} isolated wire-vertices.  This means
    that, since $K$ is a pushout of $D$ and $G$ from $I$, each
    elementary subgraph of $K$ must be an elementary subgraph of exactly
    one of $D'$ or $G$.  Similarly, $W_1(K)$ (defined in section
    \ref{sec:sg-val-value}) is the disjoint union of $W_1(G)$ and
    $W_1(D')$, and likewise for $W_2(K)$.  The same results hold for
    $L$.

    It then suffices to show that, when calculating the value of a
    framed cospan $\hat K$ of $K$, we can choose framings $(\hat G,\hat
    H)$ and $(\hat K,\hat L)$ for the equations such that $v(\hat K)$ is
    of the form
    \[ C_{a_n}^{a_n}(\cdots C_{a_1}^{a_1}(d_0 \otimes v(\hat G))\cdots) \]
    and $v(\hat L)$ is
    \[ C_{a_n}^{a_n}(\cdots C_{a_1}^{a_1}(d_0 \otimes v(\hat H))\cdots) \]
    for some indexed morphism $d_0$.  The result then follows immediately.

    We choose an order for the elementary subgraphs of $D'$ (none of
    which are isolated wire-vertices), and an order for the
    variable-arity edges of those subgraphs.  Tensoring the elementary
    subgraphs together in order gives us $d_0$, which can be indexed by
    $W_1(D')$ and $W_2(D')$, appropriately ordered.  We do the same for
    $G$ and $H$ to get $g_0$ and $h_0$, indexed similarly.

    Now $k_0 = d_0 \otimes g_0$ and $l_0 = d_0 \otimes h_0$ are valid
    choices for the same values for $K$ and $L$, and we can order
    $W_1(K)$ using the order of $W_1(D')$ followed by the order of
    $W_1(G)$, and similarly for $W_2(K)$, $W_1(L)$ and $W_2(L)$.  These
    can then be used to index $k_0$ and $l_0$, and these indexings agree
    with the indexings for $d_0$, $g_0$ and $h_0$.

    Next, we arbitrarily order $W_1(G) \cap W_2(G)$ as $b_1,\ldots,b_p$
    and $W_1(H) \cap W_2(H)$ as $c_1,\ldots,c_q$.  We apply the
    contractions $C_{b_i}^{b_i}$ in order to $k_0$ and $C_{c_j}^{c_j}$
    to $h_0$ to get $k_p$ and $h_q$ respectively.  At this point, we
    note that, by proposition \ref{prop:index-split},
    \[ k_p = d_0 \otimes C_{b_p}^{b_p}(\cdots C_{b_1}^{b_1}(g_0)\cdots)
    = d_0 \otimes g_p \]
    and similarly $l_q = d_0 \otimes h_q$.  But if we choose the framing
    $(\hat G,\hat H)$ so that the ordering of the frames matches the
    ordering of $W_1(G)$ and $W_2(G)$, we have that $v(\hat G) = g_p$.
    To get $v(\hat H)$, we may need to pre- and post-compose with swap
    maps to get an order that agrees with the frames.  Call these $s_1$
    and $s_2$.  Then $v(\hat H) = s_2 \circ h_q \circ s_1$.

    If we order $(W_1(K) \cap W_2(K))\setminus (W_1(G) \cap W_2(G))$ as
    $a_1,\ldots,a_n$, a similarly suitable choice of the framing $(\hat
    K,\hat L)$ gives us
    \[ v(\hat K) = C_{a_n}^{a_n}(\cdots C_{a_1}^{a_1}(d_0 \otimes v(\hat
    G))\cdots) \]
    It remains to show that
    \[ v(\hat L) = C_{a_n}^{a_n}(\cdots C_{a_1}^{a_1}(d_0 \otimes (s_2
    \circ h_q \circ s_1))\cdots) \]
    We need to consider the decomposition of $s_1$ and $s_2$ into their
    individual swap operations.  If any given swap operates on an index
    that is then contracted (ie: one of the $a_i$), it can be subsumed
    into that contraction, due to the way contractions are defined.  So
    we can construct reordering maps $s'_1$ and $s'_2$ that only operate
    on inputs and outputs (respectively) of $L$ such that
    \[ C_{a_n}^{a_n}(\cdots C_{a_1}^{a_1}(d_0 \otimes (s_2 \circ h_q
    \circ s_1))\cdots) = C_{a_n}^{a_n}(\cdots C_{a_1}^{a_1}(d_0 \otimes
    (s'_2 \circ h_q \circ s'_1))\cdots)\]
    But this is the same as
    \[ (1 \otimes s'_2) \circ C_{a_n}^{a_n}(\cdots C_{a_1}^{a_1}(d_0 \otimes
    h_q)\cdots) \circ (1 \otimes s'_1)\]
    We will refer to this (indexed) morphism as $\lambda : A \rightarrow
    B$.

    We need to show that indexing of this morphism agrees with the order
    on the frames of $\hat L = X \hookrightarrow L \hookleftarrow Y$.
    We will just treat the $X$ frame, as the $Y$ frame is symmetric.
    Note that we have the following commuting diagram, by the definition
    of a framing, which allows us to consider $X$ a subgraph of $D$.
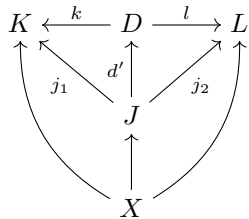
    \[ \begin{tikzcd}
        K & D \arrow[swap]{l}{k} \arrow{r}{l} & L
        \\
        & J \arrow{u}{d'} \arrow{ul}{j_1} \arrow[swap]{ur}{j_2} &
        \\
        & X \arrow{u} \arrow[bend left]{uul} \arrow[bend right]{uur} &
    \end{tikzcd} \]
    We split $X$, which consists of the inputs of $K$ (and therefore of
    $L$) into two parts: $X_1$ are those vertices that are inputs of
    $D'$, and $X_2$ are the inputs of $G$ (and hence of $H$).
    Note that we can consider $X_1$ a subgraph of $D'$ and $X_2$ a
    subgraph of $I$.

    Suppose $a < b \in X$.  Then either both are in $X_1$, both are in
    $X_2$ or $a \in X_1$ and $b \in X_2$.  In the latter case, we know
    that $a$ appears before $b$ in the indexing of the domain of
    $\lambda$ due to the construction of $\lambda$.  Suppose both are in
    $X_1$.  Then we know that $A_a$ appears before $A_b$ in the indexing
    of $d_0$, because that is how we chose the ordering on $X$.
    Otherwise, both are in $X_2$.  We know that $A_a$ appears before
    $A_b$ in the indexing of $g_0$, and hence in the indexing of $g_p$.
    If the same is true of the indexing of $h_q$, it must be the case
    that $s_1$, and hence $s'_1$, does not swap the order of $A_a$ and
    $A_b$, and so $A_a$ appears before $A_b$ in the indexing of
    $\lambda$.  Otherwise, if $A_b$ is before $A_a$ in the indexing of
    $h_q$, $s_1$ (and hence $s'_1$) must swap $A_a$ with $A_b$, and so
    we still have that $A_a$ is before $A_b$ in the indexing of
    $\lambda$.

    So $v(\hat L) = \lambda$, as required.
\end{proof}

\begin{proposition}
    The \textsc{Homeo} rule for string graphs is sound.  In other words,
    given wire-homeomorphic string graph equations $G \rweq_{i,j} H$ and
    $G' \rweq_{i',j'} H'$ with respective framings $(\hat G,\hat H)$ and
    $(\hat G',\hat H')$,
    \[ v(\hat G) = v(\hat H)\;\Rightarrow\;v(\hat G') = v(\hat H') \]
\end{proposition}
\begin{proof}
    The precondition on the wire homeomorphisms allows us to construct
    framings of the equation that are consistent with each other in the
    manner required by proposition \ref{prop:value-wire-homeo}.  From
    that proposition, we get
    \[ v(\hat G') = v(\hat G) \qquad \textrm{and} \qquad
    v(\hat H) = v(\hat H') \]
    and it follows by transitivity of $=$ that
    \[ v(\hat G) = v(\hat H)\;\Rightarrow\;v(\hat G') = v(\hat H') \]
\end{proof}

\section{Graph Rewriting} 
\label{sec:dpo}

Graph rewriting has been around in some form since at least the 60s;
it made an early appearance in \textit{graph grammars}, a field that
aimed to apply formal language theory to multidimensional objects
rather than just linear strings, notably in \cite{Pfaltz1969} and
\cite{Schneider1970}.

In term rewriting, we replaced a subterm of a term with another term
that we viewed as equivalent in some way.  In a similar manner, graph
rewriting (at least as far as this thesis is concerned) aims to replace
a subgraph of a graph with another ``equal'' graph.  There are several
approaches to this problem, almost all of which involve a graph $L$ to
be searched for and removed from the target graph $G$, a graph $R$ to
replace it with and some \textit{embedding transformation} $E$
describing how to map from the connectivity of $L$ with $G \graphminus
L$ to edges between $R$ and $G \graphminus L$.

Algebraic graph rewriting systems express $E$ using constructs from
category theory, which provides a useful set of tools for reasoning
about such systems.  Using these systems allows us to make use of
established work, such as proofs that independent rewrites (which
operate on different parts of $G$) can be applied in either order and
acheive the same result\cite{Ehrig1975}.

There are two main approaches to algebraic graph rewriting:
single-pushout (SPO) and double-pushout (DPO).  DPO was the first
algebraic graph rewriting system, introduced by Ehrig, Pfender and
Schneider in \cite{Ehrig1973}.  A rewrite rule is represented as a span
of graphs
\[ L \xleftarrow{i} I \xrightarrow{j} R \]
and a rewrite from $G$ to $H$ is a pair of overlapping pushouts
\[
    \begin{tikzcd}
        L \arrow[swap]{d}
        & I \arrow[swap]{l}{i} \arrow{r}{j} \arrow{d}
        & R \arrow{d}
        \\ G \NEbracket
        & D \arrow{l} \arrow{r}
        & H \NWbracket
    \end{tikzcd}
\]
The embedding transformation is described by $i$, $j$ and $I$.  We
can think of $I$ as being the interface of the rewrite rule, the part
allowed to interact with the rest of $G$.  This part is held in place
when the rest of $L$ is removed from $G$, allowing the connectivity with
$G \graphminus L$ to be maintained.

SPO, as presented by L\"owe in \cite{VandenBroek1991},
generalises DPO by using a single pushout
\[
    \begin{tikzcd}
        L \dar \rar
        & R \dar
        \\ G \rar
        & H \NWbracket
    \end{tikzcd}
\]
in the category of partial graph morphisms.  This allows for rewrites
that would be prohibited under DPO, such as deleting vertices in unknown
contexts.  For example the SPO rewrite
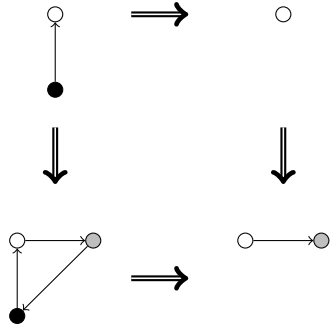
\[ \input{spo-delete.tikz} \]
has no equivalent DPO rewrite, because of the edge between the grey and
black vertices.

We will make use of DPO as our graph rewriting system, as we do not
require (and, in fact, wish to avoid) the extra flexibility of SPO.

Similar to the term case, we direct a string graph equation (definition
\ref{def:graph-eq}) to get a \textit{string graph rewrite rule}.  Given
a string graph rewrite rule $L \xleftarrow{i} I \xrightarrow{j} R$, a
\textit{string graph rewrite} is a pair of pushouts of monomorphisms in
$\catSGraph_T$
\begin{equation}\label{eq:sg-rw-rule}
    \begin{tikzcd}
        L \arrow[swap,right hook->]{d}{m}
        & I \arrow[swap]{l}{i} \arrow{r}{j} \arrow[right hook->]{d}
        & R \arrow[right hook->]{d}
        \\ G \NEbracket
        & D \arrow[left hook->]{l} \arrow[right hook->]{r}
        & H \NWbracket
    \end{tikzcd}
\end{equation}
(which should also be familiar from the rules of section
\ref{sec:graph-eqs}).

Note that the inclusion of $D$ into $G$ must contain $\Bound(G)$ in its
image.  This is because any input or output of $G$ in the image of $m$
must map from an input or output of $L$, and must therefore be in the
image of $i$.  The same is true of the inclusion of $D$ into $H$.  What
is more, the constraints on $i$ and $j$ mean that the preimages of
$\Bound(G)$ and $\Bound(H)$ under these inclusions coincide.  The same
is also true for $\In(G)$ and $\In(H)$ and for $\Out(G)$ and $\Out(H)$.
If we let $I'$ be the shared preimage of the boundaries of $G$ and $H$,
then -- providing $G$ and $H$ have no isolated wire-vertices -- the span
at the bottom of
\[ \begin{tikzcd}
    G
    & D \lar[left hook->] \rar[right hook->]
    & H
    \\
    & I' \uar[right hook->] \ar{ul}{i'} \ar[swap]{ur}{j'} &
\end{tikzcd} \]
is a string graph rewrite rule.

As the term ``rewriting'' suggests, what we actually want to do is start
with the string graph rewrite rule $L \rewritesto_{i,j} R$ and $G$, and
find the remainder of the diagram.  We start by finding a
\textit{matching}, a map $m : L \rightarrow G$ (see
\eqref{eq:sg-rw-rule}) such that $i$ and $m$ can be extended to a
pushout; ie: they have a \textit{pushout complement}.

\begin{definition}[Pushout complement]
    A \textit{pushout complement} for a pair of arrows
    \[ A \xrightarrow{f} B \xrightarrow{g} C \]
    is an object $B'$ and a pair of arrows
    \[ A \xrightarrow{f'} B' \xrightarrow{g'} C \]
    such that
    \[
        \begin{tikzcd}
            A \arrow{r}{f} \arrow[swap]{d}{f'}
            & B \arrow{d}{g}
            \\ B' \arrow[swap]{r}{g'}
            & C \NWbracket
        \end{tikzcd}
    \]
    is a pushout.
\end{definition}

The following result is due to Dixon and Kissinger\cite{Dixon2010a},
based on a similar result for adhesive categories in \cite{Lack2005}.

\begin{proposition}
    If a pair of arrows $(b,f)$ in $\catSGraph_T$, where $b$ is monic,
    have a pushout complement, it is unique up to isomorphism.
\end{proposition}

\begin{definition}[String graph matching]
    For a string graph $G$ and a string graph rewrite rule $L
    \xleftarrow{i} I \xrightarrow{j} R$, a monomorphism $m : L
    \rightarrow G$ is called a \textit{string graph matching} if $I
    \xrightarrow{i} L \xrightarrow{m} G$ has a pushout complement.
\end{definition}

We require a stronger form of arity-matching in order to guarantee
pushout complements:
\begin{definition}\label{def:sg-local-iso}
    A map $f$ between $\mathcal{G}_T$-typed graphs $G$ and $H$ is a
    \textit{local isomorphism} if for every $v \in N(G)$, the
    restriction of $f$ to the edge neighbourhood of $v$ is a bijection
    onto the edge neighbourhood of $f(v)$.
\end{definition}

\begin{proposition}[Kissinger]
    Let $L \rewritesto R$ be a string graph rewrite rule.  Then any
    monic local isomorphism $m : L \rightarrow G$ is a string graph
    matching.
\end{proposition}

So if we find a monomorphism $m : L \rightarrow G$, we can assemble the
first pushout square of \eqref{eq:sg-rw-rule}: since $i$ and $m$ are
monic, we can view them as expressing the subgraph relation.  Then $D$
is obtained by removing from $G$ everything in the image of $m$ but not
in the image of $m \circ i$.  But we still need to be able to form the
second pushout.

\textit{Boundary coherence} proves to be the right notion for
determining when we can form pushouts in $\catSGraph_T$.

\begin{definition}\label{def:sg-bound-coher}
    A span of morphisms $G \xleftarrow{g} I \xrightarrow{h} H$ in
    $\catSGraph_T$ is called \textit{boundary-coherent} if:
    \begin{enumerate}
        \item for all $v \in \In(I)$, at least one of $g(v)$ and $h(v)$
            is an input; and
        \item for all $v \in \Out(I)$, at least one of $g(v)$ and $h(v)$
            is an output.
    \end{enumerate}

    A single morphism $f : G \rightarrow H$ is called boundary-coherent
    if the span $H \xleftarrow{f} G \xrightarrow{f} H$ is.
\end{definition}

\begin{theorem}[Kissinger]\label{thm:sg-bcoh-pushout}
    A span of morphisms $G \xleftarrow{g} I \xrightarrow{h} H$ in
    $\catSGraph_T$ has a pushout if and only if it is boundary-coherent.
\end{theorem}

Our definition of a string graph rewrite rule means that the boundary
coherence of the maps in the left-hand pushout propagates to the maps on
the right, meaning that if we can construct one of the pushouts in
\eqref{eq:sg-rw-rule}, we can always construct the other one.

\begin{lemma}\label{lem:sg-bc-rw}
    Let $L \xleftarrow{i} I \xrightarrow{j} R$ be a string graph rewrite
    rule and $f : I \rightarrow G$ be a string graph morphism.  Then $i$
    and $f$ as a span are boundary-coherent if and only if $f$ and $j$
    are.
\end{lemma}
\begin{proof}
    This follows directly from the fact that, for all vertices $w$ in
    $I$, $i(w)$ is an input (resp. output) of $L$ if and only if $j(w)$
    is an input (resp. output) of $R$.
\end{proof}

\begin{corollary}
    Let $L \rewritesto R$ be a string graph rewrite rule and $m : L
    \rightarrow G$ a monic string graph local isomorphism.  Then $L
    \rewritesto R$ has an $S$-rewrite at $m$.
\end{corollary}

Constructing the pushout is easy: pushouts along monomorphisms in
$\catGraph$ (and its slice categories) are just graph unions (with an
explicit intersection).  Likewise, pushout complements from
monomorphisms can be calculated using graph subtraction.  The above
results guarantee that we will get valid string graphs when we do so.

\begin{example}
    The following string graph rewrite is an application of the string
    graph rewrite rule for the unit law in the Z/X calculus:
    \[ \input{x-unit-rewrite.tikz} \]
\end{example}

\subsection{Equational Reasoning with Rewrites} 
\label{sec:eq-reas-rw}

Rewriting is typically done with respect to some set of rewrite rules
$S$, called a \textit{(graph) rewrite system}.  We say that a
string graph $G$ rewrites to a string graph $H$ under the graph rewrite
sytem $S$ (or $G \rewritestoin{S} H$) if there is a string graph rewrite
rule $L \rewritesto R \in S$ and a string graph matching $m : L
\rightarrow G$ such that
\[
    \begin{tikzcd}
        L \arrow[swap,right hook->]{d}{m}
        & I \arrow[swap]{l}{i} \arrow{r}{j} \arrow[right hook->]{d}
        & R \arrow[right hook->]{d}
        \\ G \NEbracket
        & D \arrow[left hook->]{l} \arrow[right hook->]{r}
        & H \NWbracket
    \end{tikzcd}
\]
is a string graph rewrite.

If we define
\[ S = \{G \rewritesto_{i,j} H | G \rweq_{i,j} H \in E\} \]
then $\rewritestoin{S}$ implements the \textsc{Leibniz} and
\textsc{Axiom} rules of the previously-presented equational logic for
string graphs.  In particular, string graph rewriting implements a
single application of \textsc{Axiom} followed by a single application of
\textsc{Leibniz} (which, by proposition \ref{prop:leibniz-compact}, is
equivalent to any proof tree containing just \textsc{Axiom} and
\textsc{Leibniz} axioms).

It then follows that $\rewriteequivin{S}$, the reflexive,
symmetric, transitive closure of $\rewritestoin{S}$, implements
the previously presented equational logic with \textsc{Homeo} omitted.

\subsection{Matching up to Wire Homeomorphism} 
\label{sec:sg-match-wire-homeo}

There is one more rule we need to account for in our mechanisation
process: \textsc{Homeo}.  Ideally, we would like to do rewriting ``up to
wire homeomorphism''.

Given a string graph rewrite rule $L \rewritesto R$ and a graph $G$, we
say that $G$ rewrites to $H$ by $L \rewritesto R$ \textit{up to wire
homeomorphism} if there are $G'$, $H'$ and $L' \rewritesto R'$, wire
homeomorphic to $G$, $H$ and $L \rewritesto R$ respectively, such that
$G'$ rewrites to $H'$ by $L' \rewritesto R'$.

We would like to find the set $\mathcal{H}$ of graphs $H$ such that $G$
rewrites to $H$ by $L \rewritesto R$ up to wire homeomorphism,
quotiented out by wire homeomorphism (ie: so there are no two graphs in
$\mathcal{H}$ wire homeomorphic to each other).

We can eliminate duplicate wire homeomorphic graphs as a last step
by settling on some standard form (such as have at most one wire-vertex
on each wire other than bare wires, where we must have two).  However,
this does not help to ensure we have found all possible matches without
trying the (likely infinite number of) graphs wire homeomorphic to $G$.
We demonstrate a technique for constraining the search space at the
matching step by making a careful choice of $G'$ and $L' \rewritesto
R'$.

Any monomorphism from $L$ to $G$ must necessarily map wires of $L$ onto
wires of $G$ of the same type in the following manner:
\begin{itemize}
    \item A circle of $L$ must map to a circle of $G$, in which case
        nothing else can map to that same circle
    \item An interior wire of $L$ can only map to an interior wire of
        $G$, and this must be the only thing mapping to the wire.
    \item A bare wire of $L$ can map to \textbf{any} wire of $G$, and
        may not be the only thing to map to that wire
    \item An input wire of $L$ can map to either an input wire or an
        interior wire of $G$; in both cases, it may not be the only
        thing mapping to that wire
    \item Output wires are similar to input wires
\end{itemize}

Note that multiple bare wires can match the same wire simultaneously
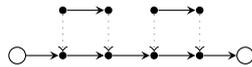
\[ \input{double-bare-wire-matching.tikz} \]
as well as matching an interior wire that has already been matched by an
input and/or an output wire
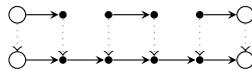
\[ \input{input-bw-output-onto-interior.tikz} \]
or an input wire that has been matched by an input wire
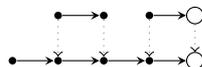
\[ \input{input-and-bw-onto-input.tikz} \]
or an output wire that has been matched by an output wire.

Let $n$ be the number of bare wires in $L$.  We choose $L'$ to be the
(unique, up to isomorphism) graph that is wire homeomorphic to $L$ such
that
\begin{itemize}
    \item every circle has $\textrm{max}(2n,1)$ wire-vertices on it
    \item every interior wire has exactly $2n + 2$ wire-vertices on it
    \item every bare wire has no wire-vertices other than the input
        and output
    \item every input or output wire has no wire-vertices other than
        the input or output itself
\end{itemize}
We then choose $G'$ to be the (unique, up to isomorphism) graph that is
wire homeomorphic to $G$ such that
\begin{itemize}
    \item every circle has $\textrm{max}(2n,1)$ wire-vertices on it
    \item every interior wire has exactly $2n + 2$ wire-vertices on it
    \item every bare wire has exactly $2n$ wire-vertices on it
        (including boundary vertices)
    \item every input or output wire has exactly $2n + 1$ wire-vertices
        on it (including the boundary vertex)
\end{itemize}
We say that these graphs are \textit{normalised}.

Any algorithm for finding graph monomorphisms (such as those in
\cite{Ullmann1976} or \cite{Akinniyi1986}) can then be used (together
with a suitable filter) to find local isomorphisms from $L'$ to $G'
\graphminus \Bound(G')$.  Now consider an arbitrary wire $w$ of $G'$.
\begin{itemize}
    \item If $w$ is a circle, it can be matched by any circle of the
        same type in $L'$, since circles in both graphs have the same
        number of wire-vertices.  Alternatively, if $n > 0$, it can be
        matched by any or all of the $n$ bare wires of $L'$.
    \item If $w$ is an interior wire, it can either be matched
        by one of the interior wires of $L'$ (which have the same number
        of wire-vertices) or it can be matched by any combination of an
        input wire (consuming a single wire-vertex at one end), an
        output wire (similarly) and any or all of the bare wires (on the
        remaining $2n$ wire-vertices).
    \item If $w$ is a bare wire, it can be matched by any or all of the
        $n$ bare wires of $L'$.
    \item If $w$ is an input wire, it can be matched by an input of
        $L'$, which will consume one of the $2n+1$ wire-vertices; this
        will still allow any or all of the bare wires of $L'$ to match.
    \item Finally, if $w$ is an output wire, it can similarly be matched
        by an output of $L'$ and/or any or all of the bare wires of $L'$.
\end{itemize}

Rewriting using all of the matches found in this way will produce a
finite collection $\mathcal{H}'$ of rewritten graphs.  Contracting all
the wires in these graphs so that there is only a single wire-vertex on
each, except for bare wires which need two wire-vertices, and
eliminating the duplicates will give us the set $\mathcal{H}$ we
originally wanted.

\begin{remark}
    The proliferation of wire-vertices and edges between them in $L'$
    and $G'$ when $L$ has bare wires means that this is not the most
    efficient way of doing matching up to wire homeomorphism.
    Quantomatic takes a slightly different approach, initially assuming
    $n = 0$ and delaying matching bare wires from $L'$ until everything
    else has been matched; each bare wire is then matched against an
    edge of $G'$, which is replaced by two wire-vertices and three edges
    at this point.  See section \ref{sec:quanto-matching} for details.
\end{remark}

\begin{example}
    The graph
    \[ \input{graph-norm-orig.tikz} \]
    would be normalised to
    \[ \input{graph-norm-pattern.tikz} \]
    if it were the graph $L$ from the rule and to
    \[ \input{graph-norm-target.tikz} \]
    if it were the target graph $G$.
\end{example}


\chapter{$!$-Graphs}
\label{ch:bang-graphs}

String graph rewrite systems have a shortcoming in that they do not have
a finite way of representing infinitary rewrite rules.  In this chapter,
we present $!$-graphs (pronounced ``bang graphs''), an extension of
string graphs, as a way of expressing certain infinite families of
string graphs in a finite way.  In the next chapter, we will use these
to construct infinite families of string graph rewrite rules and use
them to rewrite string graphs, both directly and by rewriting $!$-graphs.

Recall the spider laws of section \ref{sec:quantum} that arose from
commutative Frobenius algebras (CFAs):
\begin{equation}\label{eq:spider-law-informal}
    \input{spider-law-informal.tikz} 
\end{equation}

In a single diagram, this describes an infinite family of equations.
However, while its meaning may be clear to a human, we need a more
precise encoding for a proof assistant to be able to handle equations
like this.  In \cite{Dixon2009}, Dixon and Duncan briefly described a
system of \textit{graph patterns} to handle laws like this.  However,
they provided no mathematical foundation for these.  Kissinger, Merry
and Soloviev provided such a foundation in \cite{Kissinger2012a} (as
\textit{pattern graphs}), making the system more powerful in the process.
This joint work with the author of this thesis, which contained no
proofs, forms the basis of this chapter, although we use the term
\textit{$!$-graphs} instead of pattern graphs to avoid confusion when
describing matching (where the two graphs are often called the pattern
and target graphs).

$!$-graphs formalise the notion of repetition indicated by the ellipses
and encode it in the graph itself.  What will actually be encoded is a
collection of labelled subgraphs called $!$-boxes.  A $!$-graph can be
instantiated into a string graph by removing (``killing'') or
duplicating (``copying'') these subgraphs (including incoming and
outgoing edges) some (finite) number of times.  In this way, a $!$-graph
can stand in the place of an infinite family of string graphs:
\begin{equation*}
    \input{bang-graph-ex.tikz} 
\end{equation*}

We can even express the bipartite graph from the generalised bialgebra
law
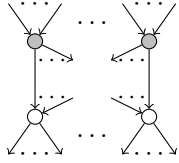
\[ \input{gen-bialg-attempt-lhs.tikz} \]
by nesting $!$-boxes inside other $!$-boxes:
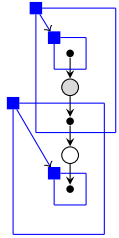
\[ \input{gen-bialg-graph-lhs.tikz} \]

Section \ref{sec:nesting-overlapping} goes into detail about the effect
of nested and overlapping $!$-boxes.

\section{Open Subgraphs} 
\label{sec:open-subgraphs}

Note: when we refer to a subgraph of a string graph in this section, we
mean a subgraph in $\catSGraph_T$ -- ie: the subgraph is also a string
graph.

It should be clear that we cannot allow $!$-boxes to be arbitrary
subgraphs.  Copying $b_1$ in the following graph, for example, will lead
to an invalid string graph, as there will be a wire vertex with multiple
outputs:
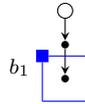
\begin{equation*}
    \input{bad-bbox.tikz}
\end{equation*}
Copying or killing a $!$-box must also not affect the fixed-arity edges
of node-vertices, otherwise the arity-matching property of the typing
morphism will be broken.

So we must restrict our $!$-boxes to those that can be copied.  For
string graphs, the correct notion is that of \textit{open subgraphs}.
Note that the definition here differs from, but is equivalent to, the
one in \cite{Kissinger2012a} (pp 5).

\begin{definition}[Open Subgraph]
    A subgraph $O$ of a string graph $G$ is said to be \textit{open} if
    it is not adjacent to any wire-vertices or incident to any
    fixed-arity edges.
\end{definition}

This is also related to boundary-coherence in \catSGraph:
\begin{proposition}\label{prop:sg-open-iff-bc}
    Let $G$ be a string graph, and $H$ a subgraph of $G$ in
    $\catSGraph_T$.  Then $H$ is open if and only if it has a complement
    $G\graphminus H$ (in $\catSGraph_T$), and the embedding of this
    complement into $G$ is boundary-coherent.
\end{proposition}
\begin{proof}
    Suppose $(H,\tau_H)$ is open in $(G,\tau_G)$, and consider
    $(G\graphminus H,\tau_{G\graphminus H})$, which is a graph in
    $\catGraph/\mathcal{G}_T$.  Let $n$ be a node-vertex in
    $G\graphminus H$, and let $\tau^n_{G\graphminus H}$ be the
    restriction of $\tau_{G\graphminus H}$ to the fixed edge
    neighbourhood of $n$, and similarly for $\tau^n_G$. Since
    $G\graphminus H$ is a subgraph of $G$, $\tau_{G\graphminus H}$ is a
    restriction of $\tau_G$.  So $\tau^n_{G\graphminus H}$ is a
    restriction of $\tau^n_G$, and hence is injective and its image
    consists only of fixed-arity edges of the typegraph.  Suppose it is
    not surjective, so there is some fixed-arity edge $e$ incident to
    $\tau_{G\graphminus H}(n)$ that is not in the image of
    $\tau_{G\graphminus H}$.  Then the preimage of $e$ under $\tau_G$,
    $e'$, must not be in $G\graphminus H$.  Since $n$ is in
    $G\graphminus H$, and $G\graphminus H$ is full in $G$, the other end
    of $e'$, which we will call $w$, must be in $H$.  But then $H$ is
    incident to a fixed-arity edge (namely $e'$), and hence is not open
    in $G$.  So $G\graphminus H$ is a string graph, since (being a
    subgraph of $G$) it must have at most one incoming and one outgoing
    edge on each wire-vertex.

    For boundary-coherence of the embedding $\iota : G\graphminus H
    \rightarrow G$, we need to show that any input of $G\graphminus H$
    is also an input of $G$, and likewise for outputs.  But the only way
    an input $w$ in $G\graphminus H$ could not be an input in $G$ is if
    there is an edge $e$ in $G$ from some vertex $v$ to $w$.  Because
    $e$ is not in $G\graphminus H$, neither can $v$ be.  So $v$ must be
    in $H$, and so $H$ is incident to a wire-vertex ($w$), contradicting
    the openness of $H$.  The argument for outputs is analagous, and so
    we have boundary-coherence as required.

    Conversely, suppose $(H,\tau_H)$ is a (string) subgraph of
    $(G,\tau_G)$, and that its complement is also a string graph with a
    boundary-coherent embedding $\iota$ into $G$.  Suppose $H$ is
    adjacent to a wire-vertex $w$, via an edge $e$.  Then $w$ is in
    $G\graphminus H$, but $e$ is not.  If $w$ is the target of $e$, $w$
    must be an input of $G\graphminus H$ but not of $G$.  Otherwise, $w$
    is the source of $e$ and $w$ is an output of $G\graphminus H$ but
    not of $G$.  Either way, $\iota$ cannot be boundary-coherent.  So
    there is no such $w$.

    Now suppose $H$ is incident to a fixed-arity edge $e$.  Let $v$ be
    the end of $e$ that is not in $H$.  $v$ cannot be a wire-vertex,
    since we have already shown that $H$ cannot be adjacent to a
    wire-vertex.  So $v$ is a node-vertex.  Let $e' = \tau_G(e)$.  Then
    $e'$ is in the fixed neighbourhood of $\tau_G(v)$, but is not in the
    image of $\tau_{G\graphminus H}$.  So $\tau_{G\graphminus H}$ is not
    arity-matching, contradicting the assumption that $(G\graphminus
    H,\tau_{G\graphminus H})$ is a string graph.  Thus there can be no
    such edge $e$, and $H$ is open in $G$.
\end{proof}

The use of the term ``open'' brings with it the expectation that the
property is preserved under unions and intersections (since we are
dealing with finite graphs).  The following proposition shows that this
is, indeed, the case.

\begin{proposition}[\cite{Kissinger2012a}, pp 5]
    \label{prop:open-subgraph-props}
    If $O,O' \subseteq G$ are open subgraphs, and $H \subseteq G$ is an
    arbitrary subgraph, then $O \cap O'$ and $O \cup O'$ are open in $G$
    and $H \cap O$ is open in $H$.
\end{proposition}
\begin{proof}
    Proposition \ref{prop:sg-subgraph-ops} gives us that the graphs are
    all string graphs.

    Suppose $O \cap O'$ is adjacent to a wire-vertex.  Then there is a
    wire-vertex $w$ in $G$ but not in $O \cap O'$ adjacent to a vertex
    $v$ in $O \cap O'$.  If $w$ is not in $O$, we have a contradition
    for the openness of $O$.  So $w$ must be in $O$.  Similarly, $w$
    must be in $O'$.  But then $w$ is in $O \cap O'$, which contradicts
    our assumption.  So there can be no such $w$.  Instead, suppose $O
    \cap O'$ is incident to a fixed-arity edge $e$.  Again, $e$ must be
    in $O$, since $O$ cannot be incident to a fixed-arity edge, and must
    also be in $O'$ for the same reason, and hence is in $O \cap O'$,
    and so not incident to it.  So there is no such $e$ and $O \cap O'$
    is open in $G$.

    Suppose $O \cup O'$ is adjacent to a wire-vertex.  Then there is a
    wire-vertex $w$ in $G$ but not in $O \cup O'$ adjacent to a vertex
    $v$ in $O \cup O'$.  $v$ must either be in $O$ or $O'$; WLOG assume
    it is in $O$.  Then, since $w$ is not in $O$, $O$ is adjacent to a
    wire-vertex.  But this cannot be the case, since $O$ is open, so
    there can be no such $w$.  Similarly, if there is a fixed-arity edge
    $e$ incident to $O \cup O'$, it must be incident to either $O$ or
    $O'$, which cannot be the case.  So there is no such $e$, and $O
    \cup O'$ is open in $G$.

    Suppose there is a wire-vertex $w$ in $H$ but not in $H \cap O$
    adjacent to a vertex $v$ in $H \cap O$.  Then $w$ is in $G$ but not
    in $O$, and $v$ is in $O$, and they are joined by an edge in $G$.
    So $O$ is adjacent to a wire-vertex in $G$, contradicting its
    openness in $G$, so there is no such $w$.  Suppose now that $H \cap
    O$ is incident to a fixed-arity edge $e$ in $H$, and let $v$ be the
    vertex in $H \cap O$ it is incident to.  Then $v$ is in $O$, and $e$
    is in $G$ but not $O$, and so $O$ cannot be open in $G$, which
    contradicts the assumption.  So there can be no such $e$, and $H
    \cap O$ is open in $H$.

\end{proof}

\section{The Category of $!$-Graphs} 
\label{sec:pg-defs}

The $!$-boxes are encoded into the graph itself as a new vertex type.
These \textit{$!$-vertices} act as the labels for the $!$-boxes, and the
edges from the $!$-vertices to other graph vertices indicate the $!$-box
contents.  We therefore need to extend the derived typegraph of a
monoidal signature $T$ to include this new vertex.

\begin{definition}[Derived Compressed $!$-typegraph]
    \label{def:pg-typegraph}
    The \textit{derived compressed $!$-typegraph} $\mathcal{G}_{T!}$
    of a compressed monoidal signature $T$ has the same vertices as
    $\mathcal{G}_T$ (see definition \ref{def:sg-monoidal-typegraph})
    together with a distinct vertex $!$, and the same edges as
    $\mathcal{G}_T$
    together with a self-loop $p_!$ on $!$ and, for each vertex $v$ of
    $\mathcal{G}_T$, an edge $p_v$ from $!$ to $v$.
\end{definition}

Note that the definition of $\mathcal{G}_{T!}$ ensures that there can be
no edges from wire- or node-vertices to $!$-vertices.

The derived compressed $!$-typegraph of the compressed monoidal
signature
\begin{align*}
    f &: [A^\varar] \rightarrow [A^\fixar,A^\fixar] \\
    g &: [B^\varar,A^\fixar] \rightarrow [A^\fixar]
\end{align*}
would be
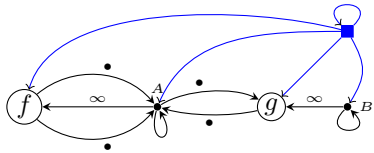
\[ \input{example-bang-typegraph.tikz} \]

The wire-vertex and node-vertex terminology is inherited from string
graphs. In addition, for a graph $(G,\tau:G \rightarrow
\mathcal{G}_{T!})$ in $\catGraph/\mathcal{G}_{T!}$, we call any vertex
$v$ where $\tau(v) =\:!$ a \textit{$!$-vertex}, and denote the set of
all $!$-vertices in $G$ as $!(G)$.

We alter the definition of an \emph{input} slightly from the
string-graph case, due to the new vertex type: a wire-vertex is an input
if the only in-edges are from $!$-vertices.

For a $!$-vertex $b \in \:!(G)$, let $B(b)$ be its associated $!$-box.
This is the full subgraph whose vertices are the set $\textrm{succ}(b)$
of all of the successors of $b$. We also define the parent graph of a
$!$-vertex $B^\uparrow(b)$ as the full subgraph of predecessors, that
is, the full subgraph generated by $\textrm{pred}(b)$.  Note that the
typegraph constrains $B^\uparrow(b)$ to only contain $!$-vertices (and
edges between them).

As a notational convenience, when dealing with subgraphs of
$\mathcal{G}_{T!}$-typed graphs where the inclusion map is implicit, we
will use subscripts to identify which graph we are referring to.  So if
$H$ is a subgraph of $G$ and $b$ is a $!$-vertex in $H$, $B_H(b)$ will
be its $!$-box in $H$, and $B_G(b)$ will be the $!$-box of $b$ when
considered as a vertex of $G$.

Note that, for a given monoidal signature $T$, $\mathcal{G}_T$ is a
subgraph of $\mathcal{G}_{T!}$.  Since \catGraph has pullbacks, theorem
\ref{thm:slice-cat-mono-func} gives us a full, coreflective embedding $E
: \catGraph/\mathcal{G}_T \rightarrow \catGraph/\mathcal{G}_{T!}$, whose
right adjoint is the forgetful functor $U : \catGraph/\mathcal{G}_{T!}
\rightarrow \catGraph/\mathcal{G}_T$ that simply drops all $!$-vertices.

\begin{definition}[$!$-Graph and $\catBGraph_T$; \cite{Kissinger2012a}, pp 5]
    \label{def:bang-graph}
    A $\mathcal{G}_{T!}$-typed graph $G$ is called a \textit{$!$-graph}
    if:
    \begin{enumerate}
        \item $U(G)$ is a string graph;
        \item the full subgraph with vertices $!(G)$, denoted
            $\beta(G)$, is posetal;
        \item for all $b \in\, !(G)$, $U(B(b))$ is an open subgraph of
            $U(G)$; and
        \item for all $b,b'\! \in\, !(G)$, if $b'\! \in B(b)$ then
            $B(b') \subseteq B(b)$.
            \label{def-item:bg-bbox-nesting}
    \end{enumerate}

    $\catBGraph_T$ is the full subcategory of
    $\catGraph/\mathcal{G}_{T!}$ whose objects are $!$-graphs.
\end{definition}

Recall that a graph is posetal if it is simple (at most one edge between
any two vertices) and, when considered as a relation, forms a partial
order. Note in particular that this implies $b \in B(b)$ (and
$B^\uparrow(b)$), by reflexivity. This partial order allows $!$-boxes to
be nested inside each other, provided that the subgraph defined by a
nested $!$-vertex is totally contained in the subgraph defined by its
parent (condition \ref{def-item:bg-bbox-nesting}).

As with string graphs, we only consider finite $!$-graphs in the sequel.

We introduce special notation for $!$-graphs. $!$-vertices are drawn
as squares, but rather than drawing edges to all of the node-vertices
and wire-vertices in $B(b)$, we simply draw a box around it.
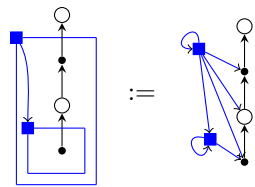
\begin{equation*}
    \input{nested_def.tikz} 
\end{equation*}

In this notation, we retain edges between distinct $!$-vertices to
indicate which $!$-boxes are nested as opposed to simply overlapping.
This distinction is important, as nested $!$-boxes are copied whenever
their parent is copied.
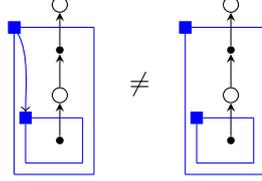
\begin{equation*}
    \input{nested_v_overlap.tikz} 
\end{equation*}

We will prove a few useful results about $!$-graphs and about
$\catBGraph_T$.  Firstly, we have a sufficient condition for a subgraph
of a $!$-graph to be a $!$-graph.

\begin{lemma}\label{lem:pg-full-subgraph}
    Let $G$ be a $!$-graph and $H$ a full subgraph of $G$ such that
    $U(H)$ is a string graph.  Then $H$ is a $!$-graph.
\end{lemma}
\begin{proof}
    $U(H)$ is a string graph by assumption.

    Since $H$ is full in $G$, $\beta(H)$ must be a full subgraph
    of $\beta(G)$.  But a full subgraph of a posetal graph is itself
    posetal, so $\beta(H)$ is posetal.

    Let $b \in\:!(H)$.  Then $U(B_G(b))$ is an open subgraph of $U(G)$,
    and $U(B_H(b)) = U(B_G(b))\cap U(H)$, which is open in $H$ by
    proposition \ref{prop:open-subgraph-props}.

    Let $b,c \in\:!(H)$, with $c \in B_H(b)$.  Then we must have $c \in
    B_G(b)$, since $G$ is a $!$-graph.  So $B_G(c) \subseteq B_G(b)$.
    Since $H$ is full in $G$, we know that $B_H(b) = B_G(b) \cap H$ and
    similarly for $c$. Then $B_G(c) \cap H \subseteq B_G(b)
    \cap H$, so $B_H(c) \subseteq B_H(b)$.

    So $H$ is a $!$-graph, as required.
\end{proof}

\begin{corollary}\label{cor:pg-graphminus}
    Let $G$ be a $!$-graph, and $H$ a subgraph of $G$ such that
    $U(H)$ is open in $U(G)$.  Then $G \graphminus H$ is a $!$-graph.
\end{corollary}

It is also useful to know under what conditions a $!$-graph morphism is
monic or (strongly) epic in $\catBGraph_T$.  These conditions turn out
to be the same as for $\catSGraph_T$.

\begin{lemma}\label{lemma:spg-mono-inj}
    A morphism in $\catBGraph_T$ is monic iff it is injective.
\end{lemma}
\begin{proof}
    Since this holds in $\catGraph/\mathcal{G}_{T!}$, any injective map in
    $\catBGraph_T$ must be monic in $\catGraph/\mathcal{G}_{T!}$, and hence
    also monic in $\catBGraph_T$.

    Suppose we have a non-injective morphism $f: G \rightarrow H$ in
    $\catBGraph_T$.  Then $f$ must either map two or more edges in
    $G$ to the same edge in $H$ or map two or more vertices in $G$ to
    the same vertex in $H$.  In fact, since $G$ and $H$ are both simple,
    $f$ must map two or more vertices $v_i$ in $G$ to a single vertex
    $v$ in $H$.  Consider the smallest sub-$!$-graph $K$ of $H$
    containing $f(v_i)$.  This is the vertex itself, plus a self-loop if
    it is a $!$-vertex or any fixed-arity edges (plus their terminating
    wire-vertices) if it is a node-vertex.  Then for each $v_i$, we
    construct the $!$-graph morphism $g_i$ that takes the single vertex
    in $K$ to $v_i$ (this will extend to any incident edges of the
    vertex in $K$ in a unique way, dictated by the typing morphisms).
    Now all the $f \circ g_i$ are the same morphism, but the $g_i$
    morphisms are distinct, so $f$ is not monic.
\end{proof}

\begin{lemma}
    \label{lemma:pattgraph-morphism-img}
    Let $f : G \rightarrow H$ be a $!$-graph morphism.  Then
    $f[G]$, the image of $f$, is a $!$-graph.
\end{lemma}
\begin{proof}
    First we note that $f[G]$ is a subgraph of $H$, and let $\iota :
    f[G] \rightarrow H$ be the inclusion.  Then $U(f[G])$ is a subgraph
    of the string graph $U(H)$.  It therefore satisfies the wire-vertex
    conditions for string graphs, and we only need check that $\tau_H$
    restricts to an arity-matching morphism.  But we know that $f$ is
    arity-matching, and hence restricts to a bijection on the
    fixed-arity neighbourhood of every node-vertex $v$ in $G$ and its
    image in $H$, and $\tau_G$ is also arity-matching, and $\tau_G$ and
    $\tau_H \circ f$ commute.  So $\tau_H \restriction_{f[G]}$ must be
    arity-matching, and hence $U(f[G])$ must be a string graph.

    $\beta(f[G])$ is a subgraph of $\beta(H)$.  But then $\beta(f[G])$
    is simple and anti-symmetric, since $H$ is, and is reflexive and
    transitive, since $G$ is, and hence is posetal.

    Let $b' \in !(f[G])$.  Then there is a $b \in !(G)$ with $\iota(b')
    = f(b)$.  Now $B(b') = B(f(b)) \cap f[G]$, which is open in $f[G]$
    by proposition \ref{prop:open-subgraph-props}.

    Let $b,c \in !(f[G])$, with $c \in B(b)$.  Then we must have
    $\iota(c) \in B(\iota(b))$.  So $B(\iota(c)) \subseteq B(\iota(b))$.
    But then $B(b) = B(\iota(b)) \cap f[G]$ and similarly for $c$, and
    $B(\iota(c)) \cap f[G] \subseteq B(\iota(b)) \cap f[G]$, so $B(c)
    \subseteq B(b)$.

    So $f[G]$ is a $!$-graph, as required.
\end{proof}

\begin{lemma}\label{lemma:spg-epi-surj}
    A morphism in $\catBGraph_T$ is a strong epimorphism iff it is
    surjective.
\end{lemma}
\begin{proof}
    Since this holds in $\catGraph/\mathcal{G}_{T!}$, any surjective map
    in $\catBGraph_T$ must be a strong epimorphism in
    $\catGraph/\mathcal{G}_{T!}$, and hence a strong epimorphism in
    $\catBGraph_T$.

    Now suppose $e : A \rightarrow B$ is a strong epimorphism.  Let $e'
    : A \rightarrow e[A]$ be the restriction of $e$ to its image (which
    is also a $!$-graph by lemma \ref{lemma:pattgraph-morphism-img})
    and let $\iota : e[A] \rightarrow B$ be the inclusion of the image
    of $e$ in $B$.  Then $\iota$ is monic, so there exists a map $d$
    making the following diagram commute:
    \begin{center}
        \begin{tikzpicture}[-latex]
            \matrix(m)[cdiag]{
            A    & B \\
            e[A] & B  \\};
            \path [arrs] (m-1-1) edge [->>] node {$e$} (m-1-2)
                         (m-2-1) edge [right hook-latex] node [swap] {$\iota$} (m-2-2)
                         (m-1-1) edge [->>] node [swap] {$e'$} (m-2-1)
                         (m-1-2) edge node {$1_B$} (m-2-2)
                         (m-1-2) edge [dashed] node [swap] {$d$} (m-2-1);
        \end{tikzpicture}
    \end{center}
    But then $\iota \circ d = 1_B$, so $\iota$ must be surjective and
    hence so is $e$.
\end{proof}

Finally, it is worth noting how $\catSGraph_T$ and $\catBGraph_T$ are
related.

\begin{definition}[Concrete Graph]
    A $!$-graph with no $!$-vertices is called a \textit{concrete
    graph}.
\end{definition}

The full subcategory of $\catBGraph_T$ consisting of concrete graphs is,
in fact, the image of $\catSGraph_T$ under the previously-mentioned
embedding functor $E$.  Concrete graphs and string graphs will therefore
be considered interchangable.

Indeed, the restriction of $E$ to $\catSGraph_T$ is a full, coreflective
embedding, and its right adjoint is the restriction of $U$ to
$\catBGraph_T$.  As a result, we will also use $E$ and $U$ to refer to
these restrictions; which is meant should be clear from context.  We use
$\Sigma$ to mean $E \circ U$ (and note that $\catBGraph_T$ is closed
under $\Sigma$).

We can extend the $\beta(G)$ notation of definition \ref{def:bang-graph}
to morphisms of $\catGraph/\mathcal{G}_{T!}$ (and hence of
$\catBGraph_T$) by making their operation be the obvious restrictions.
More precisely, if $\mathcal{G}_!$ is the subgraph of $\mathcal{G}_{T!}$
consisting of only the $!$-vertex and its self-loop (with $\iota_T$
being its inclusion map), we can view $\beta$ as $E_{\iota_T} \circ
U_{\iota_T}$ (in the terminology of theorem
\ref{thm:slice-cat-mono-func}).

\subsection{Wire Homeomorphisms}
\label{sec:bg-wire-homeo}

We will need to extend our notion of wire-homeomorphism to $!$-graphs.

A \textit{wire homeomorphism} $f : G \sim H$ between two $!$-graphs
$G$ and $H$ consists of four bijective type-preserving functions
\begin{itemize}
    \item $f_N : N(G) \leftrightarrow N(H)$
    \item $f_! :\; !(G) \leftrightarrow !(H)$
    \item $f_B : \Bound(G) \leftrightarrow \Bound(H)$
    \item $f_W : \Wires(G) \leftrightarrow \Wires(H)$
\end{itemize}
such that
\begin{itemize}
    \item for each wire $w$ in $\Wires(G)$ with a source $v$ in $N(G)$
        (resp. $\Bound(G)$), the source of $f_W(w)$ in $H$ is $f_N(v)$
        (resp. $f_B(v)$), and similarly for targets of wires
    \item for each $b \in\:!(G)$, there is an edge in $G$ from $b$ to a
        vertex $v$ in $!(G)$ (resp. $N(G)$, $\Bound(G)$) if and only if
        there is an edge in $H$ from $f_!(b)$ to $f_!(v)$ (resp.
        $f_N(v)$, $f_B(v)$)
    \item for each $b \in\:!(G)$ and each wire $w \in \Wires(G)$, there
        is an edge in $G$ from $b$ to each $v \in w_V$ if and only if
        there is an edge in $H$ from $f_!(b)$ to each $v \in (f_W(w))_V$
\end{itemize}

This is the same as the definition for string graphs, but with a map for
$!$-boxes and some constraints to preserve $!$-box containment.  This
means that if $f : G \sim H$ is a wire-homeomorphism of $!$-graphs,
discarding $f_!$ will give a wire-homeomorphism $U(G) \sim U(H)$.

\section{The $!$-Box Operations} 
\label{sec:pg-instantiation}

We have already mentioned that $!$-boxes are intended to be subgraphs
that can be copied or deleted.  The precise set of operations allowed on
a $!$-box are the following (\cite{Kissinger2012a}, pp 6, although we
omit $\MERGE$ from that list and revisit it in section
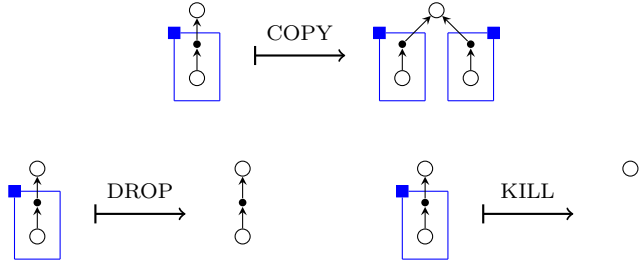
\ref{sec:bbox-merge}):
\[ \input{three-ops.tikz} \]

These are inspired by the rules for the ``bang'' operation from linear
logic\cite{Girard1987}.  In particular, $\COPY$ corresponds to
contraction, $\DROP$ to dereliction and $\KILL$ to weakening.

\begin{definitions}[$!$-Box Operations; \cite{Kissinger2012a}, pp 6]
    \label{def:bbox-ops}
    For $G$ a $!$-graph and $b\! \in\, !(G)$,
    the three $!$-box operations are defined as follows:
    \begin{itemize}
        \item $\COPY_b(G)$ is defined by a pushout of inclusions in
            $\catGraph/\mathcal{G}_{T!}$:
            \begin{equation}
                \label{eq:copy-pushout}
                \begin{tikzcd}
                    G\graphminus B(b) \rar[right hook->] \dar[right hook->]
                    & G \dar \\
                    G \rar & \COPY_b(G) \NWbracket
                \end{tikzcd}
            \end{equation}
        \item $\DROP_b(G) := G\graphminus b$
        \item $\KILL_b(G) := G\graphminus B(b)$
    \end{itemize}
\end{definitions}

Before we show that these operations preserve the property of being a
$!$-graph, it will be useful to adapt the notion of boundary-coherence
to $!$-graphs.  The important thing is that, in addition to the
restrictions on how the morphisms interact with the boundaries of the
graphs, we also need to restrict how they interact with $!$-box
containment.  Note that the definition of boundary-$!$-coherence we give
here is stronger than is actually needed for producing pushouts of
$!$-graphs; this gives us a simpler definition and simpler proofs.

\begin{definition}\label{def:pg-bound-coher}
    Let $f : G \rightarrow H$ be a $!$-graph morphism.  $f$ is said to
    \textit{reflect $!$-box containment} if whenever we have an edge $e$
    in $H$ whose source is a $!$-vertex and whose target is in the image
    of $f$, $e$ is also in the image of $f$.

    A span of monomorphisms $G \xleftarrow{g} I \xrightarrow{h} H$ in
    $\catBGraph_T$ is called \textit{boundary-$!$-coherent} if the
    span formed from $U(g)$ and $U(h)$ is boundary-coherent and both $g$
    and $h$ reflect $!$-box containment.

    As in $\catSGraph_T$, a single morphism $f : G \rightarrow H$ of
    string graphs is called boundary-$!$-coherent if the span $H
    \xleftarrow{f} G \xrightarrow{f} H$ is.
\end{definition}

The main purpose of this definition is to provide us with a sufficient
condition for a span to have a pushout in $\catBGraph_T$.

\begin{proposition}\label{prop:pg-bound-coh-p-pushout}
    If a span of monomorphisms $G \xleftarrow{g} I \xrightarrow{h} H$ in
    $\catBGraph_T$ is boundary-$!$-coherent, then it has a pushout in
    $\catBGraph_T$.
\end{proposition}
\begin{proof}
    Suppose $g$ and $h$ are boundary-$!$-coherent.  Since they are
    monomorphisms, they have a pushout in $\catGraph/\mathcal{G}_{T!}$:
    \begin{equation}\label{eq:pg-bc-po}
            \posquare{I}{H}{G}{D}{h}{p_g}{g}{p_h}
    \end{equation}
    It suffices to show that $D$ is a $!$-graph.

    The span formed by $U(g)$ and $U(h)$ is boundary-coherent, so the
    graph that results from pushing out this span in
    $\catGraph/\mathcal{G}_T$ is in $\catSGraph_T$ (theorem
    \ref{thm:sg-bcoh-pushout}), and so is a string graph.  But this is
    just $U(D)$, by proposition \ref{prop:slice-cat-func-pb}.

    Since all the morphisms are monic (and the square commutes), we can
    treat them as containment relations.  So we will consider $G$ and
    $H$ to be subgraphs of $D$, and $I$ to be their common subgraph.

    Now we note that if we have an edge $e_b$ from a $!$-vertex $b$ to a
    vertex $v$ in $D$, and another edge $e$ incident to $v$ in $D$,
    both edges must be in $G$ or else both edges must be in $H$.  To see
    this, suppose not.  Then one edge ($e_b$, say) must be in $G$ and
    the other ($e$) in $H$.  So $v$ must be in $I$, since it is in both
    $G$ and $H$.  But then, as $g$ reflects $!$-box containment, $e_b$
    must be in $I$ and hence in $H$.

    We need to show that $\beta(D)$ is posetal.  $\beta(D)$ must be
    simple, since if we have two edges from a $!$-vertex $b_1$ to a
    $!$-vertex $b_2$ in $D$, they must both be in $G$ or both in $H$ by
    the above reasoning, which cannot be the case as these are
    $!$-graphs.  So there are no such edges.  Anti-symmetry is shown in
    the same way by reversing one of the edges.  Reflexivity is
    inherited from $G$ and $H$.  Transitivity follows by considering
    that if we have the following pair of edges
    \[ \input{bg-bcoh-pushout-ex-trans.tikz} \]
    they must both be in $G$ or both in $H$, and hence there must be an
    edge from $b_1$ to $b_3$ inherited from that graph.  So $\beta(D)$
    is posetal.

    Let $b \in\:!(D)$, and suppose $U(B(b))$ is not open in $U(D)$.
    Then there is either a wire-vertex adjacent to $U(B(b))$, or a
    fixed-arity edge incident to it.  The former case is impossible,
    because if we have two edges such that
    \[ \input{bg-bcoh-pushout-ex-open.tikz} \]
    where $w$ is a wire-vertex, then $e_w$ and $e_b$ must both be in $G$
    or both be in $H$, and hence there must be an edge from $b$ to $w$.
    Similarly, if there is a fixed-arity edge $e$ incident to $v$,
    $e_b$ and $e$ must be both in $G$ or both in $H$, and hence there
    must be an edge from $b$ to the other end of $e$, and so $e$ is not
    incident to $U(B(b))$.

    Let $b,c \in\:!(D)$ with $c \in B(b)$, and let $v \in B(c)$.  We need
    to show that $v \in B(b)$.  But the edge from $b$ to $c$ and the
    edge from $c$ to $v$ must both be in $G$ or both be in $H$, and so
    there must be an edge from $b$ to $v$, as required.
\end{proof}

We now proceed to prove that all the $!$-box operations produce
$!$-graphs.  This theorem was stated, but not proved, in
\cite{Kissinger2012a} (pp 7).

\begin{theorem}\label{thm:instantiation}
    Let $G$ be a $!$-graph and $b \in\, !(G)$.  Then the
    $\mathcal{G}_{T!}$-typed graphs $\COPY_b(G)$, $\DROP_b(G)$ and
    $\KILL_b(G)$ are all $!$-graphs.
\end{theorem}
\begin{proof}
    Corollary \ref{cor:pg-graphminus} gives us that $G \graphminus B(b)$
    is a $!$-graph, and hence its embedding $\iota$ into $G$ is a
    $!$-graph monomorphism.  This embedding is boundary-$!$-coherent:
    since $U(B(b))$ is an open subgraph of $U(G)$, proposition
    \ref{prop:sg-open-iff-bc} gives us that $U(\iota)$ is a
    boundary-coherent morphism, and requirement
    \ref{def-item:bg-bbox-nesting} of definition \ref{def:bang-graph}
    ensures that $\iota$ reflects $!$-box containment.  So we know the
    pushout exists in $\catBGraph_T$ by proposition
    \ref{prop:pg-bound-coh-p-pushout}, and hence $\COPY_b(G)$ is a
    $!$-graph.

    The $\DROP$ and $\KILL$ cases follow from lemma
    \ref{lem:pg-full-subgraph} and corollary \ref{cor:pg-graphminus}
    respectively.
\end{proof}

Since one of the reaons for using string graphs is their ability to
model inputs and outputs, it is useful to know how these operations
affect the inputs and outputs of $!$-graphs.  Firstly, we extend
$\In(G)$, $\Out(G)$ and $\Bound(G)$ to include $!$-boxes in the
following way:

\begin{definition}\label{def:bang-graph-bounds}
    For a $!$-graph $G$, we define $\In_!(G)$ to be the full
    subgraph of $G$ containing the vertices $!(G)$ and $\In(U(G))$;
    $\Out_!(G)$ to be the full subgraph containing $!(G)$ and
    $\Out(U(G))$; and $\Bound_!(G)$ is the subgraph containing $!(G)$
    and $\Bound(U(G))$, and any edges with a source in $!(G)$ and a
    target in either $!(G)$ or $\Bound(U(G))$.
\end{definition}

These are all valid $!$-graphs.  Now we can show that the $!$-graph
operations affect the boundary in the expected manner.

\begin{lemma}
    \label{lemma:copy-maps-bounds}
    Let $G$ be a $!$-graph, and $b \in\:!(G)$, and consider the copy
    operation:
    \[ \posquare{G\graphminus B(b)}{G}{G}{\COPY_b(G)}{i_1}{p_2}{i_2}{p_1} \]

    Then all the maps in the pushout take inputs to inputs and outputs to
    outputs.
\end{lemma}
\begin{proof}
    The square is symmetric, and the cases for inputs mirrors the case for
    outputs, so we just consider inputs for $i_1$ and $p_1$.

    That $i_1$ maps inputs to inputs follows from the fact that it is
    boundary-$!$-coherent.

    Suppose $w \in \In(G)$, and consider $p_1(w)$.  If there is an edge
    $e$ in $\COPY_b(G)$ with $t(e) = p_1(w)$, then $e$ cannot be in the
    image of $p_1$, since $w$ is an input of $G$ and $p_1$ is injective.
    So $e$, and hence $p_1(w)$, must be in the image of $p_2$.  But then
    there must be a vertex $w'$ in $G\graphminus B(b)$ with $i_1(w') =
    w$ and $p_2(i_2(w'))$ = $p_1(w)$, since this is a pushout.  $w'$
    must be an input of $G\graphminus B(b)$, since if there were an
    incoming edge to $w'$, it would have to map to an incoming edge of
    $w$ in $G$, and there is no such edge.  So $i_2(w')$ must also be an
    input of $G$, since $i_2$ maps inputs to inputs.  Then $i_2(w')$
    cannot have an incoming edge, and so $e$ cannot be in the image of
    $p_2$.  So there is no such edge $e$, and $p_1(w) \in
    \In(\COPY_b(G))$.
\end{proof}

%

\begin{theorem}\label{thm:pg-ops-commute-bounds}
    Let $G$ be a $!$-graph, and $b \in\:!(G)$.  Then
    \begin{itemize}
        \item $\Bound_!(\COPY_b(G)) \cong \COPY_b(\Bound_!(G))$;
        \item $\Bound_!(\KILL_b(G)) \cong \KILL_b(\Bound_!(G))$; and
        \item $\Bound_!(\DROP_b(G)) \cong \DROP_b(\Bound_!(G))$
    \end{itemize}
    and similarly for $\In_!$ and $\Out_!$.
\end{theorem}
\begin{proof}
    For the $\COPY_b$ case, since $\Bound_!(G)$ is a $!$-graph, we can
    construct the pushout
    \[
    \begin{tikzcd}
        \Bound_!(G)\graphminus B(b)
            \arrow[right hook-latex]{r}
            \arrow[right hook-latex]{d}
        & \Bound_!(G) \arrow{d}
        \\
        \Bound_!(G) \arrow{r}
        & D \NWbracket
    \end{tikzcd}
    \]
    in $\catBGraph_T$, where $D = \COPY_b(\Bound_!(G))$.  This is a
    subgraph of $\COPY_b(G)$:
    \[
    \begin{tikzcd}
        \Bound_!(G)\graphminus B(b)
            \arrow[right hook-latex]{r}
            \arrow[right hook-latex]{d}
        & \Bound_!(G) \arrow{d} \arrow[right hook-latex]{dr}
        &
        \\
        \Bound_!(G) \arrow{r} \arrow[left hook-latex]{dr}
        & D \arrow[right hook-latex,dashed]{dr}{\iota} \NWbracket
        & G \arrow{d}{i_1^G}
        \\
        & G \arrow[swap]{r}{i_2^G}
        & \COPY_b(G)
    \end{tikzcd}
    \]
    where $i_1^G$ and $i_2^G$ are the inclusions from the pushout that
    defines $\COPY_b(G)$.

    Now every edge and vertex in $\beta(\COPY_b(G))$ is in the image of
    either $i_1^G$ or $i_2^G$, and its preimage is in $\Bound_!(G)$.
    Then it must be in the image of $\iota$, so $D$ contains all
    $!$-vertices and edges between them from $\COPY_b(G)$.

    Now let $w$ be a boundary vertex of $\COPY_b(G)$.  It must have a
    preimage under at least one of $i_1^G$ and $i_2^G$, and lemma
    \ref{lemma:sg-morphism-bounds} gives us that any such preimage in
    $G$ is also a boundary vertex, and hence in $\Bound_!(G)$, and so in
    the image of $\iota$.  Conversely, if $w$ is a wire-vertex in the
    image of $\iota$, it must be in the image of a boundary vertex of
    $G$ under either $i_1^G$ or $i_2^G$.  Then lemma
    \ref{lemma:copy-maps-bounds} gives us that $w$ is a boundary-vertex
    of $\COPY_b(G)$.

    If $e$ is an edge from a $!$-vertex to a boundary vertex in
    $\COPY_b(G)$, it must be in the image of either $i_1^G$ or $i_2^G$,
    and its preimage $e'$ under these map(s) must likewise be an edge
    from a $!$-vertex to a boundary vertex in $G$.  So $e'$ is in
    $\Bound_!(G)$, and hence in the image of $\iota$, and we have that
    $\iota[D] = \Bound_!(\COPY_b(G))$.

    For the remaining cases, we note that $\beta(\Bound_!(G))
    = \beta(G)$.  It then follows that
    \[ \beta(\Bound_!(\KILL_b(G))) = \beta(\KILL_b(\Bound_!(G))) \]
    and similarly for $\DROP_b$.  Then the result for
    $\KILL_b$ follows from the fact that $U(B(b))$ is an open subgraph
    of $U(G)$, and $\DROP$ from the fact that it does not affect $U(G)$.

    The arguments for $\In_!$ and $\Out_!$ are almost identical.
\end{proof}

\section{Matching String Graphs With $!$-Graphs} 
\label{sec:matching-with-bgs}

A $!$-graph can be considered to be a pattern for string graphs, similar
in spirit to a regular expression, although less powerful.  Implicit in
the term ``pattern'', though, is that there should be a notion of
matching.  We already have a notion of what it means for a string graph
to match another string graph: the first string graph must be a subgraph
of the second (with some additional constraints).  So what does it mean
for a $!$-graph to match a string graph?

If $G$ is a $!$-graph, an obvious candidate for a string graph that
should be matched by $G$ is $U(G)$, the concrete part of the graph.
Recalling the discussion at the start of the chapter, our aim is to
match the concrete part of $!$-graphs with ``any number of copies'' of
each $!$-box in $G$.

This is where the $!$-box operations come in.  $\COPY_b$ and $\KILL_b$
provide a way of getting ``any number of copies'' (including zero) of
the $!$-box $B(b)$.  $\DROP_b$ then provides a way of producing a
concrete graph (which we have already noted is essentially the same as a
string graph).  These provide the a notion of \textit{instantiation},
by which we can produce a string graph from a $!$-graph.

Note that this definition is slightly different to the one in
\cite{Kissinger2012a} (pp 7).  The notions of \textit{instance} and
\textit{instantiation} in that paper correspond to \textit{concrete
instance} and \textit{concrete instantiation} here.

\begin{definition}[Instantiation]
    \label{def:instantiation}
    For $!$-graphs $G$, $H$, we let $G \succeq H$, and say $H$ is an
    \textit{instance} of $G$, if and only if $H$ can be obtained from
    $G$ (up to isomorphism) by applying the operations from
    definition~\ref{def:bbox-ops} zero or more times.  This sequence of
    operations is called an \textit{instantiation} of $H$ from $G$.  If
    $H$ is a concrete graph, it is called a \textit{concrete instance}
    of $G$, and any instantiation of it is also called
    \textit{concrete}.
\end{definition}

\begin{remark}
    \label{rem:instantiation-notation}
    Given an instantiation $S$ of $H$ from $G$, we will sometimes use
    the notation $S(G)$ to refer to $H$.  Note that $S$ is not really a
    function: because the first operation of $S$ refers to a $!$-vertex
    of $G$, $S$ cannot be applied to anything other than $G$ (although
    we can sometimes relax this constraint in the presence of certain
    morphisms, as we will see in later chapters).
\end{remark}

We can then say that a $!$-graph $G$ matches a string graph $H$ if and
only if there is a concrete instance of $G$ that matches $H$:

\begin{definition}[Matching; \cite{Kissinger2012a}, pp 8]
    \label{def:bg-sg-matching}
    Let $P$ be a $!$-graph, and $H$ a string graph.  If there is a
    concrete instance $G$ of $P$, with instantiation $S$, and a monic
    local isomorphism $m : U(G) \rightarrow H$, $P$ is said to
    \textit{match $H$ at $m$ under $S$}, and $m$ is said to be a
    \textit{matching of $P$ onto $H$ under $S$}.
\end{definition}

Note that all the $!$-box operations yield a pattern that is at least as
specific; if we apply a $!$-box operation to a $!$-graph $G$ to get
another $!$-graph $H$, any string graph matched by $H$ will also be
matched by $G$.  $\succeq$ can therefore be viewed as a refinement
(pre-)ordering on $!$-graphs.

Given a $!$-graph $P$ that satisfies a couple of constraints on
$!$-boxes that ensure there are only ever a finite number of matchings
onto any given graph (for example, there can be no empty $!$-boxes), it
is possible to decide for any string graph $G$ whether there is a
matching of $P$ onto $G$.  An algorithm for this (that also produces the
required instantiation) is given in chapter \ref{ch:implementation}.

\begin{example}
    Suppose we wish to use a rewrite rule implementing the Z spider law
    \[ \input{z-spider-rw-ex-rule.tikz} \]
    to rewrite
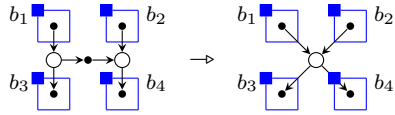
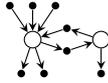
    \[ \input{z-spider-rw-ex-graph.tikz} \]

    We start by normalising the target graph:
    \[ \input{z-spider-rw-ex-graph-expand-hl.tikz} \]
    The highlighted wire-vertices are the ones that need to be matched
    by vertices that are in $!$-boxes in the LHS of the rule.  So we
    create four copies of $b_1$ and two each of $b_3$ and $b_4$
    \[ \input{z-spider-rw-ex-rule-copied.tikz} \]
    then kill $b_2$
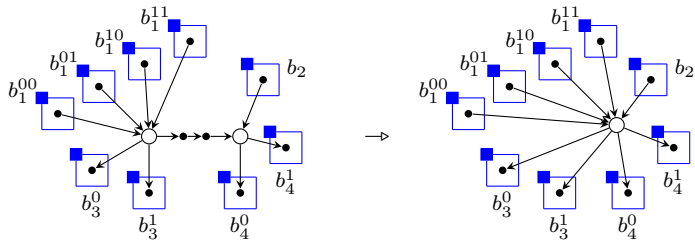
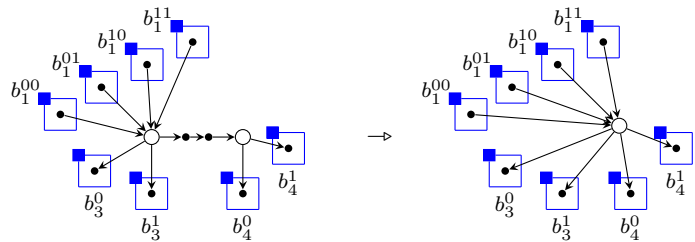
    \[ \input{z-spider-rw-ex-rule-copied-killed.tikz} \]
    and drop all the remaining $!$-boxes
    \[ \input{z-spider-rw-ex-rule-copied-killed-dropped.tikz} \]
    This corresponds to the instantiation
    \begin{align*}
        &\COPY_{b_1}; \COPY_{b_1^0};
        \COPY_{b_1^1}; \COPY_{b_3};
        \COPY_{b_4}; \KILL_{b_2};
        \DROP_{b_1^{00}};\\
        &\DROP_{b_1^{01}};
        \DROP_{b_1^{10}}; \DROP_{b_1^{11}};
        \DROP_{b_3^0}; \DROP_{b_3^1};
        \DROP_{b_4^0}; \DROP_{b_4^1}
    \end{align*}
    where, as a notational convenience, we refer to the copies of $b$ in
    $\COPY_b(G)$ as $b^0$ and $b^1$ (in abitrary order).  The resulting
    rule matches most of the graph:
    \[ \input{z-spider-rw-ex-graph-match-hl.tikz} \]
    and rewriting (and removing extraneous wire-vertices) results in
    \[ \input{z-spider-rw-ex-result.tikz} \]
\end{example}

\subsection{Nested and Overlapping \texorpdfstring{$!$}{!}-boxes} 
\label{sec:nesting-overlapping}

The effect of edges between $!$-vertices on the $!$-box operations
deserves special attention.  These edges can be thought of as indicating
a parent-child relation, with the edge going from the parent $!$-box to
the child $!$-box (although this view does not entirely hold up when we
consider the self-loop that exists on each $!$-vertex).

The edges from a $!$-vertex determine which other vertices should be
copied or killed when the $!$-vertex has $\COPY$ or $\KILL$,
respectively, applied to it.  So an edge from a $!$-vertex $b$ to
another $!$-vertex $c$ means that $\COPY_b$ will also copy $c$ (and its
associated $!$-box).

An example that demonstrates the difference that an edge between
$!$-vertices makes is the following:
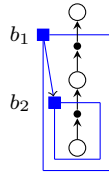
\begin{equation*}
    \input{tree-pattern.tikz}
\end{equation*}

This will match any tree of white nodes of depth at most two.  However,
if we remove the edge between the $!$-vertices, we get only the
\emph{balanced} trees.
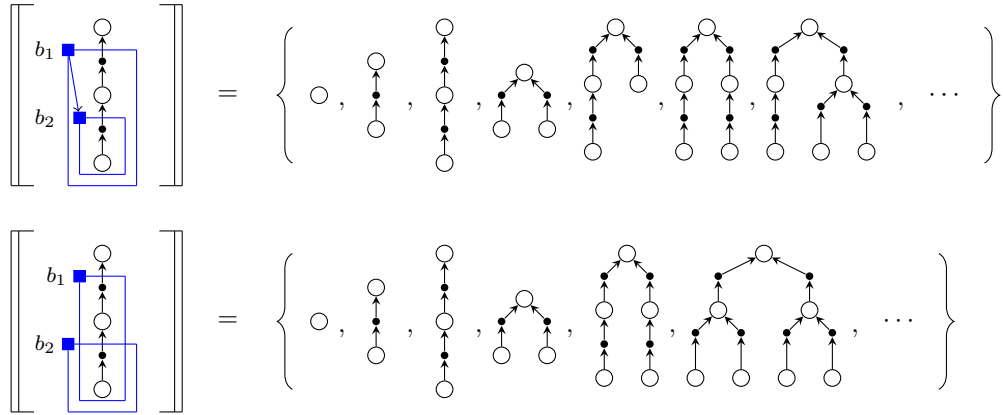
\begin{equation*}
    \input{tree-expand.tikz}
\end{equation*}

The reason for this is that when we copy the outer $!$-box in the case
where there is an edge, we also copy the inner $!$-box.
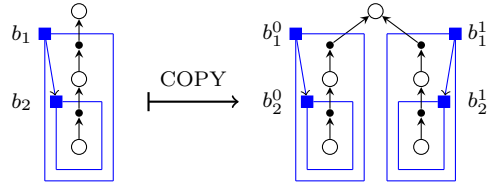
\begin{equation*}
    \input{tree-copy-nested.tikz} 
\end{equation*}
This produces a separate $!$-box round each second-level node, allowing
different numbers of nodes on different branches (eg: one could copy
$b_2^0$ and kill $b_2^1$).  In the case without an edge, however,
copying $b_1$ will simply extend $b_2$, and every copy of $b_2$ will
have exactly one second-level node in each first-level branch of the
tree.
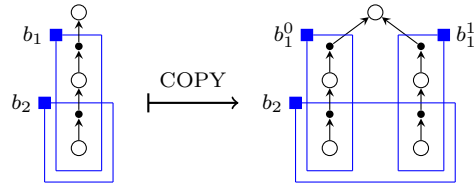
\begin{equation*}
    \input{tree-copy-overlapping.tikz} 
\end{equation*}
In this case, copying or killing $b_2$ would have the same effect on
every branch of the tree, ensuring that it remains balanced.

\begin{example}
    Recall the generalised bialgebra graph equation from the start of
    the chapter:
    \[ \input{gen-bialg-graph-labelled.tikz} \]
    Ignoring, for the moment, the details of the structure used to
    represent this equation (which we detail in section
    \ref{sec:rw-sgs-with-pgs}) we can (informally) show that the
    spiderised bialgebra law is an instance of this.  First, we apply
    $\COPY_{b_1}$:
    \[ \input{gen-bialg-graph-exp-1.tikz} \]
    Next we apply $\COPY_{b_2}$.  At this point, our shorthand $!$-box
    notation causes more confusion than it alleviates, so we use a
    hybrid notation where we show some of the $!$-box containment edges
    explicitly.
    \[ \input{gen-bialg-graph-exp-2.tikz} \]
    Finally, we $\DROP$ all the top-level $!$-vertices, giving us the
    spiderised bialgebra law.
    \[ \input{gen-bialg-graph-exp-3.tikz} \]
    If the $!$-boxes had not been nested, we would have ended up with
    the weaker statement
    \[ \input{gen-bialg-graph-exp-4.tikz} \]
\end{example}

\section{Related Work} 
\label{sec:bg-related}

Boneva et al describe\cite{Boneva2007} a system of \textit{graph shapes}
to generalise over classes of graphs in the context of model checking.
Their approach is similar to the typegraph constructions we introduced
for string graphs; a shape is a graph with some multiplicity constraints
on edges and vertices, and a graph $G$ is a \textit{concretisation} of
a shape $S$ if there is a morphism from $G$ to $S$ satisfying those
constraints (for example, a node of $S$ could require that exactly $3$
nodes of $G$ map to it, or an edge could require more than $5$ edges to
match it).  There is a shape-wide degree of approximation in these
multiplicities.

The aim of graph shapes is different from that of $!$-graphs, even
though both attempt to capture families of graphs in a single graph, and
this leads to some substantial differences between them.  For example,
graph shapes are mostly local (although there is some weak grouping of
nodes); multiplicities, in particular, are specified on a per-edge or
per-node basis.  $!$-graphs, by contrast, operate on arbitrarily-large
subgraphs.

Taentzer introduced the idea of amalgamated graph transformations in
\cite{Taentzer1996}.  These provide a framework for applying multiple,
potentially overlapping, rewrite rules in parallel by identifying common
sub-rules.  For example, one could amalgamate two copies of the string
graph rewrite rule
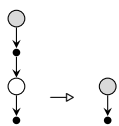
\[ \input{amalg-rule-1.tikz} \]
using the common sub-rule
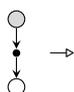
\[ \input{amalg-subrule-1.tikz} \]
to produce the familier rule
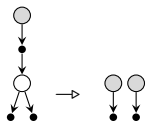
\[ \input{amalg-joined-1.tikz} \]
Doing this an arbitrary number of times would allow the construction of
any one of the rules we intend
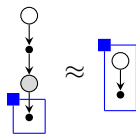
\[ \input{x-copies-z-spider-graph.tikz} \]
to represent.  However, we cannot allow arbitrary amalgamations of this
sort.  While
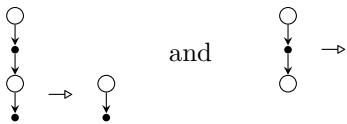
\[ \input{amalg-rule-2.tikz} \qquad \textrm{and} \qquad \input{amalg-subrule-2.tikz} \]
are both valid rules,
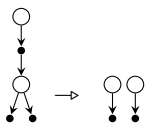
\[ \input{amalg-joined-2.tikz} \]
certainly is not.

Rensink's nested graph transformation
rules\cite{Rensink2006,Rensink2009} extend this amalgamation idea using
trees of rewrite rules, each rule being a subrule of its children.  A
predicate language over the left-hand side of these rules allows finer
control of the matching process (and, by using a tree of graphs rather
than of rewrite rules, the expression of complex graph predicates).

While the predicate language is unnecessary to codify spider-like
equations, the idea of rule trees could be used as a way of encoding the
$!$-graphs we described above, even including the nesting of $!$-boxes.
For example, using Rensink's notation, we could describe the spidered
version of X copies Z using the nested graph transformation rule
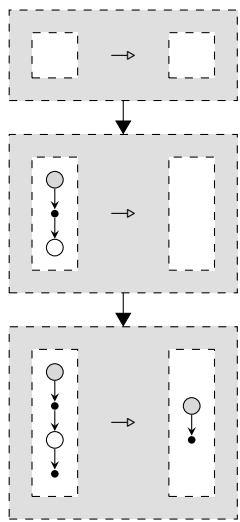
\[
\begin{tikzpicture}[string graph]
    \begin{pgfonlayer}{background}
        \draw[dashed,fill=lightgray!50] (-2.5,4.5) rectangle (2.5,2.5);
        \draw[dashed,fill=white] (-2,4) rectangle (-1,3);
        \draw[dashed,fill=white] (1,4) rectangle (2,3);
        \draw[dashed,fill=lightgray!50] (-2.5,1.75) rectangle (2.5,-1.75);
        \draw[dashed,fill=white] (-2,1.25) rectangle (-1,-1.25);
        \draw[dashed,fill=white] (1,1.25) rectangle (2,-1.25);
        \draw[dashed,fill=lightgray!50] (-2.5,-2.5) rectangle (2.5,-6.75);
        \draw[dashed,fill=white] (-2,-3) rectangle (-1,-6.25);
        \draw[dashed,fill=white] (1,-3) rectangle (2,-6.25);
        \draw[-triangle 60] (0,2.5) -- (0,1.75);
        \draw[-triangle 60] (0,-1.75) -- (0,-2.5);
    \end{pgfonlayer}
    \begin{pgfonlayer}{nodelayer}
        \node [style=sg grey vertex] (0) at (-1.5, 0.75) {};
        \node [style=sg wire vertex] (1) at (-1.5, 0) {};
        \node [style=none] (2) at (0, 0) {$\rewritesto$};
        \node [style=sg vertex] (3) at (-1.5, -0.75) {};
        \node [style=none] (4) at (0, 3.5) {$\rewritesto$};
        \node [style=sg wire vertex] (5) at (1.5, -5) {};
        \node [style=sg vertex] (6) at (-1.5, -5) {};
        \node [style=sg wire vertex] (7) at (-1.5, -4.25) {};
        \node [style=sg grey vertex] (8) at (1.5, -4.25) {};
        \node [style=sg wire vertex] (9) at (-1.5, -5.75) {};
        \node [style=sg grey vertex] (10) at (-1.5, -3.5) {};
        \node [style=none] (11) at (0, -4.625) {$\rewritesto$};
    \end{pgfonlayer}
    \begin{pgfonlayer}{edgelayer}
        \draw [style=sg diredge] (0) to (1);
        \draw [style=sg diredge] (1) to (3);
        \draw [style=sg diredge] (10) to (7);
        \draw [style=sg diredge] (7) to (6);
        \draw [style=sg diredge] (6) to (9);
        \draw [style=sg diredge] (8) to (5);
    \end{pgfonlayer}
\end{tikzpicture}
\]
The fact that an \textit{instance} of this tree can have any number of
copies of any given branch (including none) gives us the family of
rewrite rules we desire.  However, this technique could not be used to
encode the generalised bialgebra law, as the tree structure does not
allow for overlapping $!$-boxes, and hence cannot express arbitrary
bipartite graphs.  The approach we take (encoding $!$-boxes in the graph
itself) also allows us to use graph rewriting to reason directly about
$!$-graphs, as we will see in later chapters.

Stallmann produced an extension to the UML-based \textit{Story Patterns}
employed by FuJaBa\cite{fujaba} in \cite{Stallmann2008}.  These
\textit{enhanced Story Patterns} resemble Rensink's nested graph
transformation rules in character, and resemble $!$-graphs in
presentation.  They allow parts of a Story Pattern to be marked with
modifiers familiar from predicate logic, such as $\neg$, $\wedge$ or
$\forall$.  For example, the following enhanced Story Pattern captures
the requirement that every instance of class B is contained in an
instance of class A:
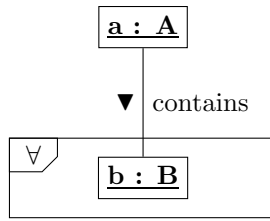
\[
    \begin{tikzpicture}
        \node[rectangle,draw,inner ysep=1.5em,inner xsep=5em] (box) at
        (0,-1) {};
        \node[chamfered rectangle,
              chamfered rectangle corners=south east,
              anchor=north west,
              inner ysep=0,
              draw] (l1)
              at (box.north west) {$\forall$};
        \node[rectangle,draw] (a) at (0,1) {\underline{\textbf{a : A}}};
        \node[rectangle,draw] (b) at (0,-1) {\underline{\textbf{b : B}}};
        \draw (a) -- (b) node[pos=0.5,auto=left]{contains}
        node[pos=0.5,rotate=-90,anchor=north]{$\blacktriangleright$};
    \end{tikzpicture}
\]


\chapter{$!$-Graph Rewriting}
\label{ch:rw-bang-graphs}

$!$-graphs provide the foundation for representing infinite families of
string graph equations and rewrite rules.  In this chapter, we will
demonstrate the construction of equations and rewrite rules composed of
$!$-graphs and show how they can be used to rewrite string graphs,
producing new string graph equations, and even $!$-graphs, producing new
$!$-graph equations.  Importantly, the latter is sound with respect to
the interpretation of $!$-graph equations as families of string graph
equations.

The first two sections of this chapter are again based on
\cite{Kissinger2012a} (pp 9-11).  The other two sections have not
previously been published.

\section{$!$-Graph Equations} 
\label{sec:bg-eqs}

We can combine $!$-graphs to form $!$-graph equations in much the same
way we did string graph equations.  We will extend the $!$-box
operations (and hence the idea of instantiation) to these equations, and
this will determine the string graph equations represented by a given
$!$-graph equation.

We can view a string graph equation as a pair of string graphs with
correlated inputs and outputs.  In the same way, we want a $!$-graph
equation to encode a pair of $!$-graphs and a correlation of inputs,
outputs and $!$-vertices.  For example, the spider law given in
\eqref{eq:spider-law-informal} could be represented as
\begin{equation}\label{eq:spider-law-pattern}
    \input{spider-law-graph.tikz}
\end{equation}

If we $\COPY_{b_2}$, $\DROP_{b_3}$ and $\KILL_{b_4}$, we get the
following equation.
\[ \input{spider-law-graph.tikz} \]
The rest of this section will be devoted to formalising these
constructions and showing that we can apply the $!$-box operations to
them.

We will use a span of monomorphisms $L \leftarrow I \rightarrow R$,
where $I$ equates $\Bound_!(L)$ with $\Bound_!(R)$.  Before we get to
the technical definition, we need a short lemma.
\begin{lemma}\label{lemma:pg-bounds-pushout}
    If $G$ is a $!$-graph such that $U(G)$ has no isolated
    wire-vertices, then we have the following pushout:
    \[
    \begin{tikzcd}
        \beta(G)
            \arrow[right hook-latex]{r}
            \arrow[right hook-latex]{d}
        & \In_!(G) \arrow[right hook-latex]{d}
        \\
        \Out_!(G) \arrow[right hook-latex]{r}
        & \Bound_!(G) \NWbracket
    \end{tikzcd}
    \]
\end{lemma}
\begin{proof}
    If $U(G)$ has no isolated wire-vertices, $\In(U(G))$ and
    $\Out(U(G))$ are disjoint.  The result then follows from the
    definitions and the fact that pushouts are unions in
    $\catGraph/\mathcal{G}_{T!}$.
\end{proof}

The definition of a $!$-graph equation we give here looks quite
different from the definition of a \textit{rewrite pattern} given in
\cite{Kissinger2012a} (pp 9).  We are using this formulation to better
demonstrate the similarities with string graph equations (definition
\ref{def:graph-eq}), and because it is easier to use when constructing
proofs.

\begin{definition}[$!$-graph equation]\label{def:bg-rw-rule}
    A \emph{$!$-graph equation} $L \rweq_{i_1,i_2} R$ is a span
    ${\cspan{L}{i_1}{I}{i_2}{R}}$ in $\catBGraph_T$ where
    \begin{enumerate}
        \item $U(L)$ and $U(R)$ contain no isolated wire-vertices;
        \item $\In_!(L) \cong \In_!(R)$ and $\Out_!(L) \cong \Out_!(R)$;
            \label{def-item:bg-rw-rule-io-iso}
        \item $\Bound_!(L) \cong I \cong \Bound_!(R)$; and
            \label{def-item:bg-rw-rule-bounds-iso}
        \item the following diagram commutes, where $j_1,j_2,k_1$ and
            $k_2$ are the inclusions from lemma
            \ref{lemma:pg-bounds-pushout} composed with the above
            isomorphisms:
            \begin{equation}
                \begin{tikzcd}[ampersand replacement=\&]
                    \& \In_!(L)
                        \arrow[left hook->]{dl}
                        \arrow{dr}{j_1}
                        \arrow{rr}{\cong}
                    \&\& \In_!(R)
                        \arrow[swap]{dl}{j_2}
                        \arrow[right hook->]{dr}
                    \&
                    \\
                    L \&\& I \arrow[swap]{ll}{i_1} \arrow{rr}{i_2} \&\& R
                    \\
                    \& \Out_!(L)
                        \arrow[left hook->]{ul}
                        \arrow[swap]{ur}{k_1}
                        \arrow[swap]{rr}{\cong}
                    \&\& \Out_!(R)
                        \arrow{ul}{k_2}
                        \arrow[right hook->]{ur}
                    \&
                \end{tikzcd}
                \label{eq:bg-rw-rule-diag}
            \end{equation}
            \label{def-item:bg-rw-rule-diag}
    \end{enumerate}

    A \textit{$!$-graph rewrite rule} is a directed $!$-graph equation.
\end{definition}

It can be seen by comparing the definitions that applying $U$ to a
$!$-graph equation will yield a string graph equation, and similarly for
rewrite rules.

Note that the isomorphisms in the above definition are forced to agree
with $i_1$ and $i_2$.  In particular, lemma
\ref{lemma:pg-bounds-pushout} can be used to show that the following
diagram commutes
\begin{equation}\label{eq:pg-rw-b-isos-span-agree}
\begin{tikzcd}
    \Bound_!(L)
        \arrow[right hook->]{d}
        \arrow{r}{\cong}
    & I
        \arrow{dl}{i_1}
        \arrow[swap]{dr}{i_2}
    & \Bound_!(R)
        \arrow[right hook->]{d}
        \arrow[swap]{l}{\cong}
    \\
    L && R
\end{tikzcd}
\end{equation}
and if we expand the definitions of $j_1$, $j_2$, $k_1$ and $k_2$ in
\eqref{eq:bg-rw-rule-diag}, we get
\begin{equation}\label{eq:pg-rw-isos-agree}
\begin{tikzcd}
    \In_!(L)
        \arrow{rr}{\cong}
        \arrow[right hook->]{d}
    && \In_!(R)
        \arrow[right hook->]{d}
    \\
    \Bound_!(L)
        \arrow{r}{\cong}
    & I
    & \Bound_!(R)
        \arrow[swap]{l}{\cong}
    \\
    \Out_!(L)
        \arrow{rr}{\cong}
        \arrow[right hook->]{u}
    && \Out_!(R)
        \arrow[right hook->]{u}
\end{tikzcd}
\end{equation}

We will now extend the $!$-box operations to $!$-graph equations.  The
intuitive rule is that the same operation must be performed on all three
graphs in the span, with the correspondence between $!$-boxes determined
by the morphisms of the span.  Of course, to get an equation, we need
some way of updating the morphisms of the span in addition to the graphs.

It turns out we can do this with any monomorphism of $!$-graphs that
reflects $!$-box containment:

\begin{lemma}\label{lem:bbox-m-rbc-ops-in-image}
    Let $f : G \rightarrow H$ be a $!$-graph monomorphism that reflects
    $!$-box containment, and let $b$ be a $!$-vertex in $G$.
    Then there are $!$-graph monomorphisms
    \begin{align*}
        \DROP_{f(b)}(f) &: \DROP_b(G) \rightarrow \DROP_{f(b)}(H) \\
        \KILL_{f(b)}(f) &: \KILL_b(G) \rightarrow \KILL_{f(b)}(H) \\
        \COPY_{f(b)}(f) &: \COPY_b(G) \rightarrow \COPY_{f(b)}(H)
    \end{align*}
    that reflect $!$-box containment and that commute with $f$ in the
    following ways:
    \begin{equation}\label{eq:bbox-rbc-ops-image-drop}
        \begin{tikzcd}[column sep=huge]
            \DROP_b(G) \rar{\DROP_{f(b)}(f)} \dar[right hook->]
            & \DROP_{f(b)}(H) \dar[right hook->]
            \\ G \arrow{r}{f} & H
        \end{tikzcd}
    \end{equation}
    \begin{equation}\label{eq:bbox-rbc-ops-image-kill}
        \begin{tikzcd}[column sep=huge]
            \KILL_b(G) \rar{\KILL_{f(b)}(f)} \dar[right hook->]
            & \KILL_{f(b)}(H) \dar[right hook->]
            \\ G \arrow{r}{f} & H
        \end{tikzcd}
    \end{equation}
    \begin{equation}\label{eq:bbox-rbc-ops-image-copy}
        \begin{tikzcd}[column sep=huge]
            \COPY_b(G) \rar{\COPY_{f(b)}(f)}
            & \COPY_{f(b)}(H)
            \\
            G \rar{f} \uar{p^G_i}
            & H \uar[swap]{p^H_i}
        \end{tikzcd}
    \end{equation}
    where $i \in \{1,2\}$ and $p^G_1$ and $p^G_2$ are the maps from the
    pushout defining $\COPY_b(G)$:
    \[
        \begin{tikzcd}
            G\graphminus B(b) \rar[right hook->] \dar[right hook->]
            & G \dar{p_1^G} \\
            G \rar[swap]{p_2^G} & \COPY_b(G) \NWbracket
        \end{tikzcd}
    \]
    and $p^H_1$ and $p^H_2$ are defined similarly.

    These maps are all the unique ones satisfying these diagrams, up to
    unique isomorphism.
\end{lemma}
\begin{proof}
    Let $G' = \DROP_b(G)$ and $H' =
    \DROP_{f(b)}(H)$.  Note that the image of $G'$ under $f$
    is contained in $H'$, since $f$ is injective (so no vertex of $G$
    other than $b$ maps to $f(b)$).  So there is a unique monomorphism
    $f' = \DROP_{f(b)}(f)$ satisfying
    \eqref{eq:bbox-rbc-ops-image-drop}.  $f'$ reflects $!$-box
    containment, since if $f'(v) \in B_H(c)$, $f(v) \in B_{H'}(c)$, and
    $c \neq b$.

    Now we deal with $\KILL_{f(b)}(f)$.  Let $G'
    = \KILL_b(G)$ and $H' = \KILL_{f(b)}(H)$.  We start by showing that
    the image of $G'$ under $f$ is contained in $H'$.  These are full
    subgraphs of $G$ and $H$, respectively, so we only need to consider
    vertices.  Let $v$ be a vertex in $G'$.  So $v$ is not in $B(b)$.
    Then, since $f$ reflects $!$-box containment, $f(v)$ is not in
    $B(f(b))$, and so is in $H'$.  So we can construct the monomorphism
    $f' = \KILL_{f(b)}(f)$ by simply restricting the domain and codomain
    of $f$.  Since the inclusion of $\KILL_{f(b)}(H)$ into $H$ is monic,
    this morphism must be the unique one satisfying
    \eqref{eq:bbox-rbc-ops-image-kill}.

    Now we show that $f'$ reflects $!$-box containment.  Let $v$ be a
    vertex of $G'$ and $c \in\:!(H')$, with $f'(v) \in B_{H'}(c)$.  Then
    $f(v) \in B_H(c)$, and so $c$ and the edge from $c$ to $f(v)$ are
    both in the image of $f$.  Further, $c$ cannot be in $B_H(f(b))$,
    since it is in $H'$.  So the preimage of $c'$ cannot be in $B_G(b)$,
    and hence must be in $G'$, along with the joining edge, and so $c$
    and the edge from $c$ to $v$ are in the image of $f'$.

    Finally, we consider $\COPY_{f(b)}(f)$.  Let $G' = \COPY_b(G)$ and
    $H' = \COPY_{f(b)}(H)$.  We demonstrate the existance of
    $\COPY_{f(b)}(f)$ by pushout.  If we let ${\iota_G : G \graphminus
    B(b) \rightarrow G}$ and ${\iota_H : H \graphminus B(f(b))
    \rightarrow H}$ be the natural inclusion maps,
    \eqref{eq:bbox-rbc-ops-image-kill} gives us that
    \[ p_1^H \circ f \circ \iota_G =
        p_1^H \circ \iota_H \circ \KILL_{f(b)}(f) =
        p_2^H \circ \iota_H \circ \KILL_{f(b)}(f) =
        p_1^H \circ f \circ \iota_G
    \]
    and there is then a unique $f'$ making the following diagram
    commute:
    \[
    \begin{tikzcd}
        G\graphminus B(b)
            \arrow[right hook->]{r}
            \arrow[left hook->]{d}
        & G \arrow{ddrr}{f} \arrow{d} &&
        \\
        G \arrow{r} \arrow[swap]{ddrr}{f}
        & G' \arrow[dashed]{dr}{f'} \NWbracket &&
        \\
        && H' \SEbracket
        & H \arrow{l}
        \\
        && H \arrow{u}
        & H\graphminus B(f(b))
            \arrow[right hook->]{l}
            \arrow[left hook->]{u}
    \end{tikzcd}
    \]
    Since all the other arrows in the diagram are monic, $f'$ must be as
    well.  For a given choice of $p_1^H$ and $p_2^H$ (relative to
    $p_1^G$ and $p_2^G$), it is the unique morphism satisfying
    \eqref{eq:bbox-rbc-ops-image-copy}.  If we do swap $p_1^H$ and
    $p_2^H$, we can note that the pushout is symmetric, and pushouts are
    unique up to a unique isomorphism, so $f'$ is unique up to a unique
    isomorphism.

    We must show that $f'$ reflects $!$-box containment.  Suppose $v$ is
    a vertex in $G'$ and $f(v) \in B(c)$, where $c$ is a $!$-vertex in
    $H'$, with $e$ the edge from $c$ to $f(v)$.  $e$ is in the image of
    one of the $p_i^H$; we call the preimage of $e$ under it $e'$, and
    its source (which maps to $c$) $c'$.  $f'(v)$ must also be in the
    image of $i_H$, and its preimage must be in the image of $f$.  What
    is more, the preimage of this under $f$ must map to $v$ by $p_i^G$,
    due to
    \eqref{eq:bbox-rbc-ops-image-copy}.
    \[ \input{matching-copy-in-image-fig2.tikz} \]
    We know that, since $f$ reflects $!$-box containment, $e'$ must be
    in the image of $f$.  But this means that $e$ must be in the image
    of $f'$, and hence $f'$ also reflects $!$-box containment.  So then
    we just let $\COPY_{f(b)}(f) = f'$.
\end{proof}

Given a $!$-graph equation $L \rweq_{i_1,i_2} R$, we know that $i_1$ and
$i_2$ are isomorphisms when their codomains are restricted to
$\Bound_!(L)$ and $\Bound_!(R)$, respectively.  This means that they
must reflect $!$-box containment, so we can now make the following
definitions.

\begin{definition}[\cite{Kissinger2012a}, pp 10]\label{def:bg-rw-rule-ops}
    Let $L \rweq R$ be a $!$-graph equation defined by the
    span $\cspan{L}{i_1}{I}{i_2}{R}$, and let $b \in\:!(I)$, with $b_1 =
    i_1(b)$ and $b_2 = i_2(b)$.  Then the $!$-box operations on
    $!$-graphs have the following equivalents on $!$-graph equations:
    \begin{itemize}
        \item $\COPY_b(L \rweq R)$ is defined to be
            \[ \cspan{\COPY_{b_1}(L)}{\COPY_{b_1}(i_1)}{\COPY_b(I)}{\COPY_{b_2}(i_2)}{\COPY_{b_2}(R)} \]
        \item $\DROP_b(L \rweq R)$ is the span
            \[ \cspan{\DROP_{b_1}(L)}{\DROP_{b_1}(i_1)}{\DROP_b(I)}{\DROP_{b_2}(i_2)}{\DROP_{b_2}(R)} \]
        \item $\KILL_b(L \rweq R)$ is the span
            \[ \cspan{\KILL_{b_1}(L)}{\KILL_{b_1}(i_1)}{\KILL_b(I)}{\KILL_{b_2}(i_2)}{\KILL_{b_2}(R)} \]
    \end{itemize}
    In the case of $!$-graph rewrite rules, the direction is preserved
    in the obvious way.
\end{definition}

Note that the equivalent definitions in \cite{Kissinger2012a} are the
same; the definition here appears simpler because we have made use of
lemma \ref{lem:bbox-m-rbc-ops-in-image}.

We must, of course, ensure that these operations produce $!$-graph
equations.  We start with some useful results about the operations of
lemma \ref{lem:bbox-m-rbc-ops-in-image}.

\begin{lemma}\label{lem:bbox-f-ops-composition}
    All the operations from lemma \ref{lem:bbox-m-rbc-ops-in-image}
    respect composition, at least up to unique isomorphism.  For
    example,
    \[ \DROP_{g(f(b))}(g \circ f) = \DROP_{g(f(b))}(g) \circ
        \DROP_{f(b)}(f) \]
\end{lemma}
\begin{proof}
    Consider the case of $\DROP$.  Suppose we have $!$-box-reflecting
    monomorphisms
    \begin{align*}
        f &: G_1 \rightarrow G_2 \\
        g &: G_2 \rightarrow G_3
    \end{align*}
    and suppose $b \in G_1$. Let $f' = \DROP_{f(b)}(f)$, $G'_1 =
    \DROP_b(G_1)$ and so on.  Then we just consider the diagram
    \begin{equation*}
        \begin{tikzcd}[column sep=huge]
            G'_1 \arrow{r}{f'} \arrow[right hook->]{d}
            & G'_2 \arrow{r}{g'} \arrow[right hook->]{d}
            & G'_3 \arrow[right hook->]{d}
            \\ G_1 \arrow{r}{f} & G_2 \arrow{r}{g} & G_3
        \end{tikzcd}
    \end{equation*}
    But $\DROP_{g(f(b))}(g \circ f)$ is the unique map from $G'_1$ to
    $G'_3$ that makes the outside edges of that diagram commute, and so
    must be the same as $g' \circ f'$.  The same argument holds for
    $\KILL_{g(f(b))}(g \circ f)$.

    A similar argument can be used for $\COPY$, but care must be taken,
    as the result of applying $\COPY$ to a map is only unique up to
    isomorphism.  It is possible, however, to make the relevant choices
    such that $\COPY_{g(f(b))}(g \circ f) = \COPY_{g(f(b))}(g) \circ
    \COPY_{f(b)}(f)$, which is sufficient to get the result we want.
\end{proof}

\begin{lemma}\label{lem:bbox-f-ops-iso}
    All the operations from lemma \ref{lem:bbox-m-rbc-ops-in-image}
    preserve isomorphisms.
\end{lemma}
\begin{proof}
    This follows from lemma \ref{lem:bbox-f-ops-composition} and the
    fact that $\DROP_b(1_G)$ and $\KILL_b(1_G)$ are the identity, and
    $\COPY_b(1_G)$ is an isomorphism of $\COPY_b(G)$.
\end{proof}

\begin{lemma}\label{lem:bbox-ops-isolated-verts}
    Let $G$ be a graph with no isolated vertices.  Then $\COPY_b(G)$,
    $\KILL_b(G)$ and $\DROP_b(G)$ do not have any isolated vertices.
\end{lemma}
\begin{proof}
    If $w$ were an isolated wire-vertex in $U(\COPY_b(G))$, it must be
    the image of a map from $U(G)$, but $U(G)$ does not contain any
    isolated wire-vertices, and so the map must carry an incident edge
    with it.

    $U(\KILL_b(G))$ is a full subgraph of $U(G)$ and $U(B(b))$ is not
    adjacent to any wire-vertices, so any wire-vertex in $U(G')$ must
    have the same incident edges as the corresponding wire-vertex in
    $U(G)$.

    $U(\DROP_b(G))$ is isomorphic to $U(G)$, and so has no isolated
    wire-vertices.
\end{proof}

\begin{theorem}
    \label{thm:rw-patt-preserved}
    Let $L \rweq R$ be a $!$-graph equation. Then applying any of the
    !-box operations of definition \ref{def:bg-rw-rule-ops} yields
    another $!$-graph equation.
\end{theorem}
\begin{proof}
    Let $\cspan{L}{i_1}{I}{i_2}{R}$ be a $!$-graph equation,
    and $\cspan{L'}{i_1'}{I'}{i_2'}{R'}$ be the result of applying one
    of the $!$-box operations to it.  We will demonstrate the proof for
    $\COPY_b$, but all of the arguments apply to $\KILL_b$ and $\DROP_b$
    as well.

    It is clear from lemma \ref{lem:bbox-m-rbc-ops-in-image} that this
    is a span of $!$-graphs.

    $U(L')$ and $U(R')$ do not contain any isolated wire-vertices, by
    lemma \ref{lem:bbox-ops-isolated-verts}.

    Since $\cspan{L}{i_1}{I}{i_2}{R}$ is a $!$-graph equation, we
    have isomorphisms
    \begin{align*}
        \phi_I &: \In_!(L) \rightarrow \In_!(R) \\
        \phi_O &: \Out_!(L) \rightarrow \Out_!(R) \\
        \phi_L &: \Bound_!(L) \rightarrow I \\
        \phi_R &: \Bound_!(R) \rightarrow I
    \end{align*}
    These all preserve $!$-box containment, and hence we have
    isomorphisms (by lemma \ref{lem:bbox-f-ops-iso} and theorem
    \ref{thm:pg-ops-commute-bounds})
    \begin{align*}
        \phi'_I &: \In_!(L') \rightarrow \In_!(R') \\
        \phi'_O &: \Out_!(L') \rightarrow \Out_!(R') \\
        \phi'_L &: \Bound_!(L') \rightarrow I' \\
        \phi'_R &: \Bound_!(R') \rightarrow I'
    \end{align*}

    Finally, we need to show that the following diagram commutes:
    \begin{equation}\label{eq:rw-rule-diag-copy}
    \begin{tikzcd}[ampersand replacement=\&]
        \& \In_!(L')
            \arrow[left hook->]{dl}
            \arrow{dr}{j'_1}
            \arrow{rr}{\phi'_I}
        \&\& \In_!(R')
            \arrow[swap]{dl}{k'_1}
            \arrow[right hook->]{dr}
        \&
        \\
        L' \&\& I' \arrow[swap]{ll}{i'_1} \arrow{rr}{i'_2} \&\& R'
        \\
        \& \Out_!(L')
            \arrow[left hook->]{ul}
            \arrow[swap]{ur}{j'_2}
            \arrow[swap]{rr}{\phi'_R}
        \&\& \Out_!(R')
            \arrow{ul}{k'_2}
            \arrow[right hook->]{ur}
        \&
    \end{tikzcd}
    \end{equation}
    where $j'_1$ is $\phi'_L$ composed with the inclusion of $\In_!(L')$
    into $\Bound_!(R')$, and similarly for $j'_2$, $k'_1$ and $k'_2$.

    The proofs for the four triangles on the left and right are all
    similar, so we will demonstrate the top left triangle.  Consider the
    following diagram:
    \begin{equation}\label{eq:copy-rule-diag-j1}
        \begin{tikzcd}[ampersand replacement=\&]
            \& \In_!(\COPY_b(L))
                \arrow[swap,bend right]{ddl}{j_1'}
                \arrow[bend left,right hook->]{ddr}
                \arrow{d}{\cong}
            \& \\
            \& \COPY_b(\In_!(L))
                \arrow[sloped,pos=0.8]{dl}{\COPY_b(j_1)}
                \arrow[sloped,pos=0.8,swap]{dr}{\COPY_b(\iota_L)}
            \& \\
            \COPY_b(I) \arrow[swap]{rr}{\COPY_b(i_1)} \&\&
            \COPY_b(L)
        \end{tikzcd}
    \end{equation}
    where $\iota_L : \In_!(L) \rightarrow L$ is the normal inclusion
    map.  The bottom triangle commutes by lemma
    \ref{lem:bbox-f-ops-composition}, since $\iota_L = i_1 \circ j_1$ by
    our assumption that $\cspan{L}{i_1}{I}{i_2}{R}$ is a !-graph
    equation.  That the upper right triangle commutes can be seen from the
    proof of theorem \ref{thm:pg-ops-commute-bounds}.  For the upper
    left triangle, consider the following:
    \begin{equation*}
        \begin{tikzcd}[column sep=huge]
              \In_!(\COPY_b(L))
                \arrow{r}{\cong}
                \arrow[right hook->]{d}
            & \COPY_b(\In_!(L))
                \arrow[swap]{d}{\COPY_b(\iota_L)}
                \arrow[sloped,pos=0.2]{dr}{\COPY_b(j_1)}
            & \\
              \Bound_!(\COPY_b(L))
                \arrow[swap]{r}{\cong}
                \arrow[bend right]{rr}{\phi'_L}
            & \COPY_b(\Bound_!(L))
                \arrow[swap]{r}{\COPY_b(\phi_L)}
            & \COPY_b(I)
        \end{tikzcd}
    \end{equation*}
    where we have abused $\iota_L$ to refer to the inclusion of
    $\In_!(L)$ into $\Bound_!(L)$.  Now we can again use the proof of
    theorem \ref{thm:pg-ops-commute-bounds} to see that the square on
    the left commutes, the triangle on the right commutes by lemma
    \ref{lem:bbox-f-ops-composition} and the definition of $j_1$, and
    the bottom triangle is simply the definition of $\phi'_L$.  But now
    $j'_1$ is $\phi'_L$ composed with the inclusion on the left of the
    diagram, by definition, and so we have that the top left triangle of
    \eqref{eq:copy-rule-diag-j1} commutes, and hence the upper left
    triangle of \eqref{eq:rw-rule-diag-copy} commutes.

    We now look at the upper middle triangle of
    \eqref{eq:rw-rule-diag-copy}; the lower middle triangle follows
    similarly.  Consider the following diagram:
    \begin{equation*}
        \begin{tikzcd}
            \COPY_b(\In_!(L))
            \arrow{rr}{\COPY_b(\phi_I)}
            \arrow[in=180,out=-135,looseness=1.75,swap]{ddr}{\COPY_b(j_1)}
            &&
            \COPY_b(\In_!(R))
            \arrow[in=0,out=-45,looseness=1.75]{ddl}{\COPY_b(k_1)}
            \\
            \In_!(L')
            \arrow[swap]{u}{\cong}
            \arrow{rr}{\phi'_I}
            \arrow{dr}{j'_1}
            &&
            \In_!(R')
            \arrow[swap]{dl}{k'_1}
            \arrow{u}{\cong}
            \\
            & I'
            &
        \end{tikzcd}
    \end{equation*}
    We have already demonstrated that the triangle on the left commutes,
    and the triangle on the right follows in a similar manner.  The top
    square is just the definition of $\phi'_I$, and the outside edges of
    the diagram commute by lemma \ref{lem:bbox-f-ops-composition} and
    the assumption that $\cspan{L}{i_1}{I}{i_2}{R}$ is a !-graph
    equation.  It follows that the triangle at the bottom commutes, and
    this is the top middle triangle of \eqref{eq:rw-rule-diag-copy}.

    Hence we have shown that $\cspan{L'}{i'_1}{I'}{i'_2}{R'}$ is,
    indeed, a $!$-graph equation.
\end{proof}
\begin{corollary}[\cite{Kissinger2012a}, pp 11]
    Let $L \rewritesto R$ be a $!$-graph rewrite rule. Then applying any
    of the !-box operations of definition \ref{def:bg-rw-rule-ops}
    yields another $!$-graph rewrite rule.
\end{corollary}

The concept of instantiation from definition \ref{def:instantiation}
then extends to $!$-graph equations and rewrite rules in the obvious
way.  We will consider a $!$-graph equation to hold if and only if all
its concrete instantiations hold.

\section{Rewriting String Graphs with $!$-Graph Equations} 
\label{sec:rw-sgs-with-pgs}

For a $!$-graph equation $L \rweq R$ and a string graph $G$,
we want to use the directed form of the equation, $L \rewritesto R$, to
rewrite $G$, and the resulting graph should be a valid rewrite of $G$ in
the (usually infinite) string graph rewrite system comprised of the
concrete instantiations of $L \rewritesto R$.

Because applying a $!$-box operation to a $!$-graph rewrite rule also
applies it to the components of that rule, and the LHS in particular, we
can use the notion of matching from section \ref{sec:matching-with-bgs}
to find a concrete instance of the rewrite rule that can be applied to a
given string graph $G$.

Suppose we have a $!$-graph rewrite rule $L \rewritesto R$ and a string
graph $G$.  If $m$ is a matching of $L$ onto $G$ under $S$, we can apply
$S$ to $L \rewritesto R$ to get a string graph rewrite rule $L'
\rewritesto R'$.  By considering definition \ref{def:bg-rw-rule-ops}, we
can see that $L'$ is the string graph obtained by instantiating $L$ with
$S$.  So we have a matching from $L'$ to $G$, and can therefore rewrite
$G$ with $L' \rewritesto R'$.

Providing finding the instantiation and matching is decidable (which we
demonstrate in chapter \ref{ch:implementation}), rewriting $G$ with $L
\rewritesto R$ is decidable, and our requirement that this is a valid
rewrite under the rewrite system of instantiations of $L \rewritesto R$
is satisfied.

So we have acheived our aim of expressing certain infinite families of
string graph rewrite rules in a finite (and useful) way.  In the next
section, we will consider how $!$-graph rewrite rules can actually be
used to rewrite $!$-graphs, and thereby derive new $!$-graph rewrite
rules.

\section{Rewriting $!$-Graphs With $!$-Graph Equations}
\label{sec:rwpg-matching}

Let $L \xrightarrow{m} G$ be a $!$-graph monomorphism.  Then we have
the following pushout complement
\begin{equation} \label{eq:pushout-complement}
\begin{tikzcd}
    \Bound_!(L) \rar[right hook->] \dar[dashed,swap]{d}
    & L \dar{m}
    \\ D \rar[swap,dashed]{g}
    & G \NWbracket
\end{tikzcd}
\end{equation}
in $\catGraph/\mathcal{G}_{T!}$ if and only if $m$ satisfies the
no-dangling-edges condition. This condition is that for every vertex $v$
in $L \graphminus \Bound_!(L)$, and every edge $e$ in $G$ incident to
$m(v)$, $e$ is in the image of $m$.  This is stronger than $U(m)$ being
a local isomorphism, as it also requires that if a $!$-vertex $b$ and a
vertex $v$ in $L \graphminus \Bound_!(L)$ are both in the image of $m$,
any edge between them must also be in the image of $m$.

\begin{proposition}
    If the pushout complement \eqref{eq:pushout-complement} exists,
    $D$ is a $!$-graph.
\end{proposition}
\begin{proof}
    As all the morphisms in the diagram are monic, we consider $D$
    and $L$ to be subgraphs of $G$, and $\Bound_!(L)$ their common
    subgraph.

    Consider applying $U$ to the pushout square.  This is a pushout in
    $\catGraph/\mathcal{G}_T$, by proposition
    \ref{prop:slice-cat-func-pb}.  But pushout complements are unique up
    to isomorphism in $\catGraph/\mathcal{G}_T$, and we know that the
    pushout complement of
    \[ \Bound(U(L)) \hookrightarrow U(L) \xrightarrow{U(m)} U(G) \]
    is a string graph (since $U(m)$ is a local isomorphism, by the
    no-dangling-edges condition), so $U(D)$ must be a string graph.

    We apply $\beta$ to the pushout square, giving us another pushout
    (again, by proposition \ref{prop:slice-cat-func-pb}).  Then, since
    $\beta(\Bound_!(L)) = \beta(L)$, we must have that $\beta(D) \cong
    \beta(G)$, and hence $\beta(D)$ must be posetal.

    Let $b \in\:!(D)$.  Then $b \in\:!(G)$.  We know that $U(B_G(b))$ is
    open in $U(G)$, and $B_D(b) = B_G(b) \cap D$, so $U(B_D(b)) =
    U(B_G(b)) \cap U(D)$.  But this is open in $U(D)$ by proposition
    \ref{prop:open-subgraph-props}.

    Let $b,c \in\:!(D)$, with $c \in B_D(b)$, and let $v \in B_D(c)$.
    Then $b,c \in\:!(G)$ and $v \in B_G(c)$, and so $v \in B_G(b)$. The
    edge linking $b$ to $v$ must either be in $D$ or $L$. suppose it is in
    $L$. Then so must $b$ and $v$ be, and so $b$ and $v$ must be in
    $\Bound_!(L)$.  But then the edge linking them must also be in
    $\Bound_!(L)$ (by definition of $\Bound_!$) and hence in $D$. So $v
    \in B_D(b)$.

    So $D$ is a $!$-graph.
\end{proof}

Note that, while the no-dangling-edges condition will always produce a pushout
complement, it will not always lead to a valid rewrite. For example, in
the following commutative diagram in $\catGraph/\mathcal{G}_{T!}$, the
map on the left satisfies the no-dangling-edges condition, but the
resulting (bottom right) graph is not a valid $!$-graph, as $B(b)$ is
not an open subgraph.
\begin{equation}
    \label{eq:bad-rewrite}
    \input{bad-rewrite.tikz}
\end{equation}

\begin{definition}[$!$-graph matching]\label{def:pg-matching}
    A monomorphism $L \xrightarrow{m} G$ between two $!$-graphs is a
    \textit{$!$-graph matching} when $U(m)$ is a local isomorphism and
    $m$ reflects $!$-box containment.  If $L_0 \succeq L$, with
    instantiation $S$, $m$ is said to be a $!$-graph matching of $L_0$
    onto $G$ under $S$.
\end{definition}

Of course, if $m$ is a $!$-graph matching, then $m$ and the inclusion of
$\Bound_!(L)$ in $L$ have a pushout complement.  But, unlike in
\eqref{eq:bad-rewrite}, a $!$-graph matching always produces a rewrite.

\begin{theorem}\label{thm:bg-rw-exists}
    Let $L \rewritesto R$ be a $!$-graph rewrite rule and $m : L
    \rightarrow G$ be a $!$-graph matching for it.  Then $L \rewritesto
    R$ has a rewrite at $m$.
\end{theorem}
\begin{proof}
    Let $L \xleftarrow{i_1} I \xrightarrow{i_2} R$ be the $!$-graph
    rewrite rule, and let
    \[
    \begin{tikzcd}
        L \arrow[swap]{d}{m}
        & I \arrow[swap]{l}{i_1} \arrow{d}{d}
        \\ G \NEbracket
        & D \arrow{l}{g}
    \end{tikzcd}
    \]
    be the pushout complement of $i_1$ and $m$.  Then it is sufficient
    (by proposition \ref{prop:pg-bound-coh-p-pushout}) to show that the
    span formed from $d$ and $i_2$ is boundary-$!$-coherent.

    Theorem \ref{thm:sg-bcoh-pushout} gives us that $U(i_1)$ and $U(d)$
    are boundary-coherent, and then lemma \ref{lem:sg-bc-rw} gives us
    that $U(d)$ and $U(i_2)$ must also be boundary-coherent.

    $i_2$ reflects $!$-box containment by the definition of a $!$-graph
    rewrite rule.  Let $e$ be an edge in $D$ from a $!$-vertex $b$ to a
    vertex $v$, where $v$ is in the image of $d$.  Then $g(v)$ must be
    in the image of $m$, and so $g(e)$ must be in the image of $m$,
    since $m$ reflects $!$-box containment.  But then, since this is a
    pushout square, $e$ must be in the image of $d$ as required, and so
    $d$ also reflects $!$-box containment.  Thus $d$ and $i_2$ are
    boundary-$!$-coherent, as required.
\end{proof}

\begin{example}
    Recall the following proof from section \ref{sec:quantum}:
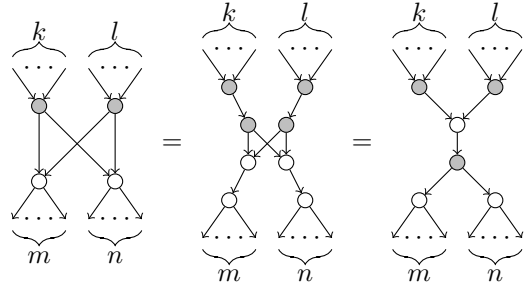
    \[ \input{bialg-spider.tikz} \]
    Consider the left equality.  We can perform the equivalent rewrite
    by using instances of the X and Z spider laws
    \begin{mathpar}
        \input{x-spider-rule.tikz}
        \and
        \input{z-spider-rule.tikz}
    \end{mathpar}
    These instances are
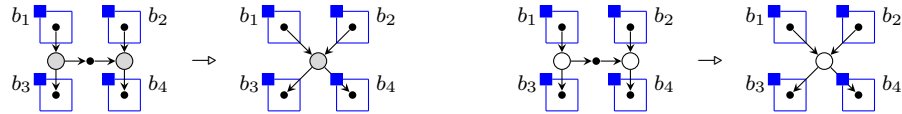
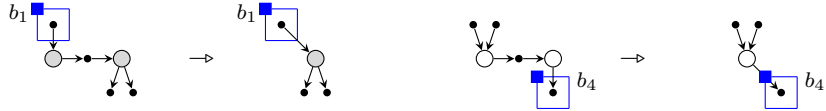
    \begin{mathpar}
        \input{x-spider-rule-exp.tikz}
        \and
        \input{z-spider-rule-exp.tikz}
    \end{mathpar}
    We can use the first of these to construct the following rewrite
    (ignoring issues of wire homeomorphism):
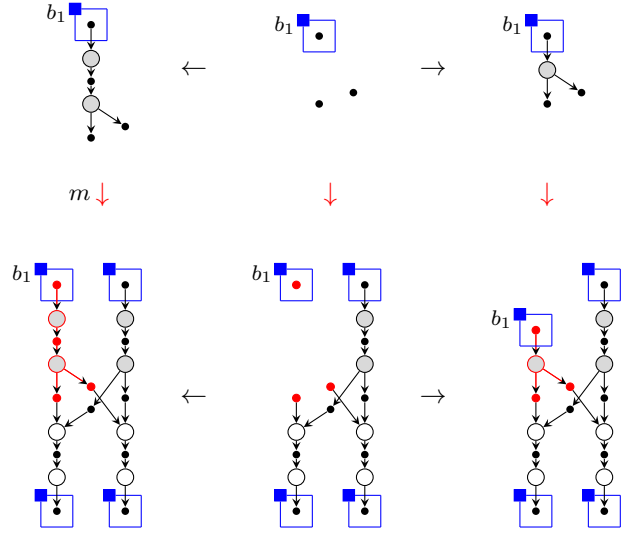
    \[ \input{bialg-spider-rewrite-pos-1.tikz} \]
    $U(m)$ is a local isomorphism, and $m$ can be seen to reflect
    $!$-box containment (for $b_1$, in this case).  As expected, we get
    a valid $!$-graph at the end.  One more application of the X spider
    law and two of the Z spider law will produce the chain of rewrites
    we want.
    \label{ex:bialg-spider-1}
\end{example}

Suppose we have the following rewrite (where $m$ is not necessarily a
$!$-graph matching):
\begin{equation}\label{eq:rewrite}
    \begin{tikzcd}
        L \arrow[swap]{d}{m}
        & I \arrow[swap]{l}{i_1} \arrow{r}{i_2} \arrow{d}{d}
        & R \arrow{d}{r}
        \\
        G \NEbracket
        & D \arrow{l}{g} \arrow[swap]{r}{h}
        & H \NWbracket
    \end{tikzcd}
\end{equation}
The bottom span will almost certainly not be a $!$-graph rewrite rule
(as $U(D)$ will not, in general, have only isolated points), but it is
trivial to produce such a rewrite rule by discarding unwanted vertices
and edges from $D$.

\begin{theorem}\label{lem:rw-patt-pres}
    In \eqref{eq:rewrite}, if $U(G)$ has no isolated points then there
    is a subgraph $J$ of $D$ such that
    \[ G \xleftarrow{g\restriction_{J}} J
        \xrightarrow{h\restriction_{J}} H \]
    is a $!$-graph rewrite rule.
\end{theorem}
\begin{proof}
    We start by showing that $\Bound_!(G)$ is in the image of $g$.

    $\beta(i_1)$ and $\beta(i_2)$ are isomorphisms since $i_1$ and $i_2$
    form a $!$-graph rewrite rule. It then follows from the pushout
    squares that $\beta(g)$ and $\beta(h)$ are also isomorphisms.  So we
    must have that $\beta(G)$ is in the image of $g$.

    If $w$ is in the boundary of $U(G)$, it is either not in the image
    of $m$, in which case it must be in the image of $g$, or its
    preimage under $m$ is in the boundary of $U(L)$ (since string graph
    monomorphisms reflect boundaries), in which case this preimage is in
    the image of $i_1$ and hence $w$ is in the image of $g$.  Now if $b
    \in\:!(G)$ and there is an edge $e$ from $b$ to $w$ in $G$, we also
    know that if $e$ is not in the image of $m$, it must be in the image
    of $g$.  So suppose $e$ is in the image of $m$.  Then its preimage
    must be in the image of $i_1$ by the same reasoning as above, and
    hence $e$ is in the image of $g$.  So $\Bound_!(G)$ is in the image
    of $g$.

    We can use the same argument to show that $\Bound_!(H)$ is in the
    image of $h$.

    Let $J$ be the preimage of $\Bound_!(G)$ under $g$, and $\iota : J
    \hookrightarrow D$ the inclusion of $J$ in $D$.  $g \circ \iota$ is
    monic, and hence injective, so we have $J \cong \Bound_!(G)$.

    We already know that the image of $\beta(J)$ under $h \circ \iota$
    is $\beta(H)$, since $\beta(\iota)$ and $\beta(h)$ are both
    isomorphisms.  We now show that, for all vertices $w$ in $J$,
    $h(\iota(w))$ is an input (resp.  output) if and only if
    $g(\iota(w))$ is an input (resp.  output).  This will then give us
    that the image of $h \circ \iota$ is exactly $\Bound_!(H)$ (since we
    already know that $\Bound_!(H)$ is in that image).

    So let $w$ be a wire-vertex in $J$.  We know that $g(\iota(w))$
    is either an input or an output.  Suppose it is an input (the output
    case is symmetric).  Then $\iota(w)$ must also be an input, and if
    $g(\iota(w))$ is in the image of $m$, its preimage must be an input.
    Let $w'$ be the preimage of $g(\iota(w))$ under $m \circ i_1$.  Then
    we know that $d(w') = \iota(w)$ and hence $r(i_2(w')) =
    h(\iota(w))$, and $i_2(w')$ is an input (since $i_1(w')$ is).  So
    the preimages of $h(\iota(w))$ under both $h$ and $r$ are inputs,
    and hence $h(\iota(w))$ must also be an input.  If $g(\iota(w))$ is
    not in the image of $m$, then $\iota(w)$ is not in the image of $d$,
    and hence $h(\iota(w))$ is not in the image of $r$.  Thus $r$ cannot
    introduce an incoming edge for $h(\iota(w))$, which must therefore
    be an input.

    Injectivity of $h \circ \iota$ gives us the isomorphism $J \cong
    \Bound_!(H)$.  This restricts to $\In_!(G) \cong \In_!(H)$ and
    $\Out_!(G) \cong \Out_!(H)$, as we just proved that $h(\iota(w))$ is
    an input if and only if $g(\iota(w))$ is an input.  The diagram from
    definition \ref{def:bg-rw-rule} commutes by construction, and hence
    \[ G \xleftarrow{g \circ \iota} J \xrightarrow{h \circ \iota} H \]
    is a $!$-graph rewrite rule.
\end{proof}

\begin{example}
    Continuing example \ref{ex:bialg-spider-1}, we get the following
    commutative diagram:
    \[ \input{bialg-spider-rewrite-to-rule-1.tikz} \]
    Thus we have the $!$-graph rewrite rule
    \[ \input{bialg-spider-rule-1.tikz} \]
\end{example}

\section{Soundness of $!$-Graph Rewriting} 
\label{sec:bbox-ops-rewrites}

We have shown how we can rewrite $!$-graphs using $!$-graph rewrite
rules, and in doing so produce further $!$-graph rewrite rules.
However, we also want to ensure that the resulting rewrite rule is
sound with respect to its interpretation as a family of string graph
rewrite rules.

Note that this differs from soundness of string graph
equations or rewrite rules, as discussed in chapter \ref{ch:rewriting},
where soundness was with respect to a particular valuation.  Instead, we
are interested in soundness with respect to a particular rewrite sytem
of string graphs.

Specifically, if we have a $!$-graph rewrite rule $L \rewritesto R$
which has been used to rewrite a $!$-graph $L'$ to another $!$-graph
$R'$, and there is an instance of the resulting $!$-graph rewrite rule
$L' \rewritesto R'$ that can rewrite the string graph $G$ to $H$, we
require there to be an instance of $L \rewritesto R$ that also rewrites
$G$ to $H$.

Luckily, not only does our definition of $!$-graph matching guarantee
that we get a rewrite, it also ensures that we can preserve the
rewrite under $!$-box operations.  We know that the top span ($L$, $I$
and $R$) share $!$-vertices (in that we get isomorphisms if we apply
$\beta$ to the maps), and the same is true of the bottom span ($G$, $D$
and $H$; see the proof of theorem \ref{lem:rw-patt-pres}).  For
notational convenience, we will identify the $!$-vertices across each of
these spans (eg: if $b$ is a $!$-vertex in $I$, we will consider $b$ to
be in $L$ as well, by which we mean $i_1(b)$).

Further, because the maps $m$, $d$ and $r$ of \eqref{eq:rewrite} are
monic, we can identify each $!$-vertex in the top span with one in the
bottom span.  So, given that we wish to perform a $!$-box operation
(which we will denote $\OP_b$) on $G$ and propagate it to the entire
rewrite, we have two possible cases: either the $!$-vertex being
operated on is in the image of $m$, in which case we consider it as
existing across the whole rewrite, or it is not, in which case it is
only in the bottom span.

If $b$ is in the image of $m$, lemma \ref{lem:bbox-m-rbc-ops-in-image}
provides us with a way of updating all the maps of the rewrite:
\begin{equation}\label{eq:rewrite-op-in-image}
    \begin{tikzcd}
        \OP_b(L) \arrow[swap]{d}{\OP_b(m)}
        & \OP_b(I) \arrow[swap]{l}{\OP_b(i_1)}
                   \arrow{r}{\OP_b(i_2)}
                   \arrow{d}{\OP_b(d)}
        & \OP_b(R) \arrow{d}{\OP_b(r)}
        \\
        \OP_b(G)
        & \OP_b(D) \arrow{l}{\OP_b(g)}
                   \arrow[swap]{r}{\OP_b(h)}
        & \OP_b(H)
    \end{tikzcd}
\end{equation}
We will, of course, need to check that this is still a rewrite.

If $b$ is not in the image of $m$, lemma
\ref{lem:bbox-m-rbc-ops-in-image} is insufficient.  However, the
following lemma will allow us to extend $!$-box operations on maps to
$!$-vertices that are not in the image of the map.

\begin{lemma}\label{lem:bbox-m-rbc-ops-disjoint}
    Let $f : G \rightarrow H$ be a $!$-graph monomorphism that reflects
    $!$-box containment, and let $b$ be a $!$-vertex in $H$ but not in
    the image of $f$.  Then there
    are $!$-graph monomorphisms
    \begin{align*}
        \DROP_b(f) &: G \rightarrow \DROP_b(H) \\
        \KILL_b(f) &: G \rightarrow \KILL_b(H) \\
        \COPY_b(f) &: G \rightarrow \COPY_b(H)
    \end{align*}
    that reflect $!$-box containment and that commute with $f$ in the
    following ways:
    \begin{equation}
        \begin{tikzcd}[column sep=huge]
            G \arrow{r}{\DROP_b(f)} \arrow[swap]{dr}{f}
            & \DROP_b(H) \arrow[right hook->]{d}
            \\ & H
        \end{tikzcd}
    \end{equation}
    \begin{equation}
        \begin{tikzcd}[column sep=huge]
            G \arrow{r}{\KILL_b(f)} \arrow[swap]{dr}{f}
            & \KILL_b(H) \arrow[right hook->]{d}
            \\ & H
        \end{tikzcd}
    \end{equation}
    \begin{equation}\label{eq:bbox-m-rbc-copy-disjoint}
        \begin{tikzcd}[column sep=huge]
            G \arrow{r}{\COPY_b(f)} \arrow[swap]{dr}{f}
            & \COPY_b(H)
            \\ & H \arrow[swap]{u}{p^H_i}
        \end{tikzcd}
    \end{equation}
    where $i \in \{1,2\}$ and $p^H_1$ and $p^H_2$ are the maps from the
    pushout that defines $\COPY_b(H)$.
\end{lemma}
\begin{proof}
    Since $b$ is not in the image of $f$ (and so neither are any of its
    edges), we let $\DROP_b(f)$ be $f$, but with codomain $\DROP_b(H)$.
    Simply restricting the codomain cannot break the property of
    reflecting $!$-box containment, so this is inherited from $f$.

    $\KILL_b(f)$ follows in the same way once we note that, since $f$
    reflects $!$-box containment, the image of $f$ cannot intersect
    $B(b)$.

    Recall the definition of $\COPY_b(H)$:
    \[ \posquare{H\graphminus B(b)}{H}{H}{\COPY_b(H)}{}{p^H_2}{}{p^H_1} \]
    Let $\iota$ be the induced inclusion of $H \graphminus B(b)$ into
    $\COPY_b(H)$ (equal to both composed maps in the diagram).

    Now we define
    \[ \COPY_b(f) := \iota \circ \KILL_b(f) \]
    and so we have a monomorphism that satisfies
    \eqref{eq:bbox-m-rbc-copy-disjoint}.  For the remainder of this
    proof, we will use $f'$ to refer to $\COPY_b(f)$ for convenience.

    Let $v$ be a vertex in $G$ and $e$ an edge from a $!$-vertex $b$ to
    $f'(v)$ in $\COPY_b(H)$.  $e$ must be in the image of at least one
    of the $p^H_i$.  Then (one of) its preimage(s) $e'$ is an edge from
    a $!$-vertex to $f(v)$ in $H$, and so $e'$ must be in the image of
    $f$ (as $f$ reflects $!$-box containment).  So $e$ is in the image
    of $f'$, and hence $f'$ also reflects $!$-box containment.
\end{proof}

Armed with this result, if $b$ is not in the image of $m$, we can do the
following:
\begin{equation}\label{eq:rewrite-op-disjoint}
    \begin{tikzcd}
        L \arrow[swap]{d}{\OP_b(m)}
        & I \arrow[swap]{l}{i_1}
                   \arrow{r}{i_2}
                   \arrow{d}{\OP_b(d)}
        & R \arrow{d}{\OP_b(r)}
        \\
        \OP_b(G)
        & \OP_b(D) \arrow{l}{\OP_b(g)}
                   \arrow[swap]{r}{\OP_b(h)}
        & \OP_b(H)
    \end{tikzcd}
\end{equation}
which we again need to show is a rewrite.  The first step (for both
\eqref{eq:rewrite-op-in-image} and \eqref{eq:rewrite-op-disjoint}) is to
show that the $!$-box operations preserve the property of being a
$!$-graph matching.

\begin{lemma}\label{lem:bbox-matching-ops-in-image}
    Let $f : G \rightarrow H$ be a $!$-graph matching and let $b$ be a
    $!$-vertex in $H$.  Then $\DROP_b(f)$, $\KILL_b(f)$ and $\COPY_b(f)$
    are all $!$-graph matchings.
\end{lemma}
\begin{proof}
    Lemmas \ref{lem:bbox-m-rbc-ops-disjoint} and
    \ref{lem:bbox-m-rbc-ops-in-image} give us that they are
    monomorphisms that reflect $!$-box containment.  So we just need to
    show that they are local isomorphisms.

    $U(\DROP_b(f)) = U(f)$, which we already know to be a local
    isomorphism.

    If $b$ is not in the image of $f$, we note that since $f$ reflects
    $!$-box containment, the image of $f$ cannot intersect $B(b)$.
    Then, since simply restricting the codomain cannot break the
    property of being a local isomorphism, $U(\KILL_b(f))$ must be a
    local isomorphism, since $U(f)$ is.

    Otherwise, suppose $b$ is in the image of $f$, and $U(\KILL_b(f))$
    is not a local isomorphism.  Let $n$ be a node-vertex in the image
    of $\KILL_b(f)$ with an incident edge $e$ not in its image.  We know
    that $e$ is in the image of $f$; let its preimage be $e'$.  Then the
    end of $e'$ that does not map to $n$ must be in $B(b')$; call this
    vertex $v$.  But, since $f(b') = b$, we must have that $f(v) \in
    B(b)$, and so $e$ cannot be in $\KILL_b(H)$ and hence there can be
    no such $n$.  So $U(\KILL_b(f))$ is a local isomorphism when $U(f)$
    is.

    Now we just have to deal with $\COPY_b(f)$.  For convenience, we
    will call this $f'$ for the remainder of the proof.

    Suppose $b$ is not in the image of $f$, and $U(f')$ is not a local
    isomorphism.  So there is a node-vertex $n$ in $U(G)$ such that
    $U(f')(n)$ is incident to an edge $e$ not in the image of $U(f')$.
    $e$ must be in the image of at least one of the $p^H_i$.  Call (one
    of) its preimage(s) $e'$.  Then $e'$ must not be in the image of
    $U(f)$ (since otherwise $e$ would be in the image of $U(f')$).  But
    $e'$ is incident to $U(f)(n)$, and so $U(f)$ is not a
    local isomorphism.  Contrapositively, if $U(f)$ is is a local
    isomorphism, then so is $U(f')$.

    Now suppose $b$ is in the image of $f$.  Let $n$ be a vertex in
    $\COPY_{b'}(G)$ with an edge $e$ incident to $f'(n)$ in
    $U(\COPY_b(H))$.  $e$ is in the image of one of the $p_i^H$; call
    the preimage of $e$ under it $e'$.  $f'(n)$ must also be in the
    image of $p_i^H$, and its preimage must be in the image of $f$.
    What is more, the preimage of this under $f$, which we will call
    $n'$, must map to $n$ by $p_i^G$.
    \[ \input{matching-copy-in-image-fig1.tikz} \]
    We know that, since $U(f)$ is a local isomorphism, all the incident
    edges of $U(f)(n')$ must be in the image of $U(f)$, including $e'$.
    So $e$ is in the image of $p_i^H \circ f$, and hence in the image
    of $f' \circ p_i^G$.  So $e$ is in the image of $f'$, and hence
    $U(f')$ is a local isomorphism.
\end{proof}

We know that the two squares of \eqref{eq:rewrite-op-in-image} commute
(up to an isomorphism) by lemma \ref{lem:bbox-f-ops-composition}.  For
\eqref{eq:rewrite-op-disjoint}, we need the following lemma:

\begin{lemma}\label{lem:bbox-f-disjoint-ops-composition}
    If $f : G_1 \rightarrow G_2$, $g : G_2 \rightarrow G_3$ and $h : G_3
    \rightarrow G_4$ are $!$-graph monophisms that reflect
    $!$-box containment, and $b$ is a $!$-vertex of $G_3$ not in the
    image of $g$ then
    \[ \DROP_{b}(g \circ f) = \DROP_{b}(g) \circ f \]
    and
    \[ \DROP_{h(b)}(h \circ g) = \DROP_{h(b)}(h) \circ \DROP_b(g) \]
    and similarly for $\KILL$ and $\COPY$.
\end{lemma}
\begin{proof}
    We can deduce the first $\DROP$ case from the following diagram,
    where the right and outer triangles commute by definition, and hence
    the left triangle commutes, since the inclusion map on the right is
    monic.
    \begin{equation*}
        \begin{tikzcd}[ampersand replacement=\&]
            \&\&\DROP_b(G_3) \dar[right hook->]
            \\
            G_1 \arrow[bend left]{urr}{\DROP_b(g \circ f)}
                \arrow[swap]{r}{f}
            \& G_2 \arrow[sloped,pos=0.8]{ur}{\DROP_b(g)}
                \arrow[swap]{r}{g}
            \& G_3
        \end{tikzcd}
    \end{equation*}
    The second drop case follows in the same way from
    \begin{equation*}
        \begin{tikzcd}[ampersand replacement=\&,column sep=huge]
            \&\DROP_b(G_3) \dar[right hook->]
            \rar{\DROP_{h(b)}(h)}
            \&\DROP_b(G_4) \dar[right hook->]
            \\
            G_2 \arrow[bend left=45]{urr}{\DROP_{h(b)}(g \circ f)}
            \arrow[sloped,pos=0.8]{ur}{\DROP_b(g)}
                \arrow[swap]{r}{g}
            \& G_3 \arrow[swap]{r}{h}
            \& G_4
        \end{tikzcd}
    \end{equation*}
    and the $\KILL$ cases use the same argument.

    For $\COPY$, we recall that $\COPY_b(g \circ f)$ is defined to be
    $\iota_3 \circ \KILL_b(g \circ f)$, where $\iota_3 : \KILL_b(G_3)
    \rightarrow \COPY_b(G_3)$ is the inclusion map induced by the
    pushout square that defined $\COPY_b(G_3)$.  We know that this is
    then $\iota_3 \circ \KILL_b(g) \circ f$, which is $\COPY_b(g) \circ
    f$ as required.

    From the proof of \ref{lem:bbox-m-rbc-ops-in-image}, we note that
    \[ \COPY_{h(b)}(h) \circ \iota_3 = \iota_4 \circ \KILL_{h(b)}(h) \]
    where $\iota_4 : \KILL_b(G_4) \rightarrow \COPY_b(G_4)$ is defined
    similarly to $\iota_3$.  Then we have
    \begin{align*}
        \COPY_{h(b)}(h \circ g) &= \iota_4 \circ \KILL_{h(b)}(h \circ g)
        \\
        &= \iota_4 \circ \KILL_{h(b)}(h) \circ \KILL_b(g) \\
        &= \COPY_{h(b)}(h) \circ \iota_3 \circ \KILL_b(g) \\
        &= \COPY_{h(b)}(h) \circ \COPY_b(g)
    \end{align*}
\end{proof}

From this, we can deduce that the two squares of
\eqref{eq:rewrite-op-disjoint} commute.  All we need now is that both
squares in each diagram are pushouts.  For this, we recall that pushouts
of graphs are just unions of graphs.  So a commuting square of
monomorphisms
\[ \begin{tikzcd}
    A \arrow{r}{f} \arrow[swap]{d}{g} & B \arrow{d}{h} \\
    C \arrow[swap]{r}{i} & D
\end{tikzcd} \]
is a pushout if and only if $h$ and $i$ cover $D$ and the image of $h
\circ f = i \circ g$ is exactly the intersection of the images of $h$
and $i$.  To get that this is the case for
\eqref{eq:rewrite-op-disjoint}, we use the following lemma.

\begin{lemma}\label{lem:m-bbox-ops-images}
    Let $f : G \rightarrow H$ be a $!$-graph monomorphism that reflects
    $!$-box containment, and let $b \in\:!(H)$.  If $\iota_D :
    \DROP_b(H) \rightarrow H$ is the normal inclusion, then for any
    vertex or edge $x$ of $\DROP_b(H)$, $\iota_D(x)$ is in the image of
    $f$ if and only $x$ is in the image of $\DROP_b(f)$, and similarly
    for $\KILL_b$, and if $p_i^H : H \rightarrow \COPY_b(H)$ is one of
    the pushout maps from the definition of $\COPY_b(H)$, for any vertex
    or edge $x$ of $H$, $p_i^H(x)$ is in the image of $\COPY_b(H)$ if
    and only if $x$ is in the image of $f$.
\end{lemma}
\begin{proof}
    If $b$ is not in the image of $f$, the arguments for all the
    operations are very similar; we will use $\DROP$ as the
    example.  We know that the following diagram commutes (from lemma
    \ref{lem:bbox-m-rbc-ops-disjoint}):
    \begin{equation*}
        \begin{tikzcd}[column sep=huge]
            G \arrow{r}{\DROP_b(f)} \arrow[swap]{dr}{f}
            & \DROP_b(H) \arrow[right hook->]{d}{\iota}
            \\ & H
        \end{tikzcd}
    \end{equation*}
    Let $x$ be a vertex or edge of $\DROP_b(H)$.  If $x$ is in the image
    of $\DROP_b(f)$, $\iota(x)$ must be in the image of $f$ by
    commutativity of the diagram.  Conversely, if $\iota(x)$ is in the
    image of $f$, its preimage must map to $x$ under $\DROP_b(H)$, since
    $\iota$ is monic (and hence injective) and the diagram commutes.

    So suppose $b$ is in the image of $f$, and let $b'$ be its preimage.
    We will start with the $\DROP$ case again.  We have the following
    diagram (from lemma \ref{lem:bbox-m-rbc-ops-in-image}):
    \begin{equation*}
        \begin{tikzcd}[column sep=huge]
            \DROP_{b'}(G) \arrow{r}{\DROP_b(f)}
                       \arrow[right hook->,swap]{d}{\iota_G}
            & \DROP_b(H) \arrow[right hook->]{d}{\iota_H}
            \\ G \arrow{r}{f} & H
        \end{tikzcd}
    \end{equation*}
    If $x$ is a vertex or edge of $\DROP_b(H)$ that is in the image of
    $\DROP_b(f)$, then $\iota_H(x)$ is in the image of $f$, by
    commutativity of the diagram.  For the converse, we will consider
    vertices and edges separately.

    Suppose $v$ is a vertex of $\DROP_b(H)$ with $\iota_H(v)$ in the
    image of $f$.  Since $f$ is monic, there is a unique $v'$ with
    $f(v') = \iota_H(v)$.  Now $\iota_H(v)$ cannot be $b$, so $v'$
    cannot be $b'$, and hence $v'$ is in the image of $\iota_G$.  But
    then the preimage of $v'$ under $\iota_G$ must map to $v$ by
    $\DROP_b(H)$, since all the maps are monic and the diagram commutes,
    so $v$ is in the image of $\DROP_b(H)$.

    Now suppose $e$ is an edge of $\DROP_b(H)$ with $\iota_H(e)$ in the
    image of $f$.  Again, there is a unique preimage $e'$ under $f$, and
    the same must be true of the endpoints of the edge.  So, by what we
    have already proved, the endpoints of $e'$ must be in the image of
    $\iota_G$, and hence $e'$ must be (since $\iota_G$ is full in $G$),
    and we have what we require by commutativity of the diagram.

    The $\KILL$ case is almost identical, except that in the vertex
    case, we must note that $v'$ is not in $B(b')$, as $v$ is not in
    $B(b)$.

    The $\COPY$ case is also similar.  This time we have the following
    diagram (actually two diagrams, with $i \in \{1,2\}$).
    \begin{equation*}
        \begin{tikzcd}[column sep=huge]
            \COPY_{b'}(G) \arrow{r}{\COPY_b(f)}
            & \COPY_{b}(H)
            \\
            G \arrow{r}{f} \arrow{u}{p^G_i}
            & H \arrow[swap]{u}{p^H_i}
        \end{tikzcd}
    \end{equation*}
    As in the $\DROP$ case, if $x$ is a vertex or edge of $H$ that is in
    the image of $f$, then $p^H_i(x)$ is in the image of $\COPY_b(f)$.

    Suppose $v$ is a vertex of $H$ with $p^H_1(v)$ in the image of
    $\COPY_b(f)$, which we will denote $f'$ (the case for $p^H_2$ is
    similar).  We then have a vertex $v'$ of $\COPY_{b'}(G)$ that maps
    to $p^H_1(v)$ by $f'$.  Now $v'$ is in the image of at least one of
    $p^G_1$ and $p^G_2$.  If $v'$ is in the image of $p^G_2$, consider
    its preimage $u$ in $G$.  $f'(p^G_2(u)) = p^H_2(f(u)) = p^H_1(v)$,
    and hence $v$ and $f(u)$ are both in the image of the inclusion
    $\iota_H : H \graphminus B(b) \rightarrow H$, by the pushout that
    defines $\COPY_b(H)$, with a common preimage.  But this means they
    must, in fact, be the same vertex, and hence $v$ is in the image of
    $f$.  The only other option is that $v'$ is in the image of $p^G_1$;
    we will again call the preimage $u$.  But then we have $f(u) = v$ by
    the above diagram of monomorphisms, and so $v$ is in the image of
    $f$.

    The argument for edges is similar to the $\DROP$ case.
\end{proof}

\begin{theorem}\label{thm:rw-bbox-op-preservation}
    Consider the following rewrite, where $m$ is a $!$-graph
    matching,
    \begin{equation*}
        \begin{tikzcd}
            L \arrow[swap]{d}{m}
            & I \arrow[swap]{l}{i_1} \arrow{r}{i_2} \arrow{d}{d}
            & R \arrow{d}{r}
            \\
            G \NEbracket
            & D \arrow{l}{g} \arrow[swap]{r}{h}
            & H \NWbracket
        \end{tikzcd}
    \end{equation*}
    and let $b$ be a $!$-vertex of $G$ (we will also consider it to be a
    $!$-vertex of $D$ and $H$).  Then for any $!$-box operation $\OP_b$
    on $b$, if $b$ is not in the image of $m$, the following is a
    rewrite
    \begin{equation*}
        \begin{tikzcd}
            L \arrow[swap]{d}{\OP_b(m)}
            & I \arrow[swap]{l}{i_1}
                       \arrow{r}{i_2}
                       \arrow{d}{\OP_b(d)}
            & R \arrow{d}{\OP_b(r)}
            \\
            \OP_b(G)
            & \OP_b(D) \arrow{l}{\OP_b(g)}
                       \arrow[swap]{r}{\OP_b(h)}
            & \OP_b(H)
        \end{tikzcd}
    \end{equation*}
    and if $b$ is in the image of $m$, we consider it to be a
    $!$-vertex of $L$, $I$ and $R$, and the following is a rewrite:
    \begin{equation*}
        \begin{tikzcd}
            \OP_b(L) \arrow[swap]{d}{\OP_b(m)}
            & \OP_b(I) \arrow[swap]{l}{\OP_b(i_1)}
                       \arrow{r}{\OP_b(i_2)}
                       \arrow{d}{\OP_b(d)}
            & \OP_b(R) \arrow{d}{\OP_b(r)}
            \\
            \OP_b(G)
            & \OP_b(D) \arrow{l}{\OP_b(g)}
                       \arrow[swap]{r}{\OP_b(h)}
            & \OP_b(H)
        \end{tikzcd}
    \end{equation*}
\end{theorem}
\begin{proof}
    We have already established that both of these are commuting
    diagrams of $!$-graph monomorphisms, that the top span in each
    diagram is a $!$-graph rewrite rule and that $\OP_b(m)$ is a
    $!$-graph matching.  All we have left to show is that all the
    squares are pushouts.

    We will first consider the $\DROP_b$ case.  Let $\iota_G :
    \DROP_b(G) \rightarrow G$ be the inclusion map, and similarly for
    the other $!$-graphs.  Suppose we have a vertex or edge $x$
    in $\DROP_b(G)$. Then $\iota_G(x)$ must be in the image of at least
    one of $m$ or $g$, and so $x$ must (by lemma
    \ref{lem:m-bbox-ops-images}) be in the image of at least one of
    $\DROP_b(m)$ or $\DROP_b(g)$.  So these maps cover $\DROP_b(G)$.
    If $x$ is in the image of both these maps, $\iota_G(x)$ must be in
    the image of both $m$ and $g$.  Let $x'$ be the preimage of $x$
    under $\DROP_b(g)$.  Then $\iota_D(x')$ is the preimage of
    $\iota_G(x)$ under $g$, which is in the image of $d$, since this is
    a pushout of graphs.  But then applying lemma
    \ref{lem:m-bbox-ops-images} again gives us that $x'$ must be in the
    image of $\DROP_b(d)$.  Then commutativity of the diagrams and the
    fact that all the maps are monic gives us that the preimage of $x$
    under $\DROP_b(m)$ must be in the image of $i_1$ or $\DROP_b(i_1)$
    (depending on whether $b$ is in the image of $m$).  So $\DROP_b(m)$
    and $\DROP_b(g)$ intersect exactly on the part of the graph mapped
    to from $I$ (or $\DROP_b(I)$), and so the left part of the diagram
    is a pushout.  The same argument applies to the right squares, and
    so we have the result for $\DROP_b$.  $\KILL_b$ follows in the same
    way.

    The argument for $\COPY_b$ is similar; here we have to note that if
    we have $x$ in $\COPY_b(G)$, it must be in the image of at least one
    of the pushout maps from $G$ that define $\COPY_b(G)$.  We can call
    this map $p_1^G$ and its preimage $x'$.  Then, as before, we use
    lemma \ref{lem:m-bbox-ops-images} to get that $x$ is the in the
    image of at least one of $\COPY_b(m)$ or $\COPY_b(g)$, since $x'$ is
    in the image of at least one of $m$ or $g$.  Similarly, if $x$ is in
    the image of both maps, it must be the image of something in
    $\COPY_b(I)$ or $I$ (depending on whether $b$ is in the image of
    $m$) since then $x'$ is in the image of both $m$ and $g$ and is
    hence the image of something in $I$.
\end{proof}

In the above theorem, let $G \rewritesto H$ be the induced $!$-graph
rewrite of the span $G \leftarrow D \rightarrow H$ in the rewrite.
It should be clear that the rewrite rule induced by the span
\[ \OP_b(G) \xleftarrow{\OP_b(g)} \OP_b(D) \xrightarrow{\OP_b(h)}
    \OP_b(H) \]
is exactly the result of applying $\OP_b$ to $G \rewritesto H$.  So we
have that if $G \rewritesto H$ is a $!$-graph rewrite derived from a
rewrite $L \rewritesto R$, $\OP_b(G \rewritesto H)$ is derived from
either $L \rewritesto R$ or $\OP_b(L \rewritesto R)$, and hence have the
following corollary:

\begin{corollary}
    \label{cor:rw-semantics}
    Suppose $L \rewritesto R$ is a $!$-graph rewrite rule that rewrites
    $G$ to $H$ at a $!$-graph matching, and $G \rewritesto H$ is the
    associated $!$-graph rewrite rule.  Then for any concrete instance
    $G' \rewritesto H'$ of this rule, there is a concrete instance of $L
    \rewritesto R$ that rewrites $G'$ to $H'$, and the associated string
    graph rewrite rule $G' \rewritesto H'$ is a concrete instance of $G
    \rewritesto H$.
\end{corollary}

This is the semantics of rewriting $!$-graphs using $!$-graph rewrite
rules: that the rewrite holds for all concrete instances of the
rewritten graphs.  In this way, just as $!$-graph rewrite rules allowed
us to represent inifinitely many string graph rewrite rules at once,
$!$-graph rewriting allows us to do infinitely many string graph
rewrites at once.

It is important to note that if we can apply an instantiation to a
rewrite, we can apply it to a sequence of rewrites.  If we consider two
consecutive rewrites \[G \rewritesto H \rewritesto J\] then we know that
$G$ and $H$ contain the same $!$-vertices, and the same goes for $H$ and
$J$.  If we can instantiate $G \rewritesto H$ to $G' \rewritesto H'$,
applying the same instantiation to $H \rewritesto J$ will produce a
rewrite from $H'$ to some graph $J'$, so we will have a rewrite sequence
\[G' \rewritesto H' \rewritesto J'\]

\begin{example}
    Recall the rewrite in example \ref{ex:bialg-spider-1}:
    \[ \input{bialg-spider-rewrite-pos-labelled.tikz} \]
    If we perform $\COPY_{b_1}$ and $\KILL_{b_4}$, we get the following
    rewrite:
    \[ \input{bialg-spider-rewrite-pos-exp.tikz} \]
\end{example}


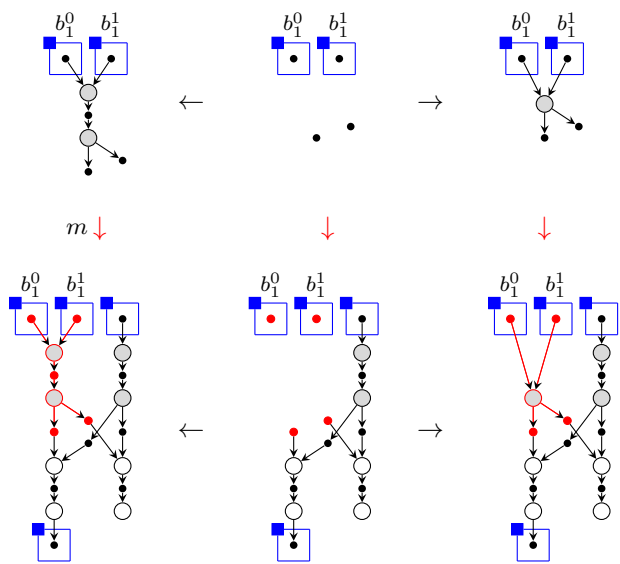
\chapter{A Logic of $!$-Graphs} 
\label{ch:rules}

In this chapter, we will use the results of chapter \ref{ch:rw-bang-graphs} to
build a logic of $!$-graphs that can be implemented by a computer, using graph
rewriting as the core operation.  This will include an equational logic much
like that in section \ref{sec:graph-eqs}, together with some inference rules
that build on that foundation.

In the term world, these rules would be based on predicates involving logical
operators rather than purely on equations (although an equation is, of course,
a type of predicate).  For example, a $\bigwedge$-introduction rule might
state that if you can provide a proof of a predicate $P$ and a proof of a
predicate $Q$, you can deduce that there is a proof of the predicate $P \wedge
Q$.  This would be written
\[ \inferrule{P \\ Q}{P \wedge Q} \]

Another common rule is induction, where $P$ is a predicate over the natural
numbers (or some other well-founded set):
\[ \inferrule{P(0) \\ P(n) \Rightarrow P(n+1)}{\forall n.P(n)} \]

Our graphical language does not contain logical connectives, nor does it
contain the natural numbers used in the above induction rule.  However,
$!$-boxes do induce a well-founded structure amenable to a graphical analogue
of induction.  In this chapter, we will present this and other rules, and
prove their soundness.

\section{$!$-Graph Equational Logic}
\label{sec:bg-eqs}

Let $E$ be a set of $!$-graph equations.  Our equational logic should look
very familiar:

\begin{mathpar}
    (\textsc{Axiom}) \enskip
    \inferrule{G \rweq_{i,j} H \in E}{E \vdash G \rweq_{i,j} H}
    \and
    (\textsc{Refl}) \enskip
    \inferrule{ }{E \vdash G \rweq_{b_G,b_G} G}
    \and
    (\textsc{Sym}) \enskip
    \inferrule{E \vdash G \rweq_{i,j} H}{E \vdash H \rweq_{j,i} G}
    \and
    (\textsc{Trans}) \enskip
    \inferrule{E \vdash G \rweq_{i,j} H \\ E \vdash H \rweq_{k,l} K
             }{E \vdash G \rweq_{p,q} K}
    \and
    (\textsc{Leibniz}) \enskip
    \inferrule{E \vdash G \rweq_{i,j} H}{E \vdash G' \rweq_{i',j'} H'}
    \and
    (\textsc{Homeo}) \enskip
    \inferrule{E \vdash G \rweq_{i,j} H}{E \vdash G^* \rweq_{i^*,j^*} H^*}
\end{mathpar}
where $p := i$ and $q := l \circ k^{-1} \circ j$,
\[ \begin{tikzcd}
    G \arrow[swap,right hook->]{d}
    & I \arrow[swap]{l}{i} \arrow{r}{j} \arrow[right hook->]{d}
    & H \arrow[right hook->]{d}
    \\ G' \NEbracket
    & D \arrow[left hook->]{l} \arrow[right hook->]{r}
    & H' \NWbracket
    \\
    & I' \arrow{ul}{i'} \arrow[swap]{ur}{j'}
         \arrow[right hook->]{u} &
\end{tikzcd} \]
is a $!$-graph rewrite (where, in particular, all morphisms reflect $!$-box
containment), and
\[ G \xleftarrow{i} I \xrightarrow{j} H \qquad \textrm{and} \qquad
   G^* \xleftarrow{i^*} I \xrightarrow{j^*} H^* \]
are wire-homeomorphic.

To these, we add inference rules for instantiations of $!$-graph equations:
\begin{mathpar}
    (\textsc{Copy}) \enskip
    \inferrule{E \vdash G \rweq_{i,j} H}{E \vdash \COPY_b(G \rweq_{i,j} H)}
    \and
    (\textsc{Drop}) \enskip
    \inferrule{E \vdash G \rweq_{i,j} H}{E \vdash \DROP_b(G \rweq_{i,j} H)}
    \and
    (\textsc{Kill}) \enskip
    \inferrule{E \vdash G \rweq_{i,j} H}{E \vdash \KILL_b(G \rweq_{i,j} H)}
\end{mathpar}

Since our rewriting process implements $\textsc{Axiom}$ and $\textsc{Leibniz}$
up to wire homeomorphism and up to instantiation of $!$-graph rewrite rules,
the above logic is equivalent to the relation $\rewriteequivin{E}$.

For a set of $!$-graph equations $E$, let $\overset{\sim}{E}$ be the set of
concrete instances of the equations of $E$.  Then the semantics of $E \vdash G
\rweq H$ we use is that, for any concrete instance $G' \rweq H'$ of $G \rweq
H$,
\[ \overset{\sim}{E} \vdash G' \rweq H' \]
in the equational logic of string graphs.  Soundness of the above logic with
respect to these semantics follows from corollary \ref{cor:rw-semantics} and
the fact that string graph rewriting implements the equational logic for
string graphs.

\section{Regular Forms of Instantiations}
\label{sec:bbox-op-seq-reg-forms}

A lot of the proofs in this chapter rely on reasoning about instantiations.
However, multiple instantiations can produce the same instance of a $!$-graph,
making them harder than necessary to reason about.  In particular, copying a
$!$-box and then killing one of the copies results in the original graph (up
to isomorphism).

In this section, we will demonstrate that it is possible to, if not completely
remove these redunancies from the set of allowed instantiations, at least
require a certain structure that is easy to reason about without affecting the
set of possible instances of a $!$-graph.

We will concentrate on instantiations of $!$-graphs, but the results extend to
$!$-graph equations.

\subsection{Depth-Ordered Form}

\begin{definition}[Depth]
    Let $b$ be a $!$-vertex of a $!$-graph, and $P(b) = \textrm{pred}(b)
    \graphminus b$ the set of parent $!$-vertices of $b$.  Then the
    \textit{depth} of $b$ is defined by
    \[ \delta(b) =
        \begin{cases}
            0 & P(b) = \varnothing \\
            1 + \max\{\delta(c) | c \in P(b)\} & \textrm{otherwise}
        \end{cases}
    \]

    The depth of a $!$-box operation is the depth of the $!$-vertex it
    operates on.
\end{definition}

Another way of viewing the depth of a $!$-vertex is that it is the longest
ancestor-path to a top-level $!$-vertex.

We can now define a \textit{depth-ordered form} for instantiations.

\begin{definition}[Depth-ordered form]
    A instantiation is \textit{depth-ordered} if no $!$-box operation of depth
    $n$ is preceded by an operation of depth greater than $n$.
\end{definition}

So a depth-ordered instantiation has some number of $!$-box operations of
depth $0$, followed by some number of depth $1$, then of depth $2$ and so on.

We will demonstrate that any instantiation has a depth-ordered form that is
equivalent in the sense that the resulting graphs are isomorphic.  To do this,
we will present an algorithm that converts any instantiation into a
depth-ordered one, where each change that the algorithm makes preserves the
resulting graph.

We will need some results about what changes we can make, but first we need a
way to talk coherently about those changes.  For example, most of the changes
we make will be commuting two operations in an instantiation.

An instantiation is a sequence of operations that progressively transforms a
graph.  The operations in that sequence cannot be divorced from the graphs
they operate on.  Consider an intermediate stage of a hypothetical
instantiation; the operations up to this point have produced the graph
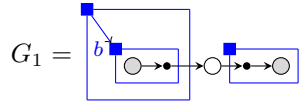
\[ \input{dnf-ex-1.tikz} \]
Suppose that the next operation in the sequence is $\COPY_b$.  Then the next
intermediate graph will be
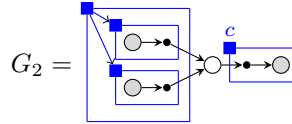
\[ \input{dnf-ex-1-copy.tikz} \]
The depth of $b$ in $G_1$, and hence of $\COPY_b$, is $1$.  If the next
operation is $\DROP_c$, it will have depth $0$, so the sequence cannot be
depth-ordered.

What we would like to do is exchange the order of $\COPY_b$ and $\DROP_c$,
while still reaching the same graph $G_3$:
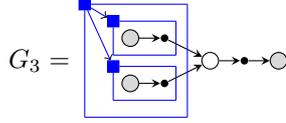
\[ \input{dnf-ex-1-copy-drop.tikz} \]
However, $c$ is a vertex of $G_2$, and does not actually exist in $G_1$.  $c$
\textit{is} derived from a $!$-vertex of $G_1$, though; which $!$-vertex is
obvious from the pictures of the graphs.

More generally, we can construct a morphism from $G_2$ to $G_1$ that maps each
vertex or edge to the one it was derived from.
\begin{definition}[Origin Map]
    Let $G_1$ be a $!$-graph, $b \in\;!(G_1)$ and $G_2 = \COPY_b(G_1)$.  Then
    $f : G_2 \rightarrow G_1$, as defined by the following pushout:
    \begin{equation}
        \begin{tikzcd}
            G_1 \graphminus B(b) \rar[right hook->] \dar[right hook->] &
            G_1 \dar \ar[bend left,double,-]{ddr} & \\
            G_1 \rar \ar[bend right,swap,double,-]{drr} &
            G_2 \ar{dr}{f} \NWbracket & \\
            && G_1
        \end{tikzcd}
        \label{eq:copy-origin-map}
    \end{equation}
    is called the \textit{origin map} for $\COPY_b$.

    The origin maps for $\KILL_b(G_1)$ and $\DROP_b(G_1)$ are the natural
    inclusion morphisms into $G_1$ (which exist because these operations are
    graph subtractions).
\end{definition}

Armed with the origin map $f : G_2 \rightarrow G_1$, we can apply
$\DROP_{f(c)}$ to $G_1$, obtaining another graph $G'_2$:
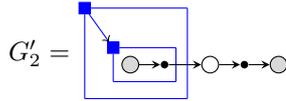
\[ \input{dnf-ex-1-drop.tikz} \]
Now we need to find the $!$-vertex of $G'_2$ that corresponds to $b$.  In this
case, since the origin map of a $\DROP$ operation is injective and $b$ is in
its image, it has a unique preimage $b'$ in $G'_2$
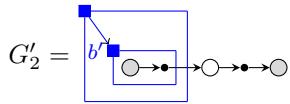
\[ \input{dnf-ex-1-drop-labelled.tikz} \]
and we can apply $\COPY_{b'}$ to $G'_2$ to get a graph $G'_3$
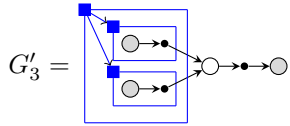
\[ \input{dnf-ex-1-drop-copy.tikz} \]
which is clearly isomorphic to $G_3$, and the origin maps for both sets of
sequences are the same.  $\COPY_b$ and $\DROP_c$ are an example of
\textit{commutable} operations.

\begin{definition}[Commutable operations]
    We say a sequence of two operations $\OP^1_b;\OP^2_c$ (with origin maps
    $f_1$ and $f_2$, respectively) is \textit{commutable} when $c$ is the
    only vertex that maps to $f_1(c)$ under $f_1$ and $b$ has a unique
    preimage under $f'_2$, the origin map of $\OP^2_{f_1(c)}$.
\end{definition}

Note that whether two operations are commutable depends on the $!$-vertices
they operate on.  Of course, if we say that the operations are commutable, we
expect them to have the same result in either order.

\begin{proposition}\label{prop:ops-commute}
    Let $H_1$ be a $!$-graph, $b \in\;!(H_1)$, $H_2 = \OP^1_b(H_1)$ with
    origin map $f_1$, $c \in\;!(H_2)$ and $H_3 = \OP^2_c(H_2)$ with origin map
    $f_2$.  If $\OP^1_b$ and $\OP^2_c$ are commutable with commuted form
    $\OP^2_{c'};\OP^1_{b'}$ (with origin maps $f'_2$ and $f'_1$ respectively)
    resulting in $H'_3$, then $H_3 \cong H'_3$, and this commutes $f_1 \circ
    f_2$ and $f'_2 \circ f'_1$.
\end{proposition}
\begin{proof}
    If both operations are $\DROP$ or $\KILL$, we just note that the net
    effect is to remove the union of the subgraphs, regardless of the order.
    For example, if $\OP^1$ is $\DROP$ and $\OP^2$ is $\KILL$,
    \[ H_3 = H'_3 = H_1 \graphminus (\{b\} \cup B(f_1(c))) \]

    If $\OP^1$ is $\COPY$, the conditions on commutable operations mean that
    $f_1(c) \notin B(b)$.  Suppose $\OP^2$ is $\DROP$.  Consider the effect of
    $\DROP_{f_1(c)};\COPY_b$ on $H_1$:
    \[
        \begin{tikzcd}
            (H_1 \graphminus \{f_1(c)\}) \graphminus B(b) \rar \dar &
            H_1 \graphminus \{f_1(c)\} \dar \\
            H_1 \graphminus \{f_1(c)\} \rar &
            H'_3 \NWbracket
        \end{tikzcd}
    \]
    Since $(H_1 \graphminus \{f_1(c)\}) \graphminus B(b) = (H_1 \graphminus
    B(b)) \graphminus \{f_1(c)\}$, $H'_3$ must be the same as $H_3$, ie:\\
    ${\COPY_b(H_1) \graphminus \{\iota(f_1(c))\}}$, where $\iota$ is the
    inclusion of $H_1 \graphminus B(b)$ into $\COPY_b(H_1)$.

    A similar argument works when $\OP^2$ is $\KILL$.  We have
    \[
        \begin{tikzcd}
            (H_1 \graphminus B(b)) \graphminus B(f_1(c)) \rar \dar &
            H_1 \graphminus B(f_1(c)) \dar \\
            H_1 \graphminus B(f_1(c)) \rar &
            H'_3 \NWbracket
        \end{tikzcd}
    \]
    and, because every vertex in $B(\iota(f_1(c)))$ in $\COPY_b(H_1)$ must be
    in $B(f_1(c))$ in all copies of $H_1$ it is contained in, $H'_3$ must
    again be the same as ${\COPY_b(H_1) \graphminus B(\iota(f_1(c)))}$.

    These arguments also work when $\OP^2$ is $\COPY$ and $\OP^1$ is $\DROP$
    or $\KILL$.  That just leaves the case where both operations are $\COPY$.

    In their original order, we have two pushouts:
    \[
        \begin{tikzcd}
            H_1 \graphminus B(b) \rar \dar &
            H_1 \dar{p_1} \\
            H_1 \rar[swap]{p_2} &
            H_2 \NWbracket
        \end{tikzcd}
        \qquad
        \begin{tikzcd}
            H_2 \graphminus B(c) \rar \dar &
            H_2 \dar{q_1} \\
            H_2 \rar[swap]{q_2} &
            H_3 \NWbracket
        \end{tikzcd}
    \]
    Once commuted, we have
    \[
        \begin{tikzcd}
            H_1 \graphminus B(f_1(c)) \rar \dar &
            H_1 \dar{p'_1} \\
            H_1 \rar[swap]{p'_2} &
            H'_2 \NWbracket
        \end{tikzcd}
        \qquad
        \begin{tikzcd}
            H'_2 \graphminus B(b') \rar \dar &
            H'_2 \dar{q'_1} \\
            H'_2 \rar[swap]{q'_2} &
            H'_3 \NWbracket
        \end{tikzcd}
    \]
    Note that the requirements for commutability of operations mean that $c$
    is in the image of both $p_1$ and $p_2$, and $b'$ is likewise in the image
    of both $p'_1$ and $p'_2$.  We construct the isomorphism from $H_3$ to
    $H'_3$ by initially constructing a bijective map $\phi_v$ from the
    vertices of $H_3$ to the vertices of $H'_3$.

    Let $v_3$ be a vertex of $H_3$, $v_2 = f_2(v_3)$ and $v_1 = f_1(v_2)$.  We
    then choose $v'_2$ to be $p'_1(v_1)$ if $v_3$ is in the image of $q_1$ and
    $p'_2(v_1)$ if $v_3$ is in the image of $q_2$.  Likewise, $v'_3 =
    q'_1(v'_2)$ if $v_2$ is in the image of $p_1$ and $v'_3 = q'_2(v'_2)$ if
    $v_2$ is in the image of $p_2$.  We then set $\phi_v(v_3) = v'_3$.

    First, we need to check that we have defined a coherent function.  If
    $v_3$ is in the image of both $q_1$ and $q_2$, $v_2$ cannot be in $B(c)$.
    Then $v_1$ is not in $B(f_1(c))$, since if it were, whichever of $p_1$ and
    $p_2$ maps $v_1$ to $v_2$ must also map the edge from $f_1(c)$ to $v_1$ to
    an edge from $c$ to $v_2$.  Thus $p'_1(v_1) = p'_2(v_1)$.  Likewise, if
    $v_2$ is in the image of both $p_1$ and $p_2$, $v'_2$ cannot be in $B(b')$
    and $q'_1(v'_2) = q'_2(v'_2)$.  So $\phi_v$ is a valid function.

    Note that the definition of an origin map means that $f'_1(v'_3) = v'_2$
    and $f'_2(v'_2) = v_1$.  Thus the same construction works in reverse,
    allowing us to construct the inverse map $\phi^{-1}_v$, and so $\phi_v$ is
    a bijection.

    Since $!$-graphs are simple, we just need to show that, for any two
    vertices $v_3$ and $w_3$ of $H_3$, there is an edge from $v_3$ to $w_3$ if
    and only if there is one from $\phi_v(v_3)$ to $\phi_v(w_3)$.  We
    construct a $w$ family of vertices in the same way we constructed the $v$
    family.

    If there is an edge $e_3$ from $v_3$ to $w_3$, this edge is mapped to an
    edge $e_1$ from $v_1$ to $w_1$ by $f_1 \circ f_2$.  The existence of $e_3$
    means that $v_3$ and $w_3$ must both be in the image of $q_1$ or both be
    in the image of $q_2$.  Thus $v'_2$ and $w'_2$ must both be in the image
    of $p'_1$ or both be in the image of $p'_2$, and that morphism must map
    $e_1$ to an edge $e'_2$ from $v'_2$ to $w'_2$.  Similarly, either $q'_1$
    or $q'_2$ must map $e'_2$ to an edge from $v'_3$ to $w'_3$.

    The converse case is symmetric, and hence $\phi_v$ extends to an
    isomorphism $\phi : H_3 \cong H'_3$.
\end{proof}

We have generalised commutable operations beyond just being on entirely
separate parts of the graph, like in the example above, but we still have to
deal with non-commutable operations.  If the operations are already
depth-ordered, such as if the next operation were $\DROP_e$
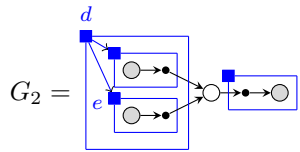
\[ \input{dnf-ex-1-copy-a.tikz} \]
we can safely ignore them.

Instead, we will look at operations where the preimage of the $!$-vertex the
second operation acts on is a parent of the $!$-vertex the first acts on.  For
example, suppose the next operation on $G_2$ were $\KILL_d$. Then the
operations would not be depth ordered ($\KILL_d$ has depth $0$, like $\DROP_c$
did, while $\COPY_b$ has depth $1$), but they cannot be commuted; once again,
$b$ is not in the image of the origin map of $\KILL_{f(d)}$.
\[ \input{dnf-ex-1-kill.tikz} \]
In this case, we can simply discard the $\COPY_c$ operation, since $\KILL_d$
is clearly isomorphic to $\KILL_{f(d)}$.  The same would be true if the first
operation were $\DROP_c$ or $\KILL_c$.

\begin{lemma}\label{lem:kill-cancel}
    Let $H_1$ be a $!$-graph, $b \in\;!(H_1)$, $H_2 = \OP_b(H_1)$ with
    origin map $f$, $c \in\;!(H_2)$ and $H_3 = \KILL_c(H_2)$ with origin map
    $k$, where $b \in B(f(c))$.  Then $H_3$ is isomorphic to $\KILL_{f(c)}$
    (with origin map $k'$), and this isomorphism commutes $k'$ and $f \circ
    k$.
\end{lemma}
\begin{proof}
    If $\OP$ is $\DROP$ or $\KILL$, this has the same net effect, since $b \in
    B(f(c))$ and hence $B(b) \subseteq B(f(c))$.  If it is $\COPY$, it is
    sufficient to note that
    \[ H_1 \graphminus B(f(c)) \subseteq H_1 \graphminus B(b) \]
    and hence the $\COPY$ operation has no effect on the part of the graph
    preserved by the $\KILL$ operation.
\end{proof}

If the next operation is $\COPY_d$, $b$ will have \textit{two} preimages under
the origin map of $\COPY_{f(d)}$, which we will call $b^0$ and $b^1$.  Then
the replacement sequence will be
\[ \COPY_{f(d)};\COPY_{b^0};\COPY_{b^1} \]

\begin{lemma}
    Let $H_1$ be a $!$-graph, $b \in\;!(H_1)$, $H_2 = \OP_b(H_1)$ with
    origin map $f$, $c \in\;!(H_2)$ and $H_3 = \COPY_c(H_2)$ with origin map
    $g$, where $b \in B(f(c))$.

    Let $H_2' = \COPY_{f(c)}(H_1)$, with origin map $g'$.  Then $b$ has two
    preimages under $g'$, $b^0$ and $b^1$.  Let $H_3' = \OP_{b^0}$, with
    origin map $f'_0$.  $b^1$ has a unique preimage under $f'_0$, which we
    call $b'^1$.  Let $H_4' = \OP_{b'^1}$, with origin map $f'_1$.  Then
    $H_4'$ is isomorphic to $H_3$, and this isomorphism commutes $f \circ g$
    and $g' \circ f'_0 \circ f'_1$.
\end{lemma}
\begin{proof}
    Intuitively, whatever we did to $b$ before the $\COPY$ operation, we do
    instead to both copies of $b$ after it.

    If $\OP$ is $\DROP$, we remove both $b^0$ and $b^1$, and this is the same
    as removing just $b$ before applying the $\COPY$.  If $\OP$ is $\KILL$,
    the preimage of $B(b)$ under $g'$ is exactly the union of $B(b^0)$ and
    $B(b^1)$, and hence the net effect is likewise maintained.  We note that
    $B(b^0) \cap B(b^1) = \varnothing$, and then the $\COPY$ case follows from
    a similar proof to that of two commutable $\COPY$ operations (in
    proposition \ref{prop:ops-commute}).
\end{proof}

We have not dealt with the second operation being $\DROP_d$, because this
commutes with $\COPY_c$, and indeed $\DROP$ commutes with any preceding
operation on a distinct $!$-vertex, including child $!$-vertices.

When we are dealing with rewrite sequences longer than two operations,
explicitly referring to the origin maps will become extremely cumbersome.
Therefore, if we have an origin map $f: G_2 \rightarrow G_1$ and a $!$-vertex
$b$ of $G_1$ with a unique preimage under $f$, we will also refer to this
preimage as $b$.  If it has two preimages under $f$, we will refer to one as
$b^0$ and the other as $b^1$.

We can now set out algorithm \ref{alg:depth-ordering}, which takes an
arbitrary instantiation of a $!$-graph and gives a depth-ordered one that
produces the same graph.

\begin{algorithm}[h]
    \caption{Order $!$-box operations by depth}
    \label{alg:depth-ordering}
\begin{algorithmic}
    \WHILE{has unmarked operations}
        \STATE $\mathit{op} \leftarrow$ rightmost unmarked operation
        \STATE $\mathit{b} \leftarrow$ $!$-box $\mathit{op}$ operates on
        \STATE look at $\mathit{op}'$, immediately to right of $\mathit{op}$, operating on $\mathit{b}'$
        \IF[case $1$]{there is no such $\mathit{op}'$}
            \STATE mark $\mathit{op}$
        \ELSIF[case $2$]{$\delta(b)$ $\leq$ $\delta(b')$}
            \STATE mark $\mathit{op}$
        \ELSIF[case $3$]{$b'$ is not a parent of $b$}
            \STATE commute $\mathit{op}$ and $\mathit{op}'$
        \ELSIF[case $4$]{$\mathit{op}'$ is $\KILL$}
            \STATE discard $\mathit{op}$
        \ELSIF[case $5$]{$\mathit{op}'$ is $\DROP$}
            \STATE commute $\mathit{op}$ and $\mathit{op}'$
        \ELSIF[case $6$]{$\mathit{op}'$ is $\COPY$}
            \STATE remove $\mathit{op}$ and add two copies of $\mathit{op}$ to the
                right of $\mathit{op}'$, acting on the two copies of $\mathit{b}$
        \ENDIF
    \ENDWHILE
\end{algorithmic}
\end{algorithm}

The preceding discussion gives rise to the following proposition, stating that
the algorithm does not alter the effect of the instantiation:
\begin{proposition}
    Applying algorithm \ref{alg:depth-ordering} to an instantiation that
    results in a graph $G$ yields (if it terminates) an instantiation that
    results in a graph isomorphic to $G$.
\end{proposition}

The main invariants of the algorithm are that the marked operations are
depth-ordered relative to each other, and that no unmarked operation of depth
$d$ is ever to the right of a marked operation of depth $> d$.  This can be
seen by noting that $\mathit{op}$ is always unmarked and everything to the
right of it, including $\mathit{op}'$, is marked.  The first part of the
invariant is preserved because the algorithm never alters the order of the
marked operations, and an operation is only marked when the marked operations
to its right have the same or greater depth (and the second part of the
invariant requires that the marked operations to the left have the same or
lesser depth).  The second part of the invariant is maintained because we only
ever move $\mathit{op}$ to the right of $\mathit{op}'$ when
$\delta(\mathit{op}) > \delta(\mathit{op}')$ and while this can reduce the
depth of $\mathit{op}$, it can only reduce it to $\delta(\mathit{op}')$.

\begin{proposition}
    Algorithm \ref{alg:depth-ordering} terminates.
\end{proposition}
\begin{proof}
    We will place a (tight) upper bound on the number of iterations, as well
    as on the length of the final instantiation.  Specifically, given the
    recurrence relations
    \begin{align*}
        s_0 &= 0 \\
        s_{i+1} &= s_i + 2^{s_i} \\
        c_0 &= 0 \\
        c_{i+1} &= c_i + 2^{s_i} - 1
    \end{align*}
    we will show that if the input instantiation is of length $i$, the output
    instantiation is no longer than $s_i$ operations and the algorithm
    takes no more than $t_i = c_i + s_i$ iterations.  We will do this by
    induction on $i$.  The $i = 0$ case is trivial.

    The main thing to note for the step case is that if we have the
    instantiation $S = \OP_b;T$, because the algorithm always deals with the
    rightmost unmarked operation (and the operation to its right), it will not
    consider $\OP_b$ until it is the last remaining unmarked operation.  This
    means that the first part of a run of the algorithm on $S$ is the same as
    a run of the algorithm on $T$.

    So suppose the length of $S$ is $i + 1$.  By the inductive hypothesis,
    after at most $t_i$ iterations, the instantiation will look like $\OP_b;T'$,
    where $T'$ consists of at most $s_i$ marked (and no unmarked) operations,
    and $\OP_b$ remains unmarked.

    Let $l$ be the actual length of $T'$.  We will show that after at most
    $2^l - 1$ iterations that hit cases 3-6, the instantiation will be
    depth-ordered.  Since only case 6 iterations can introduce new operations,
    and only one at that, there can be at most $2^l$ case 1 or 2 iterations
    before the entire sequence is marked.  So there are at most $l + 2^l \leq
    s_i + 2^{s_i}$ operations in the final sequence, and at most
    \[ t_i + 2^l + 2^l - 1 = s_i + 2^l + c_i + 2^l -1 \leq s_{i+1} + c_{i+1} =
    t_{i+1} \]
    iterations in total.

    We need to note here that any operations in the sequence, marked or
    unmarked, at any subsequent point that are not part of $T'$ must be
    $\OP_b$ or a copy of $\OP_b$.  If two copies of $\OP_b$ are adjacent, they
    must have the same depth.  This is because the depth of a copy of $\OP_b$
    can only be changed by being commuted with a $\DROP_c$.  But any other
    copy of $\OP_b$ adjacent to it must have been commuted with the same
    $\DROP_c$, and its depth will have been changed in the same way.
    It follows that cases 3-6 can only be triggered if $\mathit{op}'$ is a
    member of $T'$.

    Suppose we are at some iteration of the algorithm (at or after the point
    where we have $\OP_b;T'$).  Consider $\mathit{op}$, the rightmost unmarked
    operation in the sequence.  Call the subsequence to the left of
    $\mathit{op}$ $S_1$, and the subsequence to its right $S_2$, so the whole
    sequence is $S_1;\mathit{op};S_2$.  Let $n$ be the number of operations of
    $T'$ that are in $S_2$.

    We will show by induction on $n$ that, in no more than $2^n - 1$ case 3-6
    iterations (and some number of case 1 or 2 iterations), the instantiation
    will be $S_1$ followed by marked operations.

    For $n = 0$, either $S_2$ is empty, in which case we trigger case 1 and
    are done, or $S_2$ consists only of operations derived from $\OP_b$, which
    must therefore have the same depth as $\mathit{op}$.  This means we
    trigger case 2, and again are done.

    Suppose $n = k + 1$.  If cases 1, 2 or 4 are triggered, we are done.
    Cases 3 and 5 require the inductive hypothesis, but are straightforward.
    So suppose we hit case 6.  We now get two copies of $\mathit{op}$, each of
    which has $k$ operations from $T'$ to the right of it.  We apply the
    inductive hypothesis to the right hand one to get that in at most $2^k -
    1$ case 3-6 iterations, the left hand one will be the rightmost unmarked
    operation.  As already discussed, this will still have $k$ operations from
    $T'$ to the right of it, so after no more than $2^k - 1$ further case 3-6
    iterations, we will have the sequence we want ($S_1$ followed by only
    marked operations).  Then we see that
    \[ 1 + (2^k - 1) + (2^k - 1) = 2^{k+1} - 1 \]
    which is what we wanted.
\end{proof}

\begin{corollary}
    \label{cor:depth-ordered-instantiation}
    Any instance $H$ of a $!$-graph $G$ has a depth-ordered instantiation.
\end{corollary}

The upper bound for this algorithm may appear worryingly high, but it is not
an algorithm that would ever be implemented.  The fact that such an algorithm
exists, however, allows us to assume that any instantiation is depth-ordered;
if it is not, we can produce an equivalent one that is.  In fact, in chapter
\ref{ch:implementation} we will demonstrate matching techniques that directly
produce depth-ordered instantiations.

\subsection{Expansion-Normal Form}
\label{sec:expn-normal-form}

A depth-ordered concrete instantiation has only operations of depth $0$.  To
see this, note that any operation of depth greater than $0$ must be on a
$!$-box that has one or more parents, and at least one of those parents must
be of depth $0$ (by posetality of the subgraph of $!$-vertices).  However,
since the instantiation is depth-ordered, there can be no further operations
on that $!$-box, and so it must exist in the final graph, which cannot then be
a concrete graph.  So the instantiation cannot be concrete.

We will introduce $\EXP_b$ as a shorthand for $\COPY_b;\DROP_{b^1}$.
As a notational convenience, we will consider $b^0$ to be the same as $b$,
since it is the unique preimage of $b$ under the composition of the origin
maps.
\[ \input{enf-ex-1.tikz} \]

\begin{definition}[Expansion-normal form]
    A concrete instantiation is in \textit{expansion-normal form} when it is
    composed entirely of $\EXP$ and $\KILL$ operations, and the depth of every
    operation is $0$.
\end{definition}

\begin{example}
    A sequence in expansion normal form:
    \[ \input{enf-ex-2.tikz} \]
    Every operation is on a top-level $!$-box (depth $0$), and the only
    operations are $\EXP$ and $\KILL$.
\end{example}

We generally identify the ``constant'' parts of the graph for each operation
(in a manner consistent with the origin maps), and view $\EXP$ operations as
adding to the graph and $\KILL$ operations as removing part of the graph.

\begin{theorem}
    \label{thm:enf}
    Any concrete instance $H$ of a $!$-graph $G$ has an instantiation in
    expansion-normal form.
\end{theorem}
\begin{proof}
    We know from corollary \ref{cor:depth-ordered-instantiation} that there
    must be a depth-ordered instantiation of $H$ from $G$; we will transform
    this into an equivalent instantiation (ie: resulting in the same graph, up
    to isomorphism) that is in expansion-normal form.

    In order to do this, we will need to be able to, for a given $\COPY_b$
    operation, move any operation on $b^0$ or $b^1$ to immediately after
    $\COPY_b$.  We will demonstrate this for arbitrary depth-ordered
    instantiations, not just concrete ones.

    Consider an operation $\OP_{b^0}$ (recalling that $b^0$ and $b^1$ are
    arbitrary labels, and hence interchangable).  The depth of $\OP_{b^0}$
    must be the same as the depth of $\COPY_b$.  The only way the depth of
    $\OP_{b^0}$ could deviate from that of $\COPY_b$ is if there were a
    $\DROP$ operation between them on a parent of $b^0$.  But the depth of
    this $\DROP$ operation would have to be strictly less than the depth of
    $\COPY_b$, which contradicts the depth-ordering.  Indeed, there can be no
    operations on any parent of $b^0$ after $\COPY_b$.

    Similarly, there can be no operations on any child of $b^0$ between
    $\COPY_b$ and $\OP_{b^0}$, since that operation would have a depth
    strictly greater than that of $b^0$, and hence of $\OP_{b^0}$, which
    violates depth-ordering.  Therefore, proposition \ref{prop:ops-commute}
    allows $\OP_{b^0}$ to be freely exchanged with the operation preceeding it
    until it reaches $\COPY_b$, and this will maintain the depth-ordering of
    the sequence.

    Note that this also allows us to move $\EXP_{b^0}$ back, since we can move
    $\COPY_{b^0}$ back to $\COPY_b$, and then move $\DROP_{b^{01}}$ back to
    $\COPY_{b^0}$.  Likewise, given $\EXP_c$, we can move any subsequent
    $\EXP_c$ back to join it, so we can move all occurrences of $\EXP_{b^0}$
    and $\EXP_{b^1}$ to immediately follow $\COPY_b$.

    Starting with a depth-ordered concrete instantiation, we can eliminate
    $\COPY$ (in favour of $\EXP$) with the following procudure (note that we
    treat $\EXP$ as an opaque operation in its own right, not a composite of
    $\COPY$ and $\DROP$), and so obtain a concrete instantiation in
    expansion-minimal form.

    We eliminate $\COPY$ operations starting from the right.  At each stage,
    we consider the rightmost $\COPY$ operation, $\COPY_b$.  If the next
    operations on $b^0$ and $b^1$ are both $\EXP$, we first move all the
    $\EXP_{b^0}$ and $\EXP_{b^1}$ operations to immediately after $\COPY_b$.
    We then replace all the $\EXP_{b^1}$ operations by $\EXP_{b^0}$, which is
    an equivalent operation in this context.

    This done, we know that, after $\COPY_b$, the first operation on $b^1$
    must be either $\KILL_{b^1}$ or $\DROP_{b^1}$.  In both cases, we move it
    back so that it immediately follows $\COPY_b$.  In the case of
    $\DROP_{b^1}$, we can just replace both operations with $\EXP_b$ (and
    replace all later references to $b^0$ with $b$).  Since $\COPY_b$ was
    depth $0$, $\EXP_b$ must also be depth $0$.  In the case of $\KILL_{b^1}$,
    we can eliminate both operations (and, again, replace all later references
    to $b^0$ with $b$).  Then we will have eliminated this occurrence of
    $\COPY$, and we can proceed onto the next rightmost occurrence, until none
    remain.

    Once we have eliminated all occurrences of $\COPY$, we can eliminate
    $\DROP$ operations by replacing each $\DROP_b$ with $\EXP_b;\KILL_b$,
    which both have depth $0$ since $\DROP_b$ does.  $\EXP_b;\KILL_b$ is
    $\COPY_b;\DROP_{b^0};\KILL_{b^1}$, and this is equivalent to
    $\COPY_b;\KILL_{b^1};\DROP_{b^0}$, which is equivalent to $\DROP_b$.

    So we now have a concrete instantiation containing only $\EXP$ and $\KILL$
    operations, and all the operations are of depth $0$, so it is in
    expansion-normal form.
\end{proof}

Once we have an instantiation in expansion-normal form, if $b$ is a top-level
$!$-vertex in the starting graph $G$, we can pull all the operations on $b$ to
the start of the instantiation.  So the instantiation will start with some
number of $\EXP_b$ operations, followed by $\KILL_b$ (the argument for this is
identical to the one that allowed us to move operations on $b^0$ back to
$\COPY_b$).

\section{$!$-box Introduction}
\label{sec:bbox-intro}

We will demonstrate a simple inference rule that allows us to wrap an entire
$!$-graph equation in a $!$-box.

\begin{definition}[$\BOX$]
    Suppose $G$ is a $!$-graph.  Then $\BOX(G)$ is the
    $\mathcal{G}_{T!}$-typed graph consisting of $G$ together with a fresh
    $!$-vertex $b$ and an edge from $b$ to every vertex in the graph
    (including itself).
\end{definition}

\begin{example}
    If we consider the graph $G$:
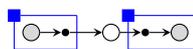
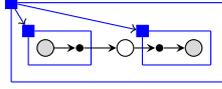
    \[ \input{bbox-intro-ex.tikz} \]
    then $\BOX(G)$ is
    \[ \input{bbox-intro-ex-boxed.tikz} \]
\end{example}

We will sometimes write $\BOX_b(G)$ to specify a name for the fresh
$!$-vertex; in contrast to the other $!$-box operations, $b$ is not in $!(G)$
but is in $!(\BOX_b(G))$.

$\BOX(G)$ is easily seen to be a $!$-graph: $U(\BOX(G)) \cong U(G)$ and hence
is a string graph; extending a poset with a fresh element that is $\leq$ every
element of the poset and to itself yields another poset; $U(B_{\BOX(G)}(c))
\cong U(B_G(c))$ for every $c \in\:!(G)$ and $U(B(b)) \cong U(G)$, so these
are all trivially open subgraphs of $U(\BOX(G))$; and $c' \in B(c) \Rightarrow
B(c') \subseteq B(c)$ follows from the fact this is true for $G$ and $B(b) =
\BOX(G)$.

As with the other operations, we can apply $\BOX$ to monomorphisms that
reflect $!$-box containment:

\begin{proposition}
    Let $f : G \rightarrow H$ be a monomorphism in $\catBGraph_T$ that
    reflects $!$-box containment.  Then $f$ can be uniquely extended to a
    $!$-box-reflecting monomorphism $\BOX(f) : \BOX(G) \rightarrow \BOX(H)$
    that commutes with the inclusions of $G$ into $\BOX(G)$ and $H$ into
    $\BOX(H)$.

    Further, if $g : H \rightarrow K$ is another such morphism, $\BOX(g \circ
    f) = \BOX(g) \circ \BOX(f)$.
\end{proposition}
\begin{proof}
    If we have $\BOX_b(G)$ and $\BOX_c(H)$, the only possible extension of $f$
    is the one that maps $b$ to $c$, and maps each edge between $b$ and some
    vertex $v$ in $\BOX_b(G)$ to the edge from $c$ to $f(v)$ in $\BOX_c(H)$.
    This then trivially reflects $!$-box containment and is monic, since $f$
    is.  It also trivially respects composition.
\end{proof}

Applying $\BOX$ to every morphism in definition \ref{def:bg-rw-rule} gives us
the following:

\begin{corollary}
    If $G \xleftarrow{i} I \xrightarrow{j} H$ is a $!$-graph equation, then so
    is \[ \BOX(G) \xleftarrow{\BOX(i)} \BOX(I) \xrightarrow{\BOX(j)} \BOX(H)
    \].
\end{corollary}

We can now state the inference rule
\[
    \textsc{Box} \enskip
    \inferrule{
        E \vdash G \rweq H
    }{
        E \vdash \BOX(G \rweq H)
    }
\]

\begin{theorem}
    \textsc{Box} is sound.  In particular, if for all concrete instances $G'
    \rweq H'$ of $G \rweq H$,
    \[ \overset{\sim}{E} \vdash G' \rweq H' \]
    then for any concrete instance $G'' \rweq H''$ of $\BOX_b(G \rweq H)$,
    \[ \overset{\sim}{E} \vdash G'' \rweq H'' \]
\end{theorem}
\begin{proof}
    By theorem \ref{thm:enf}, there is an instantiation $S$ of $G'' \rweq H''$
    from $\BOX_b(G \rweq H)$ in expansion-normal form, and we can further
    require that $S$ decomposes into $T;U$ where $T$ is a sequence of $\EXP_b$
    operations, followed by $\KILL_b$.

    Applying $\EXP_b$ to $\BOX_b(G)$ produces the disjoint union of
    $\BOX_b(G)$ and $G$.  The next $\EXP_b$ adds another copy of $G$, and so
    on.  Finally, $\KILL_b$ removes the copy of $\BOX_b(G)$, leaving zero or
    more disjoint copies of $G$.  So applying $T$ to $G \rweq H$ will produce
    $G_n \rweq H_n$, where $G_n$ is $n$ disjoint copies of $G$ and $H_n$ is
    $n$ disjoint copies of $H$.

    Now if we apply $U$ to $G_n \rweq H_n$, we know that each operation can
    only affect a single copy of $G$ in $G_n$ and the corresponding copy of $H$
    in $H_n$.  There are thus instantiations $U_1,\ldots,U_n$ of $G \rweq H$
    such that $U(G_n)$ is $U_1(G) \uplus \cdots \uplus U_n(G)$ and similarly
    for $U(H_n)$ (where we are using the notation from remark
    \ref{rem:instantiation-notation}).

    We can then use the assumption to construct proofs of $\overset{\sim}{E}
    \vdash U_i(G \rweq H)$ for each $1 \leq i \leq n$, and use
    \textsc{Leibniz} to apply these to the decomposed parts of $U(G_n \rweq
    H_n)$.
\end{proof}

This is mostly useful in the construction of larger proofs.  We will
demonstrate its use in section \ref{sec:gen-bialg-prf}.

\section{$!$-box Induction}
\label{sec:bbox-induction}

In this section, we introduce an analogue of induction for $!$-graphs.  This
will follow the usual induction scheme of having a base case and a step case
that must be proved for a particular $!$-vertex $b$.  However, in the absence
of the named variables of terms, we need some way to link a $!$-vertex in the
inductive hypothesis to one in the equation we need to prove for the step
case.

We will introduce the possibility of \textit{fixing} a $!$-vertex, which
prevents any $!$-box operations being performed on it during the matching
process.  The idea is that if we take the following graph
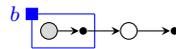
\[ \input{fix-ex-1.tikz} \]
and fix $b$, it should match itself and
\[ \input{fix-ex-2.tikz} \]
but not
\[ \input{fix-ex-3.tikz} \]
or
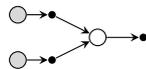
\[ \input{fix-ex-4.tikz} \]
as it normally would.  The result of this is that, given a rule like
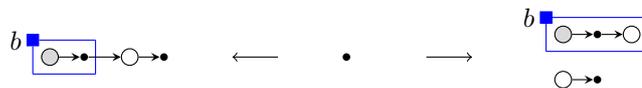
\[ \input{fix-ex-rule.tikz} \]
where $b$ is fixed, the only possible rewrites are ones like
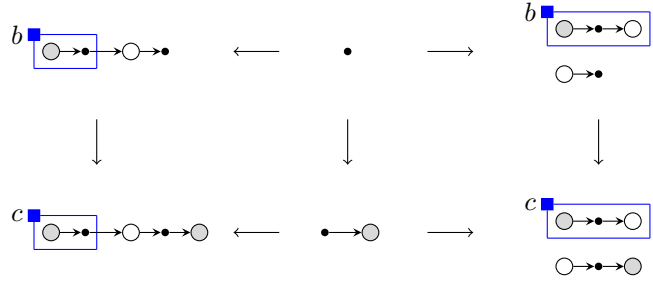
\[ \input{fix-ex-rewrite.tikz} \]
where the $!$-box operations that can subsequently be performed on $b$ and $c$
(and hence the possible concrete instantiations) are linked.

In order to deal with multiple applications of the induction rule, we will tag
these fixed $!$-vertices with arbitrary symbols drawn from a (countable) set
$\mathcal{A}$.  So we could tag $b$ above with $x \in \mathcal{A}$; in this
case, we say that $b$ is \textit{$x$-fixed}, and we only allow $b$ to match
$c$ in
\[ \input{fix-ex-2-labelled.tikz} \]
if $c$ is also $x$-fixed.

To store this information, we augment a $!$-graph $G$ with a partial map
\[ \fix_G : !(G) \nrightarrow \mathcal{A} \]
that maps each fixed $!$-vertex to its tag.  We require $!$-graph morphisms to
preserve these tags, and the morphisms of a $!$-graph equation to also reflect
them.  In particular, if we have the $!$-graph equation $G \rweq H$, the fact
we can consider the $!$-vertices of $G$ and $H$ to be the same means we can
view $\fix_G$ as being the same as $\fix_H$, and call these $\fix_{G\rweq H}$.

This will have an impact on the $!$-graph logic, and the notion of soundness
we are considering.  Specifically, we will consider what we mean for $E \vdash
G \rweq H$ to hold.  Let $A = \im(\fix_{G \rweq H})$.  For each $x \in A$,
each non-negative integer $n \in \mathbb{N}_0$, and each $!$-graph equation $L
\rweq R$, we define $X^x_n(L \rweq R)$ to be $L \rweq R$ with $n$ $\EXP_b$
operations and one $\KILL_b$ operation applied to it for each $x$-fixed
$!$-vertex $b$ in the equation.  Then for each function $\lambda \in
\mathbb{N}_0^A$ from $A$ to the non-negative integers, we define $X_\lambda(L
\rweq R)$ to be $L \rweq R$ with $X^x_{\lambda(x)}$ applied to it for each $x
\in A$ (the order is irrelevant, as these are all top-level $!$-vertices).
Then the semantics of $E \vdash G \rweq H$ is that for each $\lambda \in
\mathbb{N}_0^A$ and each instance $G' \rweq H'$ of $X_\lambda(G \rweq H)$, if
we let
\[ E_\lambda = \{X_\lambda(L \rweq R) | L \rweq R \in E. L \rweq R \textrm{ has no
$x$-fixed $!$-vertices for $x \notin A$}\} \]
then
\[ \overset{\sim}{E_\lambda} \vdash G' \rweq H' \]
in the equational logic of string graphs.

We need to adjust the $!$-graph logic and rewriting to keep it sound under
these modified semantics.  We place a side constraint on \textsc{Copy},
\textsc{Drop} and \textsc{Kill} that the $!$-vertex they operate on cannot be
fixed (hence the term ``fixed''), and on \textsc{Box} that the equation
contains no fixed $!$-vertices.  Likewise, when rewriting, we do not allow
any $!$-box operations to be applied to fixed $!$-vertices.

Our previous requirement that morphisms preserve fixing tags means that the
matching morphisms of a $!$-graph rewrite can only map $x$-fixed $!$-vertices
to other $x$-fixed $!$-vertices (although an unfixed $!$-vertex may match an
$x$-fixed one).  This also constrains the \textsc{Leibniz} rule.

We also introduce a $!$-box operation $\FIX^x_b(G)$ that produces a $!$-graph
identical to $G$ but where the $!$-vertex $b$ is $x$-fixed (and similarly for
$!$-graph equations), and a rule
\[
    (\textsc{Fix}) \enskip
    \inferrule{E \vdash G \rweq_{i,j} H}{E \vdash \FIX^x_b(G \rweq_{i,j} H)}
\]
with the condition that $b$ is a top-level $!$-vertex of $G \rweq_{i,j} H$
that is not fixed.

In the pictoral representations of $!$-graphs, we will mark $x$-fixed
$!$-vertices with the label $F_x$.

The following is the inference rule for $!$-box induction, based on one
originally suggested as a possibility by Kissinger in his
thesis (\cite{KissingerDPhil}, pp 179-181).  It holds providing $b$ is a
top-level, unfixed $!$-vertex in $G \rweq H$ and $x$ is a fresh tag (that does
not already occur in $G \rweq H$ or any of the equations in $E$).
\[
    (\textsc{Induct}) \enskip
    \inferrule{
        E \vdash \KILL_b(G \rweq H) \\
        E\cup\{\FIX^x_b(G \rweq H)\} \vdash \FIX^x_b(\EXP_b(G \rweq H))
    }{
        E \vdash G \rweq H
    }
\]

\begin{theorem}
    \textsc{Induct} is sound.
\end{theorem}
\begin{proof}
    Let $A = \im(\fix_{G \rweq H})$.  We need to show that for each
    $\lambda \in \mathbb{N}_0^A$ and each concrete instance $G' \rweq H'$ of
    $X_\lambda(G \rweq H)$,
    \[ \overset{\sim}{E_\lambda} \vdash G' \rweq H' \]
    in the equational theory of string graphs.

    We can choose an instantiation of $G' \rweq H'$ in expansion-normal form
    with all the operations on $b$ at the start (since $b$ is top-level).
    Denote this $T;U$, where $T$ contains all the operations on $b$, and $U$
    all the other operations.  $T$ is comprised of $n$ expansions, for some
    $n$, followed by $\KILL_b$.

    We show by induction on $n$ that, for all concrete instances $G_V \rweq
    H_V$ of $X_\lambda(G \rweq H)$ with an instantiation of the form $T;V$,
    \[ \overset{\sim}{E_\lambda} \vdash G_V \rweq H_V \]

    $n=0$: $T$ is just $\KILL_b$.  It follows from the definition of
    $X_\lambda$ that there is an instantiation $W$ of $X_\lambda(G \rweq H)$
    from $G \rweq H$ such that $W$ only operates on top-level $!$-vertices of
    $G \rweq H$.  Then $W;T;V$ is equivalent to $T;W;V$, and so $G_V \rweq
    H_V$ is an instance of $\KILL_b(G \rweq H)$.  The assumption $E \vdash
    \KILL_b(G \rweq H)$ then gives us that
    \[ \overset{\sim}{E_\lambda} \vdash G_V \rweq H_V \]

    $n=k+1$: let $T'$ be $T$ with the first expansion removed (so $T =
    \EXP_b;T'$).  We know that $T'$ is equivalent to $X^x_k$, and so $V$ is an
    instantiation of $G_V \rweq H_V$ from $X_k^x(\EXP_b(X_\lambda(G \rweq
    H)))$.  But
    \[ X_k^x(\EXP_b(X_\lambda(G \rweq H))) = X_\lambda(X_k^x(\EXP_b(G \rweq H))) \]
    since all the operations are on top-level $!$-vertices of $G \rweq H$.
    Further, we can define
    \[ \mu(y) =
        \begin{cases}
            \lambda(y) & y \in A \\
            k & x = y
        \end{cases}
    \]
    and then
    \[ X_\lambda(X_k^x(\EXP_b(G \rweq H))) = X_\mu(\FIX_b(\EXP_b(G \rweq H))) \]
    So now we have that $V$ is an instantiation of $G_V \rweq H_V$ from
    $X_\mu(\FIX_b(\EXP_b(G \rweq H)))$, and then the step case assumption
    means we can construct a proof of
    \[ \overset{\sim}{E_\mu} \vdash G_V \rweq H_V \]

    Now the way $\mu$ was defined means that
    \[ E_\mu = E_\lambda \cup \{X_\lambda(X_k^x(G \rweq H))\} \]
    and so the proof may contain invocations of \textsc{Axiom} with a concrete
    instance of ${X_\lambda(X_k^x(G \rweq H))}$.  Consider one such
    invocation, where the axiom ${G_W \rweq H_W}$ is a (concrete) instance of
    ${X_\lambda(X_k^x(G \rweq H))}$, which is the same as ${X_k^x(X_\lambda(G
    \rweq H))}$.  Then ${G_W \rweq H_W}$ is an instance of ${X_\lambda(G \rweq
    H)}$ and there is an instantiation of the form $T';W$ witnessing this.
    The inductive hypothesis then allows us to construct a proof of
    \[ \overset{\sim}{E_\lambda} \vdash G_W \rweq H_W \]
    We can thus replace all such \textsc{Axiom} invocations with proof trees
    using only the axioms of $E_\lambda$, and so discard $X_\lambda(X_k^x(G
    \rweq H))$ from $E_\mu$, leaving us with a proof of
    \[ \overset{\sim}{E_\lambda} \vdash G_V \rweq H_V \]
    as required, and the induction is complete.

    Now we just let $V = U$, and we have
    \[ \overset{\sim}{E_\lambda} \vdash G' \rweq H' \]
    as required.
\end{proof}

This can be viewed as a sort of $!$-box introduction rule, more powerful than
\textsc{Box} (although relying on correspondingly stronger assumptions).  It
can be used to produce $!$-versions of concrete equations, such as those
generated by an automated tool like QuantoCosy\cite{KissingerDPhil}.

\begin{example}
    Suppose we have the following three graph equations, drawn from the Z/X
    calculus presented in section \ref{sec:quantum}:
    \begin{equation}
        \label{rule:z-unit-x-counit-cancel}
        \input{z-unit-x-counit-cancel-graph.tikz}
    \end{equation}
    \begin{equation}
        \label{rule:x-copies-z}
        \input{x-copies-z-graph.tikz}
    \end{equation}
    \begin{equation}
        \label{rule:x-spider}
        \input{x-spider-law-graph.tikz}
    \end{equation}
    The first two are easily found by QuantoCosy, and the latter is the spider
    theorem arising from commutative Frobenius algebras.  We wish to use
    $!$-box induction to show
    \begin{equation}
        \label{rule:x-copies-z-spider}
        \input{x-copies-z-spider-graph.tikz}
    \end{equation}

    The base case is just \eqref{rule:z-unit-x-counit-cancel}, and so is
    trivially satisfied.  Now we need to show that adding
    \[ \input{x-copies-z-spider-graph-fixed.tikz} \]
    to the other equations allows us to derive
    \[ \input{x-copies-z-spider-graph-exp.tikz} \]

    This follows using $!$-graph rewriting:
    \[ \input{x-copies-z-spider-graph-prf.tikz} \]
\end{example}

\subsection{The Generalised Bialgebra Law}
\label{sec:gen-bialg-prf}

We now demonstrate a more ambitious application of $!$-box induction, where we
derive the generalised bialgebra law for the Z/X calculus
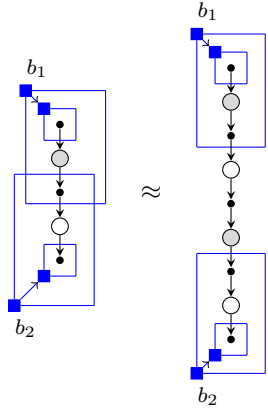
\[ \input{gen-bialg-graph-labelled.tikz} \]

Recall the axioms of the Z/X calculus:
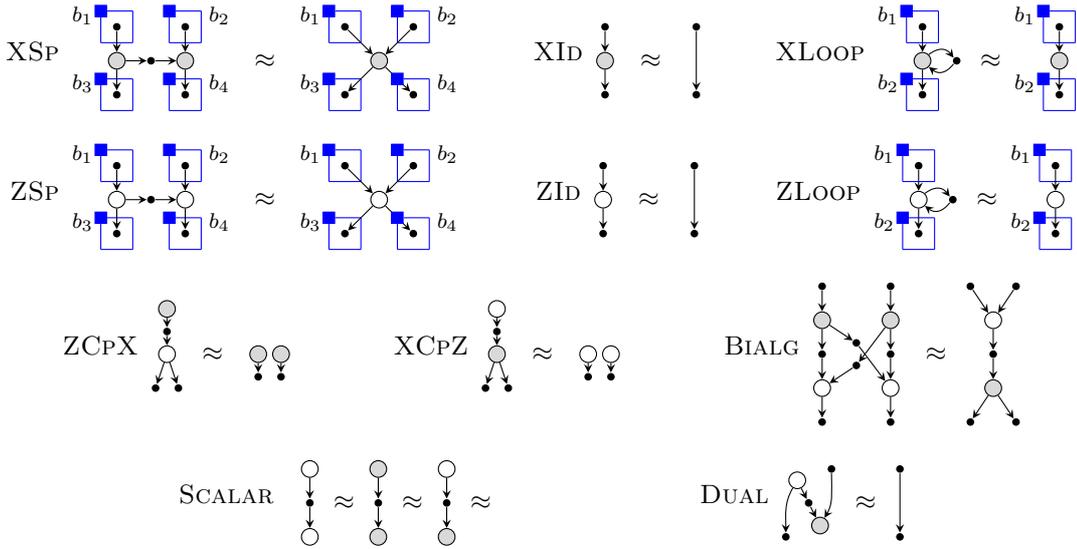
\begin{mathpar}
    \textsc{XSp} \enskip \input{x-spider-law-graph.tikz}
    \and
    \textsc{XId} \enskip \input{x-spider-id-graph.tikz}
    \and
    \textsc{XLoop} \enskip \input{x-spider-loop-graph.tikz}
    \and
    \textsc{ZSp} \enskip \input{z-spider-law-graph.tikz}
    \and
    \textsc{ZId} \enskip \input{z-spider-id-graph.tikz}
    \and
    \textsc{ZLoop} \enskip \input{z-spider-loop-graph.tikz}
    \and
    \textsc{ZCpX} \enskip \input{z-copies-x-graph.tikz}
    \and
    \textsc{XCpZ} \enskip \input{x-copies-z-graph.tikz}
    \and
    \textsc{Bialg} \enskip \input{bialg-graph.tikz}
    \\
    \textsc{Scalar} \enskip \input{z-x-mult-comult-cancel-graph.tikz}
    \and
    \textsc{Dual} \enskip \input{z-x-duals-coincide-graph.tikz}
\end{mathpar}
We have omitted the commutativity laws, as these are inherent in the fact that
all the nodes are variable-arity, and hence the inputs (resp. outputs) of a Z
or X node are indistinguishable from each other.  In particular,
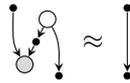
\[ \input{z-x-duals-coincide-mirror-graph.tikz} \]
is actually the same string graph equation as \textsc{Dual}.

We will start by proving some symmetries of these equations.
Note that, for brevity, we will tend to condense multiple applications of the
same rewrite rule (or of rules that differ only by colour).
\begin{lemma}
    \begin{equation}
        \label{bialg-eq:dual-swap}
        \input{z-x-duals-coincide-swap-mirror-graph.tikz}
    \end{equation}
\end{lemma}
\begin{proof}
    \[ \input{z-x-duals-swap-mirror-prf.tikz} \]
\end{proof}

\begin{lemma}
    \begin{equation}
        \label{bialg-eq:z-cp-x-upside-down}
        \input{z-copies-x-upside-down-graph.tikz}
    \end{equation}
\end{lemma}
\begin{proof}
    \[ \input{z-copies-x-upside-down-prf.tikz} \]
\end{proof}

We have already derived
\begin{equation}
    \label{bialg-eq:x-copies-z-spider}
    \input{x-copies-z-spider-graph.tikz}
\end{equation}
and
\begin{equation}
    \label{bialg-eq:z-copies-x-spider-upside-down}
    \input{z-copies-x-upside-down-spider-graph.tikz}
\end{equation}
follows from \eqref{bialg-eq:z-cp-x-upside-down} by essentially the same
proof.

\begin{lemma}
    \begin{equation}
        \label{bialg-eq:lem-1}
        \input{gen-bialg-lem-1.tikz}
    \end{equation}
\end{lemma}
\begin{proof}
    Proof by \textsc{Induct} on $b$.  The base case is
    \[ \input{gen-bialg-lem-1-base-prf.tikz} \]
    and the inductive step is
    \[ \input{gen-bialg-lem-1-prf.tikz} \]
    \[ \input{gen-bialg-lem-1-prf-2.tikz} \]
    \[ \input{gen-bialg-lem-1-prf-3.tikz} \]
\end{proof}

\begin{lemma}
    \begin{equation}
        \label{bialg-eq:lem-2}
        \input{gen-bialg-lem-2.tikz}
    \end{equation}
\end{lemma}
\begin{proof}
    Proof by \textsc{Induct} on $b$.  The base case is
    \[ \input{gen-bialg-lem-2-base-prf.tikz} \]
    and the inductive step is
    \[ \input{gen-bialg-lem-2-prf.tikz} \]
    \[ \input{gen-bialg-lem-2-prf-2.tikz} \]
    \[ \input{gen-bialg-lem-2-prf-3.tikz} \]
\end{proof}

\begin{theorem}
    \[ \input{gen-bialg-graph-labelled.tikz} \]
\end{theorem}
\begin{proof}
    First we use \textsc{Box} on \textsc{ZSp} to get
    \begin{equation}
        \input{z-spider-law-boxed.tikz}
        \label{bialg-eq:z-sp-boxed}
    \end{equation}
    and on \textsc{XSp} to get
    \begin{equation}
        \input{x-spider-law-boxed.tikz}
        \label{bialg-eq:x-sp-boxed}
    \end{equation}

    We will need to use \textsc{Induct} twice.  First, we use it on $b_1$.
    The base case is
    \[ \input{gen-bialg-prf-ind-1-base.tikz} \]
    The inductive hypothesis is
    \begin{equation}
        \input{gen-bialg-prf-ind-1-ih.tikz}
        \label{bialg-eq:ih-1}
    \end{equation}
    and we need to show
    \begin{equation}
        \input{gen-bialg-prf-ind-1-step-aim.tikz}
        \label{bialg-eq:ind-1-aim}
    \end{equation}

    We apply \textsc{Induct} again, this time on $b_2$.  The base case is
    \[ \input{gen-bialg-prf-ind-2-base.tikz} \]
    The new inductive hypothesis is
    \begin{equation}
        \input{gen-bialg-prf-ind-2-ih.tikz}
        \label{bialg-eq:ih-2}
    \end{equation}
    and we need to show
    \begin{equation}
        \input{gen-bialg-prf-ind-2-step-aim.tikz}
        \label{bialg-eq:ind-2-aim}
    \end{equation}
    This is done as follows:
    \[ \input{gen-bialg-prf-ind-2-step-prf.tikz} \]
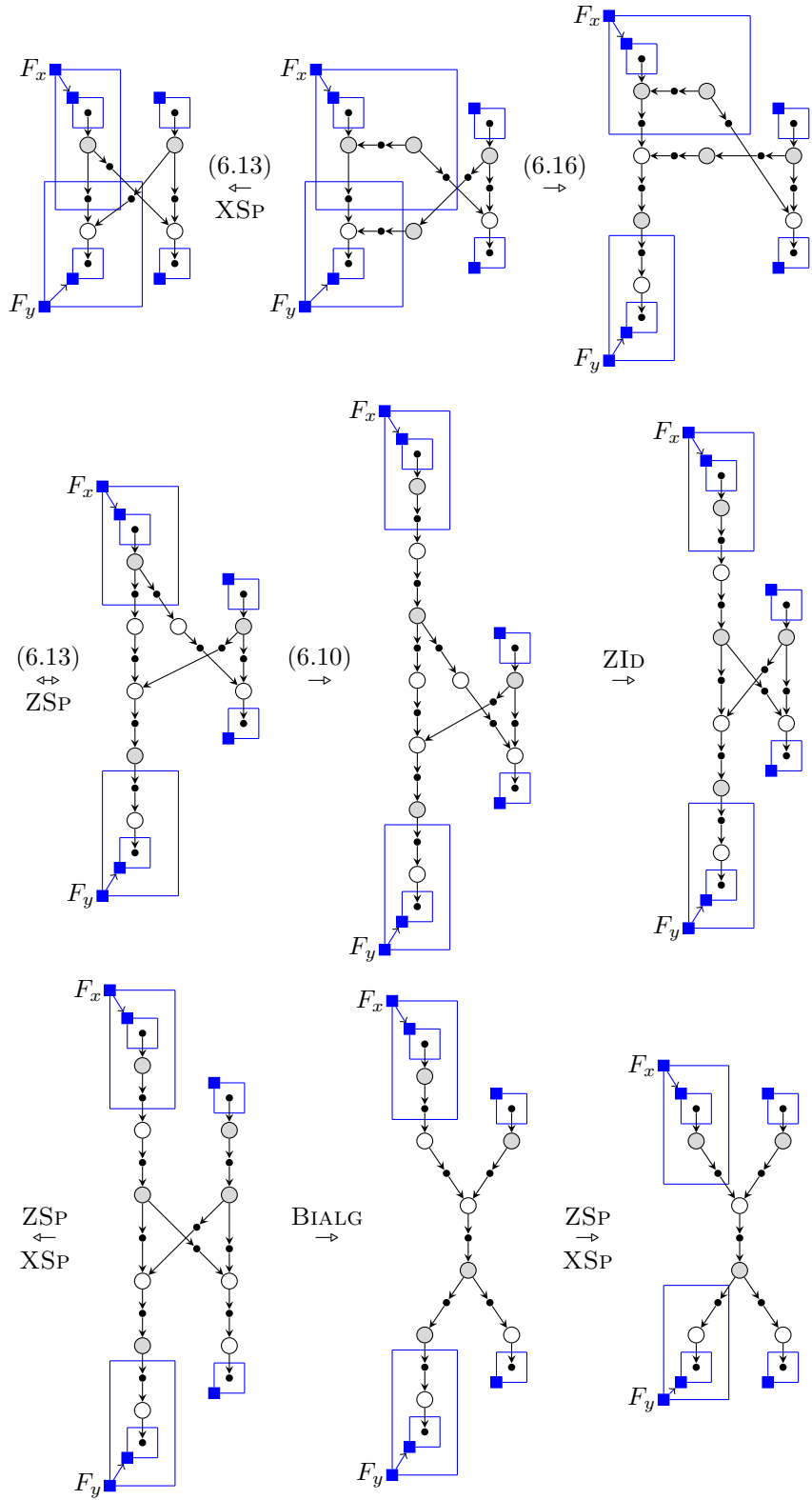
    \[ \input{gen-bialg-prf-ind-2-step-prf-2.tikz} \]
    \[ \input{gen-bialg-prf-ind-2-step-prf-3.tikz} \]

    Thus we have demonstrated that \eqref{bialg-eq:ind-2-aim} holds, so
    \textsc{Induct} gives us that \eqref{bialg-eq:ind-1-aim} holds, and
    \textsc{Induct} again gives us that the theorem holds.
\end{proof}

This example, while a demonstration of the power of \textsc{Induct}, still
requires the rewrite system to initially contain the spider laws.  Appendix
\ref{apdx:spider-law} demonstrates a possible way of using $!$-box induction
to introduce $!$-boxes into rewrite systems that previously contained only
concrete rewrite rules.

\section{Merging $!$-boxes}
\label{sec:bbox-merge}

The $\MERGE$ operation presented in this section was included as one of the
core $!$-box operations in \cite{Kissinger2012a} (pp 6-7,10-11).  Theorem
\ref{thm:merge-sound} demonstrates that this is unnecessary.

Suppose we have the equation
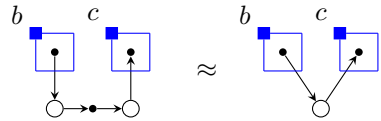
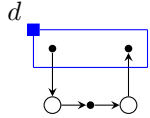
\[ \input{mergable-rule.tikz} \]
and the graph
\[ \input{merged-graph.tikz} \]
It is easy to see that for any concrete instance of the graph, there will be a
concrete instance of the equation whose LHS matches the graph: we just need to
do the same things to both $b$ and $c$ that we do to $d$.  However, the LHS of
the equation as it stands will not match the graph.

In order to allow this, we will introduce a new $!$-box operation that will
merge two disjoint $!$-boxes in a $!$-graph.  The two $!$-boxes will need to
have the same parents in order to ensure that combining them will not affect
the contents of any other $!$-box in the graph.

\begin{definition}[$\MERGE$; \cite{Kissinger2012a}, pp 7]\label{def:merge-op}
  Suppose $G$ is a $!$-graph and $b,c\! \in\, !(G)$,
  with $B^\uparrow(b)\graphminus b = B^\uparrow(c)\graphminus c$
  and $B(b) \cap B(c) = \{ \}$, then $\MERGE_{b,c}(G)$ is a
  quotient of $G$ where $B^\uparrow(b)$ and $B^\uparrow(c)$ are
  identified. More explicitly, this is the coequaliser
  \begin{equation}
    \label{eq:merge-diagram}
    \begin{tikzpicture}
      \matrix(m)[cdiag]{
      B^\uparrow(b) & G & \MERGE_{b,c}(G) \\};
      \path [arrs] (m-1-1) edge [arrow above] node {$\widehat b$} (m-1-2)
                   (m-1-1) edge [arrow below] node [swap] {$\widehat{c}$} (m-1-2)
                   (m-1-2) edge (m-1-3);
    \end{tikzpicture}
  \end{equation}
  in $\catGraph/\mathcal{G}_{T!}$ where $\widehat b$ is the normal inclusion
  map and $\widehat{c}$ is the inclusion of $B^\uparrow(c)$ into $G$
  composed with the obvious isomorphism from $B^\uparrow(b)$ to
  $B^\uparrow(c)$.
\end{definition}

This construction identifies $b$ and $c$, and leaves all the other vertices
of $G$ untouched.  The preconditions ensure that every incoming edge of $b$
corresponds exactly to an incoming edge of $c$ (and vice versa); each such
edge to $b$ is identified with its corresponding edge to $c$.  If a pair of
$!$-vertices satisfy the preconditions of $\MERGE$, we say they are
\textit{mergable}.

\begin{proposition}\label{prop:instantiation-of-merge}
    Let $G$ be a $!$-graph with mergable $!$-vertices $b$ and $c$.  Then
    $\MERGE_{b,c}(G)$ is a $!$-graph and $U(G) \cong U(\MERGE_{b,c}(G))$.
\end{proposition}
\begin{proof}
    Let $H = \MERGE_{b,c}(G)$, and let $h$ be the coequaliser map.
    \[
        \begin{tikzcd}
            B^\uparrow(b)
            \rar[arrow above]{\widehat b}
            \rar[swap,arrow below]{\widehat{c}} &
            G \rar{h} & H
        \end{tikzcd}
    \]
    $H$ exists in $\catGraph/\mathcal{G}_{T!}$, as
    $\catGraph/\mathcal{G}_{T!}$ has coequalisers of monomorphisms.  Note that
    $h$ is surjective by lemma \ref{lemma:spg-epi-surj} since it is a regular,
    and hence strong, epimorphism of $\catGraph/\mathcal{G}_{T!}$.

    We can construct a morphism $f : G \rightarrow \BOX(U(G))$ that is the
    identity on $U(G)$ (and collapses all the other $!$-boxes), and this
    trivially coequalises $\widehat{b}$ and $\widehat{c}$.
    So there is a unique $g$ making
    \[
        \begin{tikzcd}
            B^\uparrow(b)
            \rar[arrow above]{\widehat b}
            \rar[swap,arrow below]{\widehat{c}} &
            G \rar{h} \ar[swap]{dr}{f} & H \dar[dashed]{g} \\
            && \BOX(G)
        \end{tikzcd}
    \]
    commute.  Since $U(f)$ is an isomorphism, $U(h)$ is a (split) monomorphism
    as well as being a strong epimorphism, and hence is invertible.  So $U(G)
    \cong U(H)$ and hence $U(H)$ must be a string graph, since $U(G)$ is.

    Let $x,y$ be vertices in $!(G)$ such that $h(x)=h(y)$.  We can show that
    either $x = y$ or $\{x,y\} = \{b,c\}$.  Suppose $x \neq y$, let $G'$ be
    $U(G)$ together with the graph of $!$-vertices
    \begin{equation*}
        \input{full-graph-g2.tikz}
    \end{equation*}
    and edges from $a$ and $b$ to each vertex of $U(G)$.  Then let
    $f:G \rightarrow G'$ be the identity on $U(G)$ and take every vertex in
    $!(G)$ to $a$ except $y$, which it takes to $b$.  This also fixes the edge
    maps.  If $y \neq b$ and $y \neq c$, $g$ will coequalise $\widehat b$ and
    $\widehat{c}$, and so there must be a unique map $g:H \rightarrow G'$ such
    that $f = g \circ h$.  But then $g(h(x)) \neq g(h(y))$ even though $h(x) =
    h(y)$, which is impossible.  So we must have $y = b$ or $y = c$.
    Repeating with the roles of $x$ and $y$ switched, we get that $x = b$ or
    $x = c$.  Thus if $x \neq y$, then $\{x,y\}=\{b,c\}$.

    $\beta(H)$ is reflexive as it is the image of $\beta(h)$, the domain of
    which is reflexive.

    Let $x \neq y$ be vertices in $\beta(H)$ with $e_1,e_2$ edges between them
    such that $s(e_1)=t(e_2)=x$ and $t(e_1)=s(e_2)=y$.
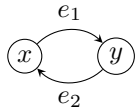
    \[ \input{merge-posetal-prf-antisym-1.tikz} \]
    $\beta(G)$ is anti-symmetric, and so the preimages of $e_1$ and $e_2$ in
    under $h$ cannot be in the same loop construction.  Thus either $x$ or $y$
    must be mapped to by more than one vertex of $G$, and hence must have the
    preimage $\{b,c\}$.  This cannot be true of both, since $x \neq y$.
    WLOG, suppose it is true of $x$.  Then $y'$, the preimage of $y$, must be
    in $B(b)$ and $B^\uparrow(c)$ (or vice versa).  But
    $B^\uparrow(b)\graphminus b = B^\uparrow(c)\graphminus c$ and (since $y'
    \neq c$) $y'$ must be in $B^\uparrow(b)$, which can only be true if $y' =
    b$, which we assumed was not the case.  So there can be no such $e_1$ and
    $e_2$.

    Simplicity of $\beta(H)$ follows a similar argument: if $e_1$ and $e_2$
    are edges from $x$ to $y$ where $x \neq y$, then if the preimage of $y$
    is $\{b,c\}$, $x$ must be in the image of $\widehat{b}$
    and $\widehat{c}$ and so the preimages of $e_1$ and $e_2$ must be
    coequalised by $h$ and hence $e_1 = e_2$, and if the preimage of $x$ is
    $\{b,c\}$ then the fact that $B(b) \cap B(c) = \varnothing$ means the
    preimages of $e_1$ and $e_2$ must be the same, and hence $e_1 = e_2$.

    For transitivity, let $x,y,z$ be vertices in $\beta(H)$ with $e_1$ an edge
    from $x$ to $y$ and $e_2$ from $y$ to $z$.  Let $e_1',e_2'$ be arbitrary
    preimages of the respective edges under $h$, as before.  Then either
    $t(e_1')=s(e_2')$, in which case the existence of an edge from $x$ to $z$
    is immediate from transitivity of $G$, or $\{t(e_1'),s(e_2')\} =
    \{b,c\}$.  WLOG, assume $t(e_1') = b$ and $s(e_2') = c$.  Then
    $B^\uparrow(b)\graphminus b = B^\uparrow(c)\graphminus c$ gives us that
    $s(e_1')$ must be a predecessor of $c$, and so $x$ must be a predecessor
    of $z$, as required.

    Let $d \in\:!(H)$.  Then, because $h$ is surjective and $U(h)$ is
    bijective, $U(B(d))$ is the union of $U(h)[U(B(d'))]$ for each preimage
    $d'$ of $d$ under $h$.  But $U(h)$ is an isomorphism, so $U(B(d))$ is the
    union of open subgraphs of $U(H)$, and hence is open in $U(H)$.

    For the final requirement of a $!$-graph, suppose $d,d' \in\:!(H)$ with
    $d' \in B(d)$.  We need to show, for any vertex $v$ of $H$ with $v \in
    B(d')$, that $v \in B(d)$.  If $v$ is a $!$-vertex, this follows from
    transitivity of $\beta(H)$.  Otherwise, the edge $e$ from $d'$ to $v$ has
    a preimage $e_0$ under $h$, as does the edge $e'$ from $d$ to $d'$.  If
    $t(e'_0) = s(e_0)$, there must be an edge from $s(e'_0)$ to $t(e_0)$ in
    $G$, which will map to an edge from $d$ to $v$ in $H$.  Otherwise, we must
    have $t(e'_0) = b$ and $s(e_0) = c$ (or vice versa, which is equivalent).
    But then the preconditions for $\MERGE$ require an edge from $s(e'_0)$ to
    $c$ as well, which also induces an edge from $s(e'_0)$ to $t(e_0)$, and
    hence from $d$ to $v$.
\end{proof}

\begin{proposition}\label{prop:merge-on-maps}
    Let $f : G \rightarrow H$ be a morphism of $\catBGraph_T$ that
    reflects $!$-box containment, and let $b$ and $c$ be mergable $!$-vertices
    of $G$ such that $f(b)$ and $f(c)$ are mergable in $H$.  Then there is a
    unique morphism
    \[ \MERGE_{f(b),f(c)}(f) : \MERGE_{b,c}(G) \rightarrow
    \MERGE_{f(b),f(c)}(H) \]
    making the following diagram commute
    \[\begin{tikzcd}[column sep=18ex]
        G \rar{f} \dar & H \dar \\
        \MERGE_{b,c}(G) \rar[swap]{\MERGE_{f(b),f(c)}(f)} &
        \MERGE_{f(b),f(c)}(H)
    \end{tikzcd}\]
    where the down arrows are the coequaliser maps from the definition of
    $\MERGE$, and this morphism reflects $!$-box containment.  Further, this
    construction preserves monomorphisms and strong epimorphisms and respects
    composition.
\end{proposition}
\begin{proof}
    Let $G' = \MERGE_{b,c}(G)$ and $H' = \MERGE_{f(b),f(c)}(H)$.  We will
    construct the following commuting diagram
    \[\begin{tikzcd}
        B^\uparrow(b) \rar{f_b}
            \dar[arrow left,swap]{\widehat{b}}
            \dar[arrow right]{\widehat{c}} &
        B^\uparrow(f(b))
            \dar[arrow left,swap]{\widehat{f(b)}}
            \dar[arrow right]{\widehat{f(c)}} \\
        G \rar{f} \dar[swap]{g} &
        H \dar{h} \\
        G' \rar[swap,dashed]{f'} &
        H'
    \end{tikzcd}\]
    where $f'$ will be uniquely determined by universality of coequalisers.
    This will then be $\MERGE_{f(b),f(c)}(f)$.

    The top square is two separately-commuting squares: $f \circ \widehat{b} =
    \widehat{f(b)} \circ f_b$ and $f \circ \widehat{c} = \widehat{f(c)} \circ
    f_b$.  If $d$ is a vertex in $B^\uparrow(b)$, then the edge witnessing
    this is mapped by $f$ to one witnessing $f(d) \in B^\uparrow(f(b))$.  So
    if we let $f_b$ be $f \circ \widehat{b}$, the restriction of $f$ to
    $B^\uparrow(b)$, reinterpreted with codomain $B^\uparrow(f(b))$ then we
    have the first commuting square.  The second follows from the fact that
    $B^\uparrow(b) \graphminus b = B^\uparrow(c) \graphminus c$ and similarly
    for $f(b)$ and $f(c)$.

    Now $h$ coequalises $\widehat{f(b)} \circ f_b$ and $\widehat{f(c)} \circ
    f_b$, and hence coequalises $f \circ \widehat{b}$ and $f \circ
    \widehat{c}$.  So $h \circ f$ coequalises $\widehat{b}$ and $\widehat{c}$.
    Then universality of $g$ means that there is a unique $f'$ making the
    bottom square of the diagram commute.  This universality also means that
    $\MERGE$ respects composition.

    The fact that $f$ reflects $!$-box containment means that $f_b$ is
    surjective, and hence epic.  Then if we have morphisms $a : G' \rightarrow
    K$ and $b : H \rightarrow K$ such that $a \circ g = b \circ f$, we know
    \begin{align*}
        &a \circ g \circ \widehat b = a \circ g \circ \widehat c \\
        &\Rightarrow b \circ \widehat{f(b)} \circ f_b = b \circ \widehat{f(c)}
        \circ f_b \\
        &\Rightarrow b \circ \widehat{f(b)} = b \circ \widehat{f(c)}
    \end{align*}
    Then, since $b$ equalises $\widehat{f(b)}$ and $\widehat{f(c)}$, there is
    a unique arrow $k : H' \rightarrow K$ such that $b = k \circ h$.  Since
    $g$ is epic, we also have $a = k \circ f'$, and hence the bottom square of
    the diagram is a pushout.  $\catBGraph_T$ preserves monomorphisms under
    pushouts, and hence $f'$ is monic whenever $f$ is.  Also, if $f$ is
    strongly epic, and hence surjective, the same must be true of $h \circ f$
    and hence of $f' \circ g$.  So $f'$ is surjective, and hence strongly
    epic.  Since it is surjective, it trivially reflects $!$-box containment.
\end{proof}

Since $U(G) \cong U(\MERGE(G))$ for $!$-graphs, $U(f)$ is the same as
$U(\MERGE(f))$ (up to isomorphism) and so we can simply apply $\MERGE$ to the
morphisms of definition \ref{def:bg-rw-rule} to get the following:

\begin{corollary}[\cite{Kissinger2012a}, pp 10-11]\label{cor:merge-op-rweq}
    Let $L \xleftarrow{i} I \xrightarrow{j} R$ be a $!$-graph equation, and
    let $b,c$ be mergable $!$-vertices of $I$.  Then
    \[
        \MERGE_{i(b),i(c)}(L)
        \xleftarrow{\MERGE_{i(b),i(c)}(i)}
        \MERGE_{b,c}(I)
        \xrightarrow{\MERGE_{j(b),j(c)}(j)}
        \MERGE_{j(b),j(c)}(R)
    \]
    is also a $!$-graph equation.
\end{corollary}

In the case that $b$ and $c$ are mergable $!$-vertices of $G \rweq H$,
we posit the following rule:
\[
    \textsc{Merge} \enskip
    \inferrule{
        E \vdash G \rweq H
    }{
        E \vdash \MERGE_{b,c}(G \rweq H)
    }
\]

\begin{theorem}
    \label{thm:merge-sound}
    \textsc{Merge} is sound.
\end{theorem}
\begin{proof}
    We will show that every concrete instance $G' \rweq H'$ of $\MERGE_{b,c}(G
    \rweq H)$ is also a concrete instance of $G \rweq H$.  Suppose $S$ is a
    concrete instantiation of $G' \rweq H'$ from $\MERGE_{b,c}(G \rweq H)$ in
    expansion-normal form.

    While instantiations may only contain $\COPY$, $\DROP$ or $\KILL$
    operations ($\EXP$ being a composite of these), for the purposes of this
    argument we will relax this to allow $\MERGE$ operations as well.  This is
    purely a convenience to allow us to call our intermediate steps
    ``instantiations'' before we finally arrive at a genuine instantiation
    containing only those three operations.  So we start with the
    ``instantiation'' $\MERGE_{b,c};S$ of $G' \rweq H'$ from $G \rweq H$, and
    transform it into a genuine instantiation with no $\MERGE$ operations.

    The intermediate instantiations may have multiple $\MERGE$ operations, but
    we will always keep these together, and they will not interfere with each
    other: if $H = \MERGE_{b,c}(G)$ and $K = \MERGE_{b',c'}(H)$ are two
    adjacent operations in the instantiation, where $G \xrightarrow{g} H
    \xrightarrow{h} K$ are the coequaliser maps, we will have that $b' \notin
    B^\uparrow(g(b))$ and $g(b) \notin B^\uparrow(b')$, and similarly
    $h(g(b))$ will not interfere with the $!$-vertices acted on by any
    $\MERGE$ operation on $K$ and so on.  This allows us to freely reorder the
    $\MERGE$ operations.

    We will use this to collate them into a ``composite'' $\MERGE$ that
    simultaneously operates on multiple pairs of mergable $!$-vertices, where
    each pair has distinct parents to any other pair.  We will only ever have
    one such operation in the sequence; initially, it will be at the start and
    will have a single pair, $\langle b,c \rangle$.  When it has no pairs,
    there are no $\MERGE$ operations left in the sequence and we are done.

    We will move the $\MERGE$ right one operation at a time until the set of
    pairs it operates on is empty.  So consider the operation to its right.
    This is either $\KILL_d$ or $\EXP_d$ for some $!$-vertex $d$.  $d$ cannot
    be a child $!$-vertex of any of the merged vertices, since $d$ must be
    a top-level $!$-vertex.  If none of the merged vertices are in $B(d)$, the
    operations are independent, and we can move the $\MERGE$ one place to the
    right without changing the resulting graph.

    Suppose the operation is $\KILL_d$.  If $d$ is a merged vertex, call its
    preimages under the coequaliser $d_0$ and $d_1$.  If we remove the
    $\KILL_d$, remove the pair $\langle d_0,d_1 \rangle$ from the $\MERGE$ and
    put $\KILL_{d_0};\KILL_{d_1}$ before the $\MERGE$, we will get the same
    graph.  Otherwise, suppose $d$ is not a merged vertex, but $B(d)$ contains
    a merged vertex.  Then we can swap $\KILL_d$ with the $\MERGE$ and remove
    any pairs in $B(d)$ from the $\MERGE$.

    Suppose instead the operation is $\EXP_d$.  If $d$ is a merged vertex, we
    can remove $\EXP_d$ and put $\EXP_{d_0};\EXP_{d_1}$ before the $\MERGE$,
    leaving the set of merged vertices untouched.  Otherwise, suppose $d$ is
    not a merged vertex but $B(d)$ contains one, and say it results from the
    pair $\langle x,y \rangle$.  Then we can swap $\EXP_d$ and the $\MERGE$,
    providing we add $\langle x',y' \rangle$ to the set of merged vertices,
    where $x'$ is the copy of $x$ under the expansion and $y'$ is the copy of
    $y$.  $x'$ and $y'$ have no edges to or from $x$ and $y$, or the
    $!$-vertices of any other merged pair.  If there are any other such pairs,
    do the same to them as well.

    Since the number of operations to the right of the $\MERGE$ always
    decreases, this will eventually terminate, and we will have a concrete of
    instantiation of $G' \rweq H'$ from $G \rweq H$.
\end{proof}

\section{A More Traditional Logic}
\label{sec:trad-logic}

Those familiar with formal logics will have noticed something strange about
the one presented at the start of this chapter, as well as the one in section
\ref{sec:graph-eqs}: the equality predicate $\rweq$ is not really a predicate
at all.  What we have presented is a logic of ``string graph equations'' and
``$!$-graph equations'', where those ``equations'' are actually spans in a
category.

In this section, we will sketch out an equational logic of $!$-graphs, where
the equality predicate is just a formal symbol indicating that the objects on
either side of it should be considered equal.  As we noted in section
\ref{sec:sg-framed}, we need to somehow encode the correlation between
boundaries of string graphs; we also need to encode the correlation between
$!$-vertices in both graphs.  $!$-graph equations are one way of doing this,
but to get a more traditional equality predicate, we will extend the idea of a
framed cospan (definition \ref{def:framed-cospan}) to $!$-vertices.

\begin{definition}[$!$-Framed Cospan]
    A \textit{$!$-graph frame} is a triple $(X,<,\textrm{sgn})$ where
    $X$ is a $!$-graph consisting only of isolated wire-vertices, $!$-vertices
    and edges whose sources are $!$-vertices; $<$ is a total order on $V_X$,
    the vertices of $X$; and $\textrm{sgn} : V_{U(X)} \rightarrow \{+,-\}$ is
    the \textit{signing map}.

    A cospan of $!$-graph monomorphisms $X \xrightarrow{d} G \xleftarrow{c} Y$
    is called a \textit{$!$-framed cospan} if
    \begin{enumerate}
        \item $X$ and $Y$ are $!$-graph frames
        \item $G$ contains no isolated wire-vertices
        \item the following is a pushout square:
        \[
            \begin{tikzcd}
                \beta(G) \rar[right hook->] \dar[right hook->] &
                \im(d) \dar[right hook->] \\
                \im(c) \rar[right hook->] &
                \Bound_!(G) \NWbracket
            \end{tikzcd}
        \]
        \item for every $v \in V_{U(X)}$, $d(v) \in \In(G) \Leftrightarrow
        \textrm{sgn}(v) = +$
        \item for every $v \in V_{U(Y)}$, $c(v) \in \Out(G) \Leftrightarrow
        \textrm{sgn}(v) = +$
    \end{enumerate}
\end{definition}

As with the framed cospans of string graphs, we can frame a $!$-graph
equation:
\begin{definition}
    A \textit{framing} of a $!$-graph equation $L \rweq_{i_1,i_2} R$ is a pair
    of $!$-framed cospans $X \xrightarrow{a} L \xleftarrow{b} Y$ and $X
    \xrightarrow{c} R \xleftarrow{d} Y$ such that there are morphisms $j_1$
    and $j_2$ forming the coproduct (ie: disjoint union) $X \xrightarrow{j_1}
    I \xleftarrow{j_2} Y$ and making the following diagram commute:
    \[
        \begin{tikzcd}[ampersand replacement=\&]
            \& X
                \arrow[swap]{dl}{a}
                \arrow{d}{j_1}
                \arrow{dr}{c}
                \&
            \\
            L \& I \arrow[swap]{l}{i_1} \arrow{r}{i_2} \& R
            \\
            \& Y
                \arrow{ul}{b}
                \arrow[swap]{u}{j_2}
                \arrow[swap]{ur}{d}
                \&
        \end{tikzcd}
    \]
\end{definition}

Just as with framings of $!$-graph equation, such a framing is equivalent to a
$!$-graph equation, in the sense that any $!$-graph equation can be framed,
and the $!$-graph equation can be recovered from the two $!$-framed cospans.
If two $!$-framed cospans have isomorphic left and right (or domain and
codomain) $!$-graph frames, they induce a $!$-graph equation, and we call them
\textit{compatible}.

With this is mind, we can construct a logic of $!$-framed cospans that is
equivalent to the logic presented in this chapter.  We have a choice in how we
represent the set of axioms: we can use a set of $!$-graph equations and the
rule
\begin{mathpar}
    (\textsc{Axiom}) \enskip
    \inferrule{G \rweq_{i,j} H \in E}{E \vdash \hat G \rweq \hat H}
\end{mathpar}
where $(\hat G,\hat H)$ is a framing of $G \rweq_{i,j} H$, or we can use a set
of pairs of compatible $!$-framed cospans and have the rule
\begin{mathpar}
    (\textsc{Axiom}) \enskip
    \inferrule{(\hat G, \hat H) \in E}{E \vdash \hat G \rweq \hat H}
\end{mathpar}

The equivalence relation rules are simple:
\begin{mathpar}
    (\textsc{Refl}) \enskip
    \inferrule{ }{E \vdash \hat G \rweq \hat G}
    \and
    (\textsc{Sym}) \enskip
    \inferrule{E \vdash \hat G \rweq \hat H}{E \vdash \hat H \rweq \hat G}
    \and
    (\textsc{Trans}) \enskip
    \inferrule{E \vdash \hat G \rweq \hat H \\ E \vdash \hat H \rweq \hat K
             }{E \vdash \hat G \rweq \hat K}
\end{mathpar}

For the others, we will need to introduce some new notation, or at least
extend existing notation.  For example, if $f : G \rightarrow H$ is a wire
homeomorphism of $!$-graphs, and $\hat G = X \rightarrow G \leftarrow Y$ is a
framed cospan, then there is a framed cospan $\hat H = X \rightarrow H
\leftarrow Y$ where the inclusions of the framed cospans commute with $f_!$
and $f_B$.  We can then view $f$ as a wire homeomorphism from $\hat G$ to
$\hat H$.

\textsc{Homeo} can then be expressed as
\begin{mathpar}
    (\textsc{Homeo}) \enskip
    \inferrule{E \vdash \hat G \rweq \hat H}{E \vdash \hat G \rweq f(\hat H)}
\end{mathpar}
where $f$ is a wire homeomorphism (note that the presence of \textsc{Sym}
means that this is equivalent to the version of \textsc{Homeo} on $!$-graph
equations).

\textsc{Leibniz} requires a bit more work.  We will borrow the notion of
\textit{contexts} from term-rewriting, where they are terms with ``holes'' in
them: placeholders that can be substituted with other terms to produce a
valid term overall.

The idea with $!$-graphs will be to provide both an ``external'' interface to
the graph, as a $!$-framed cospan does, and an ``internal'' interface.  An
simple example (without any $!$-boxes) would be
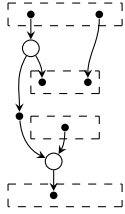
\[ \input{graph-context.tikz} \]
We could then take a $!$-framed cospan whose frames match the internal
interface of the context, such as
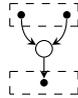
\[ \input{graph-context-inner.tikz} \]
and embed it in the context by merging the images of the frames
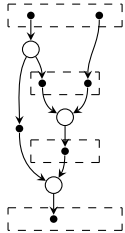
\[ \input{graph-context-filled.tikz} \]
producing the $!$-framed cospan
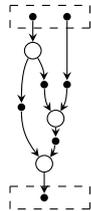
\[ \input{graph-context-complete.tikz} \]

When there are $!$-vertices to account for, the $!$-graph frames of the
external interface must contain all the $!$-vertices (and edges between them)
in order for the final construction to be a $!$-framed cospan.  The internal
interface need not contain all the $!$-vertices of the graph, though.

\begin{definition}
    A \textit{$!$-graph context} $C$ is a pair of $!$-graph cospans
    \[ X_I \xrightarrow{d_I} G \xleftarrow{c_I} Y_I \]
    \[ X_O \xrightarrow{d_O} G \xleftarrow{c_O} Y_O \]
    where
    \begin{itemize}
        \item $G$ contains no isolated wire-vertices
        \item $X_I$, $X_O$, $Y_I$ and $Y_O$ are $!$-graph frames
        \item $\beta(X_I) \cong \beta(Y_I)$
        \item $\beta(X_O) \cong \beta(Y_O) \cong \beta(G)$
        \item the morphisms are all monic and reflect $!$-box containment
        \item $\Bound(U(G))$ is the disjoint union of the images of $U(d_I)$,
            $U(c_I)$, $U(d_O)$ and $U(c_O)$
        \item for every $v \in V_{U(X_I)}$, $d_I(v) \in \In(G) \Leftrightarrow
            \textrm{sgn}(v) = +$, and similarly for $d_O$
        \item for every $v \in V_{U(Y_I)}$, $c_I(v) \in \Out(G)
            \Leftrightarrow \textrm{sgn}(v) = +$, and similarly for $c_O$
    \end{itemize}
\end{definition}

Given such a $!$-graph context and another $!$-framed cospan
\[ \hat H \; = \; Y_I \xrightarrow{d_I'} H \xleftarrow{c_I'} X_I \]
we can merge the context with $\hat H$ (we say that $C$ \textit{accepts} $\hat
H$).  We first merge $G$ and $H$, using $Y_I$ as the overlap, in a manner
similar to the composition of framed cospans from section \ref{sec:sg-framed}:
\[
    \begin{tikzcd}[ampersand replacement=\&]
        \&\&Y_I \arrow[swap]{dl}{c_I} \arrow{dr}{d_I'}\&\& \\
        X_I \arrow{r}{d_I} \& G \arrow[swap]{dr}{i_1} \&\&
        H \arrow{dl}{i_2} \& X_I \arrow[swap]{l}{c_I'} \\
        \&\& K \Nbracket \&\&
    \end{tikzcd}
\]
$K$ is a $!$-graph by a similar argument to the one that shows composition of
framed cospans to be well-defined (note the reversal of the frames for $\hat
H$ means that the interpretation of the signing map on the frames is reversed,
so inputs of $H$ will be merged with outputs of $G$ and vice versa).

Next, we merge the two images of $X_I$ using the following coequaliser:
\[
    \begin{tikzcd}
        X_I \rar[arrow above]{i_1 \circ d_I}
            \rar[arrow below,swap]{i_2 \circ c_I'} &
        K \rar[dashed]{f} &
        G+H
    \end{tikzcd}
\]
Then we define $C[\hat H]$ to be the framed cospan
\[ X_O \xrightarrow{f \circ i_1 \circ d_O} G+H \xleftarrow{f \circ i_1 \circ c_O} Y_O \]

Now we can write \textsc{Leibniz} as
\begin{mathpar}
    (\textsc{Leibniz}) \enskip
    \inferrule{E \vdash \hat G \rweq \hat H
        }{E \vdash C[\hat G] \rweq C[\hat H]}
\end{mathpar}
where $C$ is a $!$-graph context that accepts $\hat G$ and $\hat H$ (note that
the compatibility of $\hat G$ and $\hat H$ means that if $C$ accepts one then
it must accept the other).

In this version of \textsc{Leibniz}, the context graph takes the place of $D$
in the $!$-graph rewrite
\[ \begin{tikzcd}
    G \dar
    & I \lar \rar \dar
    & H \dar
    \\ G' \NEbracket
    & D \lar \rar
    & H' \NWbracket
\end{tikzcd} \]
and the internal frames ($X_I$ and $Y_I$) take the place of $I$.

The final piece of this logic of $!$-graph frames is the $!$-box operations.
We extend them to $!$-graph cospans in the same way that we extended them to
$!$-graph equations: if
\[ \hat G \; = \; X \xrightarrow{d} G \xleftarrow{c} Y \]
then
\[ \OP(\hat G) \; = \; \OP(X) \xrightarrow{\OP(d)} \OP(G) \xleftarrow{\OP(c)} \OP(Y) \]

Now we can write, for example,
\begin{mathpar}
    (\textsc{Copy}) \enskip
    \inferrule{E \vdash \hat G \rweq \hat H
        }{E \vdash \COPY_b(\hat G) \rweq \COPY_b(\hat H)}
\end{mathpar}
(since the compatibility of $\hat G$ and $\hat H$ means that we consider $b$
to be in both if it is in either), and similarly for the other $!$-box
operation rules.

Thus we have a logic with a normal equality predicate based on $!$-framed
cospans.  What is more, by considering only concrete graphs, the same
construction can be used to produce a logic based on framed cospans from the
logic presented in section \ref{sec:graph-eqs}.

In practice, a theorem prover would probably use names to store the equivalent
of the $!$-graph frames; the signing maps and the division between the domain
and codomain are not particularly important for the construction of proofs (as
can be seen from the fact that the $!$-graph equation structures do not store
this information).


\chapter{Computability of $!$-Graph Matching}
\label{ch:implementation}

Suppose we have a $!$-graph (or string graph) $G$ and a $!$-graph rewrite rule
$L \rewritesto R$ that we wish to use to rewrite $G$; as previously stated, we
want rewriting to be up to wire homeomorphism.  If we can find a rewrite rule
$L' \rewritesto R'$ that is wire-homeomorphic to an instance of $L \rewritesto
R$ such that there is a local isomorphism $m : L' \rightarrow G$ that reflects
$!$-box containment, theorem \ref{thm:bg-rw-exists} tells us that the rewrite
\begin{equation}
    \begin{tikzcd}
        L' \arrow[swap]{d}{m}
        & I' \arrow[swap]{l}{i'} \arrow{r}{j'} \arrow{d}{d}
        & R' \arrow{d}{r}
        \\
        G \NEbracket
        & D \arrow{l}{g} \arrow[swap]{r}{h}
        & H \NWbracket
    \end{tikzcd}
\end{equation}
exists.  In fact, it is easy to find this rewrite: if we remove from $G$
the image of everything in the \textit{interior} of $L$ (ie: the image of
everything in $L$ that is not in the image of $i$), we get $D$; $g$ is the
subgraph relation and $d$ is a restriction of $m$.  $H$, being a pushout of
monomorphisms, is just a graph union.  In a practical implementation, where
vertices and edges are likely to be named components, the main difficulty is
managing the names to make sure there are no clashes when building $H$.

That still leaves finding suitable instances, finding suitable rules wire
homeomorphic to those instances and finding a local isomorphism that reflects
$!$-box containment.  Section \ref{sec:sg-match-wire-homeo} discussed the
latter two (except reflecting $!$-box containment, which can be implemented by
filtering the results); in this chapter, we will demonstrate that, providing
$L \rewritesto R$ is ``well behaved'' in terms of the instances it can
produce, we can finitely enumerate all instances of the rule that can result
in a matching onto a particular $G$.  We will also describe how Quantomatic
implements these steps.

\section{Enumerating Instances}
\label{sec:enum-inst}

In order to find all possible matchings, we need a way of searching the space
of graphs $H$ such that $G \succeq H$; while this space will almost always be
infinite, it will usually be the case that only a finite number of them will
be candidates for matchings.

Suppose we wish to match a $!$-graph $G$ onto another $!$-graph $H$; we need
to find an instance $G'$ of $G$ and a monomorphism from $G'$ to $H$ (or,
rather, from a graph wire homeomorphic to $G'$ to a graph wire homeomorphic to
$H$).  The monomorphism must satisfy additional constraints (definition
\ref{def:pg-matching}), but these are not important to this discussion.

As long as $G$ has any (unfixed) $!$-vertices, it will have an infinite number
of instances (to see this, consider repeatedly applying $\COPY$ to one of the
$!$-vertices).  However, in most cases, only a finite number of these will
have a monomorphism to $H$.

In particular, we can safely ignore any instance of $G$ that has more
node-vertices, circles or fixed $!$-vertices than $H$, as the number of each
of these is constant under wire homeomorphism.  We cannot, however, restrict
the search based on the number of wire-vertices.  Consider, for example, the
following pattern and target graphs:
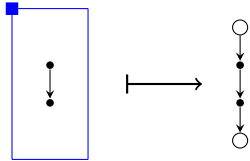
\[ \input{wild-match.tikz} \]
No matter how many times the $!$-box in the pattern is expanded, we can always
introduce enough wire-vertices in the target graph's wire to accommodate them
all.  So there are an infinite number of matchings in this case.  Empty
$!$-boxes do not even need a wire in the target graph to exhibit this
behaviour.

We call these (unfixed) $!$-boxes, containing no node-vertices or circles,
\textit{wild} $!$-boxes.  In order for it to be possible to enumerate all
candidate instances of $G$, we will require that no instance of $G$ contains
any wild $!$-boxes.  To show this for any given $G$, it is sufficient to show
that for each $b \in\:!(G)$, $B(b)$ is not wild after killing all $!$-vertices
not in $B^\uparrow(b)$.  The rest of this discussion will assume no instance
of $G$ contains wild $!$-boxes.

The first step is to consider how to explore the space of instances in an
ordered way.  There may be an infinite number of witnesses of even a single
instance, since applying $\COPY_b;\KILL_{b^1}$ to a graph $G$ with a
$!$-vertex $b$ yields a graph isomorphic to $G$, and so we cannot simply try
to build every possible instantiation.  At the same time, we cannot
depend on expansion-normal form as we need to deal with arbitrary instances,
not just concrete ones.

That said, we can produce something similar enough to expansion-normal form
for our purposes.  We will use the $\FIX$ operation used in section
\ref{sec:bbox-induction}, and we will adjust our definition of depth to ignore
$!$-vertices that have been fixed (so a $!$-vertex with only fixed parents has
depth $0$).  Since we are really using this as a notational convenience, we
will relax the constraints on $\FIX$ to allow the fixing of $!$-vertices whose
parents are all fixed.  The important thing to note is that, in the context of
instantiations, $\FIX$ behaves just like $\DROP$: after $\FIX_b^x$, there can
be no further operations on $b$, and it can potentially reduce the depth of
any children of $b$.

If we take an instantiation $S$ of $G'$ from $G'$ and a fresh fixing tag $x$,
and append a $\FIX^x_b$ operation for every $b \in\:!(G')$, then the result,
call it $S'$, is enough like a concrete instantiation for the arguments from
section \ref{sec:bbox-op-seq-reg-forms} to hold with minimal adjustment (just
read ``$\FIX^x$ or $\DROP$'' wherever $\DROP$ is mentioned).  We also need an
equivalent of $\EXP_b$; we will call this $\CFIX^x_b$, which will be a
shorthand for $\COPY_b;\FIX^x_{b^1}$.

So we can transform $S'$ into a witness for $G \succeq G'$ that is
depth-ordered (according to our revised definition of depth) and is composed
entirely of $\EXP$, $\KILL$ and $\CFIX^x$ operations.  As a result, we only
need to search for instantiations of this form in order to find all possible
instances.

\begin{remark}\label{rem:x-fixed-matching}
    Using this trick comes with a caveat: when we attempt to find a morphism
    from $G'$ to $H$, an $x$-fixed $!$-vertex would ordinarily only be allowed
    to match another $x$-fixed vertex.  However, we want $x$-fixed vertices to
    behave like unfixed vertices when it comes to building matching morphisms,
    and allow them to match \textit{any} vertex of the target graph.  This can
    be acheived by either removing all $x$-fixing tags from $G'$ before
    finding the matching morphism (which is the approach taken by the
    algorithm in section \ref{sec:quanto-matching}) or by relaxing the
    constraint with a special case for the $x$ tag.
\end{remark}

We will call the area of the pattern graph not in any $!$-boxes or only in
fixed $!$-boxes the \textit{match surface}.  This is the graph that would
result from killing all the unfixed $!$-boxes.  Importantly, no $\EXP$,
$\KILL$ or $\CFIX^x$ operation can reduce the match surface.  Furthermore, if
every $!$-box of every instance of $G$ contains at least one node-vertex or
circle, $\EXP$ and $\CFIX^x$ must always increase the number of node-vertices
or circles in the match surface.

We will take an inductive approach to finding instances of $G$ that are
candidates for a matching onto $H$, starting with the zero-length
instantiation.  To build the instantiations of length $n+1$, for each
instantiation $S$ of length $n$ (resulting in a $!$-graph $G_S$) and each
unfixed $b \in\:!(G_S)$ with depth $0$, we will create three new instantiations
by appending $\KILL_b$, $\EXP_b$ and $\CFIX^x_b$ to $S$.  If the resulting
graph has more node-vertices, circles or fixed $!$-vertices in its match
surface than are in $H$, we will discard that instantiation (and hence not
consider it as a prefix for the instantiations of length $n+2$).

\begin{remark}
    Note that there are more efficient versions of the algorithm;
    for example, it is sufficient to choose an arbitrary depth $0$ $!$-vertex
    at each stage rather than generating new instantiations for each such
    $!$-vertex.  However, it is easier to see that the algorithm here is
    correct, and this version suffices to show computability.
\end{remark}

It should be clear, given the preceeding discussion, that the algorithm will
enumerate all instances of $G$ that have a monomorphism to $H$.  It remains
for us to show that it will terminate, providing every $!$-box of every
instance of $G$ contains at least one node-vertex or circle.

Let $x_H$ be the number of node-vertices in $H$ and $y_H$ the number of
circles.  Suppose $S$ is an instantiation produced by the algorithm at some
iteration step, and let $x_S$ and $y_S$ be the corresponding properties for
the match surface of $G_S$.  Every $\EXP$ and $\CFIX^x$ operation must
increment either $x_S$ or $y_S$ (and $\KILL$ cannot decrease either), and the
algorithm will discard any instantiation where
\[ x_S + y_S > x_H + y_H \]
So there can be at most $x_H + y_H$ $\EXP$ or $\CFIX^x$ operations in $S$.

Each of these operations can, at most, double the number of unfixed
$!$-vertices in the graph, and $\KILL$ can only reduce that number.  Thus if
$z_G = |!(G)|$, $S$ cannot contain more than $z_G \cdot 2^{x_H + y_H}$ $\KILL$
operations.  $|S|$, the length of $S$, is therefore bounded by
\[ x_H + y_H + z_G \cdot 2^{x_H + y_H} \]
which is a finite number fixed by $G$ and $H$.

The same bound also places a limit on $|!(G_S)|$, limiting the number of
ways $S$ can be extended.  The total number of steps in the algorithm is then
bounded, and hence it terminates.

Note that this approach encompasses rewriting both string graphs and
$!$-graphs with $!$-graph rewrite rules; the former is simply a special case
of the latter, without any $\CFIX^x$ operations.

\section{Matching in Quantomatic}
\label{sec:quanto-matching}

\begin{figure}
    \[\input{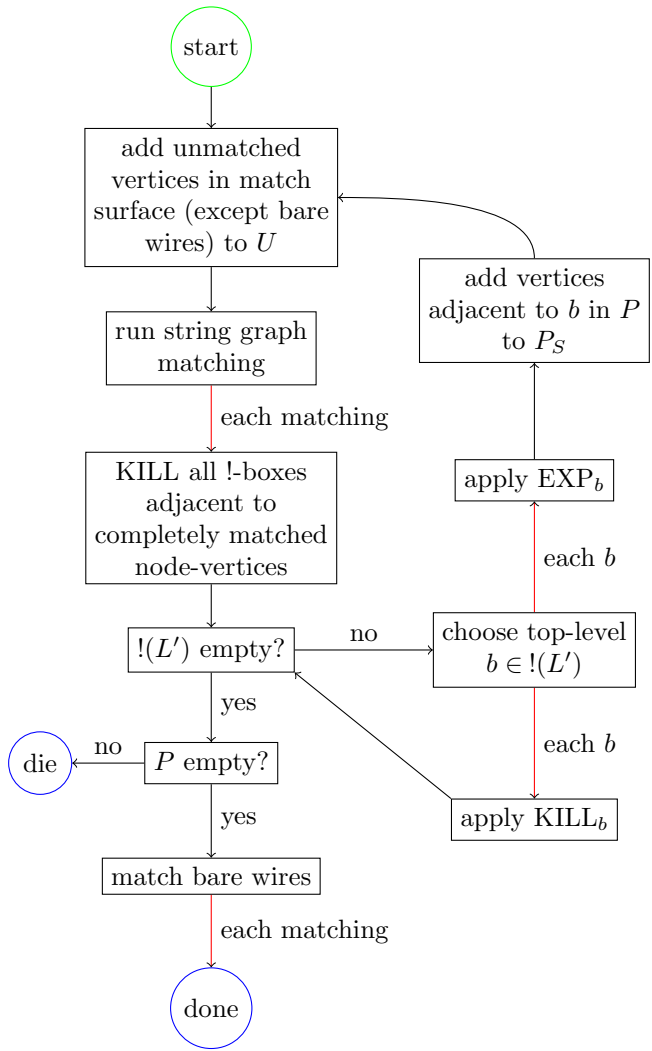}\]
    \caption{$!$-graph matching onto string graphs}
    \label{fig:bg-on-sg-match}
\end{figure}

\begin{figure}
    \[\input{string-graph-matching-subroutine.tikz}\]
    \caption{String graph matching subroutine}
    \label{fig:sg-match-sub}
\end{figure}

Quantomatic interleaves the three steps of $!$-graph matching (instantiation,
wire homeomorphism and building a suitable morphism) in a branching algorithm.
In this section, we will describe this approach.  We will start with the
algorithm as currently implemented, which assumes the target graph has no
$!$-vertices, and then describe how we plan to extend it to allow arbitrary
$!$-graphs as the target.

The algorithm maintains a partial graph monomorphism from the pattern graph to
the (concrete) target graph.  It terminates when this is a complete morphism,
and is successful if it is a local isomorphism.

The overall approach is to extend the graph morphism to the match surface
(those parts of the pattern graph not contained in any unfixed $!$-boxes),
then choose a top-level $!$-vertex and try killing it or expanding it.  In the
expansion branch, the match surface will have increased, so we repeat the
graph morphism extension step, before trying another $!$-vertex.  In the
killing branch, we proceed straight to choosing another $!$-vertex.  At each
choice that could affect the final match, the algorithm branches, so as to
explore all options.  If it fails to extend the graph morphism to the match
surface at any point, that branch terminates with failure.

Circles are matched first in each round of extending the graph morphism.
Circle matching involves taking each unmatched circle in $U_C$ and matching it
to an arbitrary unmatched circle of the same type in $T$, and removing the
vertex from $T$.  Since circles are indistinguishable under rewriting (being
wire-homeomorphic to each other), no branching needs to occur.  Doing this
first therefore reduces the amount of work duplicated across branches of the
algorithm.

Bare wires are left until the very end of the matching process, after all the
instantiation has been done.  This is because (as noted in the discussion on
wild $!$-boxes) bare wire matching does little to constrain the possible
instantiations and can generate a lot of branches.  For each bare wire $s
\xrightarrow{e} t$ in the match surface, the bare wire matching step branches
for each unmatched edge $e'$ in $G$ connected to a wire-vertex of the same
type as the bare wire.  It then replaces $e'$ with
three edges and two wire-vertices (of the appropriate type):
\[ \input{match-edge-expand.tikz} \]
and matches $s \xrightarrow{e} t$ to the new wire-vertices and the edge
between them.

To deal with wire homeomorphism, the algorithm starts with a variant of the
normalisation scheme from section \ref{sec:sg-match-wire-homeo}: every
interior wire has two wire-vertices; bare wires have only the input and output
vertices; circles have a single wire-vertex; and input and output wires have
only the input or output vertex in the pattern graph but an extra wire-vertex
in the target graph.  This ensures that a wire that can accept a match from a
bare wire (ie: any wire other than an interior wire already matched by an
interior wire or a circle matched by a circle) always has at least one
unmatched edge when the bare wire matching step is reached.

During a run of the algorithm to match a $!$-graph $L$ onto a string graph
$G$, Quantomatic maintains the following state:
\begin{itemize}
    \item $L$ : a normalised (as described above) $!$-graph
    \item $L'$ : an instance of $L$
    \item $S$ : a partial instantiation of $L$, resulting in $L'$
    \item $G$ : a normalised (as described above) string graph
    \item $m : V_{L'} \rightarrow V_G$ : a partial injective function
        describing the matching so far
    \item $U \subseteq V_{L'}$ : a set of unmatched vertices of $L'$
    \item $P \subseteq N(L')$ : a set of partially-matched node-vertices
    \item $P_S \subseteq P$ : a set of scheduled partially-matched
        node-vertices
    \item $T \subseteq V_G$: a set of vertices that can be matched onto
\end{itemize}

Note that we do not maintain a map of edges; this is because $!$-graphs are
simple, so if a function from vertices of one graph to vertices of another
extends to a graph morphism, it does so uniquely.

$U$ is partitioned into $U_C$, containing wire-vertices in circles (of which
there can only be one for each circle), $U_W$, containing other wire-vertices,
and $U_N$, containing node-vertices.

A note on terminology: a vertex is \textit{matched} if it is in either the
domain or image of $m$.  If $n$ is a matched node-vertex of $L'$, we call it
(and its image $m(n)$) \textit{partially matched} if some of the neighbourhood
of $m(n)$ is not in the image of $m$.  Otherwise, it is \textit{completely
matched}.  The same terminology applies to $m(n)$.  For example, consider the
following incomplete matching:
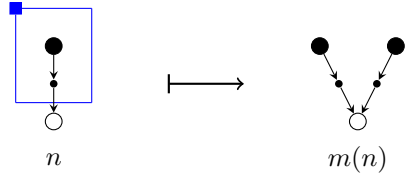
\[ \input{partial-match.tikz} \]
$n$ has been matched to a vertex in the target graph, but we will need to
expand the $!$-box twice before we are done with $n$.  This is an example of a
partially matched node-vertex.

$P$ is (a superset of) the set of partially matched vertices of $L'$, and
$P_S$ is intended to contain those node-vertices of $P$ that may be able to
become completely matched due to the application of a $!$-box operation.

We split the algorithm description into two parts.  Figure
\ref{fig:bg-on-sg-match} contains the part of the algorithm that explores the
instantiation space, and figure \ref{fig:sg-match-sub} has the part that
extends the match morphism.  This latter part can be used in a simpler wrapper
to match string graphs onto string graphs.

Red wires (with labels starting ``each'') indicate a branching point.
A branch is created for each possibility.  If there are no possibilities, the
current branch dies.  So, for example, if the algorithm takes a vertex from
$U_N$ and there are no possible matchings for it in the target graph, that
branch will be killed off.

The initial match state has $L' = L$ and $G$ normalised (in the manner
described above), $S = m = U = P = P_S = \varnothing$, and $T$ populated with
all the vertices of $G$ not in bare wires.

A proof of the correctness of this algorithm can be found in appendix
\ref{apdx:sg-matching}.

\subsection{Matching Onto $!$-Graphs}
\label{sec:quanto-match-bg}

\begin{figure}
    \[\input{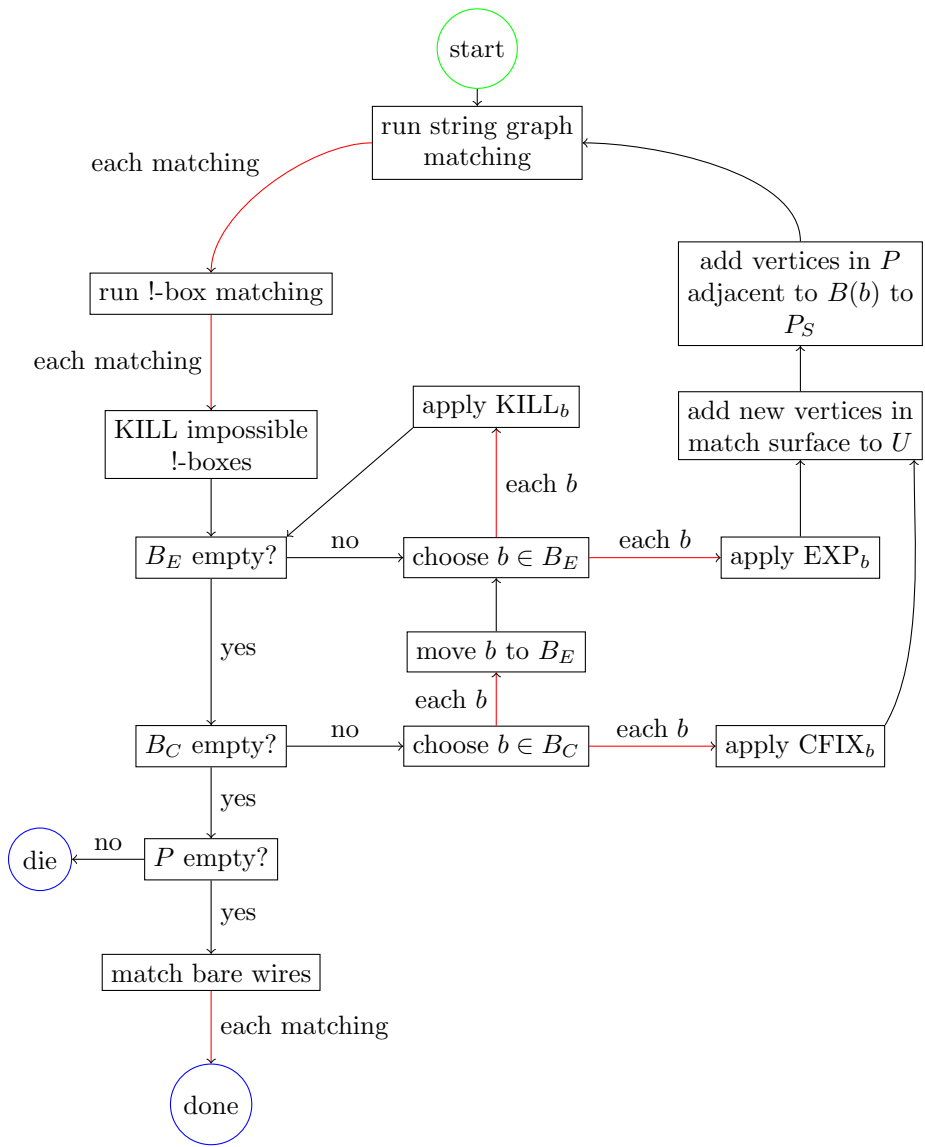}\]
    \caption{$!$-graph matching onto $!$-graphs}
    \label{fig:bg-on-bg-match}
\end{figure}

\begin{figure}
    \[\input{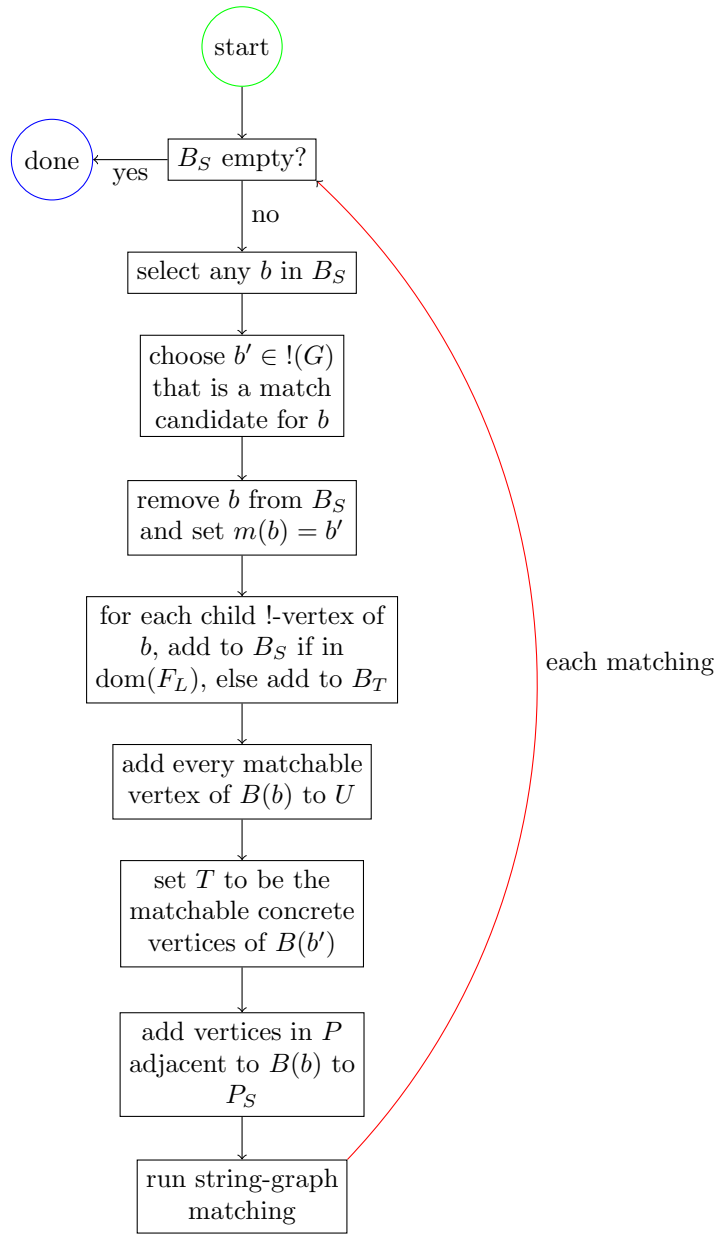}\]
    \caption{$!$-box matching subroutine}
    \label{fig:bb-match-sub}
\end{figure}

In this section, we describe how the previous algorithm can be extended to
match arbitrary $!$-graphs.  This algorithm has been implemented for
Quantomatic (albeit not in a released version yet).

Firstly, we need to maintain information about which $!$-vertices are fixed.
Since $!$-vertices are only fixed at intermediate stages of a proof (when
using \textsc{Induct} or during the matching process itself), we
maintain the information about which $!$-vertices are fixed with which tags
separately to the graphs.  Thus we extend the match state (and algorithm
inputs) with two partial maps, $F_L :\:!(L') \nrightarrow \mathcal{A}$ and
$F_G :\:!(G) \nrightarrow \mathcal{A}$.

We also add two further sets to the match state to keep track of $!$-vertices:
\begin{itemize}
    \item $B_S \subseteq \dom(F_L)$: a set of unmatched fixed $!$-vertices
    \item $B_C \subseteq !(L')$: a set of $!$-vertices that can be copied
    \item $B_E \subseteq !(L')$: a set of $!$-vertices that can be expanded
\end{itemize}

The initial match state is mostly empty as before, with $G$ and $L = L'$ again
normalised.  We also populate $U$ with the vertices of $L$ that are not in any
$!$-box, and $T$ with all the vertices of $G$ that are not in any $!$-box. The
new parts of the match state are populated as follows:
\begin{itemize}
    \item $B_S = \dom(F_L)$
    \item $B_C = \{b \in !(L) \setminus B_S : \delta(b) = 0\}$
    \item $B_E = \varnothing$
\end{itemize}

The string graph matching subroutine is extended (via ``hooks'' in the
implementation) to allow the matching to be constrained by $!$-box membership.
In particular, we only allow a vertex $v \in V_L$ to match a vertex $v' \in
V_G$ if $m[B_v] = B_{v'}$, where $B_v = \{ b \in\;!(L) : v \in B(b) \}$ and
$B_{v'} = \{ b \in\;!(G) : v' \in B(b) \}$.

Also implicit in the algorithm description is that when we kill a $!$-vertex
$b$, we always add any vertices in $P$ that were adjacent to $B(b)$ to $P_S$,
so that the next string graph matching run can remove them from $P$ if they
are now completely matched.

Additionally, one part of figure \ref{fig:bb-match-sub} refers to a
``candidate for $m(b)$''.  If $b \in\;!(L) \setminus \dom(m)$, we say that $b'
\in\;!(G) \setminus \im(m)$ is a \textit{match candidate} for $b$ if
\begin{itemize}
    \item it shares the same parents (ie:
        $m[B^\uparrow(b)\setminus b] = B^\uparrow(b')\setminus b'$), and
    \item if $b$ is fixed, $b'$ is fixed with the same tag (ie: $(b \in
        \dom(F_L) \Rightarrow F_L(b) = F_G(b'))$)
\end{itemize}

The major difference between this algorithm and the previously-presented one
is that $\CFIX_b^x$ is provided as an alternative choice to $\EXP_b$ and
$\KILL_b$ when the algorithm deals with the next top-level $!$-vertex of $L'$.


\chapter{Conclusions and Further Work}
\label{ch:conclusions}

In this dissertation, we have demonstrated how the string graph formalism for
the diagrammatic languages of traced symmetric monoidal categories and compact
closed categories can be extended to allow the finitary representation of
infinite families of string graphs and string graph equations, and how these
$!$-graphs and $!$-graph equations can be used in conjunction with
double-pushout graph rewriting to do equational reasoning both with and on
infinite families of morphisms in these categories.  We have also presented
some inference rules for a nascent logic of $!$-graphs.

We started by extending the language of monoidal signatures to allow for nodes
with variable-arity edges, representing families of morphisms such as the
spiders induced by commutative Frobenius algebras.

We then further extended string graphs with a language, internal to the graph,
describing potentially infinite families of string graphs.  We extended the
notion of string graph equations to these $!$-graphs, constructing $!$-graph
equations (and rewrite rules) to represent infinite familes of string graph
equations (and rewrite rules).

We demonstrated how these $!$-graph rewrite rules can be used to both rewrite
string graphs and to rewrite $!$-graphs, producing new $!$-graph rewrite
rules.  We showed that this rewriting is sound with respect to the
interpretation of $!$-graphs and $!$-graph rewrite rules as families of string
graphs and string graph rewrite rules.

We described an equational logic of $!$-graphs implemented by $!$-graph
rewriting and built on this with further inference rules, forming a logic of
$!$-graphs that is again sound with respect to the interpretation of
$!$-graphs as families of string graphs.  This included a graphical analogue
of induction, which we used to derive the spider law for commutative Frobenius
algebras.

We rounded off by showing how $!$-graph rewriting can be implemented.  We
demonstrated that, providing a $!$-graph rewrite rule has no instantiations
with wild $!$-boxes, it is possible to determine the finite set of
instances of the rule that apply to a given $!$-graph, and gave an example of
an algorithm to find matchings, with their associated instantiations, from the
LHS of a $!$-graph rewrite rule to a string graph or $!$-graph.

\section{Further Work}
\label{sec:further-work}

The most obvious next step would be to remove the requirement that
variable-arity edges of the same type commute by placing an order on the
(variable-arity) edges of a node and using this order when determining the
value of the elementary subgraph containing that node.  This would allow us to
have spiders for non-commutative Frobenius algebras.

A visual way to represent this could be to place a ``starting mark'' on a
node, and count clockwise from there:
\[ \input{noncom-spider.tikz} \]
Then we would have that the following two graphs are not the same:
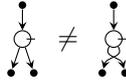
\[ \input{noncom-spider-ex.tikz} \]
and, in particular, we would not allow a graph homomorphism from one to the
other.  The implications of this need working out in detail to ensure there
are no unexpected side-effects; for example, there is a potential difference
between expanding a $!$-box clockwise and anti-clockwise.  It would allow more
expressivity, though, and the commutative version presented in this thesis
should be recoverable by placing constraints on the allowed valuations.

Another idea that could be investigated is ``defined'' generators.  For
example, suppose we have a commutative monoid $(A,\whitemult,\whiteunit)$.  We
could define a ``spidered'' version of this in the following manner:
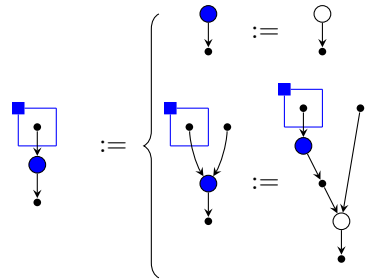
\[ \input{sp-def-mult.tikz} \]
This essentially defines two $!$-graph rewrite rules
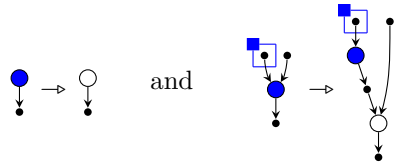
\[
    \input{sp-mult-def-base-rw.tikz}
    \qquad \textrm{and} \qquad
    \input{sp-mult-def-step-rw.tikz}
\]
Some restrictions would need to be placed on such a rewrite system, such as
termination and preservation of value, in order for it to be considered a
definition, but this could be a powerful tool for proving theorems.  Appendix
\ref{apdx:spider-law} uses it to prove the spider law for a commutative
Frobenius algebra.

The above definition has a major drawback, though, in that it implicitly
requires that the monoid is commutative.  For example, suppose we have the
graph
\[ \input{sp-mult-three.tikz} \]
This can be expanded to
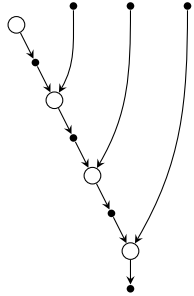
\[ \input{sp-mult-three-expn-1.tikz} \]
using the rewrite system implied by the definition, but it can also be
expanded to
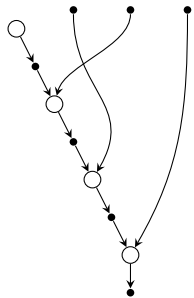
\[ \input{sp-mult-three-expn-2.tikz} \]
which, in the absense of commutativity and associativity laws, could have a
different value.  This is because all the incoming edges of the variable-arity
multiplication have the same type, and so we have a choice of matches when
rewriting.

However, if we combine this idea of definitions with the non-commutative
generators already mentioned, we can make the definition
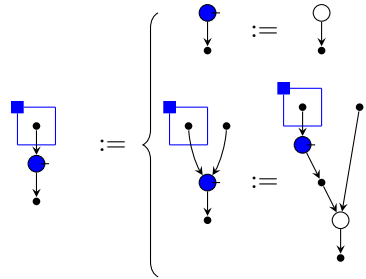
\[ \input{sp-def-mult-noncom.tikz} \]
which will produce a unique expansion for each possible starting graph.  In
particular, it forces the incoming edges to this node to be expanded in
clockwise order, so
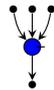
\[ \input{mult-noncom.tikz} \]
has the unique expansion
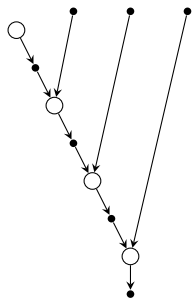
\[ \input{mult-noncom-expn.tikz} \]

There are undoubtably more inference rules to be discovered.  For example, a
variant on $\BOX(G)$ (section \ref{sec:bbox-intro}) that adds a fresh
$!$-vertex $b$ to $G$ but only adds edges to itself and to vertices in $U(G)$
is likely to give rise to a similar inference rule to \textsc{Box}, but
this has yet to be proved sound.

Kissinger introduced a notion of critical pairs for graphs in
\cite{Kissinger2008}.  The rewriting of $!$-graphs laid out in this thesis
could potentially be used, in combination with Kissinger's definition of
critical pairs, to compute confluent extensions to rewrite systems of
$!$-graphs using a Knuth-Bendix completion algorithm\cite{Huet1981}.

Of course, one of the most important steps to take is to actually implement
the ideas presented in this thesis.  In particular, we are planning to
implement and prove the correctness of the algorithm in section
\ref{sec:quanto-match-bg}.  This algorithm could no doubt be improved,
however, to make it more efficient.  For example, many graphs involving
commutative Frobneius algebras have a lot of internal symmetry, which results
in multiple matches being found that all produce the same
rewrite; it would be desirable to eliminate this duplication, preferably at
the matching stage.\footnote{Some initial work on this was done by Matvey Soloviev when
the theory of $!$-graphs and string graphs was still in flux, but nothing was
ever published; this could be revisited now that there is a secure
foundation.}

The implementation of the ideas in this thesis is work that is currently
underway in the Quantomatic project\cite{quanto}, which aims to be a useful
proof assistant for string diagrams in compact closed and traced symmetric
monoidal categories.  This encompasses QuantoCosy, a tool for generating rules
from models (as set out in \cite{Kissinger2012}), QuantoDerive, a tool for
constructing and checking graphical proofs, and QuantoTactic, an
Isabelle\cite{isabelle} tactic implementation using Quantomatic's
graph-matching engine to direct proofs and checking their correctness in
Isabelle.

\appendix
\settocdepth{chapter}

\chapter{Correctness of Quantomatic's Matching}
\label{apdx:sg-matching}

We prove the correctness of the algorithm presented in section
\ref{sec:quanto-matching} for matching $!$-graphs onto string graphs.

Recall the match state:
\begin{itemize}
    \item $L$ : the pattern graph (normalised)
    \item $L'$ : an instance of $L$
    \item $G$ : the target graph (normalised)
    \item $S$ : a partial instantiation of $L$, resulting in $L'$
    \item $m : V_{L'} \rightarrow V_G$ : a partial injective function
        describing the matching so far
    \item $U \subseteq V_{L'}$ : a set of unmatched vertices of $L'$
    \item $P \subseteq N(L')$ : a set of partially-matched node-vertices
    \item $P_S \subseteq P$ : a set of schedules partially-matched
        node-vertices
    \item $T \subseteq V_G$: a set of vertices that can be matched onto
\end{itemize}

Throughout the entire procedure, in addition to the invariants implicit in the
typing of the components of the match state, we maintain the following global
invariants:
\begin{enumerate}[label=(\arabic*)]
    \item \label{enum:g-inv-inj} $m$ is injective and respects vertex types
    \item \label{enum:g-inv-edge-match} For any $v,w \in \dom(m)$, if there is
        an edge from $v$ to $w$ in $L'$, then there is an edge from $m(v)$ to
        $m(w)$ in $G$
    \item \label{enum:g-inv-local-iso} $P \subseteq \dom(m)$ and for any
        node-vertex $v \in \dom(m)\setminus P$, all wire-vertices adjacent to
        $m(v)$ in $G$ are in the image of $m$
\end{enumerate}

Given that $!$-graphs are simple, invariants \ref{enum:g-inv-inj} and
\ref{enum:g-inv-edge-match} ensure that when $m$ is total, it uniquely extends
to a $!$-graph monomorphism $\hat m : L' \rightarrow G$.  Invariant
\ref{enum:g-inv-local-iso} ensures that, if $P$ is empty, $U(\hat m)$ will be
a local isomorphism.  Thus $\hat m$ will be a matching of $L'$ onto $G$ under
$S$.

\section{The String Graph Matching Subroutine}

The algorithm is shown in figure \ref{fig:appdx-sg-match-algo}.

\begin{figure}
    \input{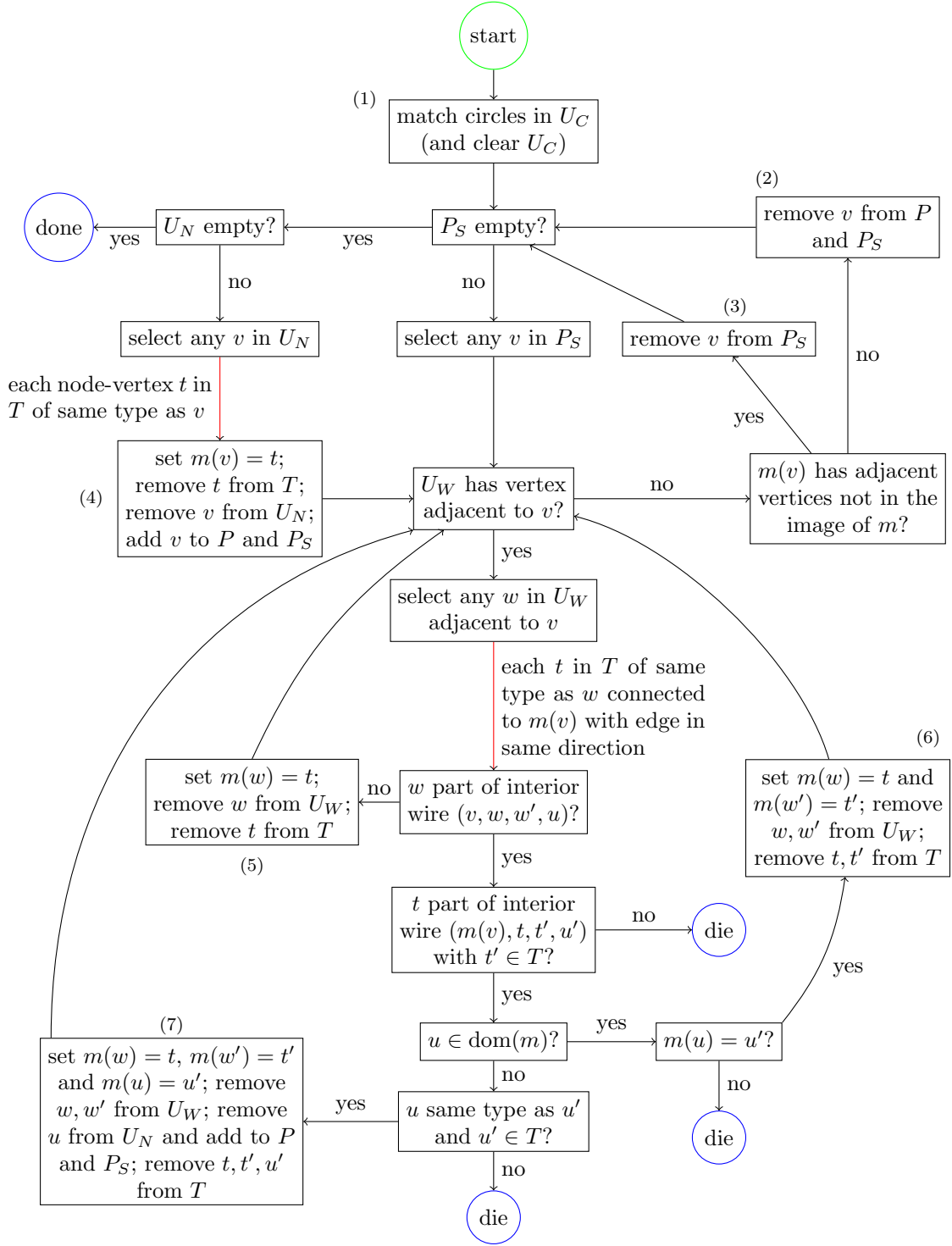}
    \caption{String Graph Matching Subroutine}
    \label{fig:appdx-sg-match-algo}
\end{figure}

\subsection{Preconditions/Invariants}
\label{sec:il-precond}

In addition to the global invariants, the string graph matching subroutine has
the following requirements on the match state.  What is more, these are
maintained throughout the procedure.
\begin{enumerate}[label=(\roman*)]
    \item \label{enum:il-inv-norm} $L'$ and $G$ are both normalised
    \item \label{enum:il-inv-unmatched} $U\cap\dom(m) = \varnothing$ (ie: the
        vertices in the unmatched set are not part of the existing matching)
    \item \label{enum:il-inv-reachable} Each vertex in $U_W$ is adjacent to
        something in $U_N\cup P_S$ (ensures we can always reach wire vertices
        in $U_W$ by starting from a node-vertex in $U_N$ or $P_S$)
    \item \label{enum:il-inv-Ps} If $v \in P$ and $N_G(m(v)) \subseteq
        \im(m)$, $v \in P_S$ (everything that is completely matched and in $P$
        is also in $P_S$)
    \item \label{enum:il-inv-wire-uw} If $v \in U_W$ and $w$ is a wire-vertex
        adjacent to $v$ in $L'$, $w \in U_W$ ($U_W$ contains complete wires)
    \item \label{enum:il-inv-wire-connected} If $v \in \dom(m)$ is a
        wire-vertex, then any vertices it is adjacent to are also in $\dom(m)$
    \item \label{enum:il-inv-t-wire-connected} If $v \in T$ is a
        wire-vertex, then any vertices it is adjacent to are in either $T$ or
        $\im(m)$
    \item \label{enum:il-inv-t-img} $T \cap \im(m) = \varnothing$
\end{enumerate}

\subsection{Postconditions}
These are true for the match state that results from each successful branch:
\begin{enumerate}[label=(\Roman*)]
    \item \label{enum:il-postcond-no-unmatched} $U = P_S =
        \varnothing$ (every vertex marked for matching was handled)
    \item \label{enum:il-postcond-all-matched} $\dom(m)$ is the union of the
        initial states of $U$ and $\dom(m)$ (we have matched exactly what we
        were asked to)
    \item \label{enum:il-postcond-T} $\im(m)$ contains only vertices that were
        in either the initial state of $\im(m)$ or the initial state of $T$
        (matches were only made against vertices in $T$)
    \item \label{enum:il-postcond-P} $P$ is exactly the set of vertices in
        $\dom(m)$ whose image is adjacent to a wire-vertex not in $\im(m)$ ($P$
        contains exactly the partially-matched vertices)
    \item \label{enum:il-postcond-untouched} $L$, $L'$, $G$ and $S$ are
        identical to their starting states
\end{enumerate}

Postcondition \ref{enum:il-postcond-all-matched} (in conjunction with the
global invariants) states that the algorithm has completed the
requested work.  Postcondition \ref{enum:il-postcond-no-unmatched} just states
that we left the match state ``clean''; it still satisfies all the
preconditions and repeating the algorithm without altering the state further
would do nothing.  Postcondition \ref{enum:il-postcond-P} says that $P$ is the
smallest set that satisfies global invariant \ref{enum:g-inv-local-iso}.

\subsection{Termination}
\label{sec:sg-matching-termination}

For termination of the inner loop exploring the wires incident to a
node-vertex $v$, we consider the variant $n = |U_W\cap N_{L'}(v)|$, where
$N_{L'}(v)$ is the neighbourhood of $v$ in $L'$.  In each of these branches,
$n$ decreases by either $1$ or $2$, and nothing is ever added to $U_W$ (and
$N_{L'}(v)$ is fixed).  So this loop must terminate.

The loop working through $P_S$ has two branches: when it is not empty and when
it is not.  Either may add more things to $P_S$.  We will take as variant $m =
2|U_N| + |P_S|$.  When this is $0$, the algorithm will terminate.  Otherwise,
if $P_S$ is empty, $U_N$ will have one vertex removed and added to $P_S$,
decreasing $m$ by one.  If $P_S$ is not empty, one vertex will be removed from
$P_S$, decreasing $m$ by one.  Every vertex that is added to $P_S$ in this
loop is taken from $U_N$, further decreasing $m$.  As a result, $m$ will
decrease by at least one with every iteration of the loop.

\subsection{Correctness}
\label{sec:sg-matching-correctness}

There are seven steps of the algorithm that modify the match state; we will
show that each of them satisfies the following properties, under the
assumption that the preconditions (section \ref{sec:il-precond}), global
invariants and match state type constraints hold at the start of the step
(note that when we say that a vertex is \textit{taken from} a set, we mean
that it was in the set at the start of the step, and is removed from the set
by the end of the step):

\begin{enumerate}[label=(\alph*)]
    \item \label{enum:step-constr-unchanged} $L$, $L'$, $G$ and $S$
        are unmodified
    \item \label{enum:step-constr-m} If $U$, $T$ or $m$ are
        modified, the only modifications are to take vertices from $U$ and
        extend $m$ to map those vertices to ones of the same type taken from
        $T$
    \item \label{enum:step-constr-Uw} If a wire-vertex is added to $\dom(m)$,
        then so is the entire wire, and for any pair of vertices $u,v$ in
        the wire with an edge from $u$ to $v$, there is an edge from $m(u)$ to
        $m(v)$ (at the end of the step) in $G$
    \item \label{enum:step-constr-Un} Any vertex taken from $U_N$ is added to
        both $P$ and $P_S$, and these are the only vertices added to either
        $P$ or $P_S$
    \item \label{enum:step-constr-P} If $P$ or $P_S$ are modified, the only
        modifications are the one from \ref{enum:step-constr-Un}, removing a
        partially matched vertex with no adjacent vertices in $U_W$ from $P_S$
        or removing a completely matched vertex with no adjacent vertices in
        $U_W$ from $P$ and $P_S$
\end{enumerate}

Given that these properties hold, we can show that all steps of the algorithm
preserve the match state type constraints, the global invariants and the
preconditions.

We can see that \ref{enum:step-constr-unchanged} holds for all seven steps by
inspection.  It is clear that the other three properties hold for steps
(1)-(3), again by inspection.  Only steps (5)-(7) need extra
justification.  We first note that the selection of $w$ demands that, since
$L$ is in normal form, it is part of an interior wire with two wire-vertices
or is the input on an input wire or the output on an output wire (and the sole
wire-vertex on the wire in both cases).

For step (5), the only non-obvious property is \ref{enum:step-constr-Uw}.
However, by what we have already noted, since $w$ is not on an interior wire
it must be an input or output.  Since $w$ and $v$ are the only vertices on the
wire and $v$ is already in $\dom(m)$, and the edge between $v$ and $w$ is the
same direction as the edge between $m(v)$ and $t$, property
\ref{enum:step-constr-Uw} holds for (5).

In step (6), it is clear that $t$ and $t'$ are in $T$ and $w$ is in $U_W$.
Precondition \ref{enum:il-inv-wire-uw} ensures that $w'$ is in $U_W$.  Since
$L$ and $G$ are $!$-graphs, $w'$ must have the same type as $w$, which has the
same type as $t$, which has the same type as $t$.  So property
\ref{enum:step-constr-m} clearly holds.  Property \ref{enum:step-constr-Uw}
can be seen to hold by noting that the edge between $u$ and $w$ is in the same
direction as the edge between $m(u)$ and $t$, and all the edges in a wire must
be in the same direction.  Properties \ref{enum:step-constr-Un} and
\ref{enum:step-constr-P} trivially hold, as $U_N$, $P$ and $P_S$ are not
modified.

For step (7), we note that precondition \ref{enum:il-inv-reachable} (together
with global invariant \ref{enum:g-inv-local-iso} and the fact that $P_S
\subseteq P$) requires that, since $w' \in U_W$ and $u \notin \im(m)$, $u$
must be in $U_N$.  Then properties \ref{enum:step-constr-m} and
\ref{enum:step-constr-Uw} follow in the same manner as for (6), and properties
\ref{enum:step-constr-Un} and \ref{enum:step-constr-P} clearly hold.

We now show that these properties mean that the match state type constraints,
global invariants and preconditions are all maintained.

\paragraph{Type constraints} The typing constraints for $U$ and $T$ are
trivially maintained as \ref{enum:step-constr-m} demands they are never added
to.  Likewise, the constraints on $L$, $L'$, $G$ and $S$ are maintained as
\ref{enum:step-constr-unchanged} means they are never changed.  Properties
\ref{enum:step-constr-Un} and \ref{enum:step-constr-P} ensure that $P$ is only
ever added to from $U_N$ and that anything added to $P_S$ is also added to
$P$, and anything removed from $P$ is also removed from $P_S$, so those
constraints hold.  Finally, the constraint on $m$ holds by property
\ref{enum:step-constr-m}.

\paragraph{Global Invariants} Injectivity of $m$ is preserved as property
\ref{enum:step-constr-m} demands that only vertices drawn from $T$ may be used
for images of additions to $m$, and precondition \ref{enum:il-inv-t-img} gives
us that these must not already be in the image of $m$.  The rest of global
invariant \ref{enum:g-inv-inj} follows from the same property.

Suppose there is a an edge in $L$ from a vertex $v$ to another vertex $w$, and
both are in $\dom(m)$ at the end of a step, but one was not at the start of
the step.  We know they cannot both be node-vertices, and so one or both are
wire-vertices.  Since at least one was not in the image of $m$ at the start of
the step, none of the wire-vertices can have been by precondition
\ref{enum:il-inv-wire-connected}.  Then property \ref{enum:step-constr-Uw}
ensures that there is an edge from $m(v)$ to $m(w)$.
So global invariant \ref{enum:g-inv-edge-match} is preserved.

Global invariant \ref{enum:g-inv-local-iso} is preserved when vertices are
removed from $P$ due to property \ref{enum:step-constr-P} and when vertices
are added to $m$ due to property \ref{enum:step-constr-m}.

\paragraph{Preconditions/Invariants}
\begin{enumerate}[label=(\roman*)]
    \item $L'$ and $G$ are both in normal form

        By property \ref{enum:step-constr-unchanged}, which ensures $L'$ and
        $G$ are never changed.
    \item $U\cap\dom(m) = \varnothing$

        Property \ref{enum:step-constr-m} requires that every vertex added to
        $\dom(m)$ is removed from $U$.
    \item Each vertex in $U_W$ is adjacent to something in $U_N\cup P_S$

        Given that property \ref{enum:step-constr-m} prevents anything being
        added to $U_W$, this precondition can only be broken by removing
        something from $U_N$ or $P_S$.  Property \ref{enum:step-constr-Un}
        ensures that anything removed from $U_N$ is added to $P_S$, and
        property \ref{enum:step-constr-P} ensures that anything removed from
        $U_N$ has no adjacent vertices in $U_W$.
    \item If $v \in P$ and $N_G(m(v)) \subseteq \im(m)$, $v \in P_S$

        Property \ref{enum:step-constr-P} ensures that when we add anything to
        $P$, we also add it to $P_S$, and we only remove completely matched
        vertices from $P_S$.
    \item If $v \in U_W$ and $w$ is a wire-vertex adjacent to $v$ in $L'$, $w
        \in U_W$

        Property \ref{enum:step-constr-Uw} (in conjunction with property
        \ref{enum:step-constr-m}) ensures that only entire wires are removed
        from $U_W$ at once (and nothing is ever added to $U_W$).
    \item If $v \in \dom(m)$ is a wire-vertex, then any vertices it is
        adjacent to are also in $\dom(m)$

        This is preserved due to property \ref{enum:step-constr-Uw}.
    \item If $v \in T$ is a wire-vertex, then any vertices it is adjacent to
        are in either $T$ or $\im(m)$

        This is maintained by property \ref{enum:step-constr-m}, which ensures
        that anything removed from $T$ is added to $\im(m)$ and, due to
        precondition \ref{enum:il-inv-unmatched}, also ensures that nothing is
        ever removed from $\im(m)$.
    \item $T \cap \im(m) = \varnothing$

        This is preserved by property \ref{enum:step-constr-m}, which ensures
        that nothing is ever added to $T$, and anything added to $\im(m)$ is
        removed from $T$.
\end{enumerate}

\paragraph{Postconditions}
\begin{enumerate}[label=(\Roman*)]
    \item $U = P_S = \varnothing$

        The conditions that lead to the ``done'' step ensure that $U_N$ and
        $P_S$ are empty.  Then precondition \ref{enum:il-inv-reachable} means
        that $U_W$ must be empty.  $U_C$ is empty at the end because it is
        empty when step (1) is complete, and property \ref{enum:step-constr-m}
        ensures it is never added to.

    \item $\dom(m)$ is the union of the initial states of $U$ and $\dom(m)$

        Property \ref{enum:step-constr-m} ensures that everything removed from
        $U$ is added to $m$, that nothing else is added to $m$ and that
        nothing is removed from $m$.  Combining this with precondition
        \ref{enum:il-inv-unmatched} gives us that no entry of $m$ is every
        overwritten, and so postcondition \ref{enum:il-postcond-no-unmatched}
        provides us with what we need.

    \item \label{enum:il-postcond-T} $\im(m)$ contains only vertices that were
        in either the initial state of $\im(m)$ or the initial state of $T$

        By property \ref{enum:step-constr-m}.

    \item $P$ is exactly the set of vertices in $\dom(m)$ whose image is
        adjacent to a vertex not in $\im(m)$

        Global invariant \ref{enum:g-inv-local-iso} gives us most of what we
        need.  We just need that if $v$ is completely matched, it is not in
        $P$.  But this is guaranteed by invariant \ref{enum:il-inv-Ps} and the
        fact that $P_S = \varnothing$.

    \item $L$, $L'$, $G$ and $S$ are identical to their starting states

        This follows immediately from property
        \ref{enum:step-constr-unchanged}.
\end{enumerate}

\subsubsection{Completeness}
\label{sec:il-complete}

Let $f$ be a matching of the full subgraph of $L'$ given by the vertices $U
\cup \dom(m)$ into $G$ that restricts to $m$ on the vertices in $\dom(m)$.
Then we need to show that, after completion of the algorithm, one of the match
states produced will agree with $f$ on vertices, up to a choice of circles.
Note that, by the argument at the start of section \ref{apdx:sg-matching},
this will then uniquely extend to $f$.

Since $f$ is a matching with $U_C$ in its domain (and $U_C\cap\dom(m) =
\varnothing$ and $f$ restricts to $m$), there must be at least $|U_C|$
unmatched circles in $G$, and so the step where we match circles must succeed.

Invariant: the vertex map part of $f$ restricts to $m$.  $m$ is extended in
two places (after circle-matching).  Each of these follows a branching point.
We need to show that in both of these places, there is a branch that will
maintain this invariant.  We also need to show that none of the ``die''
branches are taken.

When adding something from $U_N$, it is clear that $f(v)$ is a valid matching
for $v$, and so this is a branch that will maintain the invariant.

When adding something from $U_W$, the wire starting $f(w)$ must be a valid
matching for the wire starting with $w$, as the entire wire must be in the
domain of $f$: if $w$ is an input or output, this is trivial.  Otherwise,
there must be a wire-vertex $w'$ adjacent to $w$, and a node-vertex $n$
adjacent to $w'$.  Then $w'$ is in $U_W$, and hence in $\dom(f)$, by
precondition \ref{enum:il-inv-wire-uw}, and then $n$ must be in either $P_S$
or $U_N$, by precondition \ref{enum:il-inv-reachable}, and so in $\dom(m)$ or
$U_N$, and hence $n$ is in $\dom(f)$.  Once we have matched $w$ against
$f(w)$, the only possible matching for $w'$ is $f(w')$, and similarly $n$ must
match against $f(n)$.  If $n$ is already in $\dom(m)$, we know it maps to
$f(m)$, so the matchings must agree.  So this branch maintains the invariant
(and does not die).

Thus, regardless of how vertices are chosen from $P_S$, $U_N$ or $U_W$, there
must be a trace that returns a match where $m$ agrees with $f$ on vertices.

\section{The Outer Loop}

The algorithm is shown in figure \ref{fig:appdx-bg-on-sg-match}.  The steps
that ``apply'' a $!$-box operation implicitly apply the operation to $L'$ and
record it in $S$.

\begin{figure}[h]
    \[\input{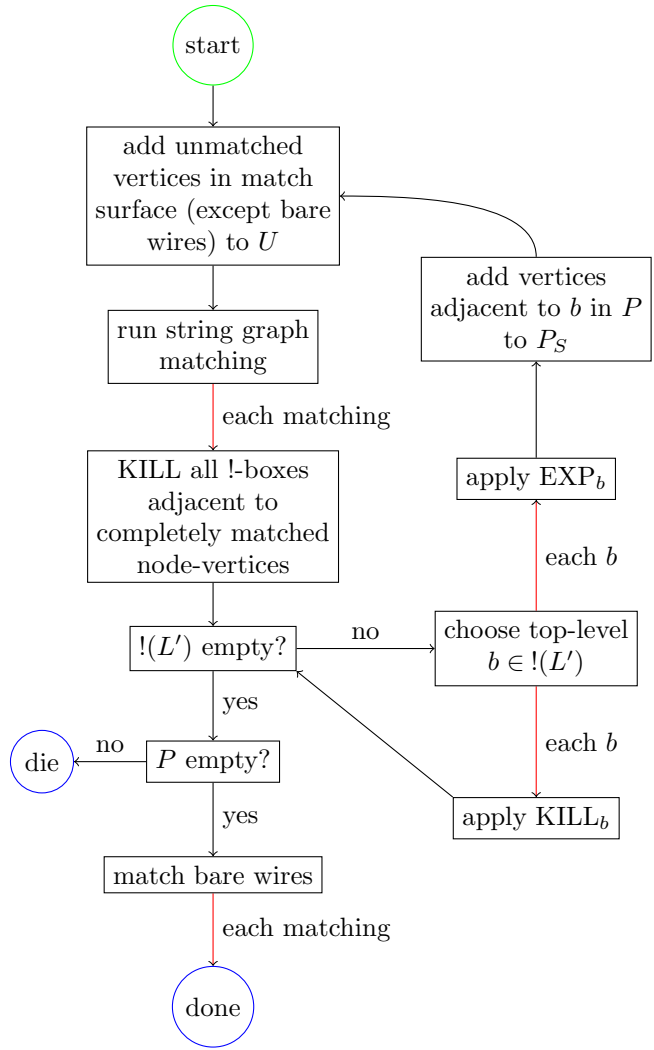}\]
    \caption{$!$-graph matching onto string graphs}
    \label{fig:appdx-bg-on-sg-match}
\end{figure}

Note that the global invariants are trivially satisfied at the start of the
algorithm, since $m$ and $P$ are empty.  They are maintained thereafter as
only the string graph matching subroutine and the bare wire matching routines
alter $m$ or $P$, and the $!$-box operations (which are the only things to
alter $L'$) can only add edges where either the source or the target is not in
$\dom(m)$.

\subsection{Termination}

The algorithm terminates, on the assumption that no sequence of $!$-box
operations will ever produce a wild $!$-box, by the following argument.

Each step terminates, and there is only one loop.  As a variant, we take
\[ n = |V_G| + |V_L\setminus\dom(m)| - |\dom(m)| \]

We calculate the variant at the point where we consider terminating the loop,
after the string graph matching step.  If the previous iteration involved
killing one or more $!$-boxes, $V_L\setminus\dom(m)$ will decrease (by the
combined size of the $!$-boxes, less bare wires --- this amount is positive
due to the assumption about wild $!$-boxes).  If the previous iteration
involved expanding a $!$-box, $|\dom(m)|$ will increase (by the size of the
$!$-box, excluding any bare wires).  $V_L\setminus\dom(m)$ is clearly bounded
below by $0$, and $|\dom(m)|$ is bounded above by $|V_G|$, by global invariant
\ref{enum:g-inv-inj}, so $n$ is bounded below by $0$ and decreases with every
iteration.  Hence the algorithm terminates.

Note that Quantomatic's implementation ensures termination for all pattern
graphs by killing any wild $!$-boxes as soon as they are produced.

\subsection{String Graph Matching Preconditions}

We need to show that the preconditions in section \ref{sec:il-precond} are
satisfied whenever the string graph matching step is run.

\begin{enumerate}[label=(\roman*),ref=(\roman*)]
  \item $L'$ and $G$ are both normalised

  We normalise the graphs at the start.  We then do not alter $G$ until the
  bare wire matching step, at which point the string graph matching routine
  can no longer be called.  Because $!$-boxes are open subgraphs, operations
  on them always remove or copy an entire wire at once, which ensures that
  $L'$ remains normalised throughout.
  \item $U\cap\dom(m) = \varnothing$

  We only ever add unmatched vertices to $U$.
  \item Each vertex in $U_W$ is adjacent to something in $U_N\cup P_S$

  This follows from openness of $!$-boxes, which (together with normalisation)
  ensures that any wire-vertex in the match surface that is not on a circle or
  bare wire must be adjacent to a node-vertex in the match surface.  The first
  time string graph matching happens, all such node-vertices are put in $U_N$.

  Since all non-bare-wire, non-circle vertices in the match surface are added
  to $U_W$ each time, and so added to $\dom(m)$ by the string graph matching
  step, on subsequent iterations of string graph matching any wire-vertex $v$
  in $U_W$ must have arisen from the expansion of a $!$-box (since this is the
  only way the match surface can be extended).  If $v$ is adjacent to a
  node-vertex that came from the $!$-box expansion, that node-vertex will be
  added to $U_N$.  Otherwise, it must be adjacent to a node-vertex in
  $\dom(m)$, and hence in $P$ (by the postconditions of the string graph
  matching step).  Then this node-vertex (being adjacent to the expanded
  $!$-box) will be added to $P_S$.  \item If $v \in P$ and $N_G(m(v))
  \subseteq \im(m)$, $v \in P_S$ (everything that is completely matched and in
  $P$ is also in $P_S$)

  The first time we run the string graph matching step, $P$ is empty.  After
  each string graph matching step, the postconditions ensure that there is no
  completely matched vertex in $P$.  As the rest of the algorithm does not add
  anything to $P$ or alter $m$, this is maintained throughout.
  \item If $v \in U_W$ and $w$ is a wire-vertex adjacent to $v$ in $L'$, $w
  \in U_W$.

  This is guaranteed by openness of $!$-boxes: we only ever add entire wires
  to $U_W$.
  \item If $v \in \dom(m)$ is a wire-vertex, then any vertices it is adjacent
  to are also in $\dom(m)$

  This is guaranteed by the postcondition \ref{enum:il-postcond-all-matched}
  in combination with precondition \ref{enum:il-inv-wire-uw} and the fact that
  $m$ is only altered by the string graph matching step.

  \item \label{enum:il-inv-t-wire-connected} If $v \in T$ is a wire-vertex,
  then any vertices it is adjacent to are in either $T$ or $\im(m)$

  $T$ and $m$ are only modified by the string graph matching subroutine, which
  maintains this invariant.  Initially, $m$ is empty, causing it to hold
  trivially.
  \item \label{enum:il-inv-t-img} $T \cap \im(m) = \varnothing$

  As above.
\end{enumerate}

\subsection{Correctness}

At the end of the algorithm, everything that was ever in the match surface
except for bare wires has been added to $m$ by the string graph matching step
(postcondition \ref{enum:il-postcond-all-matched}, together with the fact that
we always add everything unmatched in the match surface to $U_W$, $U_C$ and
$U_N$, except bare wires).  We have no more $!$-boxes, so the entire graph is
part of the match surface.  $P$ is also empty.  We matched bare wires at the
last stage, so $m$ is in fact total on $L$.  So, by the argument at the start
of section \ref{apdx:sg-matching}, $m$ is a valid match from the final state
of $L$ to $G$.  Since we have only ever performed valid $!$-box operations on
$L$, this is a valid instance of $L$, and hence we have found a valid matching
from $L$ to $G$.

\subsection{Completeness}

Providing no instance of $L$ contains a wild $!$-box (see section
\ref{sec:enum-inst}), completeness follows from the fact that any concrete
instantiation has an equivalent in expansion-normal form (section
\ref{sec:expn-normal-form}).


\chapter{The Spider Law}
\label{apdx:spider-law}

The following is joint work with Aleks Kissinger.  It is work in progress,
as we are still working out the details of the definition style used in this
section and intend to extend this to non-commutative generators, as mentioned
in section \ref{sec:further-work}.

Suppose we have a commutative Frobenius algebra in string graph form.  So we
have the following generators (node-vertices of the typegraph, with their
adjacent edges)
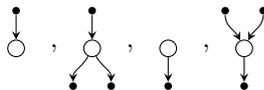
\[ \input{cfa-graph-gens.tikz} \]
where all the edges are fixed-arity, and all the wire-vertices are of the same
type.  Note that, in the sequel, we will implicitly identify distinguished
fixed-arity edges by the position in our presentation of the graph.  This is
unambiguous with these particular generators, as we can simply count clockwise
from the single output of the multiplication node, or the single input of the
comultiplication node.

We also have the following equations, where we identify inputs and outputs by
position:
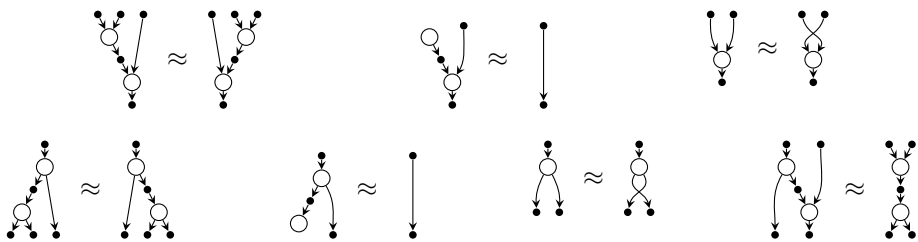
\begin{mathpar}
    \input{cfa-assoc-law.tikz}
    \and
    \input{cfa-left-unit-law.tikz}
    \and
    \input{cfa-com-law.tikz}
    \\
    \input{cfa-coassoc-law.tikz}
    \and
    \input{cfa-left-counit-law.tikz}
    \and
    \input{cfa-cocom-law.tikz}
    \and
    \input{cfa-frob-law.tikz}
\end{mathpar}
It is trivial to use the commutativity laws to derive
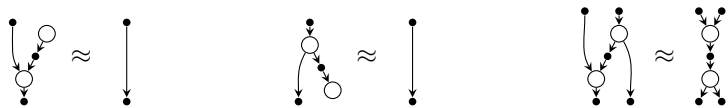
\begin{mathpar}
    \input{cfa-right-unit-law.tikz}
    \and
    \input{cfa-right-counit-law.tikz}
    \and
    \input{cfa-frob-law2.tikz}
\end{mathpar}

Suppose we now define a $!$-boxed version of the multiplication operation in
the following manner:
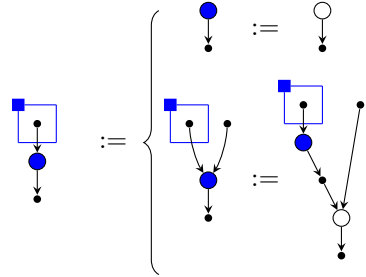
\[ \input{sp-def-mult.tikz} \]
Essentially, we have added a new generator with a variable-arity input and a
fixed-arity output, together with rules that allow us to rewrite any
instance of it (with any number of incoming edges) into a graph containing
only our original set of generators.  As noted in section
\ref{sec:further-work}, this definition assumes associativity and
commutativity in order to be consistent.

We now show that a partial form of the spider law holds for variable-arity
multiplications.

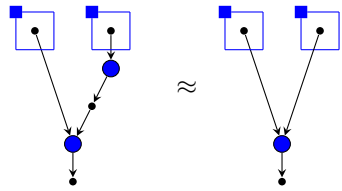
\begin{lemma}
    \[ \input{sp-lem-mult-merge.tikz} \]
\end{lemma}
\begin{proof}
    We proceed by $!$-box induction.  The base case:
    \[ \input{sp-lem-mult-merge-pf-base.tikz} \]
    and the step case:
    \[ \input{sp-lem-mult-merge-pf-step.tikz} \]
\end{proof}

We define a variable-arity comultiplication in a similar manner:
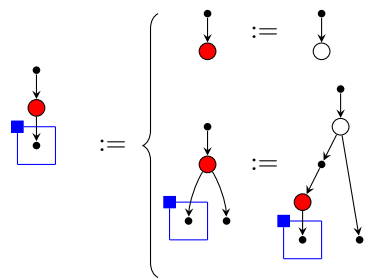
\[ \input{sp-def-comult.tikz} \]
and prove a similar result about it:

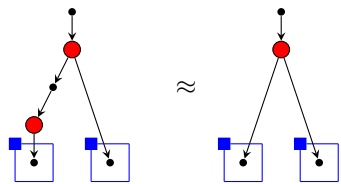
\begin{lemma}
    \[ \input{sp-lem-comult-merge.tikz} \]
\end{lemma}
\begin{proof}
    We proceed by $!$-box induction.  The base case:
    \[ \input{sp-lem-comult-merge-pf-base.tikz} \]
    and the step case:
    \[ \input{sp-lem-comult-merge-pf-step.tikz} \]
\end{proof}

Now we are ready to define our spider:
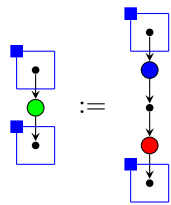
\[ \input{sp-def-spider.tikz} \]

And the spider law:
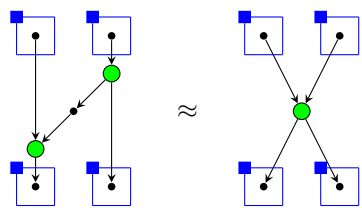
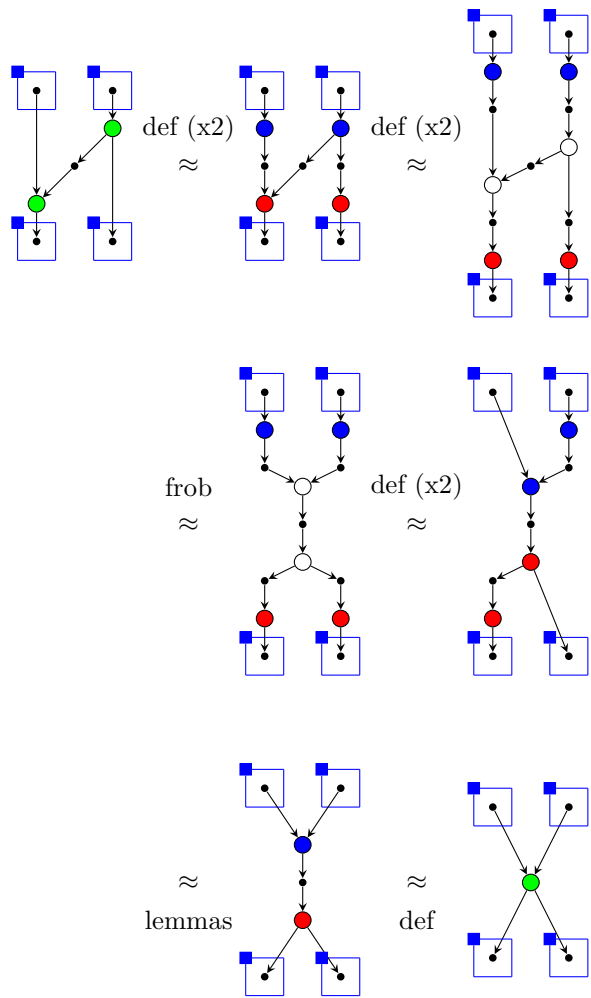
\begin{theorem}
    \[ \input{sp-thm-spider.tikz} \]
\end{theorem}
\begin{proof}
    \[ \input{sp-thm-spider-pf.tikz} \]
\end{proof}

Of course, we need to ensure that our defined spider coincides with the
original Frobenius algebra operations:
\begin{theorem}
    \[ \input{sp-thm-lift.tikz} \]
\end{theorem}
\begin{proof}
    \[ \input{sp-thm-lift-pf-mult.tikz} \]
    \[ \input{sp-thm-lift-pf-unit.tikz} \]
    \[ \input{sp-thm-lift-pf-comult.tikz} \]
    \[ \input{sp-thm-lift-pf-counit.tikz} \]
\end{proof}

\addcontentsline{toc}{chapter}{Bibliography}
\bibliographystyle{akbib}
\bibliography{bibfile}

\end{document}

%% file: dA.tikz
\begin{tikzpicture}
	\begin{pgfonlayer}{nodelayer}
		\node [style=none] (0) at (0.25, 0.25) {};
		\node [style=none] (1) at (-0.25, 0.25) {};
	\end{pgfonlayer}
	\begin{pgfonlayer}{edgelayer}
		\draw [style=diredge, bend right=90, looseness=1.75] (1.center) to (0.center);
	\end{pgfonlayer}
\end{tikzpicture}

%% file: eA.tikz
\begin{tikzpicture}
	\begin{pgfonlayer}{nodelayer}
		\node [style=none] (0) at (0.25, -0) {};
		\node [style=none] (1) at (-0.25, -0) {};
	\end{pgfonlayer}
	\begin{pgfonlayer}{edgelayer}
		\draw [style=diredge, bend right=90, looseness=1.75] (0.center) to (1.center);
	\end{pgfonlayer}
\end{tikzpicture}

%% file: d_A.tikz
\begin{tikzpicture}
	\begin{pgfonlayer}{nodelayer}
		\node [style=sg wire vertex] (0) at (-0.5, 0.5) {};
		\node [style=sg wire vertex] (1) at (0.25, 0.5) {};
	\end{pgfonlayer}
	\begin{pgfonlayer}{edgelayer}
		\draw [style=sg diredge, bend right=90, looseness=2.00] (0) to (1);
	\end{pgfonlayer}
\end{tikzpicture}

%% file: 1_A.tikz
\begin{tikzpicture}
	\begin{pgfonlayer}{nodelayer}
		\node [style=sg wire vertex] (0) at (-0.5, 0.5) {};
		\node [style=sg wire vertex] (1) at (-0.5, -0.25) {};
	\end{pgfonlayer}
	\begin{pgfonlayer}{edgelayer}
		\draw [style=sg diredge] (0) to (1);
	\end{pgfonlayer}
\end{tikzpicture}

%% file: e_A.tikz
\begin{tikzpicture}
	\begin{pgfonlayer}{nodelayer}
		\node [style=sg wire vertex] (0) at (-0.5, -0.25) {};
		\node [style=sg wire vertex] (1) at (0.25, -0.25) {};
	\end{pgfonlayer}
	\begin{pgfonlayer}{edgelayer}
		\draw [style=sg diredge, bend left=90, looseness=2.00] (0) to (1);
	\end{pgfonlayer}
\end{tikzpicture}

%% file: elem-g-wire.tikz
\begin{tikzpicture}[string graph]
	\begin{pgfonlayer}{nodelayer}
		\node [style=sg wire vertex] (0) at (0, 0.5) {};
		\node [style=sg wire vertex] (1) at (0, -0.25) {};
	\end{pgfonlayer}
	\begin{pgfonlayer}{edgelayer}
		\draw [style=sg diredge] (0) to (1);
	\end{pgfonlayer}
\end{tikzpicture}

%% file: elem-g-circle.tikz
\begin{tikzpicture}[string graph]
	\begin{pgfonlayer}{nodelayer}
		\node [style=sg wire vertex] (0) at (0, 0) {};
	\end{pgfonlayer}
	\begin{pgfonlayer}{edgelayer}
		\draw [style=sg diredge, in=135, out=45, loop] (0) to ();
	\end{pgfonlayer}
\end{tikzpicture}

%% file: elem-g-wire-vertex.tikz
\begin{tikzpicture}[string graph]
	\begin{pgfonlayer}{nodelayer}
		\node [style=sg wire vertex] (0) at (0, 0) {};
	\end{pgfonlayer}
\end{tikzpicture}